\pdfoutput=1

\documentclass[%
showpacs,%
showkeys,%
superscriptaddress,%
nofootinbib,%
longbibliography,%
aps,%
pra,%
twocolumn%
]{revtex4-2}

\usepackage[T1]{fontenc}
\usepackage[utf8]{inputenc}
\usepackage[english]{babel}

\usepackage{lipsum}
\usepackage{comment}
\usepackage{centernot}

\usepackage{graphicx}
\usepackage{epsfig}
\usepackage[all,cmtip,matrix,frame,arrow]{xy}
\usepackage{tabularx}
\usepackage{booktabs}
\usepackage[usenames,dvipsnames,svgnames,table,xcdraw]{xcolor}
\usepackage{soul}
\usepackage{bigstrut}

\usepackage{amsfonts}
\usepackage{amsthm}
\usepackage{amsmath}
\usepackage{amssymb}
\usepackage{mathrsfs}
\usepackage{faktor}

\usepackage{soul}

\usepackage{etoolbox}

\usepackage{enumerate}
\makeatletter
\newcommand\setItemnumber[1]{\setcounter{enum\romannumeral\@enumdepth}{\numexpr#1-1\relax}}
\makeatother

\usepackage{tikz}
\input{tikzit.sty}

\tikzstyle{medium rectangle}=[fill=white, draw=black, shape=rectangle, minimum width=0.75 cm, minimum height=1 cm]
\tikzstyle{red}=[fill=red, draw=black, shape=circle]
\tikzstyle{system_label}=[fill=none, draw=none, shape=circle]
\tikzstyle{small square}=[fill=white, draw=black, shape=rectangle]
\tikzstyle{big rectangle}=[fill=white, draw=black, shape=rectangle, minimum width=1.2 cm, minimum height=1.5 cm]
\tikzstyle{tall rectangle}=[fill=white, draw=black, shape=rectangle, minimum width=1.2 cm, minimum height=2.1 cm]
\tikzstyle{medium square}=[fill=white, draw=black, shape=rectangle, minimum width=0.6 cm, minimum height=0.6 cm]
\tikzstyle{huge rectangle}=[fill=white, draw=black, shape=rectangle, minimum width=2 cm, minimum height=4.5 cm]
\tikzstyle{yellow small}=[fill={rgb,255: red,255; green,252; blue,144}, draw=black, shape=rectangle]
\tikzstyle{violet node}=[fill={rgb,255: red,195; green,187; blue,255}, draw=black, shape=rectangle, minimum width=1.2 cm, minimum height=1.5 cm]
\tikzstyle{pink node}=[fill={rgb,255: red,238; green,174; blue,255}, draw=black, shape=rectangle, minimum width=1.2 cm, minimum height=1.5 cm]
\tikzstyle{pompelmo node}=[fill={rgb,255: red,255; green,187; blue,166}, draw=black, shape=rectangle, minimum width=1.2 cm, minimum height=1.5 cm]
\tikzstyle{ottano node}=[fill={rgb,255: red,131; green,201; blue,187}, draw=black, shape=rectangle]
\tikzstyle{blue}=[fill={rgb,255: red,214; green,201; blue,255}, draw=black, shape=rectangle]
\tikzstyle{pure}=[fill=blue, draw=blue, shape=circle, minimum size=8pt, inner sep=0pt, outer sep=0pt]
\tikzstyle{extremal}=[fill={rgb,255: red,0; green,233; blue,0}, draw={rgb,255: red,0; green,233; blue,0}, shape=circle, minimum size=8pt, inner sep=0pt, outer sep=0pt]
\tikzstyle{atomic - extremal}=[fill={rgb,255: red,0; green,167; blue,0}, draw={rgb,255: red,0; green,167; blue,0}, shape=circle, minimum size=8pt, inner sep=0pt, outer sep=0pt]

\tikzstyle{lambda}=[-, draw={rgb,255: red,191; green,191; blue,191}]
\tikzstyle{state}=[<-]
\tikzstyle{bluette}=[-, fill={rgb,255: red,214; green,201; blue,255}, draw={rgb,255: red,192; green,181; blue,229}]
\tikzstyle{greenish}=[-, fill={rgb,255: red,160; green,217; blue,255}, draw={rgb,255: red,137; green,188; blue,219}]
\tikzstyle{white}=[-, fill=white]
\tikzstyle{reddish}=[-, fill={rgb,255: red,255; green,143; blue,145}, dashed]
\tikzstyle{yellowish}=[-, fill={rgb,255: red,255; green,252; blue,144}, dashed, line width=1 pt]
\tikzstyle{red non-dashed}=[-, fill={rgb,255: red,255; green,143; blue,145}]
\tikzstyle{violet}=[-, dashed, fill={rgb,255: red,195; green,187; blue,255}]
\tikzstyle{green non-dashed}=[-, fill={rgb,255: red,160; green,217; blue,255}]
\tikzstyle{pompelmo non-dashed}=[-, fill={rgb,255: red,255; green,187; blue,166}]
\tikzstyle{pink}=[-, fill={rgb,255: red,238; green,174; blue,255}, dashed]
\tikzstyle{blue non-dashed}=[-, fill={rgb,255: red,178; green,255; blue,246}]
\tikzstyle{yell non-dashed}=[-, fill={rgb,255: red,255; green,248; blue,137}]
\tikzstyle{ottano}=[-, fill={rgb,255: red,164; green,184; blue,255}, draw={rgb,255: red,147; green,168; blue,229}]
\tikzstyle{pompelmo}=[-, fill={rgb,255: red,255; green,187; blue,166}, draw={rgb,255: red,213; green,156; blue,139}]
\tikzstyle{dark blue}=[-, fill={rgb,255: red,131; green,201; blue,187}, draw={rgb,255: red,119; green,183; blue,170}]
\tikzstyle{d.blue non-dashed}=[-, fill={rgb,255: red,237; green,148; blue,112}]
\tikzstyle{dashed edge}=[-, dashed]
\tikzstyle{ottano non-drawn}=[-, fill={rgb,255: red,164; green,184; blue,255}]
\tikzstyle{atomic}=[-, fill={rgb,255: red,232; green,232; blue,232}, draw={rgb,255: red,255; green,128; blue,0}, line width=1.5pt]
\tikzstyle{atomic - dashed}=[-, dashed, fill={rgb,255: red,232; green,232; blue,232}, draw={rgb,255: red,255; green,128; blue,0}, line width=0.2pt]
\tikzstyle{det}=[-, fill={rgb,255: red,232; green,232; blue,232}, draw={rgb,255: red,227; green,0; blue,3}, line width=1.5pt]
\tikzstyle{det - dashed}=[-, dashed, fill={rgb,255: red,232; green,232; blue,232}, draw={rgb,255: red,227; green,0; blue,3}, line width=0.2pt]
\tikzstyle{standard}=[-, fill={rgb,255: red,232; green,232; blue,232}, line width=1.5pt]
\tikzstyle{empty}=[-, fill={rgb,255: red,232; green,232; blue,232}, draw={rgb,255: red,232; green,232; blue,232}, line width=1.5pt]
\tikzstyle{empty- dashed}=[-, dashed, fill={rgb,255: red,232; green,232; blue,232}, draw={rgb,255: red,232; green,232; blue,232}, line width=1.5pt]

\makeatletter
\newcommand*{\hyperlinkcite}[1]{\hyper@link{cite}{cite.#1}}
\makeatother

\usepackage[]{hyperref}
\hypersetup{colorlinks, citecolor=blue, urlcolor=blue, linkcolor=blue,linktocpage,breaklinks}
\usepackage[nameinlink,noabbrev]{cleveref}

\makeatletter
\pretocmd{\NAT@citexnum}{\@ifnum{\NAT@ctype>\z@}{\let\NAT@hyper@\relax}{}}{}{}
\makeatother

\newcommand{\qw}[1][-1]{\ar @{-} [0,#1]}

\newcommand{\gate}[1]{*{\xy *+<.6em>{#1};p\save+LU;+RU **\dir{-}\restore\save+RU;+RD **\dir{-}\restore\save+RD;+LD **\dir{-}\restore\POS+LD;+LU **\dir{-}\endxy} \qw}

\newcommand{\measureD}[1]{*{\xy*+=+<.5em>{\vphantom{\rule{0em}{.1em}#1}}*\cir{r_l};p\save*!R{#1} \restore\save+UC;+UC-<.5em,0em>*!R{\hphantom{#1}}+L **\dir{-} \restore\save+DC;+DC-<.5em,0em>*!R{\hphantom{#1}}+L **\dir{-} \restore\POS+UC-<.5em,0em>*!R{\hphantom{#1}}+L;+DC-<.5em,0em>*!R{\hphantom{#1}}+L **\dir{-} \endxy} \qw}

\newcommand{\multimeasureD}[2]{*+<1em,.9em>{\hphantom{#2}}\save[0,0].[#1,0];p\save !C *{#2},p+LU+<0em,0em>;+RU+<-.8em,0em> **\dir{-}\restore\save +LD;+LU **\dir{-}\restore\save +LD;+RD-<.8em,0em> **\dir{-} \restore\save +RD+<0em,.8em>;+RU-<0em,.8em> **\dir{-} \restore \POS !UR*!UR{\cir<.9em>{r_d}};!DR*!DR{\cir<.9em>{d_l}}\restore \qw}

\newcommand{\multigate}[2]{*+<1em,.9em>{\hphantom{#2}} \qw \POS[0,0].[#1,0];p !C *{#2},p \save+LU;+RU **\dir{-}\restore\save+RU;+RD **\dir{-}\restore\save+RD;+LD **\dir{-}\restore\save+LD;+LU **\dir{-}\restore}
\newcommand{\ghost}[1]{*+<1em,.9em>{\hphantom{#1}} \qw}
\newcommand{\ustick}[1]{*!D!<0em,-.5em>=<0em>{#1}}

\newcommand{\Qcircuit}[1][0em]{\xymatrix @*=<#1>}
\newcommand{\node}[2][]{{\begin{array}{c} \ _{#1}\  \\ {#2} \\ \ \end{array}}\drop\frm{o} }

\newcommand{\pureghost}[1]{*+<1em,.9em>{\hphantom{#1}}}
\newcommand{\multiprepareC}[2]{*+<1em,.9em>{\hphantom{#2}}\save[0,0].[#1,0];p\save !C
  *{#2},p+RU+<0em,0em>;+LU+<+.8em,0em> **\dir{-}\restore\save +RD;+RU **\dir{-}\restore\save
  +RD;+LD+<.8em,0em> **\dir{-} \restore\save +LD+<0em,.8em>;+LU-<0em,.8em> **\dir{-} \restore \POS
  !UL*!UL{\cir<.9em>{u_r}};!DL*!DL{\cir<.9em>{l_u}}\restore}
\newcommand{\prepareC}[1]{*{\xy*+=+<.5em>{\vphantom{#1\rule{0em}{.1em}}}*\cir{l^r};p\save*!L{#1} \restore\save+UC;+UC+<.5em,0em>*!L{\hphantom{#1}}+R **\dir{-} \restore\save+DC;+DC+<.5em,0em>*!L{\hphantom{#1}}+R **\dir{-} \restore\POS+UC+<.5em,0em>*!L{\hphantom{#1}}+R;+DC+<.5em,0em>*!L{\hphantom{#1}}+R **\dir{-} \endxy}}

\newcommand{\braidingGhost}{\ghost{0em}}

\newcommand{\braiding}{*+<1em,.9em>{\hphantom{0em}} \qw; \POS[0,0]+L+<0.98em,-0.439em>;\POS[1,0]+R+<-0.98em,0.439em> **\dir{-}; \POS[0,0]+R+<-0.98em,-0.439em>;\POS[1,0]+L+<0.98em,0.439em> **{}?<(0.25)**@{-}; \POS[1,0]+L+<0.98em,0.439em>; \POS[0,0]+R+<-0.98em,-0.439em> **{}?<(0.25)**@{-}; \POS[0,0]*!RU{\cir<1.5em>{r_dr}}; \POS[0,0]*!LU{\cir<1.5em>{ur_r}};  \POS[1,0]*!RD{\cir<1.5em>{dl_l}}; \POS[1,0]*!LD{\cir<1.5em>{l_ul}}}

\newcommand{\braidingInv}{*+<1em,.9em>{\hphantom{0em}} \qw; \POS[0,0]+L+<0.98em,-0.439em>;\POS[1,0]+R+<-0.98em,0.439em> **{}?<(0.25)**@{-}; \POS[1,0]+R+<-0.98em,0.439em>;\POS[0,0]+L+<0.98em,-0.439em> **{}?<(0.25)**@{-}; \POS[0,0]+R+<-0.98em,-0.439em>;\POS[1,0]+L+<0.98em,0.439em> **\dir{-}; \POS[0,0]*!RU{\cir<1.5em>{r_dr}}; \POS[0,0]*!LU{\cir<1.5em>{ur_r}};  \POS[1,0]*!RD{\cir<1.5em>{dl_l}}; \POS[1,0]*!LD{\cir<1.5em>{l_ul}}}

\newcommand{\braidingInvId}{*+<1em,.9em>{\hphantom{0em}}; \POS[0,0]+R+<-0.98em,-0.439em>;\POS[1,0]+L+<0.98em,0.439em> **\dir{-}; \POS[0,0]*!LU{\cir<1.5em>{ur_r}};  \POS[1,0]*!RD{\cir<1.5em>{dl_l}}}

\newcommand{\braidingSym}{*+<1em,.9em>{\hphantom{0em}} \qw; \POS[0,0]+L+<0.98em,-0.439em>;\POS[1,0]+R+<-0.98em,0.439em> **\dir{-}; \POS[0,0]+R+<-0.98em,-0.439em>;\POS[1,0]+L+<0.98em,0.439em> **\dir{-}; \POS[0,0]*!RU{\cir<1.5em>{r_dr}}; \POS[0,0]*!LU{\cir<1.5em>{ur_r}};  \POS[1,0]*!RD{\cir<1.5em>{dl_l}}; \POS[1,0]*!LD{\cir<1.5em>{l_ul}}}

\newcommand{\gategroupColor}[7]{\POS"#1,#2"."#3,#2"."#1,#4"."#3,#4"!C*+<#5>[#7]\frm{#6}}

\newcommand{\splitterGhost}{\pureghost{0em}}

\newcommand{\splitter}{*+<1em,.9em>{\hphantom{0em}} \qw; \POS[0,0]+L+<0.98em,-0.439em>;\POS[1,0]+R+<-0.98em,0.439em> **\dir{-}; \POS[0,0]+L;\POS[0,0]+R **\dir{-}; \POS[0,0]*!RU{\cir<1.5em>{r_dr}}; \POS[1,0]*!LD{\cir<1.5em>{l_ul}}}

\newcommand{\myQcircuitSmall}[1]{\begin{aligned} \Qcircuit @C=0.6em @R=0.8em {#1} \end{aligned} }
\newcommand{\myQcircuit}[1]{\begin{aligned} \Qcircuit @C=0.8em @R=0.8em {#1} \end{aligned} }
\newcommand{\myQcircuitMedium}[1]{\begin{aligned} \Qcircuit @C=1em @R=1em {#1} \end{aligned} }
\newcommand{\myQcircuitBox}[1]{\begin{aligned} \Qcircuit @C=0.8em @R=1.5em {#1} \end{aligned} }
\newcommand{\myQcircuitSupMat}[1]{\begin{aligned} \Qcircuit @C=0.8em @R=1em {#1} \end{aligned} }
\newcommand{\myQcircuitComp}[1]{\begin{aligned} \Qcircuit @C=0.8em @R=1.5em {#1} \end{aligned} }

\theoremstyle{definition}
\newtheorem*{definition*}{Definition}
\newtheorem{definition}{Definition}

\theoremstyle{plain}
\newtheorem{procedure}{Procedure}

\newcommand{\textdef}[1]{\textit{{#1}}}

\theoremstyle{plain}
\newtheorem{property}{Property}
\newtheorem*{property*}{Property}

\newtheorem{postulate}{Postulate}

\newtheorem{theorem}{Theorem}
\newtheorem*{theorem*}{Theorem}
\newtheorem{corollary}{Corollary}
\newtheorem*{corollary*}{Corollary}

\newtheorem*{proposition*}{Proposition}

\newtheorem*{conjecture*}{Conjecture}
\newtheorem*{question*}{Question}

\newtheorem*{problem*}{Problem}

\newtheorem*{lemma*}{Lemma}
\newtheorem{lemma}{Lemma}

\newtheorem*{example*}{Example}

\newtheorem*{remark*}{Remark}

\newtheorem{remark}{Remark}
\newtheorem*{proof*}{Proof}
\newcommand{\proofRight}{\ensuremath{\mathbf{\Longrightarrow ) \; }}} 
\newcommand{\proofLeft}{\ensuremath{\mathbf{\Longleftarrow ) \; }}}

\makeatletter
\newtheorem*{rep@theorem}{\rep@title}
\newcommand{\newreptheorem}[2]{%
	\newenvironment{rep#1}[1]{%
		\def\rep@title{\autoref{##1}}%
		\begin{rep@theorem}}%
		{\end{rep@theorem}}}
\makeatother
\newreptheorem{theorem}{Theorem}

\DeclareSymbolFont{sfletters}{OML}{cmbrm}{m}{it}
\DeclareMathSymbol{\salpha}{\mathord}{sfletters}{"0B}
\DeclareMathSymbol{\sbeta}{\mathord}{sfletters}{"0C}
\DeclareMathSymbol{\sgamma}{\mathord}{sfletters}{"0D}
\DeclareMathSymbol{\sdelta}{\mathord}{sfletters}{"0E}
\DeclareMathSymbol{\sepsilon}{\mathord}{sfletters}{"0F}
\DeclareMathSymbol{\szeta}{\mathord}{sfletters}{"10}
\DeclareMathSymbol{\seta}{\mathord}{sfletters}{"11}
\DeclareMathSymbol{\stheta}{\mathord}{sfletters}{"12}
\DeclareMathSymbol{\siota}{\mathord}{sfletters}{"13}
\DeclareMathSymbol{\skappa}{\mathord}{sfletters}{"14}
\DeclareMathSymbol{\slambda}{\mathord}{sfletters}{"15}
\DeclareMathSymbol{\smu}{\mathord}{sfletters}{"16}
\DeclareMathSymbol{\snu}{\mathord}{sfletters}{"17}
\DeclareMathSymbol{\sxi}{\mathord}{sfletters}{"18}
\DeclareMathSymbol{\spi}{\mathord}{sfletters}{"19}
\DeclareMathSymbol{\srho}{\mathord}{sfletters}{"1A}
\DeclareMathSymbol{\ssigma}{\mathord}{sfletters}{"1B}
\DeclareMathSymbol{\stau}{\mathord}{sfletters}{"1C}
\DeclareMathSymbol{\supsilon}{\mathord}{sfletters}{"1D}
\DeclareMathSymbol{\sphi}{\mathord}{sfletters}{"1E}
\DeclareMathSymbol{\schi}{\mathord}{sfletters}{"1F}
\DeclareMathSymbol{\spsi}{\mathord}{sfletters}{"20}
\DeclareMathSymbol{\somega}{\mathord}{sfletters}{"21}
\DeclareMathSymbol{\svarepsilon}{\mathord}{sfletters}{"22}
\DeclareMathSymbol{\svartheta}{\mathord}{sfletters}{"23}
\DeclareMathSymbol{\svarpi}{\mathord}{sfletters}{"24}
\DeclareMathSymbol{\svarrho}{\mathord}{sfletters}{"25}
\DeclareMathSymbol{\svarsigma}{\mathord}{sfletters}{"26}
\DeclareMathSymbol{\svarphi}{\mathord}{sfletters}{"27}

\def\<{\langle}\def\>{\rangle}

\newcommand{\mathDef}{:=}

\newcommand{\prim}[1]{{#1^{\prime}}}
\newcommand{\secondE}[1]{#1^{\prime \prime}}

\newcommand*{\cardinality}[1]{\ensuremath{\vert #1 \vert}}

\newcommand*{\algebraicDual}[1]{\ensuremath{#1^{\vee}}}

\newcommand{\cartesianC}{\times}
\newcommand{\cartesianProduct}[2]{\ensuremath{#1 \cartesianC #2}}

\newcommand{\powerSetOrder}[1]{\mathscr{P}_{\text{ord}}\left\{ #1 \right\}}

\newcommand{\inverse}[1]{{#1}^{-1}}

\newcommand{\normGeneric}[2]{\ensuremath{\left\lVert {#1} \right\rVert_{#2}}}

\newcommand{\normOp}[1]{\normGeneric{#1}{op}}
\newcommand{\normSup}[1]{\normGeneric{#1}{sup}}
\newcommand{\normGen}[1]{\normGeneric{#1}{gen}}

\makeatletter

\newcommand*{\OPTMath}{\ensuremath{\Theta}}
\newcommand*{\theory}{\ensuremath{\Theta}}

\newcommand{\rbra}[1]{\left({#1}\right\vert}
\newcommand{\rket}[1]{\left\vert{#1}\right)}
\newcommand{\rbraket}[2]{\left( #1 \vphantom{#2} \right|\left. #2 \vphantom{#1} \right)}

\newcommand{\rbraSystem}[2]{{\left({#1}\right\vert}_{\system{#2}}}
\newcommand{\rketSystem}[2]{{\left\vert{#1}\right)}_{\system{#2}}}
\newcommand{\rbraketSystem}[3]{{\left( #1 \vphantom{#2} \right|\left. #2 \vphantom{#1} \right)}_{\system{#3}}}

\newcommand{\uniDetEff}{e}

\newcommand{\system}[1]{\ensuremath{\mathrm{#1}}}
\newcommand{\systemConditioned}[2]{\ensuremath{\mathrm{#1}^{\left(\outcome{#2}\right)}}}
\newcommand{\systemSequence}[2]{\ensuremath{\mathrm{#1}^{#2}}}
\newcommand{\systemIndexDown}[2]{\ensuremath{\mathrm{#1}_{#2}}}
\newcommand{\trivialSystem}{\ensuremath{\system{I}}}
\newcommand{\s}[1]{\ustick{\scriptstyle{\system{#1}}}}
\newcommand{\sSequence}[2]{\ustick{\scriptstyle{\systemSequence{#1}{#2}}}}
\newcommand{\sIndexDown}[2]{\ustick{\scriptstyle{\systemIndexDown{#1}{#2}}}}
\newcommand{\sConditioned}[2]{\ustick{\scriptstyle{\systemConditioned{#1}{#2}}}}
\newcommand{\sSequenceIndex}[3]{\ustick{\scriptstyle{\ensuremath{\mathrm{#1}_{#2}^{#3}}}}}
\newcommand{\sSequenceConditioned}[3]{\ustick{\scriptstyle{ {\systemConditioned{#1}{#2}}}^{#3} }}
\newcommand{\sSequencePrime}[2]{\ustick{\scriptstyle{\systemSequence{#1 '}{#2}}}}
\newcommand{\sEnsemble}[3]{\ustick{\scriptstyle{ \left\{ \system{#1}_{#2} \right\}_{#2 \in #3} } }}
\newcommand{\sEnsembleDouble}[5]{\ustick{\scriptstyle{ \left\{ \system{#1}_{#2} \right\}_{#2 \in #3} \left\{ \system{#1}_{#4} \right\}_{#4 \in #5}} }}

\newcommand*{\Sys}[1]{\ensuremath{\mathsf{Sys\left(\mathrm{#1}\right)}}}
\newcommand*{\@sysDimensionD}{\mathsf{D}}
\newcommand*{\sysDimension}[1]{\ensuremath{\@sysDimensionD_{\system{#1}}}}

\newcommand*{\@sysEquivC}{\cong}
\newcommand*{\sysNotEquiv}[2]{ \ensuremath{ \system{#1} \centernot{\@sysEquivC} \system{#2}} }
\newcommand*{\sysEquiv}[2]{ \ensuremath{ \system{#1} \@sysEquivC \system{#2}} }

\makeatletter
\newcommand*\bigcdot{\mathpalette\bigcdot@{.5}}
\newcommand*\bigcdot@[2]{\mathbin{\vcenter{\hbox{\scalebox{#2}{$\m@th#1\bullet$}}}}}
\makeatother

\newcommand*{\outcomeSpace}[1]{\ensuremath{\mathsf{#1}}}
\newcommand*{\outcome}[1]{\ensuremath{#1}}
\newcommand*{\outcomeDouble}[2]{\left(\outcome{#1}, \outcome{#2} \right)}
\newcommand*{\outcomeIncluded}[2]{\ensuremath{ \outcome{#1} \in \outcomeSpace{#2} }}
\newcommand*{\outcomeSpaceDouble}[2]{\ensuremath{\cartesianProduct{\outcomeSpace{#1}}{\outcomeSpace{#2}}}}
\newcommand*{\outcomeIncludedDouble}[4]{\ensuremath{ \outcomeDouble{#1}{#2} \in \outcomeSpaceDouble{#3}{#4} }}
\newcommand*{\outcomeSpaceConditioned}[2]{\ensuremath{\outcomeSpace{#1}^{\left( \outcome{#2} \right)}}}
\newcommand*{\outcomeIncludedConditioned}[3]{\ensuremath{\outcome{#1} \in \outcomeSpaceConditioned{#2}{#3}}}
\newcommand*{\outcomeSpaceSequence}[2]{\ensuremath{\outcomeSpace{#1}^{#2}}}
\newcommand*{\outcomeIncludedSequence}[3]{\ensuremath{\outcome{#1} \in \outcomeSpaceSequence{#2}{#3}}}

\newcommand{\SingletonSet}{\star}

\makeatletter
\newcommand*{\@identityI}{\eventNoDown{I}}
\newcommand{\identityTest}[1]{\ensuremath{\@identityI_{\system{#1}}}}
\makeatother

\newcommand{\probEquivC}{\sim}
\newcommand{\probabilisticallyEquivalent}[2]{\ensuremath{#1 \probEquivC #2}}

\newcommand{\testNoDown}[1]{\ensuremath{\mathsf{#1}}}
\newcommand{\testComplete}[4]{\ensuremath{\testNoDown{#1}^{\system{#3} \!\to\! \system{#4}}_{\outcomeSpace{#2}}}}
\newcommand{\test}[2]{\ensuremath{\testNoDown{#1}_{\outcomeSpace{#2}}}}
\newcommand{\conditionedTest}[3]{\ensuremath{\testNoDown{#1}^{\left(\outcome{#3}\right)}_{\outcomeSpace{#2}}}}

\newcommand{\preparationTestNoDown}[1]{{\ensuremath{#1}}}
\newcommand{\preparationTest}[2]{\ensuremath{\preparationTestNoDown{#1}_{\outcomeSpace{#2}}}}
\newcommand{\preparationTestComplete}[3]{\ensuremath{\preparationTestNoDown{#1}^{\system{#3} \!\to\! \trivialSystem}_{\outcomeSpace{#2}} } }
\newcommand{\observationTestNoDown}[1]{{\ensuremath{\text{#1}}}}
\newcommand{\observationTest}[2]{\ensuremath{\observationTestNoDown{#1}_{\outcomeSpace{#2}}}}
\newcommand{\observationTestComplete}[3]{\ensuremath{\observationTestNoDown{#1}^{\system{#3} \!\to\! \trivialSystem}_{\outcomeSpace{#2}} } }
\newcommand{\testCollection}[1]{\ensuremath{\mathsf{Test\left(\mathrm{#1}\right)}}}

\newcommand{\testCollectionAB}[2]{\ensuremath{\testCollection{\system{#1} \!\to\! \system{#2}}}}

\newcommand{\testList}[1]{\ensuremath{ \left\{ #1 \right\} }}

\newcommand{\eventNoDown}[1]{\ensuremath{ \mathscr{#1} }}
\newcommand{\event}[2]{\ensuremath{\eventNoDown{#1}_{\outcome{#2}}}}
\newcommand{\eventCG}[2]{\ensuremath{\eventNoDown{#1}_{\outcomeSpace{#2}}}}
\newcommand{\preparationEventNoDown}[1]{\ensuremath{#1}}
\newcommand{\preparationEvent}[2]{\ensuremath{\preparationEventNoDown{#1}_{\outcome{#2}}}}
\newcommand{\observationEventNoDown}[1]{\ensuremath{\text{#1}}}
\newcommand{\observationEvent}[2]{\ensuremath{\observationEventNoDown{#1}_{\outcome{#2}}}}
\newcommand{\observationUniqueDeterministic}{\observationEventNoDown{\uniDetEff}}

\newcommand{\probabilityEventNoDown}[1]{\ensuremath{ {#1} }}

\newcommand{\probabilityEvent}[2]{\ensuremath{ \probabilityEventNoDown{#1}_{\outcome{#2}} }}

\newcommand{\conditionedEvent}[3]{\ensuremath{\eventNoDown{#1}^{\left(\outcome{#3}\right)}_{\outcome{#2}}}}

\newcommand{\probabilityEventTest}[3]{\ensuremath{\left\{ \probabilityEventNoDown{#1}_{\outcome{#2}} \right\}_{\outcomeIncluded{#2}{#3}}}}
\newcommand{\conditionedEventTest}[4]{\left\{\eventNoDown{#1}^{\left(\outcome{#4}\right)}_{\outcome{#2}} \right\}_{\outcomeIncluded{#2}{#3}}}

\newcommand{\preparationEventTest}[3]{\ensuremath{\left\{\preparationEventNoDown{#1}_{\outcome{#2}} \right\}_{\outcomeIncluded{#2}{#3}}}}

\newcommand{\observationEventTest}[3]{\ensuremath{\left\{\observationEventNoDown{#1}_{\outcome{#2}} \right\}_{\outcomeIncluded{#2}{#3}}}}
\newcommand{\eventCollection}[1]{\ensuremath{\mathsf{Event\left(\mathrm{#1}\right)}}}

\newcommand{\eventCollectionAB}[2]{\ensuremath{\eventCollection{\system{#1} \!\to\! \system{#2}}}}
\newcommand{\eventCollectionAA}[1]{\eventCollectionAB{#1}{#1}}

\newcommand{\nullTransformationSymbol}{\varepsilon}
\newcommand{\nullTransformation}[2]{\nullTransformationSymbol_{\system{#1} \!\to\! \system{#2}}}
\newcommand{\nullState}[1]{\nullTransformationSymbol_{\system{#1}}}

\newcommand{\seqC}{\circ}
\newcommand{\sequentialComp}[2]{{#1} \seqC {#2}}

\newcommand*{\preparationEventNoDownSequence}[2]{\ensuremath{\preparationEventNoDown{#1}^{#2}}}
\newcommand*{\preparationEventNoDownSequenceAll}[2]{\ensuremath{ \left\{ \preparationEventNoDownSequence{#1}{#2} \right\}_{#2 \in \mathbb{N}} }}
\newcommand*{\preparationTestSequence}[3]{\ensuremath{\preparationTest{#1}{#2}^{#3}}}
\newcommand*{\observationEventNoDownSequence}[2]{\ensuremath{\observationEventNoDown{#1}^{#2}}}

\newcommand*{\observationTestSequence}[3]{\ensuremath{\observationTest{#1}{#2}^{#3}}}
\newcommand*{\eventSequenceNoDown}[2]{\ensuremath{ \eventNoDown{#1}^{#2} }}
\newcommand*{\eventSequenceNoDownAll}[2]{\ensuremath{\left\{ \eventNoDown{#1}^{#2}\right\}_{#2 \in \mathbb{N}} }}
\newcommand*{\eventSequence}[3]{\ensuremath{ \event{#1}{#2}^{#3} }}
\newcommand*{\eventSequenceAll}[3]{\ensuremath{ \left\{ \eventSequence{#1}{#2}{#3} \right\}_{#3 \in \mathbb{N}} }}
\newcommand{\conditionedEventSequence}[4]{\ensuremath{{\conditionedEvent{#1}{#2}{#3}}^{#4}}}

\newcommand{\eventTest}[3]{\ensuremath{\left\{\eventNoDown{#1}_{\outcome{#2}} \right\}_{\outcomeIncluded{#2}{#3}}}}
\newcommand*{\eventTestSequence}[4]{ \ensuremath{ \left\{ \eventNoDown{#1}_{\outcome{#2}} \right\}^{#4}_{\outcomeIncluded{#2}{#3}}} }
\newcommand*{\eventTestSequenceAll}[4]{ \ensuremath{ \left\{ \eventTestSequence{#1}{#2}{#3}{#4} \right\}_{#4 \in \mathbb{N}} }}
\newcommand*{\eventTestSequenceOut}[4]{ \ensuremath{ \left\{ \eventNoDown{#1}_{\outcome{#2}} \right\}^{#4}_{\outcomeIncludedSequence{#2}{#3}{#4}}} }
\newcommand*{\eventTestSequenceOutAll}[4]{ \ensuremath{ \left\{ \eventTestSequenceOut{#1}{#2}{#3}{#4} \right\}_{#4 \in \mathbb{N}}} }

\newcommand{\standardTensorProdC}{ \otimes }
\newcommand{\standardTensorProduct}[2]{{#1} \, \standardTensorProdC  \, {#2}}
\newcommand{\paralC}{ \boxtimes }
\newcommand{\parallelComp}[2]{{#1} \paralC {#2}}

\newcommand{\BraidingS}{\mathscr{S}}

\newcommand{\Braid}{\mathcal{S}}
\newcommand{\permutation}{\Braid}

\newcommand{\St}[1]{\ensuremath{\mathsf{St}\left(\system{#1}\right)}}					
\newcommand{\StN}[1]{\ensuremath{\mathsf{St_{1}} \left(\system{#1}\right)}}				
\newcommand{\StR}[1]{\ensuremath{\mathsf{St}_{\mathbb{R}}\left(\system{#1}\right)}}		

\newcommand{\StNOPT}[1]{\ensuremath{\mathsf{St_{1}}\left({#1}\right)}}			

\newcommand{\PurSt}[1]{\ensuremath{\mathsf{PurSt\left(\system{#1}\right)}}}         

\newcommand{\Eff}[1]{\ensuremath{\mathsf{Eff\left(\system{#1}\right)}}}					    
\newcommand{\EffR}[1]{\ensuremath{\mathsf{Eff}_{\mathbb{R}}\left(\system{#1}\right)}}	    
\newcommand{\EffN}[1]{\ensuremath{\mathsf{Eff_{1}\left(\system{#1}\right)}}}				

\newcommand{\Transf}[2]{\ensuremath{\mathsf{Transf\left(\system{#1}\!\to\!\system{#2}\right)}}}
\newcommand{\TransfR}[2]{\ensuremath{\mathsf{Transf_{\mathbb{R}}\left(\system{#1}\!\to\!\system{#2}\right)}}}
\newcommand{\TransfRA}[1]{\ensuremath{\mathsf{Transf_{\mathbb{R}}\left(\system{#1}\right)}}}
\newcommand{\TransfN}[2]{\ensuremath{\mathsf{Transf_{1}\left(\system{#1}\!\to\!\system{#2}\right)}}}
\newcommand{\TransfC}[2]{\ensuremath{\mathsf{Transf_{+}\left(\system{#1}\!\to\!\system{#2}\right)}}}

\newcommand{\TransfA}[1]{\ensuremath{\mathsf{Transf\left(\system{#1}\right)}}}
\newcommand{\RevTransfA}[1]{\ensuremath{\mathsf{RevTransf\left(\system{#1}\right)}}}
\newcommand{\RevTransf}[2]{\ensuremath{\mathsf{RevTransf\left(\system{#1}\!\to\!\system{#2}\right)}}}

\newcommand{\Instr}[2]{\ensuremath{\mathsf{Instr}\left(\system{#1}\!\to\!\system{#2}\right)}}

\newcommand{\InstrRN}[3]{\ensuremath{\mathsf{Instr}_{\mathbb{R}}^{(#3)}\left(\system{#1}\!\to\!\system{#2}\right)}}

\newcommand{\InstrA}[1]{\ensuremath{\mathsf{Instr\left(\system{#1}\right)}}}

\newcommand{\InstrEmpty}[1]{\ensuremath{\mathsf{Instr}\left({#1}\right)}}
\newcommand{\InstrOPT}[1]{\ensuremath{\mathsf{Instr}\left({\text{#1}}\right)}}

\newcommand{\Prep}[1]{\ensuremath{\mathsf{Prep}\left(\system{#1}\right)}}
\newcommand{\ObsOPT}[1]{\ensuremath{\mathsf{Obs}\left({#1}\right)}}
\newcommand{\Obs}[1]{\ensuremath{\mathsf{Obs}\left(\system{#1}\right)}}

\makeatother

\newcommand{\measurePrepare}[4]{
	\myQcircuit{
		&\s{#1}\qw&\measureD{#4}&\prepareC{#3}&\s{#2}\qw&\qw&
	}
}

\newcommand{\minimalDeterministicCausalDestroyReprep}[8]{
	\myQcircuit{
		&\s{#1}\qw&\multigate{1}{#6}&\s{#3}\qw&\measureD{\observationUniqueDeterministic}&\pureghost{}&\prepareC{#8}&\s{#4}\qw&\multigate{1}{#7}&\s{#2}\qw&\qw&
		\\
		&\pureghost{}&\pureghost{#6}&\qw&\qw&\s{#5}\qw&\qw&\qw&\ghost{#7}&
	}
}

\newcommand{\minimalDeterministicCausalDestroyReprepSequencePrime}[9]{
	\myQcircuit{
		&\s{#1}\qw&\multigate{1}{#6}&\sSequencePrime{#3}{#9}\qw&\measureD{\observationUniqueDeterministic}&\pureghost{}&\prepareC{#8}&\sSequencePrime{#4}{#9}\qw&\multigate{1}{#7}&\s{#2}\qw&\qw&
		\\
		&\pureghost{}&\pureghost{#6}&\qw&\qw&\sSequence{#5}{#9}\qw&\qw&\qw&\ghost{#7}&
	}
}

\newcommand{\equivOp}{\sim}

\newcommand{\transfArrow}[1]{\stackrel{#1}{\longrightarrow}}
\newcommand{\allowDisplayBreaks}[1]{
	\begingroup
	\allowdisplaybreaks
	
	#1
	
	\endgroup
}

\newcommand*{\MSBCTPureState}[3]{\left( \preparationEventNoDown{#1}{}_{#3} \preparationEventNoDown{#2} \right)}

\newcommand{\secRef}[1]{\hyperref[#1]{section~\ref*{#1}}}
\newcommand{\aref}[1]{\hyperref[#1]{Appendix~\ref*{#1}}}
\newcommand{\Iref}[1]{\hyperref[#1]{Item~\eqref{#1}}}

\usepackage{acronyms}

\begin{document}
	
\title{Minimal operational theories: classical theories with quantum features}

\author{Davide Rolino}
\email[Corresponding author; ]{davide.rolino01@universitadipavia.it}
\affiliation{Universit\`a degli Studi di Pavia, Dipartimento di Fisica, QUIT Group}
\affiliation{INFN Gruppo IV, Sezione di Pavia, via Bassi 6, 27100 Pavia, Italy}

\author{Marco Erba}
\email{marco.erba@ug.edu.pl}
\affiliation{International Centre for Theory of Quantum Technologies (ICTQT), Uniwersytet Gdański, ul.~Jana Bażyńskiego 1A, 80-309 Gdańsk, Poland}

\author{Alessandro Tosini}
\email{alessandro.tosini@unipv.it}
\affiliation{Universit\`a degli Studi di Pavia, Dipartimento di Fisica, QUIT Group}
\affiliation{INFN Gruppo IV, Sezione di Pavia, via Bassi 6, 27100 Pavia, Italy}

\author{Paolo Perinotti}
\email{paolo.perinotti@unipv.it}
\affiliation{Universit\`a degli Studi di Pavia, Dipartimento di Fisica, QUIT Group}
\affiliation{INFN Gruppo IV, Sezione di Pavia, via Bassi 6, 27100 Pavia, Italy}

\begin{abstract}
	We introduce a class of probabilistic theories, termed \aclp{MSOPT}, where system dynamics are constrained to the minimal set of operations consistent with the set of states and permitting conditional tests. Specifically, the allowed instruments are limited to those derived from compositions of preparations, measurements, swap transformations, and conditional operations. We demonstrate that minimal theories with conditioning and a spanning set of non-separable states satisfy two quantum no-go theorems: no-information without disturbance and no-broadcasting. As a key example, we construct \acl{MSBCT}, a classical toy-theory that lacks incompatible measurements, preparation uncertainty relations, and is noncontextual (both Kochen-Specker and generalised), yet exhibits irreversibility of measurement disturbance, no-information without disturbance, and no-broadcasting. Therefore, the latter three properties cannot be understood \emph{per se} as signatures of non-classicality. We further explore distinctions between a theory and its minimal strongly causal counterpart, showing that while the minimal strongly causal version of quantum theory diverges from full quantum theory, the same does not hold for classical theory. Additionally, we establish the pairwise independence of the properties of simpliciality, strong causality, and local discriminability.
\end{abstract}


\maketitle

\tableofcontents

\acresetall

\section{Introduction}

While classical and quantum mechanics appear to be radically different theories, their structure as theories of information processing share many features that can be understood as the minimal features of any theory that aims at describing physical systems and processes. The processing of systems ranges from preparation procedures to evolution and observations using measuring devices. As the puzzle pieces to model any experiment, the above operational primitives have been taken as the basic elements of the framework of \acp{OPT}~\cite{chiribellaProbabilisticTheoriesPurification2010,darianoQuantumTheoryFirst2016}, then profitably used to understand the origin of quantum peculiarities~\cite{chiribellaInformationalDerivationQuantum2011} among all possible theories of information processing. A similar framework is that of \acp{GPT}~\cite{hardyDisentanglingNonlocalityTeleportation1999,barnumCloningBroadcastingGeneric2006,barrettInformationProcessingGeneralized2007,spekkensEvidenceEpistemicView2007,janottaGeneralizedProbabilisticTheories2013}, which is however focused on statistical aspects arising in prepare-and-measure scenarios, neglecting in most cases the structures related to transformations. 

In order to identify a specific theory, it is necessary to make assumptions on the mathematical entities associated with its processes, either directly or by imposing requirements on tasks that can or cannot be actually implemented manipulating the systems of the theory. For example, both classical and quantum states of a composite system, possibly delocalised over separated labs, can be exhaustively probed via local observations in the two laboratories. This natural feature, can instead be violated by theories that are closely related to the quantum one, as the theory of fermionic systems~\cite{darianoFeynmanProblemFermionic2014,darianoFermionicComputationNonlocal2014}, or to the classical one, as in Ref.~\cite{darianoClassicalityLocalDiscriminability2020} where a non-trivial composition rule for classical systems is introduced. Other assumptions, as the possibility of describing any probabilistic state as part of a perfectly known state of a larger system, sharply discriminate between quantum theory, where it is possible, and classical theory, where it is forbidden. The above criteria are only two instances among several that can in principle bring together or divide classical and quantum theories, or more generally clarify the mathematical structure behind theories of physical systems.

Most studies on \acp{OPT} start from the typical assumption that the admissible transformations of each system coincide with the ``maximal set'' consistent with the set of states, i.~e.~the only requirement is that any state must be transformed into another admissible state. This form of \emph{no-restriction hypothesis} leads, for example, to identify the set of quantum transformations with the set of completely positive, trace-preserving maps, and classical transformations with stochastic matrices. In this paper we explore the opposite scenario where one keeps the minimal set of processes compatible with the structures of the framework, given the systems of the theory and their set of states (and of measurements). The latter scenario is closer to that of a real world laboratory, given that experimenters typically do not have access to all theoretically implementable transformations, but only to certain subsets of them.

The resulting class of theories is proved to preserve relevant quantum no-go theorems independently of the nature of its systems. Surprisingly, also classical systems can support quantum features as the impossibility of  gaining information without introducing irreversible disturbance, or the impossibility of broadcasting states. This shows that the observation of certain quantum phenomena in an experimental setting cannot always be taken as a definitive proof that the system under study in the experiment is actually quantum.\footnote{A similar result is also discussed in Ref.~\cite{schmidShadowsSubsystemsGeneralized2024}. However, we remark that the two considered scenarios are different. The authors of Ref.~\cite{schmidShadowsSubsystemsGeneralized2024} are interested in assessing the relationship between the classical explainability of parts of theories---in particular, between fragments and shadows of \acp{GPT}. While, instead, we consider fully-fledged theories, just with a restriction in the allowed dynamics.} 

In Ref.~\cite{erbaMeasurementIncompatibilityStrictly2024} the class of \acp{MOPT}~was introduced, namely \acp{OPT} where the only allowed operations are preparations, measurements, the identity map, the swap (systems exchange) map, and any operation that can be obtained by composing these operations sequentially or in parallel.  A further minimal operation that is missing in \acp{MOPT} is the possibility of \emph{conditioning} which experiment to perform next based on the outcome of a previous experiment. This property,  also know as \textdef{strong causality}~\cite{perinottiCellularAutomataOperational2020,perinottiCausalInfluenceOperational2021}, is arguably desirable for any reasonable physical theory. Therefore, we introduce a new class of theories, termed \acp{MSOPT}, by incorporating all possible conditional operations, while retaining the minimal resource constraints of \acp{MOPT}. The main intent is to understand which  features of minimal theories are robust under the introduction of strong causality. Remarkably, we show that almost all marking features of minimal models can survive, under suitable conditions---but they can also disappear in other circumstances.

We first systematically analyse the operational \emph{desideratum} that the spaces of operations in an \ac{OPT} must be ``complete'': if there is a procedure to prepare a transformation with arbitrary precision, then the latter is accepted as a transformation of the theory. This property is given in terms of Cauchy sequences with respect to an operational distance for all elements of the theory---states, measurements, transformations and more generally for their collections that generalise the notion of quantum instrument. We introduce a procedure to complete a theory in such a way that it is strongly causal, granting consistency between operational completeness and the compositional structure of the theory.

We then prove that, whenever a minimal theory with strong causality admits a spanning set of non-separable states, the identity transformation for every system is atomic---i.e.~it cannot be obtained performing a measurement and ignoring the outcome. This result leads to a series of consequences: \acp{MSOPT} satisfy \ac{NIWD}~\cite{buschNoInformationDisturbance2009,darianoQuantumTheoryFirst2016,heinosaariNofreeinformationPrincipleGeneral2019,darianoInformationDisturbanceOperational2020}---that is the impossibility of learning something non-trivial from a system without perturbing its state irreversibly---as well as irreversibility of measurement disturbance~\cite{erbaMeasurementIncompatibilityStrictly2024} and no-broadcasting~\cite{woottersSingleQuantumCannot1982,dieksCommunicationEPRDevices1982,yuenAmplificationQuantumStates1986,barnumNoncommutingMixedStates1996,barnumCloningBroadcastingGeneric2006,walkerClassicalBroadcastingPossible2007,barnumGeneralizedNoBroadcastingTheorem2007,pianiNoLocalBroadcastingTheoremMultipartite2008,luoQuantumNoBroadcasting2010}. 

Moreover, we show that in the above class of theories there exist instances that are classical, in the sense that they are \emph{simplicial} \acp{OPT}. In other words, we exhibit theories where the state spaces are simplexes, whose pure states (the vertices of the simplexes) are jointly perfectly discriminable, and at the same time have a wealth of features traditionally considered as signatures of non-classicality. We explicitly construct a classical \ac{MSOPT}, termed \ac{MSBCT}. This theory is a minimal strongly causal version of \ac{BCT}~\cite{darianoClassicalityLocalDiscriminability2020}. The present toy-model shares with \ac{BCT} the property of being locally equivalent to classical theory, thus it has no incompatible measurements, there is no uncertainty in preparation of states and it is noncontextual. However, based on the above results, the model here introduced must exhibit irreversibility, \ac{NIWD}, and no-broadcasting.

Finally, based on the properties of \ac{MSBCT} here introduced, we prove the independence of three main features of classical information theory: local tomography, simpliciality of the set of states and strong causality. Therefore, our results provide further insights on the relationships between different physical properties: these insights are, in turn, useful for the axiomatisation programs of \acf{QT}, and for adjudicating between \ac{QT} and alternative physical theories.

The outline of the presentation is as follows. In \autoref{sec:opt} a review of the framework of \acp{OPT}~\cite{chiribellaProbabilisticTheoriesPurification2010,chiribellaInformationalDerivationQuantum2011,darianoQuantumTheoryFirst2016,chiribellaQuantumPrinciples2016,darianoClassicalityLocalDiscriminability2020,perinottiCellularAutomataOperational2020} is provided, focusing on those aspects that will be central in this work. Big emphasis, with several original results, is placed on the topological structure that characterises \acp{OPT}, introducing the notion of \textdef{generalised instrument spaces} (\autoref{subsec:genInstr}). In \autoref{sec:cauchy} the properties of Cauchy sequences of transformations and instruments are studied. We show that sequences of instruments are Cauchy if and only if those of the transformations that compose them are (\autoref{thm:OPT:norm:instrEventConvergence}), and that the space of deterministic transformations is Cauchy complete (\autoref{lem:OPT:causal:detTransf:convergence}). Grounding on the above properties, the procedure for adding all the conditional instruments to a generic \ac{OPT} is formalised (\autoref{sec:strongCompleteness}). \autoref{thm:OPT:stronglyClosure} guarantees that adding the conditional operations and subsequently Cauchy completing the transformations spaces is sufficient for the new theory obtained to be strongly causal. After recalling the definitions of the properties of interest for this article (broadcasting and irreversibility of measurement disturbance) (\autoref{sec:properties}), in \autoref{sec:MOPT} the classes \acp{MOPT} and \acp{MSOPT} (\autoref{subsec:MSOPT}) are introduced. We show that in \acp{MOPT} the identity is atomic for every systems (\autoref{thm:OPT:minimal:symmetric:causal:idAtomicity}), while in \acp{MSOPT} this property holds for every theory where the entangled states are spanning for the set of states (\autoref{thm:MSOPT:symmetric:idAtomicity}). Then, in \autoref{sec:msbct}, we construct \ac{MSBCT}, characterizing its main properties. In conclusion, an in-depth discussion of the consequences of the results is carried out (\autoref{sec:conclusion}).

\section{Operational Probabilistic Theories}
\label{sec:opt}

In the last two decades, the way we study and understand the quantum world has profoundly changed. With the advent of quantum information~\cite{hardyDisentanglingNonlocalityTeleportation1999,hardyQuantumTheoryFive2001,fuchsQuantumMechanicsQuantum2002,brassardInformationKey2005,darianoProbabilisticTheoriesWhat2010} we started to treat \acf{QT} as a theory of information processing~\cite{barrettInformationProcessingGeneralized2007,darianoProbabilisticTheoriesWhat2010,darianoTestingAxiomsQuantum2010} selected among a universe of possible alternative theories~\cite{hardyDisentanglingNonlocalityTeleportation1999,chiribellaProbabilisticTheoriesPurification2010,chiribellaInformationalDerivationQuantum2011,masanesDerivationQuantumTheory2011,dakicQuantumTheoryEntanglement2011,spekkensEvidenceEpistemicView2007,chiribellaQuantumPrinciples2016,darianoQuantumTheoryFirst2016}. The selection criteria pertain to the ability to perform specific information processing tasks~\cite{chiribellaInformationalDerivationQuantum2011,darianoQuantumTheoryFirst2016}.

Aim of the framework of \acfp{OPT} is exactly to model \ac{QT} along with all these other alternative theories of information processing and to describe every information theory starting from its compositional (combining operations to build up experiments) structure. The same aim and scope is shared with the deeply related frameworks of \acp{GPT}~\cite{barrettInformationProcessingGeneralized2007,barnumTeleportationGeneralProbabilistic2008,barnumInformationProcessingConvex2011,plavalaGeneralProbabilisticTheories2021} and quantum picturialism~\cite{coeckeKindergartenQuantumMechanics2006,coeckePicturingQuantumProcesses2017}, with common roots dating back to Ludwig's works on the foundations of quantum mechanics~\cite{ludwigFoundationsQuantumMechanics1985,wilceTestSpacesOrthoalgebras2000}.

In this section we provide a review of the framework of \acp{OPT}~\cite{chiribellaProbabilisticTheoriesPurification2010,chiribellaQuantumPrinciples2016,darianoQuantumTheoryFirst2016,darianoClassicalTheoriesEntanglement2020,perinottiCellularAutomataOperational2020} emphasizing their \emph{linear} and \emph{topological} structure and proving some properties later used in this work. 

\subsection{Basic structure of the framework}
Every \ac{OPT} \OPTMath{} is completely characterised by a set of systems along with the set of operations that it is possible to perform on them.

\textdef{Systems} represent the physical entities which are probed in a laboratory (e.g.~an electron, a molecule, a radiation field, etc$\ldots$)~\cite{darianoClassicalTheoriesEntanglement2020}. They are denoted with capital Roman letters \system{A}, \system{B}, $\ldots \in \Sys{\OPTMath}$. In \ac{QT} systems are complex Hilbert spaces. The processes occurring between systems are captured by the notion of \textdef{tests}, which represent physical processes that can occur within a given theory. A given test $\testComplete{T}{X}{A}{B} \equiv \test{T}{X} \in \testCollectionAB{A}{B}$ models an experiment acting on a given input system \system{A} with output system \system{B}.\footnote{Most of the times, unless it is not clear from the context, the input and output systems of a test will not be specified, thus preferring the notation $\test{T}{X}$ in place of $\testComplete{T}{X}{A}{B}$.} Systems can also be depicted in diagrammatic notation as wires, while tests as wired boxes:
\begin{equation*}
	\testComplete{T}{X}{A}{B} \quad \longleftrightarrow \quad \myQcircuit{
		&\s{A}\qw&\gate{\test{T}{X}}&\s{B}\qw&\qw&
	}.
\end{equation*}
As a convention the input-output direction is taken to go from the left to the right, which does not imply a preferred direction for the flow of information.\footnote{In the subset of \textdef{causal \acp{OPT}} a preferred direction for the flow of information is instead fixed, typically from the left to the right indeed.}

The ``\outcomeSpace{X}'' appearing in the definition of a test represents the \textdef{outcome space} of the test. It is a finite set containing all the possible outcomes of the experiment. To each outcome $\outcomeIncluded{x}{X}$ is associated an \textdef{event} $\event{T}{x} \in \eventCollectionAB{A}{B}$ representing the realization of a particular occurrence in a physical process.\footnote{Here we do not include the possibility of having infinite, possibly continuous, outcome spaces. However, the framework has no bottlenecks towards non-finite outcome spaces.} Therefore, tests are finite collections of events: $\test{T}{X} \equiv \eventTest{T}{x}{X}$. Diagrammatically:
\begin{equation*}
	\event{T}{x} \quad \longleftrightarrow \quad \myQcircuit{
		&\s{A}\qw&\gate{\event{T}{x}}&\s{B}\qw&\qw&
	}, \quad \forall \outcomeIncluded{x}{X}.
\end{equation*}
There exists a particular set of events called \textdef{deterministic} which are the ones associated to tests whose outcome space has just one element---that is a \textdef{singleton set}---, which will be represented as $\SingletonSet \mathDef \left\{ * \right\}$. Tests associated to deterministic events are called \textdef{singleton tests} and operationally model processes that do not provide information. In \ac{QT}, tests are \emph{quantum instruments}, events are \emph{quantum operations}, and deterministic events are \emph{quantum channels}.

Tests, and consequently events, can be composed in two ways. Sequentially:
\begin{align*}
	\myQcircuit{
		&\s{A}\qw&\gate{\test{GT}{\outcomeSpaceDouble{X}{Y}}}&\s{C}\qw&\qw&
	} &= \quad \myQcircuit{
		&\s{A}\qw&\gate{\test{T}{X}}&\s{B}\qw&\gate{\test{G}{Y}}&\s{C}\qw&\qw&
	}.
\end{align*}
and in parallel:
\begin{align*}
	\label{eqt:OPT:parallelComp}
	\myQcircuit{
		&\s{AC}\qw&\gate{\parallelComp{\test{T}{X}}{\test{G}{Y}}}&\s{BD}\qw&\qw&
	} &= \quad 	\myQcircuit{
		&\s{A}\qw&\multigate{1}{\parallelComp{\test{T}{X}}{\test{G}{Y}}}&\s{B}\qw&\qw&
		\\
		&\s{C}\qw&\ghost{\parallelComp{\test{T}{X}}{\test{G}{Y}}}&\s{D}\qw&\qw&
	}\\[10pt]
	&= \quad \myQcircuit{
		&\s{A}\qw&\gate{\test{T}{X}}&\s{B}\qw&\qw&
		\\
		&\s{C}\qw&\gate{\test{G}{Y}}&\s{D}\qw&\qw&
                                                           },
\end{align*}
where \system{AB} is a \textdef{composite system}, obtained by composing in parallel the two systems \system{A} and \system{B}. The operation of parallel composition of systems is associative and has an identity element: the trivial system. The trivial system \trivialSystem\ is a particular system representing ``nothing the theory cares to describe''~\cite{chiribellaQuantumPrinciples2016}. Furthermore, the set of systems $\Sys{\OPTMath}$ of any \ac{OPT} is closed with respect to the latter operation.\footnote{The fact that \Sys{\OPTMath} is closed with respect to the operation of parallel composition means that if any two systems \system{A} and \system{B} are in \Sys{\OPTMath}, then also $\system{AB} \in \Sys{\OPTMath}$.} In \acp{OPT}, in general, the operation of parallel composition $\paralC$ differs from the standard tensor product, as it happens for example in the composition of fermionic systems in \acl{FQT}~\cite{darianoFeynmanProblemFermionic2014,darianoFermionicComputationNonlocal2014,bravyiFermionicQuantumComputation2002,lugliFermionicStateDiscrimination2020,perinottiShannonTheoryQuantum2023} or in classical theory with bilocal tomography~\cite{darianoClassicalityLocalDiscriminability2020}. Both operations of sequential and parallel composition are associative and have an identity element. In the former case the identities are given by a family of tests $\left\{\testComplete{I}{\SingletonSet}{A}{A}\right\}_{\system{A} \in \Sys{\OPTMath}}$ with the associated family of deterministic events $\left\{ \identityTest{A} \right\}_{\system{A} \in \Sys{\OPTMath}}$, while in the latter case the identity is given by $\identityTest{\trivialSystem}$.
F
Diagrammatically the trivial system will not be represented, leaving a blank space. Accordingly tests of the form $\preparationTestComplete{\rho}{X}{A} \in \testCollectionAB{\trivialSystem}{A}$ and $\observationTestComplete{a}{X}{A}  \in \testCollectionAB{A}{\trivialSystem}$, and the corresponding events, will be represented as follows:
\begin{align*}
	&\myQcircuit{
		&\prepareC{\preparationTest{\rho}{X}}&\s{A}\qw&\qw&
	} \quad , \quad \myQcircuit{
		&\s{A}\qw&\measureD{\observationTest{a}{X}}&
	}; \\[10pt]
	&\myQcircuit{
		&\prepareC{\preparationEvent{\rho}{x}}&\s{A}\qw&\qw&
	} \quad , \quad \myQcircuit{
		&\s{A}\qw&\measureD{\observationEvent{a}{x}}&
	}.
\end{align*}
Tests of this kind are called \textdef{preparation-} and \textdef{observation-tests} of system \system{A}, respectively, while their associated events are called \textdef{preparations} and \textdef{observations} of system \system{A}, respectively. These represent a generalisation of the notions of density matrix and \acf{POVM} of \ac{QT}, respectively. When writing the equations not in diagrammatic form, the \textdef{round ket} $\rket{\cdot}$ and \textdef{round bra} $\rbra{\cdot}$ notation will be used to represent preparations and observations, respectively. The last particular case we have to consider is where both the input and output systems are the trivial one. In this case, the tests $\probabilityEventNoDown{p}_{\outcomeSpace{X}} \in \testCollectionAB{\trivialSystem}{\trivialSystem}$ are called \textdef{scalar-tests}, while the corresponding events $\probabilityEvent{p}{x} \in \eventCollectionAB{\trivialSystem}{\trivialSystem}$ are called \textdef{scalars}.

In conclusion, we observe that the two operations of parallel and sequential composition are required to commute:
\begin{equation}
	\label{eqt:OPT:CompatParallSeq}
	\myQcircuitComp{
		&\s{A}\qw&\gate{\test{T}{X}}&\s{B}\qw&\gate{\test{G}{Y}}&\s{C}\qw&\qw&
		\\
		&\s{D}\qw&\gate{\test{W}{Z}}&\s{E}\qw&\gate{\test{F}{K}}&\s{F}\qw&\qw&
		\relax\gategroupColor{1}{3}{2}{3}{0.6em}{--}{blue}
		\relax\gategroupColor{1}{5}{2}{5}{0.6em}{--}{blue}
		\relax\gategroupColor{1}{3}{2}{5}{2.4em}{--}{orange}
	} = \quad\!
	\myQcircuitComp{
		&\s{A}\qw&\gate{\test{T}{X}}&\s{B}\qw&\gate{\test{G}{Y}}&\s{C}\qw&\qw&
		\\
		&\s{D}\qw&\gate{\test{W}{Z}}&\s{E}\qw&\gate{\test{F}{K}}&\s{F}\qw&\qw&
		\relax\gategroupColor{1}{3}{1}{5}{0.6em}{--}{orange}
		\relax\gategroupColor{2}{3}{2}{5}{0.6em}{--}{orange}
		\relax\gategroupColor{1}{3}{2}{5}{2.4em}{--}{blue}
	}.
\end{equation}
Before concluding our presentation of the compositional structure of \acp{OPT}, we introduce the notion of \textdef{reversible} event. An event $\eventNoDown{R} \in \eventCollectionAB{A}{B}$ is \textdef{reversible} if there exists an event $\inverse{\eventNoDown{R}} \in \eventCollectionAB{B}{A}$ such that $\sequentialComp{\eventNoDown{R}}{\inverse{\eventNoDown{R}}} = \identityTest{A}$ and $\sequentialComp{\inverse{\eventNoDown{R}}}{\eventNoDown{R}} = \identityTest{B}$.

The introduction of reversible transformations also allows to define the notion of \textdef{operational equivalence} between systems. Two systems \system{A} and \system{B} are said to be operationally equivalent $\sysEquiv{\system{A}}{\system{B}}$ if there exists a reversible transformation $\eventNoDown{R} \in \RevTransf{A}{B}$, i.e., whose input and output systems are \system{A} and \system{B} or viceversa.

We are now able to make the final requirement so that the compositional structure of \acp{OPT} turns out to be analogous to that of \ac{QT}. We require the existence of a family of reversible tests called \textdef{braiding} which allows a pair of agents to exchange systems between each other. In other words, given any two systems \system{A}, $\system{B} \in \Sys{\OPTMath}$ there exist two singleton reversible tests: $\testComplete{S}{\SingletonSet}{AB}{BA} = \left\{ \BraidingS_{\system{A}, \system{B}} \right\}$ and its inverse $\testComplete{\left( \inverse{S} \right)}{\SingletonSet}{BA}{AB} = \left\{ \inverse{\BraidingS}_{\system{A}, \system{B}} \right\}$. They will be pictorially represented as follows:
\begin{align*}
	\myQcircuit{
		&\s{A}\qw&\multigate{1}{\BraidingS}&\s{B}\qw&\qw&
		\\
		&\s{B}\qw&\ghost{\BraidingS}&\s{A}\qw&\qw&
	} &\longleftrightarrow \quad\!
	\myQcircuit{
		&\s{A}\qw&\braiding&\s{B}\qw&\qw&
		\\
		&\s{B}\qw&\braidingGhost&\s{A}\qw&\qw&
	},\\[10pt]
	\myQcircuit{
		&\s{B}\qw&\multigate{1}{\inverse{\BraidingS}}&\s{A}\qw&\qw&
		\\
		&\s{A}\qw&\ghost{\inverse{\BraidingS}}&\s{B}\qw&\qw&
	} &\longleftrightarrow \quad\!
	\myQcircuit{
		&\s{B}\qw&\braidingInv&\s{A}\qw&\qw&
		\\
		&\s{A}\qw&\braidingGhost&\s{B}\qw&\qw&
	}.
\end{align*}
These tests must satisfy the \textdef{naturality} property:
\begin{equation}
	\label{eqt:opt:braid:naturality}
	\myQcircuit{
		&\s{A}\qw&\gate{\test{T}{X}}&\s{B}\qw&\braiding&\s{D}\qw&\qw&
		\\
		&\s{C}\qw&\gate{\test{G}{Y}}&\s{D}\qw&\braidingGhost&\s{B}\qw&\qw&
	} = 
	\myQcircuit{
		&\s{A}\qw&\braiding&\s{C}\qw&\gate{\test{G}{Y}}&\s{D}\qw&\qw&
		\\
		&\s{C}\qw&\braidingGhost&\s{A}\qw&\gate{\test{T}{X}}&\s{B}\qw&\qw&
	},
\end{equation}
namely tests and events can slide along the wires. \acp{OPT} where $\BraidingS_{\system{A}, \system{B}} = \inverse{\BraidingS}_{\system{A}, \system{B}}$ for any pair of systems of the theory are called \textdef{symmetric}. In this case the braiding operation becomes a transposition (also referred to as the \textdef{swap} operation) and it is diagrammatically represented as follows:
\begin{equation*}
	\myQcircuit{
		&\s{A}\qw&\braidingSym&\s{B}\qw&\qw&
		\\
		&\s{B}\qw&\braidingGhost&\s{A}\qw&\qw&
	}.
\end{equation*}

The described structure is that of a braided strict monoidal category~\cite{maclaneCategoriesWorkingMathematician1978,awodeyCategoryTheory2006,heunenCategoriesQuantumTheory2019} and carries only an operational interpretation. It is merely a descriptive tool. To allow \acp{OPT} to make predictions about experiments' outcomes, we have to supply them with a \emph{probabilistic structure}. It is then required that to any acyclic circuit of events beginning with a preparation and ending with an observation, i.e., a scalar event, it is associated a conditional probability distribution
\begin{equation*}
	\begin{aligned}
		&\eventCollectionAA{\trivialSystem} \ni \text{G}\left( \event{T}{x}, \event{G}{y}, \ldots, \event{F}{z} \right) \\ & \mathDef \mathbb{P}\left( \outcome{x}, \outcome{y}, \ldots, \outcome{z} \vert \text{G}\left( \test{T}{X}, \test{G}{Y}, \ldots, \test{F}{Z} \right) \right).
	\end{aligned}
\end{equation*}
In other words, given that the experiment $\text{G}\left( \test{T}{X}, \test{G}{Y}, \ldots, \test{F}{Z} \right)$ is performed, the formula provides the probability of  occurrence of any series of events $\text{G}\left( \event{T}{x}, \event{G}{y}, \ldots, \event{F}{z} \right)$ and then reading the corresponding series of outcomes $\left(\outcome{x}, \outcome{y}, \ldots, \outcome{z}\right)$. For example:
\begin{equation*}
	\myQcircuit{
		&\prepareC{\preparationEvent{\rho}{x}}&\s{A}\qw&\gate{\event{T}{y}}&\s{B}\qw&\measureD{\observationEvent{a}{z}}&
	} \mathDef \probabilityEventNoDown{p}\left( \outcome{x}, \outcome{y}, \outcome{z} \vert \preparationTest{\rho}{X}, \test{T}{Y}, \observationTest{a}{Z} \right),
\end{equation*}
where $\event{\rho}{x}$ is an event of the test $\test{\rho}{X}$ and analogously for the others.

Within the framework it is also made the requirement that the spaces of tests and events are quotiented with respect to the following equivalence relation. For all systems \system{A}, $\system{B} \in \Sys{\OPTMath}$ and for all events $\eventNoDown{T}_{1}$, $\eventNoDown{T}_{2} \in \eventCollectionAB{A}{B}$, we define $\probabilisticallyEquivalent{\eventNoDown{T}_{1}}{\eventNoDown{T}_{2}}$ if and only if
\begin{equation}
	\label{eqt:opeEq}
	\myQcircuitBox{
		&\multiprepareC{1}{\preparationEventNoDown{\rho}}&\s{A}\qw&\gate{\eventNoDown{T}_{1}}&\s{B}\qw&\multimeasureD{1}{\observationEventNoDown{a}}&
		\\
		&\pureghost{\preparationEventNoDown{\rho}}&\qw&\s{E}\qw&\qw&\ghost{\observationEventNoDown{a}}&
	} = \myQcircuitBox{
		&\multiprepareC{1}{\preparationEventNoDown{\rho}}&\s{A}\qw&\gate{\eventNoDown{T}_{2}}&\s{B}\qw&\multimeasureD{1}{\observationEventNoDown{a}}&
		\\
		&\pureghost{\preparationEventNoDown{\rho}}&\qw&\s{E}\qw&\qw&\ghost{\observationEventNoDown{a}}&
	},
\end{equation}
for all possible $\system{E} \in \Sys{\OPTMath}$, $\preparationEventNoDown{\rho} \in \eventCollectionAB{\trivialSystem}{AE}$, and $\observationEventNoDown{a} \in \eventCollectionAB{BE}{\trivialSystem}$. This follows from the idea that whenever two events (tests) are characterised by the same statistics in any experiment they are indistinguishable. We observe that, while it is true that whenever $\eventNoDown{T}_{1} = \eventNoDown{T}_{2}$ also \eqref{eqt:opeEq} also holds, the converse is not true in general, hence the requested equivalence relation.

The quotient class of events in an \ac{OPT} are called \textdef{transformations}, and their subset having input system \system{A} and output system \system{B} is denoted by:
\begin{equation*}
	\Transf{A}{B} \mathDef \faktor{\eventCollectionAB{A}{B}}{\equivOp}.
\end{equation*}
The special case of preparations $\St{A} \mathDef \Transf{\trivialSystem}{A}$ and observations $\Eff{A} \mathDef \Transf{A}{\trivialSystem}$ are called the \textdef{states} and \textdef{effects} of system \system{A}, respectively. We define $\TransfN{A}{B}$, $\StN{A}$ and $\EffN{A}$ as the set of deterministic transformations, states and effects,  respectively. The tests from \system{A} to \system{B} become \textdef{instruments}, and their collection is denoted by:
\begin{equation*}
	\Instr{A}{B} \mathDef \faktor{\testCollectionAB{A}{B}}{\equivOp}.
\end{equation*}
Finally, the collections of \textdef{preparation-} and \textdef{observation-instruments} are denoted by $\Prep{A}$ and $\Obs{A}$, respectively.

There is one final assumption for a generic information theory to be an \ac{OPT}: the theory must be closed with respect to the \textdef{coarse-graining} operation. This operation allows one to disregard information related to the outcome of an experiment. Given any test $\test{T}{X}$ and any disjoint partition $\left\{ \outcomeSpaceConditioned{Z}{y} \right\}_{\outcomeIncluded{y}{Y}}$ of the outcome space $\outcomeSpace{X}$ there exists the \textdef{coarse-grained test} $\prim{\test{T}{Y}}$ representing the same operation, where the outcome $\outcomeIncluded{y}{Y}$ stands for ``the outcome of the test $\test{T}{X}$ belongs to $\outcomeSpaceConditioned{Z}{y}$''. The event $\prim{\event{T}{y}} = \sum_{\outcomeIncludedConditioned{x}{Z}{y}}\event{T}{x}$ is called \textdef{coarse-grained}. Obviously, given a test \test{T}{X} the full coarse-grained transformation $\eventCG{T}{X} = \sum_{\outcomeIncluded{x}{X}} \event{T}{x}$ is deterministic. The operations of scalar multiplication, sequential and parallel composition distribute over coarse-graining. As to what precisely is meant by the sum symbol used here we refer to \autoref{ssec:linear}.

In conclusion of this section, we observe that for every \ac{OPT} the event with null-probability in any experiment is an actual event of the theory: $0 \in \TransfA{\trivialSystem}$. Consequence of this fact is that for every pair of systems \system{A}, $\system{B} \in \Sys{\OPTMath}$, there exists a transformation $\nullTransformation{A}{B} \in \Transf{A}{B}$, called \textdef{null transformation}, defined by the following relation:
\begin{equation*}
	 \myQcircuitBox{
	 	&\multiprepareC{1}{\preparationEventNoDown{\rho}}&\s{A}\qw&\gate{\nullTransformation{A}{B}}&\s{B}\qw&\multimeasureD{1}{\observationEventNoDown{a}}&
	 	\\
	 	&\pureghost{\preparationEventNoDown{\rho}}&\qw&\s{E}\qw&\qw&\ghost{\observationEventNoDown{a}}&
	 } = \rbraSystem{\observationEventNoDown{a}}{BE} \left( \parallelComp{\nullTransformation{A}{B}}{\identityTest{E}} \right) \rketSystem{\preparationEventNoDown{\rho}}{AE}  = 0,
\end{equation*}
for any system $\system{E} \in \Sys{\OPTMath}$, $\preparationEventNoDown{\rho} \in \St{AE}$, and $\observationEventNoDown{a} \in \Eff{BE}$. In words, the null transformation is the transformation that always occur with null probability in any closed circuit. Clearly, transformations are invariant for coarse-graining with the null one
\begin{equation*}
	\eventNoDown{T} = \eventNoDown{T} + \nullTransformation{A}{B},
\end{equation*}
for any couple of systems \system{A}, $\system{B} \in \Sys{\OPTMath}$ and transformation $\eventNoDown{T} \in \Transf{A}{B}$.
 
\subsubsection{Atomicity, extremality and purity}
We introduce a classification of transformations based on how they can be decomposed as combinations of other transformations of the theory.

\begin{definition}[Atomic transformation]
	\label{def:opt:transf:atomic}
	A transformation $\eventNoDown{T} \in \Transf{A}{B}$ is \textdef{atomic} if, given $\eventNoDown{T}_{1}$, $\eventNoDown{T}_{2} \in \Transf{A}{B}$, one has the following implication:
	\begin{equation*}
		\eventNoDown{T} = \eventNoDown{T}_{1} + \eventNoDown{T}_{2} \implies \eventNoDown{T}_{1}, \eventNoDown{T}_{2} \propto\eventNoDown{T}.
	\end{equation*}
\end{definition}

The notion of atomic transformations captures the idea of ``indecomposable'' events from a conic point of view. These are the transformations that generate the extremal rays of the cones generated by the transformation sets:
\begin{equation*}
	\TransfC{A}{B}  \mathDef \left\{ \lambda \eventNoDown{T} \mid \lambda \geq 0, \eventNoDown{T} \in \Transf{A}{B} \right\}.
\end{equation*}

The same argument can also be made in the case of convex combinations. 

\begin{definition}[Extremal transformations]
	A transformation $\eventNoDown{T} \in \Transf{A}{B}$ is called \textdef{extremal} if, given $\eventNoDown{T}_{1}$, $\eventNoDown{T}_{2} \in \Transf{A}{B}$ and $\probabilityEventNoDown{p} \in \left(0,1\right)$, the condition $\eventNoDown{T} = \probabilityEventNoDown{p}\eventNoDown{T}_{1} + \left( 1 - \probabilityEventNoDown{p} \right)\eventNoDown{T}_{2}$ implies $\eventNoDown{T}_{1} = \eventNoDown{T}_{2}=\eventNoDown{T}$.
\end{definition}

Extremal transformations embody the notion of extreme points of convex sets. The latter property is of particular interest in the special class of theories where $\Transf{A}{B}$ is convex for every 
$\system A$ and $\system B$ 
\begin{definition}[Convex \acp{OPT}]
	An \ac{OPT} \OPTMath{} where $\Transf{A}{B}$ coincides with its convex hull for every couple of systems \system{A}, $\system{B} \in \Sys{\OPTMath}$ is called \textdef{convex}~\cite{darianoQuantumTheoryFirst2016,darianoClassicalityLocalDiscriminability2020}.
\end{definition} 

In general the two properties of atomicity and extremality of a transformation are not related. There are transformations that are extremal but not atomic and viceversa, an example being the deterministic effect of any system of \ac{CT} or \ac{QT}~\cite{darianoQuantumTheoryFirst2016}, which is clearly an extremal point of the convex set of effects but can be obtained as the coarse-graining of any observation-test. 

Slightly detaching ourselves from quantum theory's tradition, and following the nomenclature of Ref.~\cite{darianoClassicalityLocalDiscriminability2020}, we have the following definition.

\begin{definition}[Pure and mixed transformations]
	A transformation is \textdef{pure} if it is extremal and deterministic. While, a transformation is \textdef{mixed} if it is neither atomic, nor extremal.
\end{definition}

In the following,the set of pure states of a system \system{A} will be indicated with $\PurSt{A}$.

\subsection{Linear structure}
\label{ssec:linear}
The equivalence relation of transformation \eqref{eqt:opeEq} reduces to the following in the case of states:
\begin{equation*}
	\myQcircuit{
		&\prepareC{\preparationEventNoDown{\rho}_{1}}&\s{A}\qw&\measureD{\observationEventNoDown{a}}&
	} = \quad \myQcircuit{
		&\prepareC{\preparationEventNoDown{\rho}_{2}}&\s{A}\qw&\measureD{\observationEventNoDown{a}}&
	},
\end{equation*}
and analogously for effects:
\begin{equation*}
	\myQcircuit{
		&\prepareC{\preparationEventNoDown{\rho}}&\s{A}\qw&\measureD{\observationEventNoDown{a}_{2}}&
	} = \quad \myQcircuit{
		&\prepareC{\preparationEventNoDown{\rho}}&\s{A}\qw&\measureD{\observationEventNoDown{a}_{1}}&
	}.
\end{equation*}

These relations have two important consequences. First, the set of states is separating for that of effects---that is, for every pair of states $\preparationEventNoDown{\rho}_{1}$, $\preparationEventNoDown{\rho}_{2} \in \St{A}$ such that $\preparationEventNoDown{\rho}_{1} \neq \preparationEventNoDown{\rho}_{2}$, there exists an effect $\observationEventNoDown{a} \in \Eff{A}$ such that $\rbraketSystem{\observationEventNoDown{a}}{\preparationEventNoDown{\rho}_{1}}{A} \neq \rbraketSystem{\observationEventNoDown{a}}{\preparationEventNoDown{\rho}_{2}}{A}$---and viceversa for effects. Second, states can be seen as a set of functionals from \Eff{A} to the real interval $[0,1]$, and viceversa effects are a set of functionals that map \St{A} to $[0,1]$.

These two properties allows us to equip \acp{OPT} with a linear structure by extending the functional described above to the whole $\mathbb{R}$~\cite{darianoQuantumTheoryFirst2016}. Considering the collection of all the functionals defined in this way it is possible to construct the spaces of \textdef{generalised states} $\StR{A}$ and \textdef{generalised effects} $\EffR{A}$, which are the real vector spaces for which $\St{A}$ and $\Eff{A}$ are spanning sets, respectively. The separability property between states and effects induces the same property for the generalised spaces. Furthermore, each one is included in the algebraic dual of the other, $\StR{A} \subseteq \algebraicDual{\EffR{A}}$ and $\EffR{A} \subseteq \algebraicDual{\StR{A}}$. In the particular case where the dimension $\dim\StR{A}$ (or equivalently $\dim \EffR{A}$) is finite one has $\StR{A} = \algebraicDual{\EffR{A}}$ ($\EffR{A} = \algebraicDual{\StR{A}}$). The dimension of the generalised state space $\sysDimension{A} \mathDef \dim \StR{A}$ is defined to be the \textdef{dimension} (or \textdef{size}) of system \system{A}. The dimension $\sysDimension{}$ of a system represents the number of probabilities that has to be known in order to completely characterise the states of the system when represented as vectors in $\mathbb{R}^{\sysDimension{}}$. For example, in quantum theory the size of a system \system{A} defined on an Hilbert space of dimension $d_{\system{A}}$ is given by $\sysDimension{A} = d_{\system{A}}^{2}$.

It is also possible to define \textdef{generalised transformations} $\TransfR{A}{B}$. Looking at \eqref{eqt:opeEq}, one has to consider the families of transformations $\left\{\parallelComp{\eventNoDown{T}}{\identityTest{E}}\right\}_{\system{E} \in \Sys{\OPTMath}}$ seen as maps between the collections of states $\St{AE}$ and $\St{BE}$, or through their dual between the collections of effects $\Eff{BE}$ and $\Eff{AE}$. Starting from them, one can define a unique family of linear maps between the generalised spaces and consequently construct the real vector space $\TransfR{A}{B}$. The operations of parallel and sequential composition for the generalised case are induced by the same operations for transformations~\cite{darianoQuantumTheoryFirst2016}.

We observe that the summation symbol used to indicate the coarse-graining operation has now a precise meaning. Whenever, we make the coarse-graining of two operations, we are considering their sum, seeing the transformations as elements of the generalised transformation space.

The linear structure makes any \ac{OPT} ``usable''. In fact, regardless of how abstract are its operational constituents it is always possible to embed everything in a linear vector space where calculations can be made.

As already observed, in \acp{OPT} the rule for parallel composition is normally not given by the standard tensor product $\standardTensorProdC$. Therefore, in general it holds that 
$\standardTensorProduct{\StR{A}}{\StR{B}} \subseteq \StR{AB}$ (and similarly, $\standardTensorProduct{\EffR{A}}{\EffR{B}} \subseteq \EffR{AB}$).  

\subsubsection{Generalised instruments}
\label{subsec:genInstr}
The analysis of the linear structure of instruments, which is relevant for the present purposes, is less straightforward.

Instruments are a family of finite ordered $n$-tuples of transformations with constraints for the compatibility of the general structure of the theory. For example, the full coarse-graining of an instrument must be a deterministic transformation, or the sequential composition of two instruments must be another instrument. Indeed, within the framework of \acp{OPT} the most elementary concept is that of instrument, and not that of transformation: the set of instruments is defined first and the set of transformations follows. It is always true that 

\begin{equation*}
	\Instr{A}{B} \in \powerSetOrder{\Transf{A}{B}}
\end{equation*}

for all $\system{A}, \system{B} \in \Sys{\OPTMath{}}$, where $\powerSetOrder{S}$ represents the set of all ordered subsets of $S$. Therefore, the most natural way to define a sort of generalised space for instruments is:
\begin{equation}
	\InstrRN{A}{B}{N} \mathDef \bigoplus_{n \in N} \TransfR{A}{B},
\end{equation}
where $N$ is the cardinality of the outcome space of the generalised instruments. We recall that by hypothesis the cardinality of the outcome space of instruments is finite, and the operation of direct sum is always well defined. $\InstrRN{A}{B}{N}$ is still a vector space for any $N \in \mathbb{N}$ with the usual scalar multiplication:
\begin{equation*}
	\probabilityEventNoDown{p} \left( \event{T}{1}, \ldots, \event{T}{N} \right) = \left( \probabilityEventNoDown{p} \event{T}{1}, \ldots, \probabilityEventNoDown{p} \event{T}{N} \right),
\end{equation*}
and elementwise sum:
\begin{equation*}
	\left( \event{T}{1}, \ldots, \event{T}{N} \right) + \left( \event{G}{1}, \ldots, \event{G}{N} \right) = \left( \event{T}{1} + \event{G}{1}, \ldots, \event{T}{N} + \event{G}{N} \right),
\end{equation*}
where the scalar multiplication $\probabilityEventNoDown{p} \event{T}{1}$, with $\probabilityEventNoDown{p} \in \TransfRA{\trivialSystem}$, and the sum operation $\event{T}{1} + \event{G}{1}$ are the ones of \TransfR{A}{B}.

The generalised spaces of instruments have the interesting property that $\InstrRN{A}{B}{N}$ can always be seen as a subspace of $\InstrRN{A}{B}{M}$ whenever $N < M$. This comes from the fact that an instrument with $N$ outcomes can always be seen as an instrument with $M$ outcomes of which $M-N$ occur with zero probability, i.e.
\begin{equation}
	\label{eqt:opt:genInstrEmbedd}
	\eventTest{T}{x}{X}\rightarrow \eventTest{T}{x}{X} \bigcup_{n \in (M-N)} \left\{ \nullTransformation{A}{B} \right\},
\end{equation}
where $\eventTest{T}{x}{X} \in \InstrRN{A}{B}{N}$.

The coarse-graining operation on these generalised instruments is simply defined by extension through linearity from the coarse-graining defined on the allowed instruments of a theory. In practice the coarse-graining operation can be seen as a linear function from \InstrRN{A}{B}{N} into \InstrRN{A}{B}{M} with $M < N$.

\subsection{Topological structure and operational completeness}

Here we show how to induce a \emph{topological structure} on the generalised transformation space using the \textdef{operational norm} $\normOp{\cdot}$~\cite{darianoQuantumTheoryFirst2016}. The distance, induced by the norm, is related to the probability of discrimination between two transformations; the closer they are, the harder it is to discriminate between them. The operational norm allows one to express the requirement that the spaces of transformations, states, and effects are Cauchy complete, corresponding to the operational \emph{desideratum} that if there is a procedure to prepare a transformation with arbitrary precision, then it is natural to assume that the latter is a transformation of the theory. 

The operational norm is not the only norm  one can choose for an \ac{OPT}. Another norm of interest is the \textdef{sup norm} $\normSup{\cdot}$, used for example to construction quasi-local algebras of transformations in the \acp{OPT} framework~\cite{perinottiCellularAutomataOperational2020}. In general, these two norms are not equivalent and satisfy different properties. However, in the following discussion we will restrict to the class of \acp{OPT} where these two norms are equivalent, in particular, focusing on \acp{OPT} where all spaces of instruments, and consequently transformations, are finite dimensional.

Before proceeding we just state some properties of the aforementioned norms, since they will be used throughout our argument.

\begin{lemma}[Monotonicity of the operational norm]
	\label{lem:opt:norm:operational:monotonicity}
	Let $\eventNoDown{T} \in \TransfR{B}{C}$, then
	\begin{equation}
		\label{eqt:opt:norm:operational:monotonicity}
		\normOp{\eventNoDown{T}} \geq 	\normOp{\eventNoDown{E}\eventNoDown{T}\eventNoDown{C}},
	\end{equation}	
	where $\eventNoDown{E} \in \TransfN{C}{D}$, and $\eventNoDown{C} \in \TransfN{A}{B}$. The equality holds if both $\eventNoDown{E}$ and $\eventNoDown{C}$ are reversible~\cite{darianoQuantumTheoryFirst2016}.
\end{lemma}

\begin{lemma}[Invariance of the norm in the presence of ancillary systems]
	\label{lem:opt:norm:operational:invariaceAncilas}
	Given a generic \ac{OPT} \OPTMath{}, for arbitrary systems $\system{A}, \system{B}, \system{E} \in \Sys{\OPTMath}$ and any generalised transformation $\eventNoDown{T} \in \TransfR{A}{B}$ it holds that
	\begin{equation*}
		\normOp{\eventNoDown{T}} = \normOp{ \parallelComp{\eventNoDown{T}}{\identityTest{E}} }.
	\end{equation*}
\end{lemma}
\begin{proof}
	From the definition of operational norm~\cite{darianoQuantumTheoryFirst2016}, it holds that
	\begin{equation*}
		\normOp{\eventNoDown{T}} \geq \normOp{ \parallelComp{\eventNoDown{T}}{\identityTest{E}} }.
	\end{equation*}
	On the other hand, by considering a generic deterministic state $\preparationEventNoDown{\rho} \in \StN{E}$ and deterministic effect $\observationUniqueDeterministic\in \EffN{E}$, the monotonicity property of the norm implies that
	\begin{align*}
		\normOp{ \quad
			\myQcircuitBox{
				&\s{A}\qw&\gate{\eventNoDown{T}}&\s{B}\qw&\qw&
				\\
				&\qw&\s{E}\qw&\qw&\qw&
			}
		} & \geq 
		\normOp{ \quad
			\myQcircuitMedium{
				&\s{A}\qw&\gate{\eventNoDown{T}}&\s{B}\qw&\qw&
				\\
				&\prepareC{\eventNoDown{\rho}}&\s{E}\qw&\measureD{\uniDetEff}&
			}
		} \\[10pt] & =
		\normOp{ \quad
			\myQcircuitMedium{
				&\s{A}\qw&\gate{\eventNoDown{T}}&\s{B}\qw&\qw&
			}
		}.
	\end{align*}
\end{proof}

Also the sup norm satisfies the property stated in \autoref{lem:opt:norm:operational:invariaceAncilas}~\cite[Corollary 2]{perinottiCellularAutomataOperational2020}.

\begin{lemma}
	\label{lem:opt:norm:sup:inequality}
	Let $\eventNoDown{T} \in \TransfR{A}{B}$ and $\eventNoDown{G} \in \TransfR{B}{C}$, then~\cite{perinottiCellularAutomataOperational2020}
	\begin{equation*}
		\normSup{\eventNoDown{G}\eventNoDown{T}} \leq \normSup{\eventNoDown{G}} \normSup{\eventNoDown{T}}.
	\end{equation*}
\end{lemma}

\subsubsection{Operational completeness for instruments}
As already stated Cauchy completeness of the spaces of transformations ensures that whenever a transformation can be arbitrarily approximated by operations in an \ac{OPT} then it is also an admitted transformation of the theory. Here we extend this requirement to the case of instruments.

Let us start by considering the scenario of preparations. Suppose that Alice can prepare two deterministic states $\preparationEventNoDown{\rho}$, $\preparationEventNoDown{\sigma} \in \StN{A}$ and that she wants to probabilistically prepare one or the other, so that the resulting mixed state is given by $\probabilityEventNoDown{p} \preparationEventNoDown{\rho} + \left(1- \probabilityEventNoDown{p}\right)\preparationEventNoDown{\sigma} \in \StN{A}$, with $p\in[0,1]$. This scenario can also be expressed through the notion of preparation test $\testList{\probabilityEventNoDown{p} \preparationEventNoDown{\rho}, \left(1- \probabilityEventNoDown{p}\right)\preparationEventNoDown{\sigma} } \in \Prep{A}$ or, equivalently,  $\testList{ \preparationEventNoDown{\rho}', \preparationEventNoDown{\sigma}' } \in \Prep{A}$, where the non-normalized states $\preparationEventNoDown{\rho}' = \probabilityEventNoDown{p}\preparationEventNoDown{\rho}$ and $\preparationEventNoDown{\sigma}' = \left( 1 - \probabilityEventNoDown{p}\right) \preparationEventNoDown{\sigma}$ were used.

If now a second experimenter, let us call him Bob, would like to reproduce Alice's preparation procedure, he would have to approximate the two states $\preparationEventNoDown{\rho}'$, $\preparationEventNoDown{\sigma}'$ with two other states $\preparationEventNoDown{\rho}''$, $\preparationEventNoDown{\sigma}''$, or analogously with a new instrument
$\testList{ \preparationEventNoDown{\rho}'', \preparationEventNoDown{\sigma}'' } \in \Prep{A}$. At this point we wonder what is the error in Bob's approximation. We assume that the error is equal to the sum of the errors made in approximating the two preparation events separately:
\begin{equation*}
	\normGen{ \testList{ \preparationEventNoDown{\rho}', \preparationEventNoDown{\sigma}' } - \testList{ \preparationEventNoDown{\rho}'', \preparationEventNoDown{\sigma}'' } } = \normGen{ \preparationEventNoDown{\rho}' - \preparationEventNoDown{\rho}''} + \normGen{\preparationEventNoDown{\sigma}' - \preparationEventNoDown{\sigma}'' },
\end{equation*}
where we refer now to a generic norm $\normGen{\cdot}$. We observe that this error is greater than the one that would be obtained by considering the coarse-grained state corresponding to the instrument:
\begin{equation*}
	\normGen{ \preparationEventNoDown{\rho}' - \preparationEventNoDown{\rho}''} + \normGen{\preparationEventNoDown{\sigma}' - \preparationEventNoDown{\sigma}'' } \geq \normGen{ \left(\preparationEventNoDown{\rho}' + \preparationEventNoDown{\sigma}'\right) - \left(\preparationEventNoDown{\rho}'' + \preparationEventNoDown{\sigma}''\right) }.
\end{equation*}
This is not unexpected since via coarse-graining one disregards information about what particular event occurred within the instrument, and this removes one source of error.

We can proceed to equip the spaces of instruments of an \ac{OPT} also with a topological structure. The definition will be given by referring to the generic norm $\normGen{\cdot}$, reminding that the cases of interest for the framework are those where the norm is equal to the operational or sup one.

\begin{definition}[Norm of instruments]
	Let \OPTMath\ be an \ac{OPT} equipped with a generic norm $\normGen{\cdot}$ and $\left\{ \event{T}{n} \right\}_{n \in N} \in \InstrRN{A}{B}{N}$. The norm of $\left\{ \event{T}{n} \right\}_{n \in N}$ is defined as
	\begin{equation}
		\label{eqt:OPT:norm:instr}
		\normGen{\left\{ \event{T}{n} \right\}_{n \in N}} \mathDef \sum_{n \in N} \normGen{ \event{T}{n} }.
	\end{equation}
\end{definition}

In the case of instruments, the operational norm in \eqref{eqt:OPT:norm:instr} does not carry the same operational interpretation as in the case of transformations. Even though it gives an indication as to how close two instruments are, it is not related to the optimal discrimination strategy that could be implemented. The latter would require to solve a minimax problem, whose analysis is typically hard, and amenable only under special circumstances is left for future work.

Based on the topological structure we extend the requirement of Cauchy completeness also for instruments.

\section{Properties of Cauchy sequences}
\label{sec:cauchy}

In this section we analyse the properties of Cauchy sequences of instruments and transformations and how they are related within an \ac{OPT} in terms of a generic norm.

As a first step we have to define a Cauchy sequence of instruments $\eventTestSequenceOutAll{T}{x}{X}{n}$. The most formal way to see it is as a sequence of subsets $\eventTestSequenceOutAll{T}{x}{X}{n} \subset \TransfR{A}{B}$. One has to distinguish two cases:
\begin{enumerate}
	\item $\exists N \in \mathbb{N}$ such that $\cardinality{\outcomeSpaceSequence{X}{n}} \leq N, \; \forall n \in \mathbb{N}$. In this case, it is possible to treat $\eventTestSequenceOutAll{T}{x}{X}{n}$ as a sequence of elements in $\InstrRN{A}{B}{N}$, where everything is well-defined. This is the case we will consider in the following. Furthermore, it is possible to drop the dependence of the outcomes space from $n$, since, thanks to \eqref{eqt:opt:genInstrEmbedd}, it is always possible to embed every outcome space in one of cardinality $N$: $\eventTestSequenceAll{T}{x}{X}{n}$.
	\item $\nexists N \in \mathbb{N}$ such that $\cardinality{\outcomeSpaceSequence{X}{n}} \leq N, \; \forall n \in \mathbb{N}$. This case cannot be treated within the hypothesis assumed in the present work. The fact that two comparable instruments must have outcome spaces with equal cardinality imposes that the cardinality of the outcome spaces of the instruments in the sequence stabilizes to a fixed value. We highlight that the latter hypothesis also guarantees that the cardinality of the outcome space of the limit instrument is finite. 
\end{enumerate}

We can now state the completeness result for instruments and a consequent corollary on their coarse-graining.

\begin{theorem}
	\label{thm:OPT:norm:instrEventConvergence}
	Given a generic \ac{OPT}, $\eventTestSequenceAll{T}{x}{X}{n} \subset \InstrRN{A}{B}{N}$ is Cauchy with respect to the norm $\normGen{\cdot}$ if and only if each sequence of transformations $\eventSequenceAll{T}{x}{n} \subset \TransfR{A}{B}$ is Cauchy with respect to that norm $\forall \outcomeIncluded{x}{X}$. Furthermore $\lim_{n \rightarrow \infty} \left\{ \eventTestSequence{T}{x}{X}{n} \right\} = \left\{ \lim_{n \rightarrow \infty} \eventSequence{T}{x}{n} \right\}_{\outcomeIncluded{x}{X}}$.
\end{theorem}

\begin{proof}
	Let us start by proving the first statement.
	\proofRight Let $\eventTestSequenceAll{T}{x}{X}{n}$ be a Cauchy sequence, then $\forall \varepsilon$ there exist $\tilde{n}$ such that $\forall n,m \geq \tilde{n}$ it holds that
	\begin{align*}
		 \normGen{ \eventTestSequence{T}{x}{X}{n} - \eventTestSequence{T}{x}{X}{m}} = \sum_{\outcomeIncluded{x}{X}} \normGen{ \eventSequence{T}{x}{n} - \eventSequence{T}{x}{m} } \leq \varepsilon
	\end{align*}
	which implies
	\begin{equation*}
		\normGen{ \eventSequence{T}{x}{n} - \eventSequence{T}{x}{m} } \leq \varepsilon \quad \forall \outcomeIncluded{x}{X}.
	\end{equation*}	
	\proofLeft On the other hand if the sequences of transformations are Cauchy, i.e., $\forall \varepsilon$ exists $\tilde{n}$ such that $\forall n,m \geq \tilde{n}$ it holds that
	\begin{equation*}
		\normGen{ \eventSequence{T}{x}{n} - \eventSequence{T}{x}{m} } \leq \varepsilon \quad \forall \outcomeIncluded{x}{X},
	\end{equation*}
	then it immediately follows
	\begin{align*}
		\normGen{ \eventTestSequence{T}{x}{X}{n} - \eventTestSequence{T}{x}{X}{m}} = \sum_{\outcomeIncluded{x}{X}} \normGen{ \eventSequence{T}{x}{n} - \eventSequence{T}{x}{m} } \leq \cardinality{\outcomeSpace{X}} \varepsilon,
	\end{align*}
	which, being the cardinality of \outcomeSpace{X} $\left(\cardinality{\outcomeSpace{X}}\right)$ equal to $N < \infty$, implies that $\eventTestSequenceAll{T}{x}{X}{n}$ is Cauchy.	
	We then conclude the proof by showing that the limit of the Cauchy sequence of the collection is given by the collection of the limits. Let indeed $\eventTest{T}{x}{X}$ be the limit of the Cauchy sequence $\eventTestSequenceAll{T}{x}{X}{n}$---that exists by virtue of completeness of instruments. Then $\forall \varepsilon, \exists \tilde{n}$ such that $\forall n \geq \tilde{n}$
	\begin{align*}
		\normGen{ \eventTestSequence{T}{x}{X}{n} - \eventTest{T}{x}{X} }= \sum_{\outcomeIncluded{x}{X}} \normGen{ \eventSequence{T}{x}{n} - \event{T}{x} } \leq \varepsilon,
	\end{align*}
	which, by the completeness requirement for transformations, implies that each sequence of events $\left\{ \eventSequence{T}{x}{n} \right\}_{n \in \mathbb{N}}$ converges to $\event{T}{x}$.
\end{proof}

\begin{corollary}
	\label{corol:OPT:norm:instrEventConvergence}
	Given a generic \ac{OPT} and a Cauchy sequence $\eventTestSequenceAll{T}{x}{X}{n} \subset \InstrRN{A}{B}{N}$ with respect to a generic norm $\normGen{\cdot}$, each sequence of coarse-grainings $\left\{ \sum_{\outcomeIncluded{x}{\prim{X}}} \eventSequence{T}{x}{n} \right\}_{n \in \mathbb{N}} \subset \TransfR{A}{B}$ is Cauchy with respect to that norm $\forall \outcomeSpace{\prim{X}} \subseteq \outcomeSpace{X}$. Furthermore $\lim_{n \rightarrow \infty} \left( \sum_{\outcomeIncluded{x}{\prim{X}}} \eventSequence{T}{x}{n} \right) = \sum_{\outcomeIncluded{x}{\prim{X}}} \left( \lim_{n \rightarrow \infty} \eventSequence{T}{x}{n} \right)$.
\end{corollary}

\begin{proof}
	Using \autoref{thm:OPT:norm:instrEventConvergence}, the sequences of transformations $\left\{ \eventSequence{T}{x}{n} \right\}_{n \in \mathbb{N}}$ are Cauchy for each $\outcomeIncluded{x}{X}$. Therefore, $\forall \varepsilon$ exists $\tilde{n}$ such that $\forall n,m \geq \tilde{n}$ it holds that 
	\begin{equation*}
		\normOp{ \eventSequence{T}{x}{n} - \eventSequence{T}{x}{m} } \leq \varepsilon \quad \forall \outcomeIncluded{x}{X},
	\end{equation*}
	then it immediately follows that
	\begin{align*}
		 \normOp{ \sum_{x \in \prim{X}} \left( \eventSequence{T}{x}{n} - \eventSequence{T}{x}{m} \right)}  \leq \sum_{\outcomeIncluded{x}{\prim{X}}} \normOp{ \eventSequence{T}{x}{n} - \eventSequence{T}{x}{m} } \leq \cardinality{\prim{\outcomeSpace{X}}} \varepsilon, 
	\end{align*}
	and being $\cardinality{\prim{\outcomeSpace{X}}}\leq\cardinality{\outcomeSpace{X}} = N$, this implies that $\left\{ \sum_{\outcomeIncluded{x}{\prim{X}}} \eventSequence{T}{x}{n} \right\}_{n \in \mathbb{N}}$ is Cauchy. The second statement immediately follows from the possibility of exchanging limits and summations.
\end{proof}

The next two lemmas consider the limit of Cauchy sequences of deterministic transformations.

\begin{lemma}
	\label{lem:OPT:causal:detTransf:convergence}
	In a generic \ac{OPT} the limit of a Cauchy sequence, in the operational norm, of deterministic transformations is a deterministic transformation. In other words, \TransfN{A}{B} is closed with respect to the operational norm for any $\system{A}, \system{B} \in \Sys{\OPTMath}$.
\end{lemma}

\begin{proof}
	The first step is to prove that the spaces of deterministic states are 
	closed with respect to the operational norm.
	To show this it is sufficient to show that the limit $\preparationEventNoDown{\rho} = \lim_{n \rightarrow \infty} \preparationEventNoDownSequence{\rho}{n}$, where $\preparationEventNoDownSequenceAll{\rho}{n} \subset \StN{A}$ gives probability equal to one when evaluated on any deterministic effect $\observationUniqueDeterministic \in \EffN{A}$ (this is a necessary and sufficient condition for a state to be deterministic~\cite{darianoClassicalityLocalDiscriminability2020}).
	From the monotonicity of the operational norm (\autoref{lem:opt:norm:operational:monotonicity}) it follows that
	\begin{equation*}
		\normOp{\rbraketSystem{\observationUniqueDeterministic}{\preparationEventNoDown{\rho}}{A} - \rbraketSystem{\observationUniqueDeterministic}{\preparationEventNoDownSequence{\rho}{n}}{A}} \leq \normOp{\preparationEventNoDown{\rho} - \preparationEventNoDownSequence{\rho}{n}}.
	\end{equation*}
	The latter relation then implies that 
	\begin{equation*}
		\lim_{n \to \infty} \rbraketSystem{\observationUniqueDeterministic}{\preparationEventNoDownSequence{\rho}{n}}{A}  = \rbraketSystem{\observationUniqueDeterministic}{\preparationEventNoDown{\rho}}{A} \quad \forall \observationUniqueDeterministic \in \EffN{A}.
	\end{equation*}
	Since $\rbraketSystem{\observationUniqueDeterministic}{\preparationEventNoDownSequence{\rho}{n}}{A} = 1$ for all $n \in \mathbb{N}$, it immediately follows $\rbraketSystem{\observationUniqueDeterministic}{\preparationEventNoDown{\rho}}{A} = 1$ for all $\observationUniqueDeterministic \in \EffN{A}$. 
	
	To prove the result in the case of transformations, we will use the fact that a transformation $\eventNoDown{T} \in \Transf{A}{B}$ is deterministic if and only if $\left( \parallelComp{\eventNoDown{T}}{\identityTest{E}} \right)\rketSystem{\preparationEventNoDown{\rho}}{AE} \in \StN{BE}$ for any $\eventNoDown{\rho} \in \StN{AE}$, where \system{E} is a generic system of the theory.
	Then, let $\eventNoDown{T} = \lim_{n \rightarrow \infty} \eventSequenceNoDown{T}{n}$, with $\eventSequenceNoDownAll{T}{n} \subset \TransfN{A}{B}$. From \autoref{lem:opt:norm:operational:monotonicity} and \autoref{lem:opt:norm:operational:invariaceAncilas}, it immediately follows that
	$\left\{ \left(\parallelComp{\eventSequenceNoDown{T}{n}}{\identityTest{E}}\right)\rketSystem{\preparationEventNoDown{\rho}}{AE} \right\}_{n \in \mathbb{N}} \subset \StN{BE}$, for any system $\system{E}$ of the theory and any deterministic state $\preparationEventNoDown{\rho} \in \StN{AE}$, is still a Cauchy sequence that converges to 
	\begin{equation*}
		\left( \parallelComp{\eventNoDown{T}}{\identityTest{E}}\right)\rketSystem{\preparationEventNoDown{\rho}}{AE} = \lim_{n \rightarrow \infty} \left( \parallelComp{\eventSequenceNoDown{T}{n}}{\identityTest{E}} \right)\rketSystem{\preparationEventNoDown{\rho}}{AE} \in \StN{BE}.
	\end{equation*}
	In conclusion of this proof, we observe that it would have been possible to carry out the second part also referring to effects. A transformation $\eventNoDown{T} \in \Transf{A}{B}$ is deterministic if and only if $\rbraSystem{\observationUniqueDeterministic}{BE}\left( \parallelComp{\eventNoDown{T}}{\identityTest{E}} \right)\in \EffN{BE}$ for any $\observationUniqueDeterministic \in \EffN{BE}$, where \system{E} is a generic system of the theory.
\end{proof}

\begin{lemma}
	\label{lem:OPT:sequence:deterministic:factorized}
	In a generic \ac{OPT} the limit of a Cauchy sequence, in the operational norm, of deterministic product transformations $\left\{ \parallelComp{\eventSequenceNoDown{T}{n}}{\eventSequenceNoDown{G}{n}} \right\}_{n \in \mathbb{N}} \subset \parallelComp{\TransfN{A}{B}}{\TransfN{C}{D}}$ is still a deterministic product transformation. In other words, the set of deterministic product transformations is closed with respect to the operational norm.
\end{lemma}

\begin{proof}
	By \autoref{lem:OPT:causal:detTransf:convergence}, the sequence converges to a deterministic transformation. Therefore, it is sufficient to show that the limit is a product. The first step is to show that also the sequences $\eventSequenceNoDownAll{T}{n} \subset \TransfR{A}{B}$ and	$\eventSequenceNoDownAll{G}{n} \subset \TransfN{C}{D}$ are Cauchy. This can be done by direct calculation:
	\begin{align*}
		&\normOp{ \parallelComp{\eventSequenceNoDown{T}{n}}{\eventSequenceNoDown{G}{n}} - \parallelComp{\eventSequenceNoDown{T}{m}}{\eventSequenceNoDown{G}{m}} } \\ 
		&\qquad \geq  \normOp{ \parallelComp{\eventSequenceNoDown{T}{n}}{\rbraSystem{\observationUniqueDeterministic}{D}} \eventSequenceNoDown{G}{n} \rketSystem{\preparationEventNoDown{\rho}}{C} - \parallelComp{\eventSequenceNoDown{T}{m}}{\rbraSystem{\observationUniqueDeterministic}{D}} \eventSequenceNoDown{G}{m} \rketSystem{\preparationEventNoDown{\rho}}{C} } \\
		&\qquad = \normOp{ \eventSequenceNoDown{T}{n} - \eventSequenceNoDown{T}{m} },
	\end{align*}
	where $\preparationEventNoDown{\rho} \in \StN{C}$ and $\observationUniqueDeterministic \in \EffN{D}$ and the inequality follows from \autoref{lem:opt:norm:operational:monotonicity}, and similarly for the case of \eventSequenceNoDown{G}{n}.
	Defining $\eventNoDown{T} \mathDef \lim_{n \rightarrow \infty} \eventSequenceNoDown{T}{n}$ and $\eventNoDown{G} = \lim_{n \rightarrow \infty} \eventSequenceNoDown{G}{n}$,
and exploiting \autoref{lem:opt:norm:operational:monotonicity}, one can prove that
	\begin{align*}
		&\normOp{ \parallelComp{\eventNoDown{T}}{\eventNoDown{G}} - \parallelComp{\eventSequenceNoDown{T}{n}}{\eventSequenceNoDown{G}{n}} }
		\leq \normOp{ \eventNoDown{T} - \eventSequenceNoDown{T}{n} } + \normOp{ \eventNoDown{G} - \eventSequenceNoDown{G}{n} },
	\end{align*}
	which implies
	\begin{equation*}
		\parallelComp{\eventNoDown{T}}{\eventNoDown{G}} = \lim_{n \rightarrow \infty} \parallelComp{\eventSequenceNoDown{T}{n}}{\eventSequenceNoDown{G}{n}}.\qedhere
	\end{equation*}
	\end{proof}

\subsection{Compositional structure and coarse-graining vs Cauchy completion}
\label{subsec:onCauchyCompleting}
By definition, an \ac{OPT} is closed under sequential and parallel composition, and all the sets of instruments and transformations are closed in the topology induced by the operational norm. However, in the process of \emph{constructing} a theory element by element, one needs to check consistency of the mathematical structures with the axioms of the framework. We then consider whether closure under sequential and parallel composition, as well as under coarse-graining, is preserved after Cauchy-completing tentative spaces of instruments and transformations with respect to the operational norm. The following results provide useful insights in this context.

For simplicity let us denote with $\OPTMath$ the \ac{OPT} obtained from a tentative theory $\tilde{\OPTMath}$ through the operation of completion. Thanks to \autoref{corol:OPT:norm:instrEventConvergence} we know that the spaces of transformations of \OPTMath\ are still closed under the operation of coarse-graining. We have now to check what happens in the case of sequential and parallel composition. 

Let us start with the former. We will present the argument in the case of transformations, but it is immediately extended to the case of instruments through \autoref{thm:OPT:norm:instrEventConvergence}. Let $\eventNoDown{T} \in \Transf{A}{B}$ and $\eventNoDown{G} \in \Transf{B}{C}$ be two transformations of \OPTMath. We have to prove that $\sequentialComp{\eventNoDown{G}}{\eventNoDown{T}}$ is also a transformation of $\OPTMath$. To show this, let us consider $\eventNoDown{T} \in \overline{\Transf{A}{B}}$ as the limit of a Cauchy sequence of transformations $\eventSequenceNoDownAll{T}{n} \subset \Transf{A}{B}$, i.e. $\eventNoDown{T} = \lim_{n \to \infty} \eventSequenceNoDown{T}{n}$, and analogously for $\eventNoDown{G} = \lim_{n \to \infty} \eventSequenceNoDown{G}{n}$. We can then rewrite
\begin{equation*}
	\sequentialComp{\eventNoDown{G}}{\eventNoDown{T}} = \sequentialComp{\lim_{m \to \infty} \eventSequenceNoDown{G}{m} }{\lim_{n \to \infty} \eventSequenceNoDown{T}{n}}.
\end{equation*}
To prove the result, we have then to show that $\sequentialComp{\eventNoDown{G}}{\eventNoDown{T}}$ is the limit of a Cauchy sequence of transformations of $\tilde{\OPTMath}$ by showing that
\begin{equation*}
	\sequentialComp{\eventNoDown{G}}{\eventNoDown{T}} = \lim_{m \to \infty} \lim_{n \to \infty}  \sequentialComp{\eventSequenceNoDown{G}{m}}{\eventSequenceNoDown{T}{n}}.
\end{equation*}
To this end, we will exploit the fact that the operational and sup norm are equivalent, which implies that
\begin{equation}
	\label{opt:strongClosure:usefProp}
	\begin{aligned}
		\normOp{\eventNoDown{G}\eventNoDown{T}} \leq C\normSup{\eventNoDown{G}\eventNoDown{T}} &\leq C \normSup{\eventNoDown{G}} \normSup{\eventNoDown{T}} \\ & \leq \frac{C}{c^2} \normOp{\eventNoDown{G}}\normOp{\eventNoDown{T}},
	\end{aligned}
\end{equation}
where $c$ and $C$ are suitable positive real constants, and the second inequality is obtained by exploiting \autoref{lem:opt:norm:sup:inequality}. Using the latter equation it then follows that:
\begin{align*}
	&\normOp{ \eventSequenceNoDown{G}{m}\eventSequenceNoDown{T}{n} - \eventNoDown{G}\eventNoDown{T}} \\
	&= \normOp{ \eventSequenceNoDown{G}{m}\eventSequenceNoDown{T}{n} - \eventSequenceNoDown{G}{m}\eventNoDown{T} + \eventSequenceNoDown{G}{m}\eventNoDown{T} - \eventNoDown{G}\eventNoDown{T} }\\
	&\leq \normOp{ \eventSequenceNoDown{G}{m} \left( \eventSequenceNoDown{T}{n} - \eventNoDown{T} \right)} + \normOp{ \left( \eventSequenceNoDown{G}{m} - \eventNoDown{G} \right)\eventNoDown{T} }\\
	&\leq \frac{C}{c^2} \normOp{ \eventSequenceNoDown{G}{m} } \normOp{ \eventSequenceNoDown{T}{n} - \eventNoDown{T} } +
	 \frac{C}{c^2} \normOp{ \eventSequenceNoDown{G}{m} - \eventNoDown{G} } \normOp{ \eventNoDown{T} }\\
	&\leq \frac{C}{c^2} \left( \normOp{ \eventSequenceNoDown{T}{n} - \eventNoDown{T} } + \normOp{ \eventSequenceNoDown{G}{m} - \eventNoDown{G} } \right).
\end{align*}
Given that both $\left\{\eventSequenceNoDown{G}{m}\right\}_{m \in \mathbb{N}}$ and $\eventSequenceNoDownAll{T}{n}$ are Cauchy, the latter series of inequalities implies that $\left\{\eventSequenceNoDown{G}{m}\eventSequenceNoDown{T}{n}\right\}_{n,m \in \mathbb{N}}$ is Cauchy and converges to $\eventNoDown{G}\eventNoDown{T}$.

In the last inequality we used the fact that the operational norm of a generic transformation $\eventNoDown{T} \in \Transf{A}{B}$, for any couple of systems \system{A}, $\system{B} \in \Sys{\OPTMath}$, is always bounded by one. This can be checked by direct calculation using the explicit formula for the norm provided in Refs.~\cite{darianoQuantumTheoryFirst2016,perinottiCellularAutomataOperational2020}.

With respect to the closure under parallel composition, the result follows considering that two operations composed in parallel can be always seen as the sequential composition of the same transformations composed with the identity. In diagrams:
\begin{equation*}
	\myQcircuit{
		&\s{A}\qw&\gate{\eventNoDown{T}}&\s{B}\qw&\qw&
		\\
		&\s{C}\qw&\gate{\eventNoDown{G}}&\s{D}\qw&\qw&
	} = \quad\!\! \myQcircuit{
		&\s{A}\qw&\gate{\eventNoDown{T}}&\qw&\s{B}\qw&\qw&\qw&
		\\
		&\qw&\s{C}\qw&\qw&\gate{\eventNoDown{G}}&\s{D}\qw&\qw&
	},
\end{equation*}
or equivalently
\begin{equation*}
	\myQcircuit{
		&\s{A}\qw&\gate{\eventNoDown{T}}&\s{B}\qw&\qw&
		\\
		&\s{C}\qw&\gate{\eventNoDown{G}}&\s{D}\qw&\qw&
	} = \quad\!\! \myQcircuit{
		&\qw&\s{A}\qw&\qw&\gate{\eventNoDown{T}}&\s{B}\qw&\qw&
		\\
		&\s{C}\qw&\gate{\eventNoDown{G}}&\qw&\s{D}\qw&\qw&\qw&
	}.
\end{equation*}
Remembering that $\normOp{\parallelComp{\eventNoDown{T}}{\identityTest{E}}} = \normOp{\eventNoDown{T}}$ for any system \system{E} of the theory (\autoref{lem:opt:norm:operational:invariaceAncilas}) one can reduce to the case of sequential composition.

In summary, we proved that completing a tentative \ac{OPT} including limits of Cauchy sequences of the instruments and transformations spaces returns a well defined \ac{OPT} whenever the operational and sup norm are equivalent. The result is thus the following lemma.

\begin{lemma}
	\label{lem:opt:cauchyCompletion}
	Let $\tilde{\OPTMath}$ be a tentative \ac{OPT} whose instruments and transformations spaces are not Cauchy complete. Then, the theory \OPTMath{} obtained by Cauchy completing these spaces is a well defined \ac{OPT} provided that the operational and sup norm are equivalent.
\end{lemma}

The argument we just exposed also shows another important property that \acp{OPT} satisfy whenever the operational and sup norm are equivalent.

\begin{lemma}
	\label{lem:opt:instr:seqComp:limit}
	Let \OPTMath\ be an \ac{OPT} where the operational and sup norms are equivalent. Consider the two Cauchy sequences of instruments $\eventTestSequenceAll{T}{x}{X}{n} \subset \Instr{A}{B}$ and $\eventTestSequenceAll{G}{y}{Y}{n} \subset \Instr{B}{C}$, whose limits are $\eventTest{T}{x}{X} = \lim_{n \to \infty} \eventTestSequence{T}{x}{X}{n}$ and $\eventTest{G}{y}{Y} = \lim_{n \to \infty} \eventTestSequence{G}{y}{Y}{n}$, respectively. Then, the composite limit instrument is given by the composition of the limits
	\begin{equation*}
		\sequentialComp{\eventTest{G}{y}{Y}}{\eventTest{T}{x}{X}} = \lim_{m \to \infty} \lim_{n \to \infty} \sequentialComp{\eventTestSequence{G}{y}{Y}{m}}{\eventTestSequence{T}{x}{X}{n}}.
	\end{equation*}
\end{lemma}

\begin{lemma}
	\label{lem:opt:transf:seqComp:limit}
	Let \OPTMath\ be an \ac{OPT} where the operational and sup norms are equivalent. Consider the two Cauchy sequences of transformations $\eventSequenceNoDownAll{T}{n} \subset \Transf{A}{B}$ and $\eventSequenceNoDownAll{G}{n} 	\subset \Transf{B}{C}$, whose limits are $\eventNoDown{T} = \lim_{n \to \infty} \eventSequenceNoDownAll{T}{n}$ and $\eventNoDown{G} = \lim_{n \to \infty} \eventSequenceNoDownAll{G}{n}$, respectively. Then, the composite limit transformation is given by the composition of the limits
	\begin{equation*}
		\sequentialComp{\eventNoDown{G}}{\eventNoDown{T}} = \lim_{m \to \infty} \lim_{n \to \infty} \sequentialComp{\eventSequenceNoDown{G}{m}}{\eventSequenceNoDown{T}{n}}.
	\end{equation*}
\end{lemma}

These two latter results can be immediately generalised to hold for generalised instruments and transformations also, and in the case of the parallel composition.

\section{Strong completeness}
\label{sec:strongCompleteness}

We prove here how to complete a theory in such a way that it possesses all the conditional operations, or more formally to a theory that is \emph{strongly causal}.

\subsection{Strongly causal operational probabilistic theories}

The idea is to include in the theory every instrument that can originate from conditional operations based on experimental results. Thus we refer to the most general conditional operation that can occur in a theory that is a conditional instrument.

\begin{definition}[Conditional instruments]
	Let \OPTMath{} be an \ac{OPT}, $\test{T}{X} = \eventTest{T}{x}{X} \in \Instr{A}{B}$ a test of the theory, and $\left\{ \conditionedTest{G}{Y}{x} = \conditionedEventTest{G}{y}{Y}{x} \right\}_{\outcomeIncluded{x}{X}} \subset \Instr{B}{C}$ a labelled collection of instruments. We then define the following collection as a \textdef{conditional instrument}
	\begin{equation}
		\label{eqt:OPT:condInstrDef}
		\bigcup_{\outcomeIncluded{x}{X}} \left\{ \sequentialComp{\conditionedEvent{G}{y}{x}}{\event{T}{x}} \right\}_{\outcomeIncluded{y}{Y}} = \left\{ \left\{ \sequentialComp{\conditionedEvent{G}{y}{x}}{\event{T}{x}} \right\}_{\outcomeIncluded{y}{Y}} \right\}_{\outcomeIncluded{x}{X}}.
	\end{equation}
\end{definition}

\begin{remark}
	Although the outcome space of the conditional instrument may depend on \outcome{x}, this dependence can always be eliminated—thanks to \eqref{eqt:opt:genInstrEmbedd}—by considering a common outcome space \outcomeSpace{Y} with the largest cardinality.
\end{remark}

\begin{remark}
	\label{remark:condInstr}
	In general the object in~\eqref{eqt:OPT:condInstrDef} may not belong to $\InstrEmpty{\OPTMath}$, namely it may model an operation that it is not actually implementable in the theory. An example of a theory where not every conditional 
	instrument is an actual instrument of the theory is \ac{MCT} proposed in Ref.~\cite{erbaMeasurementIncompatibilityStrictly2024}.
\end{remark}

The name conditional instruments is derived from the fact that the outcome of the first test \emph{conditions} which test the experimenter implements next.\footnote{{We observe that in the literature, the families of instruments depending on a classical parameter, like $\left\{ \conditionedTest{G}{Y}{x} = \conditionedEventTest{G}{y}{Y}{x} \right\}_{\outcomeIncluded{x}{X}} \subset \Instr{B}{C}$, are also referred to as \textdef{multimeters}~\cite{selbyContextualityIncompatibility2023,selbyAccessibleFragmentsGeneralized2023,bluhmSimulationQuantumMultimeters2024}}.} We can now provide the formal definition of \textdef{strong causality}.
\begin{definition}[Strongly causal \acp{OPT}]
	\label{def:opt:causality:strong}
	An \ac{OPT} \OPTMath{} is \textdef{strongly causal} if every operation of the form \eqref{eqt:OPT:condInstrDef} belongs to $\InstrEmpty{\OPTMath}$, that is if every conditional instrument is an instrument of the theory.
\end{definition}

In the literature on quantum theory the property of strong causality is also referred to as \textdef{classical control on outcomes} and \textdef{post-processing}~\cite{leppajarviPostprocessingQuantumInstruments2021}.
The name strong causality comes from the fact, that this property is stronger than the causality property usually considered within the framework:

\begin{definition}[Causal \acp{OPT}]
	A \textdef{causal \ac{OPT}} \OPTMath{} is a theory where every system $\system{A} \in \Sys{\OPTMath}$ admits a unique deterministic effect. In symbols, $\forall \system{A} \in \Sys{\OPTMath} \; \exists ! \; \observationUniqueDeterministic \in \EffN{A}$~\cite{darianoQuantumTheoryFirst2016}.
\end{definition}
The causality condition can be proven to be equivalent to the property of \emph{no-signalling from the future}, that is that the probability distributions of preparation-instruments does not depend on the choice of the observation-instruments at their output~\cite{darianoQuantumTheoryFirst2016}.
In causal theories the flow of information is fixed from input to output (diagrammatically, from left to right): the input-output direction can then be identified  with the direction from past to future, i.e., the \emph{arrow of time}. Furthermore, the causality condition also implies the property of \emph{no-signalling without interaction}, namely, the fact that separated parties cannot influence each other if they do not interact. The causal structure induced by causality among parties in a network coincides with the so-called \emph{Einstein locality}~\cite{chiribellaProbabilisticTheoriesPurification2010,darianoQuantumTheoryFirst2016}.

\subsection{Completeness with respect to strong causality}

In the following theorem we show how to complete a causal theory under the strong causality assumption (remind that causality is a necessary, but not sufficient condition for strong causality):

\begin{theorem}[Completion under strong causality]
	\label{thm:OPT:stronglyClosure}
	Let \OPTMath\ be a causal \ac{OPT}. If one constructs a new theory $\tilde{\OPTMath}$ by adding all the conditional instruments and subsequently completing the spaces of transformations and instruments endowed with the operational norm, then $\tilde{\OPTMath}$ is a strongly causal \ac{OPT}.
\end{theorem}

The proof will be constructive, and in particular it consists in the following extension procedure, along with a consistency check.

\begin{procedure}[How to make an \ac{OPT} strongly causal]
	\label{procedure:OPT:stronglyClosure}
	The algorithm to make an \ac{OPT} strongly causal consists of three different steps:
	\begin{enumerate}[I)]
		\item Addition of all the functions obtained from conditioning a family of instruments a finite number of times. More precisely, one has to add all instruments of the form
		\begin{equation}
			\label{eqt:procedure:OPT:stronglyClosure:1}
			\begin{aligned}
				&\myQcircuit{
					&\s{A}\qw&\gate{\eventTest{T}{x}{X}}&\s{B}\qw&\gate{\conditionedEventTest{G}{y}{Y}{x}}&\s{C}\qw&\qw&\cdots
				}\\[10pt] 
				& \myQcircuit{
					&\cdots&&\s{D}\qw&\gate{\conditionedEventTest{H}{z}{Z}{x,y,\ldots}}&\s{E}\qw&\qw&
				}
			\end{aligned}
		\end{equation}
		where $\outcome{\left( x,y,\ldots \right)}$ indicates the finite space of the outcomes of the instruments that precede $\conditionedEventTest{H}{z}{Z}{x,y,\ldots}$. 
		
		\item Closure of the new instruments and transformations spaces under sequential and parallel composition, and coarse-graining,
		\item Cauchy completion of the spaces of instruments and transformations with respect to the operational norm.
	\end{enumerate}
\end{procedure}

\begin{remark}
	\label{remark:strongClosure:finiteOutcome}	
	Coherently with the requirements stated at the beginning of \autoref{sec:cauchy}, the procedure described above does not lead to include in $\tilde{\OPTMath}$  instruments obtained by conditioning an infinite number of times. Otherwise, it could lead to the definition of sequences of instruments where the cardinality of the outcome space is not bounded.
\end{remark}

\begin{remark}	
	\label{remark:strongClosure:conditioningInception}	
	The reason why in the first step of \autoref{procedure:OPT:stronglyClosure} one is adding all instruments obtained through a finite number of conditioning steps, instead of just the instruments of the form \eqref{eqt:OPT:condInstrDef}---i.e.,with just one conditioning step---as would be required by the definition of strongly causal \acp{OPT} (\autoref{def:opt:causality:strong}) is to avoid recursion in the procedure. Indeed, if one just added the instruments with one conditioning step, then the instruments obtained through conditioning on conditional instruments would be required to be added in a subsequent step, and so on. Instead, the requirement~\eqref{eqt:procedure:OPT:stronglyClosure:1} already encompasses all instruments that would be obtained through the iteration procedure.	To illustrate this we consider as an example conditioning of conditional instruments. The argument can then be straightforwardly generalised. Let
	\begin{equation*}
		\left\{ \sequentialComp{\conditionedEvent{G}{y}{x}}{\event{T}{x}} \right\}_{\outcomeIncludedDouble{x}{y}{X}{Y}},
	\end{equation*}
	be an instrument in $\Instr{A}{E}$, where 
	\begin{equation*}
		\begin{aligned}
			\eventTest{T}{x}{X} &= \left\{ \event{T}{x' , x''} \right\}_{\outcomeIncludedDouble{x'}{x''}{X'}{X''}} \\
			&= \left\{ \sequentialComp{ { \eventNoDown{T}'_{\outcome{x}''} }^{\left( \outcome{x}' \right)} }{ \eventNoDown{T}'_{\outcome{x}'} } \right\}_{\outcomeIncludedDouble{x'}{x''}{X'}{X''}},
		\end{aligned}
	\end{equation*}
	is an instrument in $\Instr{A}{C}$, and
	\begin{equation*}
		\begin{aligned}
			&\left\{\conditionedEventTest{G}{y}{Y}{x}\right\}_{\outcomeIncluded{x}{X}}\\
			&= \left\{\left\{ \event{G}{y' ,y''} \right\}^{\left(\outcome{x' , x''}\right)} _{\outcomeIncludedDouble{y'}{y''}{Y'}{Y''}}\right\}_{\outcomeIncludedDouble{x'}{x''}{X'}{X''}} \\ 
			&= \left\{ \left\{ \sequentialComp{ { \eventNoDown{G}'_{\outcome{y}''} }^{\left( \outcome{y}' \right)} }{ \eventNoDown{G}'_{\outcome{y}'} } \right\}^{\left(\outcome{x' , x''}\right)} _{\outcomeIncludedDouble{y'}{y''}{Y'}{Y''}}\right\}_{\outcomeIncludedDouble{x'}{x''}{X'}{X''}} \\ 
			&= \left\{ \sequentialComp{ {\eventNoDown{G}''_{\outcome{y''}}}^{\left( \outcome{y',x',x''} \right)} }{ { \eventNoDown{G}'_{y'} }^{\left( \outcome{x',x''} \right)} } \right\}_{\left(\outcome{x',x'',y',y''}\right) \in \cartesianProduct{\outcomeSpace{X'}}{\cartesianProduct{\outcomeSpace{X''}}{\cartesianProduct{Y'}{Y''}}}},
		\end{aligned}
	\end{equation*}
	is an instrument in $\Instr{C}{E}$. By composing the two instruments one obtains
	\begin{equation*}
		\left\{ \sequentialComp{ {\eventNoDown{G}''_{\outcome{y''}}}^{\left( \outcome{y',x',x''} \right)} }{\sequentialComp{ { \eventNoDown{G}'_{y'} }^{\left( \outcome{x',x''} \right)} }{ \sequentialComp{ { \eventNoDown{T}'_{\outcome{x}''} }^{\left( \outcome{x}' \right)} }{ \eventNoDown{T}'_{\outcome{x}'} } }} \right\}_{\outcomeIncluded{z}{Z}},
	\end{equation*}
	where $\outcomeIncluded{z}{Z}$ was used in place of $\left(\outcome{x',x'',y',y''}\right) \in \cartesianProduct{\outcomeSpace{X'}}{\cartesianProduct{\outcomeSpace{X''}}{\cartesianProduct{Y'}{Y''}}}$, which is exactly of the form \eqref{eqt:procedure:OPT:stronglyClosure:1}.
\end{remark}

We are now in a position to prove \autoref{thm:OPT:stronglyClosure}.
\begin{proof}
	One has to show that $\tilde{\OPTMath{}}$, obtained through \autoref{procedure:OPT:stronglyClosure}, is a well defined \ac{OPT}. This amounts to showing that after the operation of Cauchy completion all the conditional instruments are actual instruments of the theory and that this new collection of instruments is closed under sequential and parallel composition, and coarse-graining~\cite{perinottiCellularAutomataOperational2020,darianoClassicalTheoriesEntanglement2020}. Closure under composition and coarse-graining is guaranteed by \autoref{lem:opt:instr:seqComp:limit}, \autoref{lem:opt:transf:seqComp:limit}, and \autoref{corol:OPT:norm:instrEventConvergence}. Therefore, the first requirement is the only one that needs to be checked explicitly. We then show that, after adding all the instruments obtained through a finite number of conditioning steps to \OPTMath\ and subsequently completing with respect of the operational norm, for every instrument $\test{T}{X} \equiv \eventTest{T}{x}{X} \in \Instr{A}{B}$ and for every family $\conditionedTest{G}{Y}{x} \equiv \conditionedEventTest{G}{y}{Y}{x} \subset \Instr{B}{C}$, for any system \system{A}, \system{B}, $\system{C} \in \Sys{\tilde{\OPTMath}}$,  the conditional set of operations
	\begin{equation*}
		\left\{ \left\{ \sequentialComp{\conditionedEvent{G}{y}{x}}{\event{T}{x}} \right\}_{\outcomeIncluded{y}{Y}} \right\}_{\outcomeIncluded{x}{X}} 
	\end{equation*}
	is a test of the theory. 
	
	In this case, it is sufficient to study what happens for instruments with a single conditional step since one is only interested in proving that $\tilde{\OPTMath{}}$ satisfies \autoref{def:opt:causality:strong}. Let us suppose that there exists a particular family of tests for which this does not hold. Since we already added to $\tilde{\OPTMath{}}$ all the conditional instruments that it is possible to construct using the tests of \OPTMath, the only cases we have to consider are the ones where at least one between $\test{T}{X}$ or $\left\{ \conditionedTest{G}{Y}{x} \right\}_{\outcomeIncluded{x}{X}}$ is an instrument obtained as a limit from the procedure of Cauchy completion. Given that an instrument can always be seen as the limit of a constant sequence, we will treat both $\test{T}{X}$ and $\left\{ \conditionedTest{G}{Y}{x} \right\}_{\outcomeIncluded{x}{X}}$ as limits. Thanks to \autoref{thm:OPT:norm:instrEventConvergence} along with \autoref{lem:opt:instr:seqComp:limit} and \autoref{lem:opt:transf:seqComp:limit} we have that
	\begin{align*}
		\sequentialComp{\conditionedEvent{G}{y}{x}}{\event{T}{x}} &= \sequentialComp{\left( \lim_{m \rightarrow \infty} \conditionedEventSequence{G}{y}{x}{m} \right)}{\left( \lim_{n \rightarrow \infty} \eventSequence{T}{x}{n} \right)}\\
		&= \lim_{m \rightarrow \infty}  \lim_{n \rightarrow \infty}  \sequentialComp{\conditionedEventSequence{G}{y}{x}{m}}{\eventSequence{T}{x}{n}}.
	\end{align*}
	Hence, $\sequentialComp{\conditionedEvent{G}{y}{x}}{\event{T}{x}}$ is the limit of a sequence of transformations of $\tilde{\OPTMath{}}$, and by the requirement of Cauchy completeness of the spaces of transformations one obtains the result, i.e.
	\begin{equation*}
		\left\{ \left\{ \sequentialComp{\conditionedEvent{G}{y}{x}}{\event{T}{x}} \right\}_{\outcomeIncluded{y}{Y}} \right\}_{\outcomeIncluded{x}{X}} \in \Instr{A}{C}
	\end{equation*}
	is an actual instrument of $\tilde{\OPTMath{}}$.
\end{proof}

\section{Broadcasting, compatibility and disturbance}
\label{sec:properties}
In this section we introduce the notions of \textdef{broadcasting}, \textdef{compatibility of observations}, and \textdef{irreversibility of measurement disturbance}. 

\subsection{Broadcasting}
Let us start from the definition of broadcasting.
\begin{definition}[Broadcasting transformation]
	\label{def:opt:bradcasting}
	Considering a causal \ac{OPT}, a transformation $\eventNoDown{B} \in \TransfN{A}{AA}$ is \textdef{broadcasting} if: 
	\begin{equation}
		\label{eqt:opt:transf:broadcasting}
		\myQcircuit{
			&\s{A}\qw&\multigate{1}{\eventNoDown{B}}&\s{A}\qw&\measureD{\observationUniqueDeterministic}&
			\\
			&\pureghost{}&\pureghost{\eventNoDown{B}}&\s{A}\qw&\qw&\qw&
		} = \;
		\myQcircuit{
			&\s{A}\qw&\multigate{1}{\eventNoDown{B}}&\s{A}\qw&\qw&\qw&
			\\
			&\pureghost{}&\pureghost{\eventNoDown{B}}&\s{A}\qw&\measureD{\observationUniqueDeterministic}&
		} = \;
		\myQcircuit{
			&\s{A}\qw&\qw&
		}.
	\end{equation}
\end{definition}

Notice that, following the above definition, a transformation $\eventNoDown{B} \in \TransfN{A}{AA}$ is broadcasting if and only if
\begin{align*}
	\myQcircuit{
		&\multiprepareC{1}{\preparationEventNoDown{\rho}}&\s{A}\qw&\qw&
		\\
		&\pureghost{\preparationEventNoDown{\rho}}&\s{E}\qw&\qw&
	} &= 
	\myQcircuit{
		&\multiprepareC{2}{\preparationEventNoDown{\rho}}&\s{A}\qw&\multigate{1}{\eventNoDown{B}}&\s{A}\qw&\measureD{\observationUniqueDeterministic}&
		\\
		&\pureghost{\preparationEventNoDown{\rho}}&\pureghost{}&\pureghost{\eventNoDown{B}}&\s{A}\qw&\qw&\qw&
		\\
		&\pureghost{\preparationEventNoDown{\rho}}&\qw&\qw&\s{E}\qw&\qw&\qw&
	}\\[10pt]
	&= 
	\myQcircuit{
		&\multiprepareC{2}{\preparationEventNoDown{\rho}}&\s{A}\qw&\multigate{1}{\eventNoDown{B}}&\s{A}\qw&\qw&\qw&
		\\
		&\pureghost{\preparationEventNoDown{\rho}}&\pureghost{}&\pureghost{\eventNoDown{B}}&\s{A}\qw&\measureD{\observationUniqueDeterministic}&
		\\
		&\pureghost{\preparationEventNoDown{\rho}}&\qw&\qw&\s{E}\qw&\qw&\qw&
	},
\end{align*}
for every system $\system{E}$ of the theory and for every state $\preparationEventNoDown{\rho} \in \St{AE}$.

While \ac{QT} does not allow for a broadcasting transformation, \ac{CT} does~\cite{woottersSingleQuantumCannot1982,dieksCommunicationEPRDevices1982,yuenAmplificationQuantumStates1986,barnumNoncommutingMixedStates1996,daffertshoferClassicalNoCloningTheorem2002,barnumCloningBroadcastingGeneric2006,walkerClassicalBroadcastingPossible2007,barnumGeneralizedNoBroadcastingTheorem2007,pianiNoLocalBroadcastingTheoremMultipartite2008,luoQuantumNoBroadcasting2010} and the classical broadcasting map is given by:
\begin{equation}
	\label{eqt:opt:transf:broadcast:classical}
	\sum_{i = 1}^{\sysDimension{A}} \quad \myQcircuit{
		&\s{A}\qw&\measureD{\observationEventNoDown{i}}&\prepareC{\preparationEventNoDown{i}}&\s{A}\qw&\qw&
		\\
		&\pureghost{}&\pureghost{\observationEventNoDown{i}}&\prepareC{\preparationEventNoDown{i}}&\s{A}\qw&\qw&
              },
\end{equation}
where $\preparationEventNoDown{i} \in \PurSt{A}$ are the pure states of system \system{A} and $\observationEventNoDown{i} \in \Eff{A}$ are the measurements that jointly perfectly discriminate the pure states.\footnote{A set of states $\left\{ \preparationEventNoDown{\rho}_{n} \right\}_{n \in N}$ is \textdef{jointly perfectly discriminable} if there exists a test $\observationEventTest{a}{x}{\mathit N} \in \Obs{A}$ such that $\rbraket{\observationEvent{a}{x}}{\preparationEventNoDown{\rho}_{x}} = 1$ for all $\outcomeIncluded{x}{\mathit N}$.}

We will say that a theory satisfies the property of \textdef{broadcasting} if every system of the theory admits a broadcasting transformation. Otherwise, we will say that the theory satisfies the \textdef{no-broadcasting} theorem.

We now prove a sufficient condition for a system of an \ac{OPT} to satisfy no-broadcasting.
\begin{lemma}
	\label{lem:opt:broadcasting:idAtomic}
	Consider a causal \ac{OPT} and a system \system{A} of the theory of dimension greater than 1, $\sysDimension{A} \geq 2$. If $\identityTest{A}$ is atomic, then the system \system{A} does not admit a broadcasting channel $\eventNoDown{B} \in \TransfN{A}{AA}$.
\end{lemma}

\begin{proof}
  Suppose by contradiction that there exists a system \system{A}, $\sysDimension{A} \geq 2$, such that \identityTest{A} is atomic and there exists a broadcasting channel $\eventNoDown{B} \in \TransfN{A}{AA}$. Consider now a non-trivial decomposition $\left\{ \observationEvent{a}{0}, \observationEvent{a}{1} \right\} \neq \left\{ \probabilityEvent{p}{0} \observationUniqueDeterministic, \probabilityEvent{p}{1}\observationUniqueDeterministic \right\} \in \Obs{A}$ of the deterministic effect $\observationUniqueDeterministic \in \EffN{A}$. An observation-test of this kind always exists, since otherwise the system would have dimension equal to 1.	By \eqref{eqt:opt:transf:broadcasting}, it follows that:
	\begin{equation*}
		\myQcircuit{
			&\s{A}\qw&\multigate{1}{\eventNoDown{B}}&\s{A}\qw&\measureD{\observationEvent{a}{0}}&
			\\
			&\pureghost{}&\pureghost{\eventNoDown{B}}&\s{A}\qw&\qw&\qw&
		} + \;
		\myQcircuit{
			&\s{A}\qw&\multigate{1}{\eventNoDown{B}}&\s{A}\qw&\measureD{\observationEvent{a}{1}}&
			\\
			&\pureghost{}&\pureghost{\eventNoDown{B}}&\s{A}\qw&\qw&\qw&
		} = \;		\myQcircuit{
			&\s{A}\qw&\qw&
		},
	\end{equation*}
	which, by the atomicity of the identity transformation, implies:
	\begin{equation*}
		\myQcircuit{
			&\s{A}\qw&\multigate{1}{\eventNoDown{B}}&\s{A}\qw&\measureD{\observationEvent{a}{0}}&
			\\
			&\pureghost{}&\pureghost{\eventNoDown{B}}&\s{A}\qw&\qw&\qw&
		} \propto
		\myQcircuit{
			&\s{A}\qw&\multigate{1}{\eventNoDown{B}}&\s{A}\qw&\measureD{\observationEvent{a}{1}}&
			\\
			&\pureghost{}&\pureghost{\eventNoDown{B}}&\s{A}\qw&\qw&\qw&
		} \propto
		\myQcircuit{
			&\s{A}\qw&\qw&
		}.
	\end{equation*}
	Exploiting again \eqref{eqt:opt:transf:broadcasting} one obtains that:
	\begin{equation*}
		\myQcircuit{
			&\s{A}\qw&\multigate{1}{\eventNoDown{B}}&\s{A}\qw&\measureD{\observationEvent{a}{0}}&
			\\
			&\pureghost{}&\pureghost{\eventNoDown{B}}&\s{A}\qw&\measureD{\observationUniqueDeterministic}&
		} =
		\myQcircuit{
			&\s{A}\qw&\measureD{\observationEvent{a}{0}}&
		}
		\propto
		\myQcircuit{
			&\s{A}\qw&\measureD{\observationUniqueDeterministic}&
		},
	\end{equation*}
	and analogously for $\observationEvent{a}{1}$, which contradicts the hypothesis $\sysDimension{A} \geq 2$.
\end{proof}

\subsection{Compatibility and irreversibility}
Two operations are said to be \emph{compatible} if operating one does not preclude the possibility of performing the other on the same system. For example, position and momentum measurements are compatible in \ac{CT} as they can be performed simultaneously, but not in \ac{QT}. In the following,we use ``compatibility'' referring to observation-tests, and we define theories with full compatibility as those theories where all observation-tests are pairwise compatible. On the contrary, we will define \emph{irreversible disturbance} the existence of some operation that prevents from making another one on the same system. In this sense the first operation causes an unavoidable disturbance.

Let us look at observation-instruments and identify the class of theories where all observations are compatible,  as it happens for example in classical theory:
\begin{definition}[\acp{OPT} with full-compatibility of the observation-instruments]
	A causal \ac{OPT} \OPTMath{} is said to satisfy \textdef{full-compatibility of the observation-instruments} if every pair of observation-instruments \observationEventTest{a}{x}{X}, $\observationEventTest{b}{y}{Y} \in \Obs{A}$ of the theory, for every system $\system{A} \in \Sys{\OPTMath}$, are \textdef{compatible}, namely there exists a third test $\observationEventTest{c}{(x,y)}{\outcomeSpaceDouble{X}{Y}} \in \Obs{A}$ such that~\cite{darianoIncompatibilityObservablesChannels2022,erbaMeasurementIncompatibilityStrictly2024}
	\begin{align*}
		&\myQcircuit{
			&\s{A}\qw&\measureD{\observationEvent{a}{x}}& 
		}  = \quad \sum_{\outcomeIncluded{y}{Y}} \myQcircuit{
			&\s{A}\qw&\measureD{\observationEvent{c}{\left(x,y\right)}}&
		} \quad \forall \outcomeIncluded{x}{X}, \\	
		&\myQcircuit{
			&\s{A}\qw&\measureD{\observationEvent{b}{y}}& 
		}  = \quad \sum_{\outcomeIncluded{x}{X}} 
		\myQcircuit{
			&\s{A}\qw&\measureD{\observationEvent{c}{\left(x,y\right)}}&
		}  \quad \forall \outcomeIncluded{y}{Y}.
	\end{align*}
\end{definition}
This definition can be seen as a special instance of the more complex definition involving arbitrary instruments, however in the above form it is a straightforward extension to \acp{OPT} of the definition introduced in the quantum literature~\cite{buschNoInformationDisturbance2009,heinosaariInvitationQuantumIncompatibility2016,heinosaariNofreeinformationPrincipleGeneral2019}

In the following,the shorter nomenclature \textdef{compatibility} will be used in place of \emph{full-compatibility of the observation-instruments}.

The notion of irreversibility, on the other hand, describes processes that cause an irreducible disturbance on the systems on which they act. This is based on the notion of exclusion between instruments. 
\begin{definition}[Does not exclude]
	Let \OPTMath\ be a causal \ac{OPT}, we say that an instrument $\eventTest{T}{x}{X} \in \Instr{A}{B}$ \textdef{does not exclude} another instrument $\eventTest{G}{y}{Y} \in \Instr{A}{C}$ if there exists a test $\eventTest{C}{z}{Z} \in \Instr{A}{BE}$ and a post-processing, i.e.~a family of instruments $\left\{\conditionedEventTest{P}{y}{Y}{z}\right\}_{\outcomeIncluded{z}{Z}} \subset \Instr{BE}{C}$ such that
	\begin{align}
		\label{eqt:doesNotExclude}
		&\myQcircuit{
			&\s{A}\qw&\gate{\event{A}{x}}&\s {B}\qw&\qw&
		} = \sum_{\outcomeIncludedConditioned{z}{S}{x}}
		\myQcircuit{
			&\s{A}\qw&\multigate{1}{\event{C}{z}}&\qw&\s {B}\qw&\qw&
			\\
			&\pureghost{}&\pureghost{\event{C}{z}}&\s{E}\qw&\measureD{\observationUniqueDeterministic}
		},\\[10pt]
		&\myQcircuit{
			&\s{A}\qw&\gate{\event{B}{y}}&\s{C}\qw&\qw&
		} = \sum_{\outcomeIncluded{z}{Z}}
		\myQcircuit{
			&\s{A}\qw&\multigate{1}{\event{C}{z}}&\s {B}\qw&\multigate{1}{\conditionedEvent{P}{y}{z}}&\s{C}\qw&\qw&
			\\
			&\pureghost{}&\pureghost{\event{C}{z}}&\s{E}\qw&\ghost{\conditionedEvent{P}{y}{z}}&\pureghost{}&
		},
	\end{align}
	where $\left\{\outcomeSpaceConditioned{S}{x}\right\}_{\outcomeIncluded{x}{X}}$ is a suitable partition of \outcomeSpace{X}~\cite{darianoIncompatibilityObservablesChannels2022}. On the other hand, if the above condition fails, we say that the instrument $\eventTest{T}{x}{X} $ \textdef{excludes} $\eventTest{G}{y}{Y} $.
\end{definition}

A theory with irreversibility is defined as follows.

\begin{definition}[\acp{OPT} with irreversibility]
	A causal \ac{OPT} \OPTMath{} is said to have \textdef{irreversibility of measurement disturbance} (or in short \textdef{irreversible disturbance}) if it admits an \textdef{intrinsically irreversible} instrument, i.e., a test that excludes some other test of the theory.
\end{definition}

Notice that the kind of disturbance introduced by an intrinsically irreversible instrument is the ``strictest'' possible: if an instrument $\eventTest{T}{x}{X}$ excludes $\eventTest{G}{y}{Y}$, one cannot obtain $\eventTest{G}{y}{Y}$ even using any conceivable resources, namely any kind of post-processing involving arbitrary ancillary systems. An example of theory with irreversibility is \ac{QT}, where almost all quantum instruments are intrinsically irreversible. We remark, however, that since every quantum channel admits a unitary dilation involving ancillary systems, all quantum channels are not intrinsically irreversible~\cite{erbaMeasurementIncompatibilityStrictly2024}. 

On the contrary, \ac{CT} does not have irreversibility. This implies that whatever operation in \ac{CT} can be implemented in such a way to reverse it and restore the initial state of the system.

Two sufficient conditions for irreversibility have been proved in Ref.~\cite{erbaMeasurementIncompatibilityStrictly2024}. The first is given by the following lemma:

\begin{lemma}
	\label{opt:intrIrrev:identity}
	An instrument is intrinsically irreversible if and only if it excludes the identity~\cite{erbaMeasurementIncompatibilityStrictly2024}.
\end{lemma}

Another condition for irreversibility is the atomicity of the identity transformation:\footnote{We remark that the identity transformation is atomic if every instrument whose full coarse-graining is the identity is of the form $\left\{ \probabilityEvent{p}{x} \identityTest{} \right\}_{\outcomeIncluded{x}{X}}$, with \probabilityEventTest{p}{x}{X} a probability distribution (\autoref{def:opt:transf:atomic}).}

\begin{theorem}
	\label{thm:opt:irrev:idAtomicity}
	Let \OPTMath{} be a causal \ac{OPT} with a system $\system{A} \in \Sys{\OPTMath}$ such that $\sysDimension{A} \geq 2$, and let its identity transformation $\identityTest{A}$ be atomic. Then, there exists an instrument $\eventTest{T}{x}{X} \in \Instr{A}{B}$, for some system $\system{B} \in \Sys{\OPTMath}$ that is intrinsically irreversible. Hence, the theory has irreversibility~\cite{darianoInformationDisturbanceOperational2020}.
\end{theorem}

The proof of the above theorem is also reported in \aref{opt:atomicityAndIrreversibility} for completeness.

A particular class of theories that satisfy the hypothesis of \autoref{thm:opt:irrev:idAtomicity} is that of \acp{OPT} where the identity transformation is atomic for every system of the theory, which is equivalent to \textdef{\acf{NIWD}}~\cite{darianoInformationDisturbanceOperational2020}. The last property means that in order to extract non-trivial information from a system one must necessarily disturb it. \ac{QT} is an example of a theory with \ac{NIWD}, which in turn underlies the possibility of devising information-theoretically secure cryptographic protocols, since any intervention from an eavesdropper would be detectable by the other parties~\cite{bennettQuantumCryptographyUsing1992,bennettExperimentalQuantumCryptography1992,barnettInformationtheoreticLimitsQuantum1993,ekertEavesdroppingQuantumcryptographicalSystems1994,kentUnconditionallySecureBit1999,kentUnconditionallySecureBit1999,loUnconditionalSecurityQuantum1999,mayersUnconditionalSecurityQuantum2001,nielsenQuantumComputationQuantum2010,bennettQuantumCryptographyPublic2014}.

Regarding the relation between the notions of incompatibility of measurements and of processes (named irreversibility) it has been proven in Ref.~\cite{erbaMeasurementIncompatibilityStrictly2024} that while
\begin{align*}
	\text{incompatibility} &\Rightarrow \text{irreversibility}, 
\end{align*}
on the other hand
\begin{align*}
	\text{incompatibility} &\nLeftarrow \text{irreversibility}.
\end{align*}
A counterexample to the second implication is \ac{MCT}~\cite{erbaMeasurementIncompatibilityStrictly2024}, which has full-compatibility of the observation-instruments, yet it admits intrinsically irreversible tests.

\section{Minimal Operational Probabilistic Theories}
\label{sec:MOPT}

We will now proceed to introduce and characterise the class of \acfp{MOPT}. The idea is to explore the diametrically opposed situation as to what is typically considered in the probabilistic theories scenario. Indeed the set of operations in a theory is often taken as large as possible, so for example given the set of states, all maps sending legitimate states to legitimate states are included as maps of the theory. However, no bottlenecks to consistent theories with restricted set of instruments and transformations occur in principle. On the contrary, here we look at the ``smallest conceivable set'' of operations, in the sense that the removal of any of the instruments would no longer lead to a legitimate \ac{OPT}.

The first point we need to analyse concerns the structure of composite systems. In a general theory, systems can, in principle, be obtained in multiple ways by composing other systems, except for those that do not admit any decomposition, which we term \textdef{elementary}.

More formally, let us start from the following definition

\begin{definition}[Elementary systems and minimal decomposition]
We say that a system $\system{A}$ is \textdef{elementary} if $\system{A} = \system{BC}$ implies $\system{B} = \trivialSystem$ or $\system{C} = \trivialSystem$. Given a system $\system{S}$, we say that $\system{A} = \system{A}_{1} \system{A}_{2} \ldots \system{A}_{k}$ is a \emph{minimal decomposition} of $\system{S}$ in elementary systems if $\system{A}_{i}$ is elementary and non trivial for every $i$.
\end{definition}

In the following,we require the decomposition in elementary systems to be unique up to operational equivalence of the single elementary systems.

\begin{definition}[Unique decomposition \acsp{OPT}]
	\label{def:opt:system:uniqueDec}	
	We say that an \ac{OPT} has \textdef{unique decomposition} if for every system $\system{S}$, given two minimal decompositions in elementary systems  $\system{S} = \system{A}_{1} \system{A}_{2} \ldots \system{A}_{k}$ and $\system{S} = \system{B}_{1} \system{B}_{2} \ldots \system{B}_{l}$ implies $k = l$ and ${\system{B}_{i}}={\system{A}_{i}}, \forall i = 1,\ldots,k$.
\end{definition}

Finally, we can define minimal theories in the following way:

\begin{definition}[Minimal Operational Probabilistic Theory]
	\label{def:OPT:minimal}
	We define a \textdef{\acf{MOPT}} as an \ac{OPT} with unique decomposition where the only allowed tests are the ones obtainable by composing the elements of
	\begin{equation}
		\label{eqt:OPT:minimal:def:test}
		\left\{		
		\testComplete{I}{\SingletonSet}{A}{A} \; , \; \testComplete{S}{\SingletonSet}{AB}{BA} \; , \; \testComplete{\left( \inverse{S} \right)}{\SingletonSet}{BA}{AB} \; , \; \preparationTest{\rho}{X} \; , \; \observationTest{a}{X}
		\right\},
	\end{equation} 
	where $\preparationTest{\rho}{X} \in \Prep{A}$ and $\observationTest{a}{X} \in \Obs{A}$ are all the possible preparation- and observation-tests of the theory, and the Cauchy completion of the aforementioned set. Thus the only allowed events are those obtainable by sequential and parallel composition of the elements of
	\begin{equation}
		\label{eqt:OPT:minimal:def:event}
		\left\{ 
		\identityTest{A} \; , \; \BraidingS_{\system{A}, \system{B}} \; , \; \inverse{\BraidingS}_{\system{A}, \system{B}} \; , \; \preparationEventNoDown{\rho} \; , \; \observationEventNoDown{a}
		\right\},
	\end{equation}
	for every $\system{A}, \system{B} \in \Sys{\OPTMath}$, $\preparationEventNoDown{\rho} \in \eventCollectionAB{\trivialSystem}{A}$ and $\observationEventNoDown{a} \in \eventCollectionAB{A}{\trivialSystem}$, and the Cauchy completion of spaces of events of this type that belong to a test of the theory.
\end{definition}

\begin{remark}
	\label{remark:opt:minimal:wellDefined}
	The above class is well-defined. One should check that the spaces of instruments and transformations are closed under the operations of sequential and parallel composition, and under the operation of coarse-graining. While the closure under coarse-graining is guaranteed by \autoref{corol:OPT:norm:instrEventConvergence}, the closure under the compositional structure is guaranteed by \autoref{lem:opt:instr:seqComp:limit} and \autoref{lem:opt:transf:seqComp:limit} under the here assumed hypothesis that the operational and sup norm are equivalent.
\end{remark}

Similar restrictions on a theory, in analogy to the ones we consider for \acp{MOPT}, have been proposed in Ref.~\cite{selbyAccessibleFragmentsGeneralized2023} and Ref.~\cite{barrettComputationalLandscapeGeneral2019}. Ref.~\cite{selbyAccessibleFragmentsGeneralized2023} studies \emph{accessible \acs{GPT} fragments} which are designed to describe scenarios where states and effects are limited to just those accessible in a particular experimental setting. Even though the notion of accessible \acs{GPT} fragments is close to the notion of \acp{MOPT}, there is an important difference. \acp{MOPT}, despite the restrictions, are still fully-fledged theories, while accessible \acs{GPT} fragments are not in general \ac{GPT} themselves. One aspect in which these two definitions differ is that of the state and effect spaces of a \ac{MOPT} must be separating from each other, while this may not be true in the case of accessible \ac{GPT} fragments. In Ref.~\cite{barrettComputationalLandscapeGeneral2019} a class of operational theories is introduced, called \emph{non-free}, where restrictions on operations are not followed by closure with respect to parallel and sequential composition.

\subsection{Characterization of instruments and transformations}
As noted in \autoref{remark:opt:minimal:wellDefined}, any \ac{MOPT} is formally well-defined, although the transformations and instruments are provided only implicitly. In this section, we prove structural theorems that specify the circuital realization of the theory's instruments and transformations. Some of these results were previously presented in the supplementary material of Ref.~\cite{erbaMeasurementIncompatibilityStrictly2024}, but they are derived here in a more general framework. The proofs will be included in the appendix.

Let us start with the complete characterization of the instruments and transformations that are obtainable by composing sequentially and in parallel the elements of \eqref{eqt:OPT:minimal:def:test} and \eqref{eqt:OPT:minimal:def:event}, respectively, postponing the analysis of instruments and transformations obtained as limits of Cauchy sequences. 

Then one has the following structural results.

\begin{lemma}
	\label{lem:OPT:minimal:generalInstr}
	In every \ac{MOPT} any instrument $\eventTest{T}{x}{X} \in \Instr{A}{B}$ obtained as parallel and sequential composition of the elements of \eqref{eqt:OPT:minimal:def:test} is of the form:
	\begin{equation}
		\label{eqt:OPT:minimal:instr:general}
		\myQcircuit{
			&\prepareC{\preparationEventTest{\rho}{y}{Y}}&\s{C}\qw&\multigate{1}{\Braid}&\s{D}\qw&\measureD{\observationEventTest{a}{z}{Z}}&
			\\
			&\s{A}\qw&\qw&\ghost{\Braid}&\qw&\s{B}\qw&\qw&
		}, 
	\end{equation}
	where $\Braid \in \RevTransf{AC}{DB}$ is a suitable braid transformation, $\preparationEventTest{\rho}{y}{Y} \in \Prep{C}$, $\observationEventTest{a}{z}{Z} \in \Obs{A}$, the outcome space $\outcomeSpace{X} = \cartesianProduct{\outcomeSpace Y}{\outcomeSpace Z}$, and \system{A}, \system{B}, \system{C}, $\system{D} \in \Sys{\OPTMath}$ may also be equal to the trivial system~\cite{erbaMeasurementIncompatibilityStrictly2024}.
\end{lemma}

Formally braid transformations are defined in the following way:
\begin{definition}[Set of braid transformations] 
	\label{def:permutations}
	The \textdef{set of braid transformations}, whose representatives will be indicated with $\Braid$, is defined as the ensemble of transformations which are obtained by parallel and sequential composition of the braiding and identity transformations.
\end{definition} 

The proof of the above theorem can be found in \aref{app:mopt:instr:characterisation} and from it the analogous result for transformations immediately follows.

\begin{corollary}
	\label{cor:OPT:minimal:generalTransf}
	In every \ac{MOPT} any transformation $\eventNoDown{T} \in \Transf{A}{B}$ obtained as parallel and sequential composition of the elements of \eqref{eqt:OPT:minimal:def:event} is of the form:
		\begin{equation}
		\label{eqt:OPT:minimal:transf:general}
		\myQcircuit{
			&\prepareC{\preparationEventNoDown{\rho}}&\s{C}\qw&\multigate{1}{\Braid}&\s{D}\qw&\measureD{\observationEventNoDown{a}}&
			\\
			&\s{A}\qw&\qw&\ghost{\Braid}&\qw&\s{B}\qw&\qw&
		}, 
	\end{equation}
	where $\Braid \in \RevTransf{AC}{DB}$ is a suitable braid transformation, $\preparationEventNoDown{\rho} \in \St{C}$, $\observationEventNoDown{a} \in \Eff{A}$, and \system{A}, \system{B}, \system{C}, $\system{D} \in \Sys{\OPTMath}$ may also be equal to the trivial system~\cite{erbaMeasurementIncompatibilityStrictly2024}.
\end{corollary}

An important property that will be used throughout the discussion is that \eqref{eqt:OPT:minimal:transf:general} is invariant under parallel and sequential composition~\cite{erbaMeasurementIncompatibilityStrictly2024}.

In the case of symmetric \acp{MOPT}, braid transformations become \emph{permutations}, which allows us to further specialise the characterisation of the form of instruments and transformations, thanks to the following results.

First, let us observe that
\begin{lemma}
	\label{lem:opt:permutations:characterisation}
	Consider a symmetric \ac{OPT} with unique decomposition and let $\Braid \in \RevTransf{A}{B}$ that permutes the systems as 	$\system{B}_{j} = \system{A}_{\sigma\left(i\right)}$, where $\system{A} = \system{A}_{1}\ldots\system{A}_{n}$ and $\system{B} = \system{B}_{1}\ldots\system{B}_{n}$ are the unique decompositions of \system{A} and \system{B}, respectively. Then the action of $\Braid$ is completely characterised by the permutation $\sigma$.
\end{lemma} 

\begin{proof}
	The proof is a straightforward application of the coherence theorem for symmetric monoidal categories, to which symmetric \acp{OPT} belong. The theorem guarantees that the operations obtained through the composition of swap are completely defined by how they permute the order of their input objects~\cite{maclaneCategoriesWorkingMathematician1978,joyalBraidedTensorCategories1993}.
\end{proof}

\begin{remark}
	\label{remark:braidReorder}
	In the case of non-symmetric \acp{OPT} for a generic braid transformations the reordering of the subsystems does not completely characterise the action of the transformation. Consider, for example, the identity $\identityTest{AB}$ and the following transformation:
	\begin{equation*}
		\myQcircuit{
			&\s{A}\qw&\braiding&\s{B}\qw&\braiding&\s{A}\qw&\qw&
			\\
			&\s{B}\qw&\braidingGhost&\s{A}\qw&\braidingGhost&\s{B}\qw&\qw&
		}.
	\end{equation*}
	Despite reordering the systems in the same way, the two transformations are generally different. 
\end{remark}

\begin{remark}
	\label{remark:opt:symmetric:uniqueDec}
	Also in the case in which uniqueness of decomposition does not hold, it is impossible to completely characterise a permutation by how it permutes elementary systems. Let $\system{S}$ be a system admitting two different decompositions in elementary systems $\system{AB} = \system{CD}$ such that $\system{A} \neq \system{C}$ and $\system{A} \neq \system{D}$, and analogously for \system{B}. Consider then the following permutation
	\begin{equation}
		\label{eqt:remark:mopt:uniquenessSystem:1}
		\myQcircuit{
			&\s{B}\qw&\braidingSym&\s{A}\qw&\multigate{1}{\identityTest{AB \to CD}}&\s{C}\qw&\braidingSym&\s{D}\qw&\qw&
			\\
			&\s{A}\qw&\braidingGhost&\s{B}\qw&\ghost{\identityTest{AB \to CD}}&\s{D}\qw&\braidingGhost&\s{C}\qw&\qw&,
		}
	\end{equation}
	where $\identityTest{AB \to CD}$ is the identity test for $\system{AB} = \system{CD}$. Given that no relation is known between the systems $\system{BA}$ and $\system{DC}$, other than the fact that they are operationally equivalent, it is not possible to state that this permutation can be completely characterised by how it permutes its input systems. 
\end{remark}

From \autoref{lem:opt:permutations:characterisation} the next result immediately follows.

\begin{lemma}[Permutations on bipartite systems]
	\label{thm:OPT:symmetric:permutations:generalForm}
	In every symmetric \ac{OPT} with unique decomposition, for any permutation acting on a bipartite system $\system{AB}$ there exist suitable systems $\system{A}'$, $\system{B}'$, $\system{A}''$, $\system{B}''$, and reversible transformations $\Braid_{1}$, $\Braid_{2}$, $\Braid_{3}$, $\Braid_{4}$ such that
	\begin{equation}
		\label{eqt:OPT:permutation:symm:formula}
		\myQcircuitSupMat{
			&\s{A}\qw&\multigate{1}{\Braid}&\s{C}\qw&\qw&
			\\
			&\s{B}\qw&\ghost{\Braid}&\s{D}\qw&\qw&
		} =
		\myQcircuitSupMat{
			&\s{A}\qw&\multigate{1}{\Braid_{3}}&\qw&\s{\prim{A}}\qw&\qw&\multigate{1}{\Braid_{4}}&\s{C}\qw&\qw&
			\\
			&\pureghost{}&\pureghost{\Braid_{3}}&\s{\secondE{A}}\qw&\braidingSym&\s{\prim{B}}\qw&\ghost{\Braid_{4}}\qw&
			\\
			&\s{B}\qw&\multigate{1}{\Braid_{1}}&\s{\prim{B}}\qw&\braidingGhost&\s{\secondE{A}}\qw&\multigate{1}{\Braid_{2}}&\s{D}\qw&\qw&
			\\
			&\pureghost{}&\pureghost{\Braid_{1}}&\qw&\s{\secondE{B}}\qw&\qw&\ghost{\Braid_{2}}&
		}, 
	\end{equation}
	where \system{A}, \system{B} are generic systems of the theory and \system{C}, \system{D} are systems such that \system{CD} has the same decomposition in elementary systems as \system{AB}. In general, any of \system{A}, \system{B}, \system{C}, \system{D} can be the trivial system, and the same holds also for \system{\prim{A}}, \system{\secondE{A}}, \system{\prim{B}}, \system{\secondE{B}}.
\end{lemma}

We have now all the elements to prove the following characterisation theorems for instruments and transformations in symmetric \acp{MOPT}.

\begin{theorem}
	\label{thm:OPT:minimal:symmetric:generalInstr}
	In every symmetric \ac{MOPT} any instrument $\eventTest{T}{x}{X} \in \Instr{A}{B}$ obtained as parallel and sequential composition of the elements of \eqref{eqt:OPT:minimal:def:test} is of the form:
	\begin{equation}
		\label{eqt:OPT:minimal:symmetric:instr:generic}
		\myQcircuit{
			&\pureghost{}&\multiprepareC{1}{\preparationEventTest{\rho}{y}{Y}}&\qw&\s{C}\qw&\qw&\multimeasureD{1}{\observationEventTest{a}{z}{Z}}&
			\\			&\pureghost{}&\pureghost{\preparationEventTest{\rho}{y}{Y}}&\s{\prim{B}}\qw&\braidingSym&\s{\prim{A}}\qw&\ghost{\observationEventTest{a}{z}{Z}}&
			\\
			&\s{A}\qw&\multigate{1}{\Braid_{1}}&\s{\prim{A}}\qw&\braidingGhost&\s{\prim{B}}\qw&\multigate{1}{\Braid_{2}}&\s{B}\qw&\qw&
			\\
			&\pureghost{}&\pureghost{\Braid_{1}}&\qw&\s{E}\qw&\qw&\ghost{\Braid_{2}}&
		},
	\end{equation}
	where $\Braid_{1}, \Braid_{2} \in \RevTransfA{\OPTMath}$ are suitable permutations, $\preparationEventTest{\rho}{y}{Ys} \in \Prep{C \prim{B}}$, $\observationEventTest{a}{z}{Z} \in \Obs{C \prim{A}}$, the outcome space $\outcomeSpace{X} = \cartesianProduct{Y}{Z}$, and \system{A}, \system{B}, $\system{\prim{A}}$, $\system{\prim{B}}$, $\system{C}$, $\system{E} \in \Sys{\OPTMath}$ may also be equal to the trivial system~\cite{erbaMeasurementIncompatibilityStrictly2024}.
\end{theorem}

\begin{corollary}
\label{cor:OPT:minimal:symmetric:generalTransf}
	In every symmetric \ac{MOPT} any transformation $\eventNoDown{T} \in \Transf{A}{B}$ obtained as parallel and sequential composition of the elements of \eqref{eqt:OPT:minimal:def:event} is of the form:
	\begin{equation}
		\label{eqt:OPT:minimal:symmetric:transf:generic}
		\myQcircuit{
			&\pureghost{}&\multiprepareC{1}{\preparationEventNoDown{\rho}}&\qw&\s{C}\qw&\qw&\multimeasureD{1}{\observationEventNoDown{a}}&
			\\			&\pureghost{}&\pureghost{\preparationEventNoDown{\rho}}&\s{\prim{B}}\qw&\braidingSym&\s{\prim{A}}\qw&\ghost{\observationEventNoDown{a}}&
			\\
			&\s{A}\qw&\multigate{1}{\Braid_{1}}&\s{\prim{A}}\qw&\braidingGhost&\s{\prim{B}}\qw&\multigate{1}{\Braid_{2}}&\s{B}\qw&\qw&
			\\
			&\pureghost{}&\pureghost{\Braid_{1}}&\qw&\s{E}\qw&\qw&\ghost{\Braid_{2}}&
		},
	\end{equation}
	where $\Braid_{1}, \Braid_{2} \in \RevTransfA{\OPTMath}$ are suitable permutations, $\preparationEventNoDown{\rho} \in \St{C \prim{B}}$, $\observationEventNoDown{a} \in \Eff{C \prim{A}}$, and \system{A}, \system{B}, $\system{\prim{A}}$, $\system{\prim{B}}$, $\system{C}$, $\system{E} \in \Sys{\OPTMath}$ may also be equal to the trivial system~\cite{erbaMeasurementIncompatibilityStrictly2024}.
\end{corollary}

The diagrams in~\eqref{eqt:OPT:minimal:symmetric:instr:generic} and \eqref{eqt:OPT:minimal:symmetric:transf:generic}, are colloquially referred to as \emph{jellyfish} instruments and transformations, respectively.

\subsubsection{Limits of Cauchy sequences for deterministic instruments}
The explicit form of instruments and transformations given by limits of Cauchy sequences of \eqref{eqt:OPT:minimal:def:test} and \eqref{eqt:OPT:minimal:def:event}, which have to be included as physical objects by the completeness requirement, remains an open question. However, in the special case of symmetric causal theories, a structure theorem can be proved for Cauchy sequences of \emph{deterministic} transformations.

A first tool in this direction is the following lemma that holds for all symmetric \ac{MOPT} (not necessarily causal). 

\begin{lemma}
	\label{thm:OPT:minimal:transf:stabilization}
	In a symmetric \ac{MOPT} any Cauchy sequence of transformations obtained as parallel and sequential composition of the elements in~\eqref{eqt:OPT:minimal:def:event},
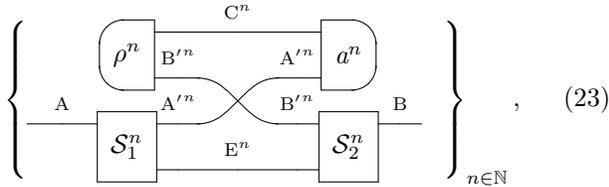
	\begin{equation}
		\label{eqt:sequenceJelly}
		\left\{
		\myQcircuit{
			&\pureghost{}&\multiprepareC{1}{\preparationEventNoDownSequence{\rho}{n}}&\qw&\sSequence{C}{n}\qw&\qw&\multimeasureD{1}{\observationEventNoDownSequence{a}{n}}&
			\\
			&\pureghost{}&\pureghost{\preparationEventNoDownSequence{\rho}{n}}&\sSequencePrime{B}{n}\qw&\braidingSym&\sSequencePrime{A}{n}\qw&\ghost{\observationEventNoDownSequence{a}{n}}&
			\\							&\s{A}\qw&\multigate{1}{\Braid_{1}^{n}}&\sSequencePrime{A}{n}\qw&\braidingGhost&\sSequencePrime{B}{n}\qw&\multigate{1}{\Braid_{2}^{n}}&\s{B}\qw&\qw&
			\\
			&\pureghost{}&\pureghost{\Braid_{1}^{n}}&\qw&\sSequence{E}{n}\qw&\qw&\ghost{\Braid_{2}^{n}}&
		}
		\right\}_{n \in \mathbb{N}},
	\end{equation}	
	admits a subsequence where the systems $\systemSequence{E}{n}$, $\systemSequence{A'}{n}$, $\systemSequence{B'}{n}$ and the where the permutations $\Braid_{1}^{n}$, $\Braid_{2}^{n}$ are fixed:
	\begin{equation*}
		\left\{
		\myQcircuit{
			&\pureghost{}&\multiprepareC{1}{\preparationEventNoDownSequence{\rho}{n}}&\qw&\sSequence{C}{n}\qw&\qw&\multimeasureD{1}{\observationEventNoDownSequence{a}{n}}&
			\\
			&\pureghost{}&\pureghost{\preparationEventNoDownSequence{\rho}{n}}&\s{B'}\qw&\braidingSym&\s{A'}\qw&\ghost{\observationEventNoDownSequence{a}{n}}&
			\\
			&\s{A}\qw&\multigate{1}{\Braid_{1}}&\s{A'}\qw&\braidingGhost&\s{B'}\qw&\multigate{1}{\Braid_{2}}&\s{B}\qw&\qw&
			\\
			&\pureghost{}&\pureghost{\Braid_{1}}&\qw&\s{E}\qw&\qw&\ghost{\Braid_{2}}&
		}
		\right\}_{n \in \mathbb{N}}.
	\end{equation*}	
\end{lemma}

\begin{remark}
Thanks to \autoref{thm:OPT:norm:instrEventConvergence}, the result of \autoref{thm:OPT:minimal:transf:stabilization} holds also in the case of instruments.
\end{remark}

Then we have a more detailed version of \autoref{cor:OPT:minimal:symmetric:generalTransf} in the case of deterministic transformations of a causal theory (not necessarily symmetric)

\begin{lemma}
	\label{lem:OPT:minimal:transf:causalDeterm}
	In a causal \ac{MOPT} (also non-symmetric) every deterministic transformation obtained as composition of the elements in \eqref{eqt:OPT:minimal:def:event} is of the form~\cite{erbaMeasurementIncompatibilityStrictly2024}
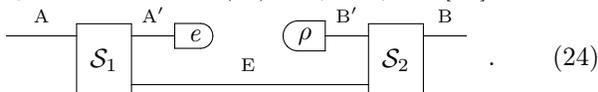
	\begin{equation}
		\label{eqt:OPT:minimal:transf:lem:causalDeterm}
		\minimalDeterministicCausalDestroyReprep{A}{B}{\prim{A}}{\prim{B}}{E}{\Braid_{1}}{\Braid_{2}}{\rho}.
	\end{equation}
\end{lemma}

\begin{proof}
	To prove the result it is sufficient to use \eqref{eqt:opt:braid:naturality}, \autoref{lem:opt:permutations:characterisation}, and the uniqueness of the deterministic effect.	
	
	The most general deterministic transformation of an \ac{MOPT} is given by
	\begin{equation*}
		\myQcircuit{
			&\prepareC{\preparationEventNoDown{\rho}}&\s{C}\qw&\multigate{1}{\Braid}&\s{D}\qw&\measureD{\observationUniqueDeterministic}&
			\\
			&\s{A}\qw&\qw&\ghost{\Braid}&\qw&\s{B}\qw&\qw&
		},
	\end{equation*}
	where $\rho \in \StN{C}$ and $\observationUniqueDeterministic \in \EffN{D}$. This immediately follows from \autoref{lem:OPT:minimal:generalInstr}. Given that we are considering also the case of non-symmetric \acp{MOPT}, \autoref{lem:opt:permutations:characterisation} cannot be exploited in this case. However, even though a generic braid transformation cannot be completely characterised by how it permutes its input wires (\autoref{remark:braidReorder}), its action is still of permuting them in some way. Consequently, one has	
   	\begin{align*}
		\myQcircuitBox{			&\prepareC{\preparationEventNoDown{\rho}}&\qw&\s{C'C''}\qw&\multigate{1}{\Braid}&\qw&\qw&\s{\sigma(C')\sigma(A')}\qw&\qw&\qw&\measureD{\observationUniqueDeterministic}&
			\\
			&\s{A'A''}\qw&\qw&\qw&\ghost{\Braid}&\qw&\qw&\qw&\qw&\qw&\s{\sigma(C'')\sigma(A'')}\qw&\qw&\qw&\qw&
		},
	\end{align*}
	where we used $\sigma\left(E\right) = \sigma\left(\system{E}_{1} \ldots \system{E}_{n}\right)$ as a shorthand notation for $\system{E}_{\sigma\left(1\right)}\ldots\system{E}_{\sigma\left(n\right)}$.	Using the uniqueness of the deterministic effect, which implies	
	
	\begin{align*}
		\myQcircuitBox{
			&\s{\sigma\left(C'\right)}\qw&\qw&\multimeasureD{1}{\observationUniqueDeterministic}&
			\\
			&\s{\sigma\left(A'\right)}\qw&\qw&\ghost{\observationUniqueDeterministic}&
		} = \quad\!\! \myQcircuitBox{
			&\s{\sigma\left(C'\right)}\qw&\qw&\measureD{\observationUniqueDeterministic}&
			\\
			&\s{\sigma\left(A'\right)}\qw&\qw&\measureD{\observationUniqueDeterministic}&
		},
	\end{align*}
	and the naturality property of the braiding \eqref{eqt:opt:braid:naturality}, one obtains the desired result.
\end{proof}	

The transformation
\begin{equation}
  \label{eqt:opt:transf:destrReprep}
  \measurePrepare{\prim{A}}{\prim{B}}{\rho}{\observationUniqueDeterministic}
\end{equation}
featuring in \eqref{eqt:OPT:minimal:transf:lem:causalDeterm} is usually referred to as \textdef{erase and prepare} since, whatever the input is, it will discard it and prepare the state $\rho$.

Finally, the key information is that for symmetric causal \acp{MOPT}, the form \eqref{eqt:OPT:minimal:transf:lem:causalDeterm} remains valid even when considering the limits of sequences of deterministic transformations:

\begin{theorem}
	\label{thm:OPT:minimal:transf:goodDeterministic}
	In a causal symmetric \ac{MOPT} the limits of Cauchy sequences of deterministic transformations are still of the form \eqref{eqt:OPT:minimal:transf:lem:causalDeterm}~\cite{erbaMeasurementIncompatibilityStrictly2024}.
\end{theorem}

As for \autoref{cor:OPT:minimal:symmetric:generalTransf}, also the last result cannot be proven in the case of non-symmetric theories due to the fact that the reordering of the subsystems does not completely characterise the braid transformations.

\subsection{Atomicity of the identity in minimal theories}

The full characterization of limits of deterministic transformations after completion in \autoref{thm:OPT:minimal:transf:goodDeterministic} is at the core of one of the main results of this work, that is the atomicity of the identity transformation. The last one in turn has a number of relevant consequence, most notably the \ac{NIWD} property (actually the atomicity of the identity has been proved equivalent to \ac{NIWD}~\cite{darianoInformationDisturbanceOperational2020}).

The relation between minimality and atomicity of the identity is proved in the following theorem:

\begin{theorem}
	\label{thm:OPT:minimal:symmetric:causal:idAtomicity}
	In every causal symmetric \ac{MOPT} \OPTMath{} the identity transformation \identityTest{A} is atomic for every system $\system{A} \in \Sys{\OPTMath}$~\cite{erbaMeasurementIncompatibilityStrictly2024}.
\end{theorem}

\begin{proof}
	To prove this result we have to show that the full coarse-graining of the limit of any Cauchy sequence of instruments of the theory that converges to the identity has to consist of transformations proportional to the identity itself. Let us start by considering a Cauchy sequence of generic instruments of the theory, which are of the form
	\begin{equation}
		\label{eqt:OPT:minimal:NIWD:proof}
		\left\{
		\myQcircuit{
			&\pureghost{}&\multiprepareC{1}{\preparationTestSequence{\rho}{X}{n}}&\qw&\sSequence{C}{n}\qw&\qw&\multimeasureD{1}{\observationTestSequence{a}{Y}{n}}&
			\\
			&\pureghost{}&\pureghost{\preparationTestSequence{\rho}{X}{n}}&\sSequencePrime{B}{n}\qw&\braidingSym&\sSequencePrime{A}{n}\qw&\ghost{\observationTestSequence{a}{Y}{n}}&
			\\								&\s{A}\qw&\multigate{1}{\Braid_{1}^{n}}&\sSequencePrime{A}{n}\qw&\braidingGhost&\sSequencePrime{B}{n}\qw&\multigate{1}{\Braid_{2}^{n}}&\s{B}\qw&\qw&
			\\
			&\pureghost{}&\pureghost{\Braid_{1}^{n}}&\qw&\sSequence{E}{n}\qw&\qw&\ghost{\Braid_{2}^{n}}&
		}
		\right\}_{n \in \mathbb{N}},
	\end{equation}
	by \autoref{thm:OPT:minimal:symmetric:generalInstr}. We now consider only instruments with the same input and output system since we are interested in sequences whose coarse-graining converges to the identity transformation. \autoref{thm:OPT:minimal:transf:stabilization} guarantees that it is always possible to find a subsequence with the permutations and part of the systems fixed
	\begin{equation*}
		\left\{
			\myQcircuit{
				&\pureghost{}&\multiprepareC{1}{\preparationTestSequence{\rho}{X}{n}}&\qw&\sSequence{C}{n}\qw&\qw&\multimeasureD{1}{\observationTestSequence{a}{Y}{n}}&
				\\
				&\pureghost{}&\pureghost{\preparationTestSequence{\rho}{X}{n}}&\s{A'}\qw&\braidingSym&\s{A'}\qw&\ghost{\observationTestSequence{a}{Y}{n}}&
				\\									&\s{A}\qw&\multigate{1}{\Braid}&\s{A'}\qw&\braidingGhost&\s{A'}\qw&\multigate{1}{\inverse{\Braid}}&\s{A}\qw&\qw&
				\\
				&\pureghost{}&\pureghost{\Braid}&\qw&\s{E}\qw&\qw&\ghost{\inverse{\Braid}}&
			}
			\right\}_{n \in \mathbb{N}}.
	\end{equation*}
	Consequently the sequence of full coarse-grainings is of the form
	\begin{equation*}
		\left\{ \minimalDeterministicCausalDestroyReprep{A}{A}{\prim{A}}{\prim{A}}{E}{\Braid}{\inverse{\Braid}}{\preparationEventNoDown{\rho}^{n}} \right\}_{n \in \mathbb{N}}.
	\end{equation*}
	In order to be equal to the identity, the limit must be 
	\begin{equation*}
		\minimalDeterministicCausalDestroyReprep{A}{A}{\prim{A}}{\prim{A}}{E}{\Braid}{\inverse{\Braid}}{\preparationEventNoDown{\rho}} = \myQcircuitSmall{&\qw&\s{A}\qw&\qw&\qw&},
	\end{equation*}
	or equivalently
	\begin{equation*}
		\myQcircuit{
			&\s{\prim{A}}\qw&\measureD{\observationUniqueDeterministic}&\pureghost{}&\prepareC{\preparationEventNoDown{\rho}}&\s{\prim{A}}\qw&\qw&
			\\
			&\qw&\qw&\s{E}\qw&\qw&\qw&\qw&
		} = \myQcircuit{
			&\qw&\s{A}\qw&\qw&\qw&
		},
	\end{equation*}
	where the limit of the sequence is deterministic and can be obtained exploiting \autoref{thm:OPT:minimal:transf:goodDeterministic}, defining $\eventNoDown{\rho} = \lim_{n \to \infty} \eventNoDown{\rho}^{n}$. The last identity can hold only if $\system{\prim{A}} = \trivialSystem$ and $\system{E} = \system{A}$, which means that a sequence of instruments whose coarse-graining converges to the identity transformation is of the form
	\begin{equation}
		\label{proof:thm:OPT:minimal:symmetric:causal:idAtomicity2}
		\left\{ \quad \myQcircuitSupMat{
			&\prepareC{{\preparationTest{\rho}{X}'}^{n}}&\sSequence{C}{n}\qw&\measureD{{\observationTest{a}{Y}'}^{n}}&
			\\
			&\qw&\s{A}\qw&\qw&\qw&
		} \right\}_{n \in \mathbb{N}},
	\end{equation}
    obtained by substituting $\system{\prim{A}} = \trivialSystem$ and $\system{E} = \system{A}$ into \eqref{eqt:OPT:minimal:NIWD:proof}. The last equation implies that every admissible decomposition of the identity is trivial, i.e.,made of transformations proportional to the identity itself. It follows that the identity is atomic.
\end{proof}

\subsection{Properties of minimal theories}

We show here the main informational consequences following the structure of minimal theories.

A first interesting fact regarding \acp{OPT} is that whenever the identity transformation is atomic for a certain system \system{A}, then any reversible transformation having that system as input or output must also be atomic. This implies that causal symmetric \acp{MOPT} do not admit reversible transformations different from permutations.

\begin{lemma}
	\label{lem:opt:idAtomic:revAtomic}
	In any \ac{OPT}  whenever the identity transformation is atomic for a certain system \system{A}, then any reversible transformation having that system as input or output must also be atomic~\cite{darianoInformationDisturbanceOperational2020}.
\end{lemma}

\begin{lemma}
	\label{lem:mopt:atomRev}
	In any causal symmetric \ac{MOPT} every reversible transformation is atomic and the set of reversible transformations coincides with that of permutations.
\end{lemma}

\begin{proof}
	The first part of the corollary derives directly from \autoref{thm:OPT:minimal:symmetric:causal:idAtomicity} and \autoref{lem:opt:idAtomic:revAtomic}.
		We now prove that the only reversible transformations are permutations. In the minimal setting, the only way in which a new reversible transformation, different from a permutation, can be included in the theory is as limit of some sequence of transformations. Moreover, since the set of deterministic transformations is Cauchy complete by \autoref{lem:OPT:causal:detTransf:convergence}, and $\eventNoDown{R}$ is deterministic, the sequence can be taken in  the set of deterministic transformations $\left\{\eventSequenceNoDown{T}{n}\right\}_{n \in \mathbb{N}} \subset \TransfN{A}{B}$:
	\begin{equation*}
		\eventNoDown{R} = \lim_{n \rightarrow \infty} \eventSequenceNoDown{T}{n} \in \RevTransf{A}{B},  
	\end{equation*}
	Now consider the sequence $\eventSequenceNoDownAll{G}{n} \mathDef \left\{ \inverse{\eventNoDown{R}}\eventSequenceNoDown{T}{n} \right\}_{n \in \mathbb{N}}$, which is still Cauchy and converges to an instrument whose coarse-graining is \identityTest{A} as a consequence of (\autoref{lem:opt:norm:operational:monotonicity}) (or, equivalently, \autoref{lem:opt:norm:sup:inequality} in the case of the sup norm). By the same argument used in the proof of \autoref{thm:OPT:minimal:symmetric:causal:idAtomicity}, it follows that any sequence of instruments whose coarse-graining converges to the identity consists of transformations that are proportional to the identity itself. In this specific case, where the sequence is composed of deterministic maps, all transformations must coincide with the identity, implying that the sequence is constant. Therefore, $\identityTest{A} = \sequentialComp{\inverse{\eventNoDown{R}}}{\eventSequenceNoDown{T}{n}}$ for all $n \in \mathbb{N}$, which further implies that $\eventNoDown{R} = \eventSequenceNoDown{T}{n}$ for any $n \in \mathbb{N}$. Since $\eventSequenceNoDownAll{T}{n}$ is a sequence of transformations included in~\eqref{eqt:OPT:minimal:def:event}, we have thus proven that Cauchy completing a causal symmetric \ac{MOPT} does not add any reversible transformation. Therefore, the set of reversible transformations coincides with the set of permutations.
\end{proof}

In Ref.~\cite{darianoInformationDisturbanceOperational2020}  it has been shown that any theory where the identity transformation is atomic for every system satisfies the property of \ac{NIWD}, namely any process that provides some non-trivial information on a system must perturb the system. Therefore, by \autoref{thm:OPT:minimal:symmetric:causal:idAtomicity} we have that any minimal theory exhibit \ac{NIWD}: 

\begin{corollary}
  \label{corol:MOPT:niwd}
	Every causal symmetric \ac{MOPT} satisfies the property of \ac{NIWD}.
\end{corollary}

From \autoref{lem:mopt:atomRev} and \autoref{thm:opt:irrev:idAtomicity} it immediately follows the irreversibility of minimal theories:
\begin{corollary}
	\label{corol:MOPT:irrev}
	Every non-trivial causal symmetric \ac{MOPT} has irreversibility.\footnote{We term \textdef{non-trivial} \acp{OPT} those theories that admit systems of dimension greater than 1.}
\end{corollary}
Furthermore, in every minimal theory there cannot exist a broadcasting channel, as it follows from \autoref{lem:opt:broadcasting:idAtomic}.
\begin{corollary}
	\label{corol:MOPT:broadcasting}
	Every causal symmetric \ac{MOPT} does not admit a broadcasting channel.
\end{corollary}
In conclusion of this section, we observe that no \acp{MOPT} can satisfy the programming theorem.
\begin{definition}[Universal simulator]
	Consider a causal \ac{OPT} and a pair of systems \system{A} and \system{B} of the theory. We define a \emph{universal simulator for $\Transf{A}{B}$}, a deterministic transformation $\eventNoDown{P}_{\system{A}, \system{B}} \in \TransfN{PA}{PB}$ for some suitable system~\system{P}, such that for every deterministic transformation $\eventNoDown{C} \in \TransfN{A}{B}$, there exists a program state $\sigma \in \StN{P}$ such that~\cite{chiribellaProbabilisticTheoriesPurification2010}:
	\begin{equation*}
		\myQcircuit{
			&\s{A}\qw&\gate{\eventNoDown{C}}&\s{B}\qw&\qw& 
		} = \myQcircuit{
			&\prepareC{\preparationEventNoDown{\sigma}}&\s{P}\qw&\multigate{1}{\eventNoDown{P}_{\system{A}, \system{B}}}&\s{P}\qw&\measureD{\observationUniqueDeterministic}&\pureghost{}&
			\\
			&\s{A}\qw&\qw&\ghost{\eventNoDown{P}_{\system{A}, \system{B}}}&\qw&\s{B}\qw&\qw&
		}.
	\end{equation*}
\end{definition}

\begin{definition}[No-programming]
	\label{def:opt:programming}
	We say that an \ac{OPT} has \emph{no-programming} if some pair of systems $\system A$ 
	and $\system B$ does not have an universal simulator.
\end{definition}

\begin{corollary}
	\label{corol:MOPT:programming}
	Every causal symmetric \ac{MOPT} has no-programming.
      \end{corollary}
\begin{proof}
      By contradiction, suppose that a universal simulator for system $\system A$ exists, and let us use it to program the identity. If we decompose the channel $\eventNoDown{P}_{\system{A}, \system{B}}$ using \eqref{eqt:OPT:minimal:transf:lem:causalDeterm}, then a program $\sigma \in \StN{P}$ for the identity must satisfy
\begin{equation*}
	\myQcircuitSmall{
		&\s{A}\qw&\qw&
	} = 
	\myQcircuitSmall{
		&\prepareC{\preparationEventNoDown{\sigma}}&\s{P}\qw&\multigate{1}{\Braid_1}&\s{\prim{A}}\qw&\measureD{\observationUniqueDeterministic}&\pureghost{}&\prepareC{\preparationEventNoDown{\rho}}&\s{\prim{A}}\qw&\multigate{1}{\Braid_2}&\s{P}\qw&\measureD{\observationUniqueDeterministic}&
		\\
		&\s{A}\qw&\qw&\ghost{\Braid_1}&\qw&\qw&\s{E}\qw&\qw&\qw&\ghost{\Braid_2}&\qw&\s{A}\qw&\qw&
	},
\end{equation*}
for some suitable states and permutations. However, by the same argument used in the proof of \autoref{thm:OPT:minimal:symmetric:causal:idAtomicity}, the latter equality can be satisfied only if $\Braid_1=\tilde\Braid_1\boxtimes\tilde\Braid_2$ and $\Braid_2=\tilde\Braid_3\boxtimes\inverse{\tilde\Braid}_2$ the following holds:
\begin{equation*}
	\myQcircuitSmall{
		&\s{A}\qw&\qw&
	} = 
	\myQcircuitSmall{
		&\prepareC{\preparationEventNoDown{\sigma}}&\s{P}\qw&\gate{\tilde\Braid_{1}}&\s{\prim{A}}\qw&\measureD{\observationUniqueDeterministic}&\pureghost{}&\prepareC{\preparationEventNoDown{\rho}}&\s{\prim{A}}\qw&\gate{{\tilde\Braid_{3}}}&\s{P}\qw&\measureD{\observationUniqueDeterministic}&
		\\
		&\s{A}\qw&\qw&\gate{\tilde\Braid_{2}}&\qw&\qw&\s{E}\qw&\qw&\qw&\gate{\inverse{\tilde\Braid_{2}}}&\qw&\s{A}\qw&\qw&
	}.
\end{equation*}
This shows that the identity can be programmed, then it is impossible to program any other transformation, thus forbidding the existence of a universal simulator.
\end{proof}

\section{Minimal theories with strong causality}
\label{subsec:MSOPT}
The results presented in the previous section pertain to operational theories that do not satisfy strong causality, i.e.~the very natural possibility of performing conditional experiments—a feature that is reasonably expected in a theory of physical systems. This raises the question we aim to address in this section: do the characteristics of minimal theories derived so far remain valid when strong causality is assumed? To explore this, we analyse the properties of what we term \acfp{MSOPT}, which can be understood as generic \acp{MOPT} completed to satisfy strong causality through the procedure outlined in \autoref{procedure:OPT:stronglyClosure}.

\begin{definition}[Minimal Strongly Causal Operational Probabilistic Theory]
	\label{def:OPT:minimalsc}
	We define as \textdef{\acf{MSOPT}} an \ac{OPT} with unique decomposition where the allowed tests are those of the \ac{MOPT} with the same systems, plus all conditional tests thereof, and finally Cauchy completed. 
\end{definition}

We remark that completion with respect to conditioning completely sets apart \acp{MSOPT} form \acp{MOPT}, with the set of instruments of a given minimal theory strictly contained in its completion under strong causality. For example, while the single transformations of the form
\begin{equation*}
	\myQcircuit{&\s{A}\qw&\gate{\event{T}{x}}&\s{B}\qw}\, \mathDef \,
	\myQcircuit{		&\s{A}\qw&\measureD{\observationEvent{a}{x}}&\prepareC{\preparationEventNoDown{\rho}^{\left(\outcome{x}\right)}}&\s{B}\qw&\qw&	
	},
\end{equation*}
where \system{A} and \system{B} are generic systems of the theory, $\observationEventTest{a}{x}{X} \in \Obs{A}$ is a generic observation-test, and $\left\{ \preparationEventNoDown{\rho}^{\left( \outcome{x} \right)} \right\}_{\outcomeIncluded{x}{X}} \subset \StN{B}$ is a collection of generic deterministic states of the theory, are legitimate in \acp{MOPT}, their collection $\eventTest{T}{x}{X}$ is an instrument only in \acp{MSOPT}.

Due to the additional operations enabled by conditioning, it is no longer possible to characterise the generic instruments and transformations of \acp{MSOPT}. This is due to the fact that generic conditional instruments do not necessarily take the simple forms given by \eqref{eqt:OPT:minimal:symmetric:transf:generic}, or, in the deterministic case, \eqref{eqt:OPT:minimal:transf:lem:causalDeterm}. However, by introducing an additional assumption, we will demonstrate that the identity transformation remains atomic for every system, even when all conditional instruments are considered.

We recall that the kind of transformations added to an \ac{MOPT} via closure under conditioning are (\autoref{procedure:OPT:stronglyClosure})
\begin{align}
	\label{eqt:proof:MSOPT:atomicity:1}
	\myQcircuit{			&\s{A}\qw&\gate{{}^{1}{\event{T}{x_{1}}}}&\s{A_{1}}\qw&\gate{{}^{2}\conditionedEvent{T}{x_{2}}{x_{1}}}&\s{A_{2}}\qw&\qw&\;&\ldots&\pureghost{}&\s{A_{k-1}}\qw&\qw&\gate{{}^{k}\conditionedEvent{T}{x_{k}}{\boldsymbol{x}}}&\s{B}\qw&\qw&
	},
\end{align}
where $\boldsymbol{x} \mathDef \outcome{x_{1}, \ldots, x_{k-1}}$ and ${}^{i}\event{T}{x_{i}}$ are transformations of the original \ac{MOPT}---thus of the form \eqref{eqt:OPT:minimal:symmetric:transf:generic} or limits of Cauchy sequences thereof. We highlight that according to \autoref{remark:strongClosure:finiteOutcome} the number of conditioning steps we consider is arbitrary but finite (in this case equal to $k$). Moreover, one has to include transformations that are limits of Cauchy sequences of transformations above,
\begin{align}
	\label{eqt:proof:MSOPT:atomicity:bis}
	\left\{ \myQcircuit{			&\s{A}\qw&\gate{{}^{1}\eventSequence{T}{x_{1}}{n}}&\sSequenceIndex{A}{1}{n}\qw&\qw&\;&\ldots&\pureghost{}&\sSequenceIndex{A}{k(n)-1}{n}\qw&\qw&\gate{{}^{k(n)}\conditionedEventSequence{T}{x_{k(n)}}{\boldsymbol{x}}{n}}&\s{B}\qw&\qw&
	} \right\}_{n \in \mathbb{N}}.
\end{align}
Notice that in \eqref{eqt:proof:MSOPT:atomicity:bis} the systems on the ``inside'' wires $\systemSequence{A_{i}}{n}$ depend on the index of the sequence, while the input and output systems are fixed. Moreover, also the number of conditioning steps $k(n)$ can vary with the index of the sequence. However, for each $n$ we can take the maximal value $k=\max_mk(m)$ for each sequence item, conditioning with the identity at the extra steps  $k-k(n)$ if necessary, thus reducing to consider transformations that are limits of Cauchy sequences of the form
\begin{align}
	\label{eqt:proof:MSOPT:atomicity:2}
	\left\{ \myQcircuit{			&\s{A}\qw&\gate{{}^{1}\eventSequence{T}{x_{1}}{n}}&\sSequenceIndex{A}{1}{n}\qw&\qw&\;&\ldots&\pureghost{}&\sSequenceIndex{A}{k-1}{n}\qw&\qw&\gate{{}^{k}\conditionedEventSequence{T}{x_{k}}{\boldsymbol{x}}{n}}&\s{B}\qw&\qw&
	} \right\}_{n \in \mathbb{N}}.
\end{align}
Clearly \eqref{eqt:proof:MSOPT:atomicity:1} can be seen as special case of \eqref{eqt:proof:MSOPT:atomicity:2} and analogously the most general instruments to include are Cauchy sequences
\begin{equation}
	\label{eqt:proof:MSOPT:atomicity:3}
	\left\{  \left\{ \sequentialComp{\conditionedEvent{G}{x_{k}}{\boldsymbol{x}}}{\event{T}{\boldsymbol{x}}} \right\}^{n}_{\outcomeIncluded{x_{k}}{X_{k}}, \outcomeIncluded{\boldsymbol{x}}{\boldsymbol{\outcomeSpace X}}} \right\}_{n \in \mathbb{N}} \subset \InstrA{A},
\end{equation}
where for simplicity we highlight here only the last conditioning step. Explicitly one has that $\{ \{ \eventNoDown{G}_{\outcome{x}_{k}}^{( \outcome{\boldsymbol{x}} )} \}_{\outcome{x}_{k} \in \outcomeSpace{X}_{k} }^{n} \}_{n \in \mathbb{N}} \subset \InstrEmpty{\system{A}_{k-1}^{n} \!\to\! \system{A}}$ and $\eventSequenceAll{T}{\boldsymbol{x}}{n} \subset \InstrEmpty{\system{A} \!\to\! \system{A}_{k-1}^{n}}$ with ${\boldsymbol{x}} \mathDef ( \outcome{x_{1}, \ldots, x_{k-1}} ) \in \boldsymbol{\outcomeSpace{X}} = \cartesianProduct{\outcomeSpace{X}_{1}}{\cartesianProduct{\cdots}{\outcomeSpace{X}_{k-1}}}$.

\subsection{Atomicity of the identity}
We can now state and prove the main theorem on the structure of minimal theories with conditioning
\begin{theorem}
	\label{thm:MSOPT:symmetric:idAtomicity}
	In every symmetric \ac{MSOPT} that admits a spanning set of entangled states for every composite system, the identity transformation is atomic for every system.
\end{theorem}

\begin{proof}
	We have to show that the full coarse-graining of every instrument cannot converge to the identity. The only cases we have to check are those of conditional instruments~\eqref{eqt:proof:MSOPT:atomicity:3}, since \autoref{thm:OPT:minimal:symmetric:causal:idAtomicity} guarantees the impossibility of decomposing the identity in any other way. The statement then follows proving that the sequence of deterministic transformations obtained by coarse-graining completely over the outcome spaces $\boldsymbol{\outcomeSpace{X}}$ and $\outcomeSpace{X}_{k}$ cannot be equal to the identity unless the instruments are trivial decompositions of $\identityTest{A}$, i.e., are of the form $\left\{ \probabilityEvent{p}{x} \identityTest{A} \right\}_{\outcomeIncluded{x}{X}}$. Thus, we characterize under which conditions the following equality can be satisfied
	\begin{equation}
		\label{eqt:proof:MSOPT:atomicity:4}
		\lim_{n \to \infty} \sum_{\outcomeIncluded{\boldsymbol{x}}{\boldsymbol{\outcomeSpace X}}} 
		\myQcircuit{
                  &\s{A}\qw&\gate{\eventSequence{T}{\boldsymbol{x}}{n}}&\qw&\sSequenceIndex{A}{k-1}{n}\qw&\qw&\gate{
                    {{\mathscr{G}_{\outcomeSpace{X}_{k}}^{(\boldsymbol{x})}}}^{n}
                  }&\s{A}\qw&\qw&
		} = \quad\!\! \myQcircuit{&\s{A}\qw&\qw&},
	\end{equation}
    where by ${{\mathscr{G}_{\outcomeSpace{X}_{k}}^{(\boldsymbol{x})}}}^{n}$ we denote the deterministic transformation given by the full coarse-graining of the test $\{ \{ \eventNoDown{G}_{\outcome{x}_{k}}^{( \outcome{\boldsymbol{x}} )} \}_{\outcome{x}_{k} \in \outcomeSpace{X}_{k} }^{n} \}_{n \in \mathbb{N}}$.

    The main issue here is that the sequences of transformations in the cascade of instruments above have a complicated dependency of the index $n$: at an arbitrary conditioning step $k'$ both the input $\system{A}_{k'-1}^{n}$ and the output $\system{A}_{k'}^{n}$ are $n$-dependent thus preventing from straightforwardly applying the stabilization \autoref{thm:OPT:minimal:transf:stabilization}. However, we can proceed by iteration from the last step, which has fixed output system $\system{A}$, as follows:
    \begin{enumerate}[(i)]
        \item\label{it:channels} Consider the limit of the coarse-graining of sequence~\eqref{eqt:proof:MSOPT:atomicity:3}. Step by step, starting from the last ($k$-th) conditional gate, each instrument can be coarse-grained to a deterministic map and  \eqref{eqt:proof:MSOPT:atomicity:4} gives information on their form, especially on the  stabilization (dropping of dependence on the sequence step $n$) of systems involved. 
        \item\label{it:transformations} With the constraints obtained in the last step, we look at the structure of the limit transformations within the conditional instruments. This is crucial to exclude non-trivial decompositions of the identity which holds only if all the limit transformations are proportional to the identity. This second step is quite involved and full details are carried out in \aref{app:proof:MSOPT:atomicity}.         
    \end{enumerate}
	\Iref{it:channels}: Thanks to \autoref{lem:OPT:minimal:transf:causalDeterm} and \autoref{thm:OPT:minimal:transf:goodDeterministic}, the deterministic map ${{\mathscr{G}_{\outcomeSpace{X}_{k}}^{(\boldsymbol{x})}}}^{n}$ is equal to
	\begin{equation}
		\label{eqt:proof:MSOPT:atomicity:5}
		\myQcircuit{
			&\sSequenceIndex{A}{k-1}{n}\qw&\multigate{1}{{\permutation_{1}^{(\outcome{\boldsymbol{x}})}}^{n}}&\qw&\sSequenceConditioned{C}{\boldsymbol{x}}{n}\qw&\measureD{\observationUniqueDeterministic}&\pureghost{}&\prepareC{{\preparationEventNoDown{\rho}^{(\outcome{\boldsymbol{x}})}}^{n}}&\qw&\sSequenceConditioned{D}{\boldsymbol{x}}{n}\qw&\multigate{1}{{\permutation_{2}^{(\outcome{\boldsymbol{x}})}}^{n}}&\s{A}\qw&\qw&
			\\
			&\pureghost{}&\pureghost{{\permutation_{1}^{(\outcome{\boldsymbol{x}})}}^{n}}&\qw&\qw&\qw&\sSequenceConditioned{E}{\boldsymbol{x}}{n}\qw&\qw&\qw&\qw&\ghost{{\permutation_{2}^{(\outcome{\boldsymbol{x}})}}^{n}}&
		},
	\end{equation}
	with $\{ {\preparationEventNoDown{\rho}^{(\outcome{\boldsymbol{x}})}}^{n} \}_{n \in \mathbb{N}} \subset \StNOPT{ {\system{D}^{\left( \outcome{\boldsymbol{x}} \right)}}^{n} }$.
	Contrarily to what happens in \autoref{thm:OPT:minimal:transf:stabilization}, the dependence on the index $n$ of the sequence cannot be removed, since the input system $\system{A}_{k-1}^{n}$ also depends on $n$. However, given that the outcome system \system{A} is fixed, it is always possible to find a subsequence where ${\system{D}^{\left(\outcome{\boldsymbol{x}}\right)}}^{n}$, ${\system{E}^{\left(\outcome{\boldsymbol{x}}\right)}}^{n}$, and ${\permutation_{2}^{(\outcome{\boldsymbol{x}})}}^{n}$ are fixed, being the number of permutations of elementary systems composing \system{A} and the ways to group them finite. Hence, one can always consider a subsequence such that ${{\mathscr{G}_{\outcomeSpace{X}_{k}}^{(\boldsymbol{x})}}}^{n}$ is of the form
	\begin{equation}
		\label{eqt:proof:MSOPT:atomicity:6}
		\myQcircuit{
			&\sSequenceIndex{A}{k-1}{n}\qw&\multigate{1}{{\permutation_{1}^{(\outcome{\boldsymbol{x}})}}^{n}}&\qw&\sSequenceConditioned{C}{\boldsymbol{x}}{n}\qw&\measureD{\observationUniqueDeterministic}&\pureghost{}&\prepareC{{\preparationEventNoDown{\rho}^{(\outcome{\boldsymbol{x}})}}^{n}}&\qw&\sConditioned{D}{\boldsymbol{x}}\qw&\multigate{1}{\permutation_{2}^{(\outcome{\boldsymbol{x}})}}&\s{A}\qw&\qw&
			\\
			&\pureghost{}&\pureghost{{\permutation_{1}^{(\outcome{\boldsymbol{x}})}}^{n}}&\qw&\qw&\qw&\sConditioned{E}{\boldsymbol{x}}\qw&\qw&\qw&\qw&\ghost{\permutation_{2}^{(\outcome{\boldsymbol{x}})}}&
		}.
	\end{equation}
	The condition expressed in \eqref{eqt:proof:MSOPT:atomicity:4} can be rewritten as
	\begin{widetext}
		\begin{equation}
			\label{eqt:proof:MSOPT:atomicity:7}
			\sum_{\outcomeIncluded{\boldsymbol{x}}{\boldsymbol{X}}} \; \lim_{n \to \infty}  \; \myQcircuit{
				&\sConditioned{D}{\tilde{\boldsymbol{x}}}\qw&\multigate{1}{\permutation_{2}^{(\outcome{\tilde{\boldsymbol{x}}})}}&\s{A}\qw&\gate{\eventSequence{T}{\boldsymbol{x}}{n}}&\sSequenceIndex{A}{k-1}{n}\qw&\multigate{1}{{\permutation_{1}^{(\outcome{\boldsymbol{x}})}}^{n}}&\qw&\sSequenceConditioned{C}{\boldsymbol{x}}{n}\qw&\measureD{\observationUniqueDeterministic}&\prepareC{{\preparationEventNoDown{\rho}^{(\outcome{\boldsymbol{x}})}}^{n}}&\qw&\sConditioned{D}{\boldsymbol{x}}\qw&\qw&\multigate{1}{\permutation_{2}^{(\outcome{\boldsymbol{x}})}}&\s{A}\qw&\multigate{1}{\inverse{\permutation_{2}^{(\outcome{\tilde{\boldsymbol{x}}})}}}&\qw&\sConditioned{D}{\tilde{\boldsymbol{x}}}\qw&\qw&
				\\
				&\sConditioned{E}{\tilde{\boldsymbol{x}}}\qw&\ghost{\permutation_{2}^{(\outcome{\tilde{\boldsymbol{x}}})}}&\pureghost{}&\pureghost{}&\pureghost{}&\pureghost{{\permutation_{1}^{(\outcome{\boldsymbol{x}})}}^{n}}&\qw&\qw&\qw&\qw&\qw&\sConditioned{E}{\boldsymbol{x}}\qw&\qw&\ghost{\permutation_{2}^{(\outcome{\boldsymbol{x}})}}&\pureghost{}&\pureghost{\inverse{\permutation_{2}^{(\outcome{\tilde{\boldsymbol{x}}})}}}&\qw&\sConditioned{E}{\tilde{\boldsymbol{x}}}\qw&\qw&
			} = \quad\! \begin{aligned} \Qcircuit @C=1em @R=2.1em {
				&\sConditioned{D}{\tilde{\boldsymbol{x}}}\qw&\qw&
				\\
				&\sConditioned{E}{\tilde{\boldsymbol{x}}}\qw&\qw&
			} \end{aligned},
		\end{equation}   
	\end{widetext}
	where we exploited \autoref{corol:OPT:norm:instrEventConvergence} to exchange the operations of limit and coarse-graining and we have applied the permutation $\permutation_{2}^{(\outcome{\tilde{\boldsymbol{x}}})}$ (its inverse) to the left (right) at both sides of the equality. Let now $\preparationEventNoDown{\omega}\in\St{\systemConditioned{D}{\tilde{\boldsymbol{x}}}\systemConditioned{{F}}{\tilde{\boldsymbol{x}}}}$ be an atomic entangled state, where $\systemConditioned{{F}}{\tilde{\boldsymbol{x}}}=\systemConditioned{{E}}{\tilde{\boldsymbol{x}}}\system{E'}$ for some arbitrary extra ancillary system $\system{E'}$. If we apply both sides of \eqref{eqt:proof:MSOPT:atomicity:7} to the state $\preparationEventNoDown{\omega}$, by atomicity it follows that any term of the coarse-graining on the right hand side of \eqref{eqt:proof:MSOPT:atomicity:7} must be proportional to $\preparationEventNoDown{\omega}$ itself, that is
	\begin{widetext}
		\begin{equation*}
			\lim_{n \to \infty} \; \myQcircuit{
				&\multiprepareC{2}{\preparationEventNoDown{\omega}}&\qw&\sConditioned{D}{\tilde{\boldsymbol{x}}}\qw&\multigate{1}{\permutation_{2}^{(\outcome{\tilde{\boldsymbol{x}}})}}&\s{A}\qw&\gate{\eventSequence{T}{\tilde{\boldsymbol{x}}}{n}}&\sSequenceIndex{A}{k-1}{n}\qw&\multigate{1}{{\permutation_{1}^{(\outcome{\tilde{\boldsymbol{x}}})}}^{n}}&\qw&\sSequenceConditioned{C}{\tilde{\boldsymbol{x}}}{n}\qw&\measureD{\observationUniqueDeterministic}&\prepareC{{\preparationEventNoDown{\rho}^{(\outcome{\tilde{\boldsymbol{x}}})}}^{n}}&\qw&\sConditioned{D}{\tilde{\boldsymbol{x}}}\qw&\qw&
				\\
				&\pureghost{\preparationEventNoDown{\omega}}&\qw&\sConditioned{E}{\tilde{\boldsymbol{x}}}\qw&\ghost{\permutation_{2}^{(\outcome{\tilde{\boldsymbol{x}}})}}&\pureghost{}&\pureghost{}&\pureghost{}&\pureghost{{\permutation_{1}^{(\outcome{\tilde{\boldsymbol{x}}})}}^{n}}&\qw&\qw&\qw&\qw&\qw&\sConditioned{E}{\tilde{\boldsymbol{x}}}\qw&\qw&
				\\
				&\pureghost{\preparationEventNoDown{\omega}}&\qw&\qw&\qw&\qw&\qw&\qw&\s{E'}\qw&\qw&\qw&\qw&\qw&\qw&\qw&\qw&
			}\,=\lim_{n\to\infty}
			\myQcircuit{
				&\prepareC{{\preparationEventNoDown{\rho}^{(\outcome{\tilde{\boldsymbol{x}}})}}^{n}}&\qw&\sConditioned{D}{\tilde{\boldsymbol{x}}}\qw&\qw&
				\\
				&\multiprepareC{1}{\sigma^{(\tilde{\boldsymbol x})}_n}&\qw&\sConditioned{E}{\tilde{\boldsymbol{x}}}\qw&\qw&
				\\
				&\pureghost{\sigma^{(\tilde{\boldsymbol x})}_n}&\qw&\s{E'}\qw&\qw&
			}
			 \propto \; \myQcircuit{
				&\multiprepareC{2}{\preparationEventNoDown{\omega}}&\qw&\sConditioned{D}{\tilde{\boldsymbol{x}}}\qw&\qw&
				\\
				&\pureghost{\preparationEventNoDown{\omega}}&\qw&\sConditioned{E}{\tilde{\boldsymbol{x}}}\qw&\qw&
				\\
				&\pureghost{\preparationEventNoDown{\omega}}&\qw&\s{E'}\qw&\qw&
			}\,.
		\end{equation*}
	\end{widetext}
	Since a Cauchy sequence whose elements are parallel composition of states is such that the sequences of marginals are Cauchy, one can easily prove that the limit is a parallel composition itself (\autoref{lem:OPT:sequence:deterministic:factorized}), that is
	\begin{align*}
		\myQcircuit{
				&\prepareC{{\preparationEventNoDown{\rho}^{(\outcome{\tilde{\boldsymbol{x}}})}}}&\qw&\sConditioned{D}{\tilde{\boldsymbol{x}}}\qw&\qw&
				\\
				&\multiprepareC{1}{\sigma^{(\tilde{\boldsymbol x})}}&\qw&\sConditioned{E}{\tilde{\boldsymbol{x}}}\qw&\qw&
				\\
				&\pureghost{\sigma^{(\tilde{\boldsymbol x})}_n}&\qw&\s{E'}\qw&\qw&
			}
			 \propto \; \myQcircuit{
				&\multiprepareC{2}{\preparationEventNoDown{\omega}}&\qw&\sConditioned{D}{\tilde{\boldsymbol{x}}}\qw&\qw&
				\\
				&\pureghost{\preparationEventNoDown{\omega}}&\qw&\sConditioned{E}{\tilde{\boldsymbol{x}}}\qw&\qw&
				\\
				&\pureghost{\preparationEventNoDown{\omega}}&\qw&\s{E'}\qw&\qw&
			}\,.		
	\end{align*}
	The only possibility for the above equality to hold is that $\systemConditioned{D}{\tilde{\boldsymbol{x}}} = \trivialSystem$. Since this must hold for every atomic entangled state $\omega$, and the latter states are spanning, we can conclude that $\systemConditioned{D}{\tilde{\boldsymbol{x}}} = \trivialSystem$ for every $\tilde{\boldsymbol x}$, and this implies that \eqref{eqt:proof:MSOPT:atomicity:6} is of the form
	\begin{equation}
		\label{eqt:proof:MSOPT:atomicity:10}
		\myQcircuit{
			&\sSequenceIndex{A}{k-1}{n}\qw&\multigate{1}{{\permutation_{3}^{(\outcome{\boldsymbol{x}})}}^{n}}&\qw&\sSequenceConditioned{C}{\boldsymbol{x}}{n}\qw&\measureD{\observationUniqueDeterministic}&
			\\
			&\pureghost{}&\pureghost{{\permutation_{3}^{(\outcome{\boldsymbol{x}})}}^{n}}&\qw&\s{A}\qw&\qw&\qw&
		}= \, \myQcircuit{
			&\sSequenceIndex{A}{k-1}{n}\qw&\multigate{1}{\permutation_{4}^{n}}&\sSequence{C}{n}\qw&\measureD{\observationUniqueDeterministic\left(\boldsymbol{x}\right)}&
			\\
			&\pureghost{}&\pureghost{\permutation_{4}^{n}}&\s{A}\qw&\qw&\qw&
		},
	\end{equation}
	where ${\permutation_{3}^{(\outcome{\boldsymbol{x}})}}^{n} \mathDef \sequentialComp{\left( \parallelComp{ \eventNoDown{I}_{{\systemConditioned{C}{\boldsymbol{x}}}^{n} }}{\permutation_{2}^{(\outcome{\boldsymbol{x}})}} \right)}{ {\permutation_{1}^{(\outcome{\boldsymbol{x}})}}^{n} }$. The equality follows reminding that a permutation is completely fixed by how it permutes the elementary subsystems (\autoref{lem:opt:permutations:characterisation}). Therefore, ${\systemConditioned{C}{\boldsymbol{x}}}^{n}$ can be replaced by a fixed system up to a local permutation depending on $\boldsymbol{x}$, which can then be absorbed within the deterministic effect. Clearly the dependence from $\boldsymbol{x}$ in the deterministic effect could be dropped. However, in the study of the transformations composing the above deterministic transformations at \Iref{it:transformations} the effects in the observation-tests summing to $\observationUniqueDeterministic$ will depend of the outcome $\outcome{\boldsymbol{x}}$, for this reason we keep it in the notation.
	
	We can now proceeded to make explicit the $k-1$ conditioning step in \eqref{eqt:proof:MSOPT:atomicity:4}, so we express also the coarse-graining of the instrument at position $k-1$ as a generic deterministic transformation (see \autoref{lem:OPT:minimal:transf:causalDeterm} and \autoref{thm:OPT:minimal:transf:goodDeterministic}) followed by the deterministic transformation at position $k$ derived in \eqref{eqt:proof:MSOPT:atomicity:10}. The result is of the form
	\begin{widetext}
		\begin{equation*}
            \lim_{n \to \infty} \sum_{\outcomeIncluded{\boldsymbol{x}'}{\boldsymbol{\outcomeSpace X}'}} \; \myQcircuit{
				&\s{A}\qw&\gate{{\event{T}{\boldsymbol{x}'}'}^{n}}&\sSequenceIndex{A}{k-2}{n}\qw&\multigate{1}{{\permutation_{5}^{(\outcome{\boldsymbol{x}'})}}^{n}}&\qw&\sSequenceConditioned{F}{\boldsymbol{x}'}{n}\qw&\measureD{\observationUniqueDeterministic}&\pureghost{}&\prepareC{{\preparationEventNoDown{\sigma}^{(\outcome{\boldsymbol{x}'})}}^{n}}&\qw&\sSequenceConditioned{G}{\boldsymbol{x}'}{n}\qw&\multigate{1}{{\permutation_{6}^{(\outcome{\boldsymbol{x}'})}}^{n}}&\sSequenceIndex{A}{k-1}{n}\qw&\multigate{1}{\permutation_{4}^{n}}&\sSequence{C}{n}\qw&\measureD{\observationUniqueDeterministic\left(\outcome{\boldsymbol{x}}\right)}&
				\\
				&\pureghost{}&\pureghost{}&\pureghost{}&\pureghost{{\permutation_{5}^{(\outcome{\boldsymbol{x}'})}}^{n}}&\qw&\qw&\qw&\sSequenceConditioned{H}{\boldsymbol{x}'}{n}\qw&\qw&\qw&\qw&\ghost{{\permutation_{6}^{(\outcome{\boldsymbol{x}'})}}^{n}}&\pureghost{}&\pureghost{\permutation_{4}^{n}}&\s{A}\qw&\qw&\qw&
			} = \quad \!\! \myQcircuit{
				&\s{A}\qw&\qw&	
			},
		\end{equation*}
	\end{widetext}
	where $\outcome{\boldsymbol{x}'} = \left( \outcome{x_{1}, \ldots, x_{k-2}} \right) \in \boldsymbol{\outcomeSpace{X}'} = \cartesianProduct{\outcomeSpace{X}_{1}}{\cartesianProduct{\cdots}{\outcomeSpace{X}_{k-2}}}$, denotes the outcome space of the first $k-2$ instruments. At the first conditioning step~\eqref{eqt:proof:MSOPT:atomicity:7}  we were able to conclude the triviality of system $\system{D}^{\left(\outcome{\boldsymbol{x}}\right)}$, here instead we cannot immediately conclude the triviality of system $\system{G}^{\left(\outcome{\boldsymbol{x}'}\right)}$, but only of part of it. To see this we highlight the fact that ${\systemConditioned{G}{\outcome{\boldsymbol{x}'}}}^{n}$ and $\systemSequence{C}{n}$ can share a subsystem $\systemSequence{J}{n}$ made of some elementary subsystems\footnote{To be precise $\systemSequence{J}{n}$ is defined up to a local permutation. However, this is not important for the discussion since it can be always absorbed within the deterministic effect.}. Let ${\permutation_{7}^{(\outcome{\boldsymbol{x}'})}}^{n} \mathDef \sequentialComp{\permutation_{4}^{n}}{{\permutation_{6}^{(\outcome{\boldsymbol{x}'})}}^{n}}$,
	\begin{widetext}
		\begin{equation*}
			\lim_{n \to \infty} \sum_{\outcomeIncluded{\boldsymbol{x}'}{\boldsymbol{\outcomeSpace X}'}} \; \myQcircuit{
				&\pureghost{}&\pureghost{}&\pureghost{}&\pureghost{}&\pureghost{}&\pureghost{}&\pureghost{}&\multiprepareC{1}{{\preparationEventNoDown{\sigma}^{(\outcome{\boldsymbol{x}'})}}^{n}}&\qw&\sSequence{J}{n}\qw&\qw&\multigate{2}{{\permutation_{7}^{(\outcome{\boldsymbol{x}'})}}^{n}}&\sSequence{J}{n}\qw&\measureD{\observationUniqueDeterministic\left(\outcome{\boldsymbol{x}}\right)}&
				\\
				&\s{A}\qw&\gate{{\event{T}{\boldsymbol{x}'}'}^{n}}&\sSequenceIndex{A}{k-2}{n}\qw&\multigate{1}{{\permutation_{5}^{(\outcome{\boldsymbol{x}'})}}^{n}}&\sSequenceConditioned{F}{\boldsymbol{x}'}{n}\qw&\measureD{\observationUniqueDeterministic}&\pureghost{}&\pureghost{{\preparationEventNoDown{\sigma}^{(\outcome{\boldsymbol{x}'})}}^{n}}&\qw&\sSequenceConditioned{G'}{\boldsymbol{x}'}{n}\qw&\qw&\ghost{{\permutation_{7}^{(\outcome{\boldsymbol{x}'})}}^{n}}&\sSequencePrime{C}{n}\qw&\measureD{\observationUniqueDeterministic\left(\outcome{\boldsymbol{x}}\right)}&
				\\
				&\pureghost{}&\pureghost{}&\pureghost{}&\pureghost{{\permutation_{5}^{(\outcome{\boldsymbol{x}'})}}^{n}}&\qw&\qw&\sSequenceConditioned{H}{\boldsymbol{x}'}{n}\qw&\qw&\qw&\qw&\qw&\ghost{{\permutation_{7}^{(\outcome{\boldsymbol{x}'})}}^{n}}&\s{A}\qw&\qw&\qw&
			} = \quad \!\! \myQcircuit{
				&\s{A}\qw&\qw&	
			}
		\end{equation*}
	\end{widetext}
	where we used the fact that the deterministic effect is unique and uniquely splits over $\systemSequence{J}{n}\systemSequence{C'}{n}$. Exploiting the naturality property of the permutations one can then make the deterministic effects slide to the left obtaining
	\begin{widetext}
		\begin{equation}
			\label{eqt:proof:MSOPT:atomicity:14}
			\lim_{n \to \infty} \sum_{\outcomeIncluded{\boldsymbol{x}'}{\boldsymbol{\outcomeSpace X}'}} \; \myQcircuit{
				&\s{A}\qw&\gate{{\event{T}{\boldsymbol{x}'}'}^{n}}&\sSequenceIndex{A}{k-2}{n}\qw&\multigate{2}{{\permutation_{5}^{(\outcome{\boldsymbol{x}'})}}^{n}}&\qw&\sSequenceConditioned{F}{\boldsymbol{x}'}{n}\qw&\measureD{\observationUniqueDeterministic}&
				\\
				&\pureghost{}&\pureghost{}&\pureghost{}&\pureghost{{\permutation_{5}^{(\outcome{\boldsymbol{x}'})}}^{n}}&\qw&\sSequencePrime{C}{n}\qw&\measureD{\observationUniqueDeterministic\left(\outcome{\boldsymbol{x}}\right)}&\pureghost{}&\prepareC{{\preparationEventNoDown{\sigma}^{(\outcome{\boldsymbol{x}'})}}^{n}}&\qw&\sSequenceConditioned{G'}{\boldsymbol{x}'}{n}\qw&\qw&\multigate{1}{{\permutation_{7}^{(\outcome{\boldsymbol{x}'})}}^{n}}&\s{A}\qw&\qw&
				\\	
				&\pureghost{}&\pureghost{}&\pureghost{}&\pureghost{{\permutation_{5}^{(\outcome{\boldsymbol{x}'})}}^{n}}&\qw&\qw&\qw&\sSequenceConditioned{H}{\boldsymbol{x}'}{n}\qw&\qw&\qw&\qw&\qw&\ghost{{\permutation_{7}^{(\outcome{\boldsymbol{x}'})}}^{n}}&
			} = \quad \!\! \myQcircuit{
				&\s{A}\qw&\qw&	
			}.
		\end{equation}
	\end{widetext}
	We can now apply the permutation $\permutation_{7}^{(\outcome{\boldsymbol{x}'})}$ (its inverse) on the left (right) side of both circuits, and repeat the argument after~\eqref{eqt:proof:MSOPT:atomicity:7} to conclude that ${\systemConditioned{G'}{\outcome{\boldsymbol{x}'}}}^{n}$ stabilizes to the trivial system. The channel corresponding to last two steps ($k$ and $k-1$) is thus of the form
	\begin{equation}
		\label{eqt:proof:MSOPT:atomicity-bis}
		\myQcircuit{
			&\sSequenceIndex{A}{k-2}{n}\qw&\multigate{2}{{\permutation_{5}^{(\outcome{\boldsymbol{x}'})}}}&\qw&\sSequence{F}{n}\qw&\measureD{\observationUniqueDeterministic\left(\outcome{\boldsymbol{x}'}\right)}&
			\\
			&\pureghost{}&\pureghost{{\permutation_{5}^{(\outcome{\boldsymbol{x}'})}}}&\qw&\sSequencePrime{C}{n}\qw&\measureD{\observationUniqueDeterministic\left(\outcome{\boldsymbol{x}}\right)}
			\\	
			&\pureghost{}&\pureghost{{\permutation_{5}^{(\outcome{\boldsymbol{x}'})}}}&\qw&\s{A}\qw&\qw&\qw&\qw&
		}.
	\end{equation}
    where the dependence of $\boldsymbol{x}'$ disappears in the deterministic effect, but again we keep track of $\boldsymbol{x}'$ since the single events in the observation-test generally depend on it.
    By iteration, the condition for the coarse-graining in~\eqref{eqt:proof:MSOPT:atomicity:4} reduces to:
	\begin{equation}
		\label{eqt:final-deterministic}
		\myQcircuit{
           &\s{A}\qw&\multigate{1}{\permutation_{8}}&\s{L}\qw&\measureD{\observationUniqueDeterministic{(\boldsymbol{x})}}&
			\\
			&\pureghost{}&\pureghost{\permutation_{8}}&\s{A}\qw&\qw&\qw&
		}= \quad \!\! \myQcircuit{
				&\s{A}\qw&\qw&	
		},
	\end{equation}
	for a suitable permutation $\permutation_{8}$ and system $\system{L}$ (notice that by \autoref{thm:OPT:minimal:transf:stabilization} the dependence on $n$ is dropped). The last condition requires the system $\system{L}$ to be the trivial system.
	\\
	
	\noindent\Iref{it:transformations}: Here we check the form of the transformations that appear in the instrument of the sequence \eqref{eqt:proof:MSOPT:atomicity:3}, and verify that all of them must be proportional to the identity if \eqref{eqt:proof:MSOPT:atomicity:4} is satisfied. The idea is to repeat the steps above for transformations instead of their coarse-graining. Since, we are now interested in the single transformations that form the instruments of sequence \eqref{eqt:proof:MSOPT:atomicity:3} we consider a fixed outcome $\left( \outcome{\boldsymbol{x}'}, \outcome{x}_{k-1}, \outcome{x}_{k} \right) \in \cartesianProduct{\boldsymbol{\outcomeSpace{X}}'}{\cartesianProduct{\outcomeSpace{X}_{k-1}}{\outcomeSpace{X}_{k}}}$.
	\begin{widetext}
		\begin{equation}
			\label{eqt:proof:MSOPT:atomicity-transf}
			\begin{aligned}
				\lim_{m_{k-1} \to \infty} \; \lim_{m_{k} \to \infty} \quad
				&\myQcircuit{
					&\pureghost{}&\pureghost{}&\pureghost{}&\multiprepareC{1}{{\preparationEventNoDown{\Phi}_{\outcome{x_{k-1}}}^{\left(\outcome{\boldsymbol{x}'}\right)}}^{m_{k-1},n}}&\qw&\qw&\qw&\qw&\sSequenceConditioned{H'}{\boldsymbol{x}'}{m_{k-1},n}\qw&\qw&\qw&\qw&\qw&\multimeasureD{1}{{\observationEventNoDown{B}_{\outcome{x_{k-1}}}^{\left(\outcome{\boldsymbol{x}'}\right)}}^{m_{k-1},n}}&
					\\
					&\pureghost{}&\pureghost{}&\pureghost{}&\pureghost{{\preparationEventNoDown{\Phi}_{\outcome{x_{k-1}}}^{\left(\outcome{\boldsymbol{x}'}\right)}}^{m_{k-1},n}}&\qw&\qw&\sSequenceConditioned{G}{\boldsymbol{x}'}{m_{k-1},n}\qw&\qw&\braidingSym&\qw&\sSequenceConditioned{F}{\boldsymbol{x}'}{m_{k-1},n}\qw&\qw&\qw&\ghost{{\observationEventNoDown{B}_{\outcome{x_{k-1}}}^{\left(\outcome{\boldsymbol{x}'}\right)}}^{m_{k-1},n}}&
					\\
					&\s{A}\qw&\gate{{\event{T}{\boldsymbol{x}'}'}^{n}}&\sSequenceIndex{A}{k-2}{n}\qw&\multigate{1}{{\permutation_{5}^{(\outcome{\boldsymbol{x}'})}}^{m_{k-1},n}}&\qw&\qw&\sSequenceConditioned{F}{\boldsymbol{x}'}{m_{k-1},n}\qw&\qw&\braidingGhost&\qw&\sSequenceConditioned{G}{\boldsymbol{x}'}{m_{k-1},n}\qw&\qw&\qw&\multigate{1}{{\permutation_{6}^{(\outcome{\boldsymbol{x}'})}}^{m_{k-1},n}}&\sSequenceIndex{A}{k}{n-1}\qw&\qw&
					\\
					&\pureghost{}&\pureghost{}&\pureghost{}&\pureghost{{\permutation_{5}^{(\outcome{\boldsymbol{x}'})}}^{m_{k-1},n}}&\qw&\qw&\qw&\qw&\sSequenceConditioned{H}{\boldsymbol{x}'}{m_{k-1},n}\qw&\qw&\qw&\qw&\qw&\ghost{{\permutation_{6}^{(\outcome{\boldsymbol{x}'})}}^{m_{k-1},n}}&\pureghost{}&
				} \cdots \\[10pt] & \cdots \quad
				\myQcircuit{
					&\pureghost{}&\multiprepareC{1}{{\preparationEventNoDown{\Psi}_{\outcome{x_{k}}}^{\left(\outcome{\boldsymbol{x}}\right)}}^{m_{k},n}}&\qw&\qw&\qw&\qw&\sSequenceConditioned{E'}{\boldsymbol{x}}{m_{k},n}\qw&\qw&\qw&\qw&\qw&\multimeasureD{1}{{\observationEventNoDown{A}_{\outcome{x_{k}}}^{\left(\outcome{\boldsymbol{x}}\right)}}^{m_{k},n}}&
					\\
					&\pureghost{}&\pureghost{{\preparationEventNoDown{\Psi}_{\outcome{x_{k}}}^{\left(\outcome{\boldsymbol{x}}\right)}}^{m_{k},n}}&\qw&\qw&\sSequenceConditioned{D}{\boldsymbol{x}}{m_{k},n}\qw&\qw&\braidingSym&\qw&\sSequenceConditioned{C}{\boldsymbol{x}}{m_{k},n}\qw&\qw&\qw&\ghost{{\observationEventNoDown{A}_{\outcome{x_{k}}}^{\left(\outcome{\boldsymbol{x}}\right)}}^{m_{k},n}}&
					\\
					&\sSequenceIndex{A}{k-1}{n}\qw&\multigate{1}{{\permutation_{1}^{(\outcome{\boldsymbol{x}})}}^{m_{k},n}}&\qw&\qw&\sSequenceConditioned{C}{\boldsymbol{x}}{m_{k},n}\qw&\qw&\braidingGhost&\qw&\sSequenceConditioned{D}{\boldsymbol{x}}{m_{k},n}\qw&\qw&\qw&\multigate{1}{{\permutation_{2}^{(\outcome{\boldsymbol{x}})}}^{m_{k},n}}&\s{A}\qw&\qw&
					\\
					&\pureghost{}&\pureghost{{\permutation_{1}^{(\outcome{\boldsymbol{x}})}}^{m_{k},n}}&\qw&\qw&\qw&\qw&\sSequenceConditioned{E}{\boldsymbol{x}}{m_{k},n}\qw&\qw&\qw&\qw&\qw&\ghost{{\permutation_{2}^{(\outcome{\boldsymbol{x}})}}^{m_{k},n}}&\pureghost{}&
				},
			\end{aligned}
		\end{equation}	
	\end{widetext}
	We highlighted the outcomes of the last two transformations, the $k$-th and the $(k-1)$-th ones, in the sequential composition which we write explicitly as limits with respect to the variables $m_{k-1}$ and $m_{k}$, respectively. This is due to the fact that the transformations within an instrument of a minimal theory are generally limits of sequences of transformations of the form \eqref{eqt:OPT:minimal:symmetric:transf:generic}. Following the same argument that led to \eqref{eqt:proof:MSOPT:atomicity-bis} in the case of transformations (\aref{app:proof:MSOPT:atomicity}), one proves that \eqref{eqt:proof:MSOPT:atomicity-transf} reduces to
	\begin{widetext}
		\begin{equation}
			\label{eqt:proof:MSOPT:atomicity:16-bis}
			\lim_{m_{k-1} \to \infty} \; \lim_{m_{k} \to \infty} \;
			\myQcircuit{
				&\s{A}\qw&\gate{{\event{T}{\boldsymbol{x}'}'}^{n}}&\sSequenceIndex{A}{k-2}{n}\qw&\multigate{2}{\permutation_{7}^{n}}&\sSequence{F}{n}\qw&\measureD{{\observationEventNoDown{b}_{\outcome{x_{k-1}}}^{\left(\outcome{\boldsymbol{x}'}\right)}}^{m_{k-1},n}}&
				\\
				&\pureghost{}&\pureghost{}&\pureghost{}&\pureghost{\permutation_{7}^{n}}&\sSequence{C'}{n}\qw&\measureD{{{\observationEventNoDown{a}_{\outcome{x_{k-1},x_{k}}}}^{\left(\outcome{\boldsymbol{x}}\right)}}^{m_{k-1},m_{k},n}}&
				\\
				&\pureghost{}&\pureghost{}&\pureghost{}&\pureghost{\permutation_{7}^{n}}&\s{A}\qw&\qw&\qw&
			},
		\end{equation}
	\end{widetext}
	where suitable observation-tests appear in place of the deterministic effect of \eqref{eqt:proof:MSOPT:atomicity-bis}. Further iteration leads, in analogy with~\eqref{eqt:final-deterministic}, to the following expression for $\lim_{n \to \infty}\mathscr{G}_{\outcomeSpace{X}_{k}}^{(\boldsymbol{x})n}\circ\eventSequence{T}{\boldsymbol{x}}{n}$
	\begin{align}
		\label{eqt:proof:MSOPT:atomicity:17-bis}
	&	\lim_{m_{1} \to \infty} \; \lim_{\overline{\boldsymbol m} \to \infty} \;
		\myQcircuit{
			&\pureghost{}&\multiprepareC{1}{\preparationEvent{\Gamma}{x_{1}}^{m_{1},n}}&\qw&\sSequencePrime{N}{m_{1},n}\qw&\qw&\multimeasureD{1}{\observationEvent{C}{x_{1}}^{m_{1},n}}&
			\\
			&\pureghost{}&\pureghost{\preparationEvent{\Gamma}{x_{1}}^{m_{1},n}}&\sSequence{P}{n}\qw&\braidingSym&\s{L}\qw&\ghost{\observationEvent{C}{x_{1}}^{m_{1},n}}&
			\\
			&\s{A}\qw&\multigate{1}{\permutation_{8}}&\s{L}\qw&\braidingGhost&\sSequence{P}{n}\qw&\measureD{{\observationEvent{d}{\overline{\boldsymbol x}}^{\left( \outcome{\boldsymbol{x}} \right)}}^{\overline{\boldsymbol m},n}}&
			\\
			&\pureghost{}&\pureghost{\permutation_{8}}&\qw&\s{A}\qw&\qw&\qw&\qw&
      	}\,,
	\end{align}
	where $\overline{\boldsymbol m} = (m_{2}, \ldots, m_{k})$, $\outcome{\overline{\boldsymbol x}} = \left( \outcome{x_{2}, \ldots, x_{k}} \right)$. Reminding the triviality of system $\system{L}$, we then obtain that any transformation in the instrument \eqref{eqt:proof:MSOPT:atomicity:3} must be of the form
	\begin{equation}
		\label{eqt:proof:MSOPT:atomicity:18-bis}
		\lim_{m_{1} \to \infty} \; \lim_{\overline{\boldsymbol m} \to \infty} \; 
			\myQcircuit{			&\multiprepareC{1}{\preparationEvent{\Gamma}{x_{1}}^{m_{1},n}}&\qw&\sSequencePrime{N}{m_{1},n}\qw&\qw&\measureD{\observationEvent{C}{x_{1}}^{m_{1},n}}&
			\\
			&\pureghost{\preparationEvent{\Gamma}{x_{1}}^{m_{1},n}}&\qw&\sSequence{P}{n}\qw&\qw&\measureD{{\observationEvent{d}{\overline{\boldsymbol x}}^{\left( \outcome{\boldsymbol{x}} \right)}}^{\overline{\boldsymbol m},n}}&
			\\
			&\qw&\qw&\s{A}\qw&\qw&\qw&\qw&
		},
	\end{equation}
	namely a probability multiplied by the identity map for system $\system{A}$.
	
	Summarizing, we started from a generic Cauchy sequence of instruments of 
	an \ac{MSOPT} obtained from elementary processes and showed that if it is coarse-grained to the identity then it is of the form
	\begin{equation*}
		\left\{ \lim_{m \to \infty} \probabilityEvent{p}{x}^{m,n} \identityTest{A}
		\right\}_{\outcomeIncluded{x}{X}} \quad \forall n \in \mathbb{N},
	\end{equation*}
	where $m = m_{1}, \ldots, m_{k}$, $\outcome{x} = \left( \outcome{x_{1}, \ldots, x_{k}} \right)$, and $\outcomeSpace{X} = \cartesianProduct{\outcomeSpace{X}_{1}}{\cartesianProduct{\cdots}{\outcomeSpace{X}_{k}}}$. Since the limit of a Cauchy sequence of probability distributions is still a probability distribution, the limit instrument is equal to  
	\begin{equation*}
		\left\{ \probabilityEventNoDown{p}'_{\outcome{x}} \identityTest{A}
		\right\}_{\outcomeIncluded{x}{X}},
	\end{equation*}
	which concludes the proof.
\end{proof}

\subsection{Properties}

As in the case of minimal theories, the atomicity of the identity leads to a series of properties of \acp{MSOPT}, which persist even after the completion by strong causality.

First, \acp{MSOPT} that satisfy the hypothesis of \autoref{thm:MSOPT:symmetric:idAtomicity} do not admit reversible transformations different from permutations:

\begin{corollary}\label{corol:msopt:revTransfPerm}
	Every \ac{MSOPT} satisfying the hypothesis of \autoref{thm:MSOPT:symmetric:idAtomicity} does not admit of reversible transformations different from permutations.
\end{corollary}

The proof follows that of \autoref{lem:mopt:atomRev}, simply noticing that also in this case the argument used in the proof of \autoref{thm:MSOPT:symmetric:idAtomicity} guarantees that the only sequence of deterministic transformations that can converge to the identity is the constant one.

Continuing to analyse the properties satisfied by the \acp{MSOPT}, from \autoref{thm:MSOPT:symmetric:idAtomicity}  
a property that follows, being equivalent to the atomicity of the identity, is \ac{NIWD}~\cite{darianoInformationDisturbanceOperational2020}:
\begin{corollary}
	\label{corol:MSOPT:niwd}
	Every \ac{MSOPT} satisfying the hypothesis of \autoref{thm:MSOPT:symmetric:idAtomicity} has \ac{NIWD}.
\end{corollary}

One also has irreversibility, as immediately follows from~\autoref{thm:opt:irrev:idAtomicity}.
\begin{corollary}
	\label{corol:MSOPT:irrev}
	Every non-trivial \ac{MSOPT} satisfying the hypothesis of \autoref{thm:MSOPT:symmetric:idAtomicity} has irreversibility.
\end{corollary}

The last feature, which derives directly from \autoref{lem:opt:broadcasting:idAtomic}, is no-broadcasting:
\begin{corollary}
	\label{corol:MSOPT:broadcasting}
	Every \ac{MSOPT} satisfying the hypothesis of \autoref{thm:MSOPT:symmetric:idAtomicity} does not admit of a broadcasting channel.
\end{corollary}

The only property of minimal theories that has not been proved also in the presence of conditional tests is the no-programming theorem (\autoref{def:opt:programming}). This result would require a complete characterisation of the deterministic transformations which is still lacking.

Finally, we would like to remark that if, on the contrary, the hypothesis of \autoref{thm:MSOPT:symmetric:idAtomicity} are not satisfied, one can have \acp{MSOPT} with none of the above properties. The straightforward example is the \ac{MSOPT} completion of \ac{MCT}~\cite{erbaMeasurementIncompatibilityStrictly2024}, which is the standard classical theory. This latter result is discussed more in detail at the end of \autoref{sec:axioms-for-ct}.

\section{Classical toy-models}

\label{sec:msbct}
In this section we introduce \acf{MSBCT}, a toy-theory that is locally classical and satisfies strong causality, but still has irreversibility and does not admit a broadcasting channel. The theory is the minimal version of \ac{BCT}~\cite{darianoClassicalTheoriesEntanglement2020} completed under strong causality. 

We formalise the theory through the following postulates.  

\subsection{Postulates of the theory}

The first postulate is on the nature of the systems of the theory. We call \textdef{classical} any theory such that the set of states of every system has a simplicial basis, all of whose vertices of the simplex belong to a jointly perfectly discriminable set~\cite{darianoInformationDisturbanceOperational2020}. If all the vertices are jointly perfectly discriminable, this also implies that the theory is convex, every state is proportional to a deterministic one, and, consequently, that the theory is causal~\cite{darianoQuantumTheoryFirst2016,darianoClassicalTheoriesEntanglement2020}. In the latter case the set of states $\St A$ is a simplex whose vertices are the states in the perfectly discriminable set and the null state $\nullState A$, and $\StN A$ is a simplex itself whose vertices are all those of $\St A$ except $\nullState A$. 
\begin{postulate}[Classicality, convexity and system types]
	\label{pos:MSBCT:1}
	The theory \theory\ is classical and the vertices of state sets are jointly perfectly discriminable. In addition to the trivial system, for every integer $\sysDimension{} \geq 1$, $\Sys{\theory}$ contains a type of system of dimension $\sysDimension{}$.
\end{postulate}
As observed above, the theory is causal as a consequence of the first postulate. In the following, $\observationUniqueDeterministic_\system A$ will denote the unique deterministic effect of system $\system A$, that amounts to 1 on the vertices of the simplex (except for the vertex given by the null state $\nullState A$). The second postulate fixes a parallel composition rule that differs from the tensor product.

\begin{postulate}[Parallel composition of systems and states]
	\label{pos:MSBCT:2}
	For every pair of systems \system{A}, $\system{B} \in \Sys{\theory}$, the dimension of the composite system \system{AB} is given by the following rule:
	\begin{equation*}
		\sysDimension{AB} = \sysDimension{BA} = 	\begin{cases}
			2 \sysDimension{A} \sysDimension{B}, \quad &\textit{if} \; \; \system{A}, \system{B} \neq \trivialSystem,\\
			\sysDimension{A}, & \textit{if} \; \; \system{B} = \trivialSystem.
		\end{cases}
	\end{equation*}
	Let $\trivialSystem \neq \system{A}, \system{B}, \system{C} \in \Sys{\theory}$. Denoting the pure states of the composite system $\system{AB}$ as $\PurSt{AB} = \left\{ \MSBCTPureState{i}{j}{-}, \MSBCTPureState{i}{j}{+} \mid 1 \leq i \leq \sysDimension{A}, 1 \leq j \leq \sysDimension{B} \right\}$, for all state $\preparationEventNoDown{i} \in \PurSt{\system{A}}$, $\observationEventNoDown{j} \in \PurSt{\system{B}}$ the following parallel composition rule holds:
	\begin{equation}
		\label{eqt:MSBCT:postulate2}
		\myQcircuit{
			&\prepareC{i}&\s{A}\qw&\qw&
			\\
			&\prepareC{j}&\s{B}\qw&\qw&
		} = \frac{1}{2} \sum_{s = +,-}
		\myQcircuit{
			&\multiprepareC{1}{\MSBCTPureState{i}{j}{s}}&\s{A}\qw&\qw&
			\\
			&\pureghost{\MSBCTPureState{i}{j}{s}}&\s{B}\qw&\qw&
		}.
	\end{equation}
	The relation between states of the same composite system defined 
	via compositions in different orders is given by 
	\begin{equation}
		\MSBCTPureState{\MSBCTPureState{i}{j}{s_1}}{l}{s_2} = \MSBCTPureState{i}{\MSBCTPureState{j}{l}{s_1 s_2}}{s_1},
	\end{equation}
	for all indices $i, j, k$ and signs $s_1, s_2$.
\end{postulate}

The former postulate also implies that \ac{MSBCT} is a bilocal theory~\cite{darianoClassicalTheoriesEntanglement2020}. In words, this means that it is not possible to discriminate two different states of a composite system by local measurements (and classical communication), as it is instead possible to do in the case of \ac{CT} or \ac{QT} which have local discriminability~\cite{10.1063/1.2219356,hardyLimitedHolismRealVectorSpace2012,darianoQuantumTheoryFirst2016,darianoClassicalTheoriesEntanglement2020}. In the case of bilocal theories, one has generally to perform measurements also on any possible couple of subsystems. For example, considering the state $\MSBCTPureState{i}{j}{s} \in \St{AB}$, with $s = \pm$, by simply measuring the local state of the system $\preparationEventNoDown{i} \in \St{A}$ and $\preparationEventNoDown{j} \in \St{B}$, one would not obtain any information about the sign $s$. The latter information could be recovered only by performing a measurement on the composite system \system{AB}.

\begin{postulate}[Swap]
	\label{pos:MSBCT:3}
	The theory is symmetric. Considering $\trivialSystem \neq \system{A}, \system{B}, \system{E} \in \Sys{\theory}$, the swap $\BraidingS_{\system A,\system B}$ is defined as follows
	\begin{equation}
		\myQcircuitSmall{
			&\multiprepareC{2}{\MSBCTPureState{\MSBCTPureState{i}{j}{s_1}}{k}{s_2}}&\s{A}\qw&\braidingSym&\s{B}\qw&\qw&
			\\
			&\pureghost{\MSBCTPureState{\MSBCTPureState{i}{j}{s_1}}{k}{s_2}}&\s{B}\qw&\braidingGhost&\s{A}\qw&\qw&
			\\
			&\pureghost{\MSBCTPureState{\MSBCTPureState{i}{j}{s_1}}{k}{s_2}}&\qw&\s{E}\qw&\qw&\qw&
		} = 
		\myQcircuitSmall{
			&\multiprepareC{2}{\MSBCTPureState{\MSBCTPureState{j}{i}{s_1}}{k}{s_1 s_2}}&\s{B}\qw&\qw&
			\\
			&\pureghost{\MSBCTPureState{\MSBCTPureState{j}{i}{s_1}}{k}{s_1 s_2}}&\s{A}\qw&\qw&
			\\
			&\pureghost{\MSBCTPureState{\MSBCTPureState{j}{i}{s_1}}{k}{s_1 s_2}}&\s{E}\qw&\qw&
		}.
	\end{equation}
\end{postulate}

\begin{postulate}[Preparation- and observation-instruments]
	\label{pos:MSBCT:4}
	Given any system $\system{A} \in \Sys{\theory}$, a collection $\preparationEventTest{\rho}{x}{X} \subset \St{A}$ is a preparation instrument if and only if $\sum_{\outcomeIncluded{x}{X}} \rbraketSystem{\observationUniqueDeterministic}{\preparationEvent{\rho}{i}}{A} = 1$. The observation-instruments of every system $\system{A} \in \Sys{\theory}$ are all the collections $\observationEventTest{a}{y}{Y} \subset \EffR{A}$ of generalised effects such that $\left\{ \parallelComp{\observationEvent{a}{y}}{\identityTest{E}} \right\}_{\outcomeIncluded{y}{Y}}$ maps preparation instruments of $\system{AE}$ to preparation instruments of $\system{E}$ for all $\system{E} \in \Sys{\theory}$.
\end{postulate}
The sets $\EffR{A}$, for every $\system{A} \in \Sys{\theory}$, are identified  by \autoref{pos:MSBCT:1} through the property of \emph{joint perfect discriminability}~\cite{darianoClassicalTheoriesEntanglement2020}.

We highlight that the postulates presented so far are in common with the ones of \ac{BCT}. The two theories are set apart by the spaces of instruments and transformations. 

\begin{postulate}[Minimal Strong-causality]
	\label{pos:MSBCT:5}
	The theory \theory\ is minimal and strongly causal as in \autoref{def:OPT:minimalsc}.
\end{postulate}
Notice that the postulate actually defines a class of theories rather than a single theory. \autoref{pos:MSBCT:5} prescribes only that the theory \theory\ must satisfy the uniqueness of the decomposition of systems into subsystems (\autoref{def:opt:system:uniqueDec}), without specifying which set of elementary systems is chosen. This aspect itself is not of particular importance, since the results discussed here are valid regardless of this choice. However, for completeness, we will make a couple of remarks in this regard.

The most conservative choice is that there exists an elementary system of every dimension. This models the possibility of the existence of high-dimensional non-composite systems. For example, in the case of \ac{MSBCT}, there is no reason to exclude that a system of dimension 8 could be elementary, and not the composition of two systems of dimension 2.

However, it is also possible to make a more minimalist choice in which high-dimensional systems are always seen as compositions of smaller-dimensional ones. In particular, in the case of \ac{MSBCT}, this would involve defining that the elementary systems of the theory are those of odd dimension, given that starting from them, through the compositional rule described in \autoref{pos:MSBCT:2}, it becomes possible to reconstruct systems of any dimension. This can be proven by observing that a system of dimension 1 that is not operationally equivalent to the trivial system enables the reconstruction of all even-dimensional systems. Furthermore, the fact that this is the minimal choice can be demonstrated by utilizing the fact that the only way to obtain an odd number as a product of other numbers is for all of those numbers to be odd.

In the following, we write \ac{MSBCT} as if there were only one theory because the properties that we illustrate actually hold in any of the theories obtained via possible choices of elementary systems.

\subsection{Properties}

Having introduced the theory, we move on to illustrate its main properties.
First of all, it is self-evident that \ac{MSBCT} is \emph{locally classical}. In other words, if we only consider local operations and measurements on composite systems, the theory cannot be discriminated from classical theory. For example, local states admit a broadcasting channel. More precisely, forgetting for a moment \eqref{def:opt:bradcasting}, if the broadcasting channel was defined disregarding the need of copying also arbitrary correlations, i.e., if a map $\eventNoDown{B}' \in \TransfN{A}{AA}$ was defined only requiring that
\begin{align*}
	\myQcircuit{
		&\prepareC{\preparationEventNoDown{\rho}}&\s{A}\qw&\measureD{\observationEventNoDown{a}}&
	} & =
	\myQcircuit{
		&\prepareC{\preparationEventNoDown{\rho}}&\s{A}\qw&\multigate{1}{\eventNoDown{B}'}&\s{A}\qw&\measureD{\observationUniqueDeterministic}&
		\\
		&\pureghost{}&\pureghost{}&\pureghost{\eventNoDown{B}'}&\s{A}\qw&\measureD{\observationEventNoDown{a}}&
	} \\[10pt]
	& =
	\myQcircuit{
		&\prepareC{\preparationEventNoDown{\rho}}&\s{A}\qw&\multigate{1}{\eventNoDown{B}'}&\s{A}\qw&\measureD{\observationEventNoDown{a}}&
		\\
		&\pureghost{}&\pureghost{}&\pureghost{\eventNoDown{B}'}&\s{A}\qw&\measureD{\observationUniqueDeterministic}&
	}
\end{align*}
for any state $\preparationEventNoDown{\rho} \in \St{A}$ and $\observationEventNoDown{a} \in \Eff{A}$, then \eqref{eqt:opt:transf:broadcast:classical} has the required features for a broadcasting channel. However, \ac{BCT}---that has the same states as \ac{MSBCT}---features entanglement \cite{darianoClassicalityLocalDiscriminability2020}. There are then data stored in entangled states that would not be broadcast by the above channel 
$\eventNoDown{B}'$. In particular, the data that would be destroyed by the above channel are represented by the sign $s$ of pure states $\MSBCTPureState{i}{j}{s}$. Therefore, if one considers the strict definition of broadcasting channel given by \eqref{def:opt:bradcasting}, then \ac{MSBCT} does not admit such a map. This descends form \autoref{lem:opt:broadcasting:idAtomic}, since \ac{MSBCT} admits of systems with dimension greater than 1 whose associated identity transformation is atomic.

In addition to satisfying the no-broadcasting property, \ac{MSBCT} also has irreversibility of measurement disturbance and it satisfies the property of \ac{NIWD}. Due to the fact that the entangled states of \ac{MSBCT} are spanning for the generalised state spaces $\StR{AB}$ of every composite system \system{AB}, \ac{MSBCT} satisfies the hypothesis of \autoref{thm:MSOPT:symmetric:idAtomicity} and, consequently, also \autoref{corol:MSOPT:irrev} and \autoref{corol:MSOPT:niwd}.

\ac{MSBCT} provides then an evidence that strong causality and classicality do not prevent no-broadcasting and \ac{NIWD}. Also, strong causality and classicality do not forbid the existence of ``incompatible'' dynamics (irreversibility), despite the fact that all observation-tests of \ac{MSBCT} are compatible (since their set is equal to that of the observation-tests of \ac{BCT} which in turn satisfy this property~\cite{darianoClassicalityLocalDiscriminability2020}).

We observe that all the properties discussed so far also hold in the case of \ac{MBCT}, since they only descend from the nature of state spaces along with the fact that the theory is an \acp{MOPT}. However, the fact that these properties hold in a minimal classical theory \emph{without} strong causality is not particularly interesting, as they were already been shown in Ref.~\cite{erbaMeasurementIncompatibilityStrictly2024}, even without bilocal discriminability, through the construction of a minimal version of \ac{CT}, named \ac{MCT}.

Both in \ac{MBCT} and in \ac{MSBCT} no mixed state has a purification, and the operational superposition principle~\cite{chiribellaInformationalDerivationQuantum2011,darianoQuantumTheoryFirst2016} is not satisfied. The latter features are a consequence of the no-go results proven in Ref.~\cite{darianoClassicalTheoriesEntanglement2020} for general simplicial theories that are not locally tomographic.

Finally, the results concerning the characterization of reversible transformations also allow us to prove that \ac{MBCT} and \ac{MSBCT} do not satisfy the property of \textdef{essential uniqueness of purification}~\cite{darianoQuantumTheoryFirst2016} here reported for the convenience of the reader: 
\begin{property}[Uniqueness of purification]
	Let \system{A} and \system{B} be a couple of generic systems of an \ac{OPT}. If there exist $\preparationEventNoDown{\Sigma}_{1}$, $\preparationEventNoDown{\Sigma}_{2} \in \PurSt{AB}$ and $\observationUniqueDeterministic_{1}$, $\observationUniqueDeterministic_{2} \in \EffN{B}$ such that:
	\begin{equation*}
		\myQcircuit{
			&\multiprepareC{1}{\preparationEventNoDown{\Sigma}_{1}}&\qw&\s{A}\qw&\qw&\qw&
			\\
			&\pureghost{\preparationEventNoDown{\Sigma}_{1}}&\s{B}\qw&\measureD{\observationUniqueDeterministic_{1}}
		} = \myQcircuit{
			&\multiprepareC{1}{\preparationEventNoDown{\Sigma}_{2}}&\qw&\s{A}\qw&\qw&\qw&
			\\
			&\pureghost{\preparationEventNoDown{\Sigma}_{2}}&\s{B}\qw&\measureD{\observationUniqueDeterministic_{2}}
		},
	\end{equation*}
	then there exists $\eventNoDown{R} \in \RevTransf{B}{B}$ such that:
	\begin{equation*}
		\myQcircuit{
			&\multiprepareC{1}{\preparationEventNoDown{\Sigma}_{1}}&\s{A}\qw&\qw&
			\\
			&\pureghost{\preparationEventNoDown{\Sigma}_{1}}&\s{B}\qw&\qw&
		} = \myQcircuit{
			&\multiprepareC{1}{\preparationEventNoDown{\Sigma}_{2}}&\qw&\s{A}\qw&\qw&\qw&
			\\
			&\pureghost{\preparationEventNoDown{\Sigma}_{2}}&\s{B}\qw&\gate{\eventNoDown{R}}&\s{B}\qw&\qw&
		}.
	\end{equation*}
\end{property}
In a causal theory one would have $\observationUniqueDeterministic_{1} = \observationUniqueDeterministic_{2}$, corresponding to the unique deterministic effect.

The fact that neither \ac{MBCT} nor \ac{MSBCT} satisfy the property just introduced derives from the fact that no permutation in $\RevTransf{B}{B}$ can map a pure state of the form $\MSBCTPureState{i}{j}{s} \in \PurSt{AB}$ into one of the form $\MSBCTPureState{i'}{j'}{s'} \in \PurSt{AB}$ whenever, for example, $\preparationEventNoDown{i} \neq \preparationEventNoDown{i}'$. This follows directly from the definition of the swap operation in these theories (\autoref{pos:MSBCT:3}), since it cannot modify the state of a system on which it is not acting directly.

In conclusion of this section, we observe that the fact that the identity transformation is atomic for every system of \ac{MSBCT} sets it apart from \ac{BCT}. In the latter theory, in fact, there always exists a non-trivial test that decomposes the identity~\cite{darianoClassicalTheoriesEntanglement2020}. The transformation spaces of \ac{MSBCT} are \emph{strictly} contained within those of \ac{BCT}. The same result can also be derived from \autoref{corol:msopt:revTransfPerm} observing that \ac{BCT} admits reversible transformations different from permutations.

\section{On the relation between theories and their minimal versions}
\label{sec:axioms-for-ct}

In general, it may not be true that an \ac{OPT} and its minimal strongly causal version are different. In this respect, we briefly discuss the case of three theories of interest: \ac{QT}, \ac{CT}, and \ac{BCT}.

Let us look first at the quantum case and observe that in this case it is indeed possible to distinguish \ac{QT} from its minimal strongly causal version. 
First, being strongly causal~\cite{darianoQuantumTheoryFirst2016} \ac{QT} differs from its minimal version; \ac{MQT}, which for sure does not admit conditional instruments. By \ac{MQT} we mean a version of \ac{QT} where the states and effects spaces are kept unchanged, while the only admissible operations are preparations, observations, the identity and the swap.\footnote{As for \ac{MSBCT}, the most conservative approach for choosing the elementary systems in both the minimal and minimal strongly causal versions of \ac{QT} and \ac{CT} is to include an elementary system for every possible dimension. Instead, the minimal choice for the set of elementary systems in the case of \ac{QT} would consist of those whose dimensions are squares of prime numbers—i.e., systems associated with Hilbert spaces of prime dimension. Meanwhile, for \ac{CT}, the minimal choice would be to consider only systems of prime dimension. Nonetheless, as with for \ac{MSBCT}, all of the results discussed here hold regardless of the particular choice of the set of elementary systems.} The procedure to build \ac{MQT} starting from \ac{QT} is the same as the one followed to build \ac{MCT} starting from \ac{CT} in Ref.~\cite{erbaMeasurementIncompatibilityStrictly2024} or \ac{MBCT} from \ac{BCT} in \autoref{sec:msbct}. Exploiting \autoref{corol:msopt:revTransfPerm} one can then show that also the minimal strongly causal version of \ac{QT} (\ac{MSQT}) differs form \ac{QT}. Indeed \ac{QT} admits reversible transformations that are different from permutations, for example the CNOT gate.
Therefore, summarizing, in the case of \ac{QT} the following strict inclusions hold
\begin{equation*}
	\InstrOPT{\ac{MQT}} \subset \InstrOPT{\ac{MSQT}} \subset \InstrOPT{\ac{QT}},
\end{equation*}
where $\InstrOPT{\ac{OPT}}$ denotes the set of instruments of a particular \ac{OPT}.

We now move on to the case of \ac{BCT}. The fact that the identity transformation is atomic for every system of \ac{MSBCT} makes it different from \ac{BCT}. In the latter theory, indeed, there always exists a non-trivial test that decomposes the identity~\cite{darianoClassicalTheoriesEntanglement2020}. Therefore, the transformations spaces of \ac{MSBCT} are \emph{strictly} contained within those of \ac{BCT}. This result can also be derived from \autoref{corol:msopt:revTransfPerm} observing that \ac{BCT} admits reversible transformations different from permutations. Moreover, given that \ac{MSBCT} is a strongly causal version of \ac{MBCT}, it possess by definition more instruments than \ac{MBCT}. In summary one then has
\begin{equation*}
	\InstrOPT{\ac{MBCT}} \subset \InstrOPT{\ac{MSBCT}} \subset \InstrOPT{\ac{BCT}},
\end{equation*}

Last, let us consider the case of standard \ac{CT}. As anticipated, in the case of \ac{MCT}, closure by conditioning recovers the original ``full'' theory. 
The reason why a strongly causal version of \ac{MCT} coincides with \ac{CT} is that this theory would include all instruments of the form:
\begin{equation*}
	\left\{
		\myQcircuit{
			&\s{A}\qw&\measureD{\observationEventNoDown{i}}&\prepareC{\preparationEventNoDown{f(i)}}&\s{B}\qw&\qw&
		}
	\right\}_{\outcomeIncluded{i}{I}},
\end{equation*}
for an arbitrary $f:\outcomeSpace{I} \to \outcomeSpace{J}$, where the sets \outcomeSpace{I} and $\outcomeSpace{J}$ label the pure states of $\system A$ and \system{B}, respectively. Coarse-grainings of the above instruments include all reversible transformation in \ac{CT}. By the fact that every transformation of \ac{CT} admits a reversible dilation~\cite{darianoQuantumTheoryFirst2016,darianoClassicalTheoriesEntanglement2020,darianoClassicalityLocalDiscriminability2020}, the result immediately follows. On the other hand, \ac{MCT} is a strict subset of \ac{MSCT} due to the lack of conditional operations in the former. In summary, for classical theory one has
\begin{equation*}
	\InstrOPT{\ac{MCT}} \subset \InstrOPT{\ac{MSCT}} \equiv \InstrOPT{\ac{CT}},
\end{equation*}

In this light, the assumption of the existence of entangled states plays a crucial role in the proof of the atomicity of the identity in \acp{MSOPT}. Removing this assumption immediately leads to a counterexample to all the statements just made. Indeed, a minimal and strongly causal version of \ac{CT} coincides with the latter.

As a byproduct we have also proved that the three properties of simpliciality, strong causality and local discriminability are a set of independent properties.

The independence of the three properties derives from the fact that any pair of them does not imply the third one. \ac{QT} is an example of a theory that is strongly causal and locally tomographic, but non-simplicial. \ac{MCT} is a locally tomographic non-strongly causal simplicial theory. Finally, \ac{MSBCT} is a non-locally-tomographic strongly casual simplicial \ac{OPT}.

Moreover, theories can be exhibited that have exactly one out of the three properties. Indeed, \acl{FQT}~\cite{bravyiFermionicQuantumComputation2002,darianoFeynmanProblemFermionic2014} has strong causality without simpliciality and local discriminability; \ac{MBCT} has simpliciality without local discriminability and strong causality; finally, \ac{MQT} has local discriminability without simpliciality and strong causality.

\section{Discussion}
\label{sec:conclusion}
Exploring the consequences of limiting the allowed dynamics in a generic theory of information processing, we have shown that even a strongly causal classical theory can consistently support features usually regarded as distinguishing marks of non-classicality. Those are the existence of intrinsically irreversible instruments---namely operations that are not compatible with each other---, \ac{NIWD}---that is the impossibility of performing non-trivial measurements on a system without perturbing it irreversibly---, and no-broadcasting---which prevents the possibility of ``copying'' the state of a system (including its remote correlations). An explicit example of theory with all the above properties has been developed and named \ac{MSBCT}, showing that they cannot be considered \emph{per se} as signatures of non-classicality.

In order to construct \ac{MSBCT}, we introduced the class of \acp{MSOPT}. These can be seen as \acp{MOPT}~\cite{erbaMeasurementIncompatibilityStrictly2024}—i.e., \acp{OPT} where all operations are obtained through sequential and parallel composition of preparations, measurements, and permutations—extended by closure under strong causality, thereby ensuring that all conditional instruments are admissible. The minimality criterion introduces the minimal set of dynamical maps consistent with the possible states and measurements of systems, while closure by strong causality enlarges such a minimal set in order to allow for conditioning of instruments on the basis of outcomes of previous experiments. Both \acp{MOPT} and \acp{MSOPT} aim to model scenarios where operations on a physical system are severely restricted, reflecting the constraints typical of real-world laboratories. However, \acp{MSOPT} address a significant limitation of \acp{MOPT}: the inability to perform conditioned operations. Such operations are not only routinely possible in experimental settings but are also considered a necessary feature of any physical theory that seeks to accurately describe and model our reality.

The properties of \ac{MSBCT} have been shown to hold for a broader class of minimal theories. Specifically, all causal symmetric \acp{MOPT} and all symmetric \acp{MSOPT} that admit a spanning set of entangled states exhibit irreversibility of measurement disturbance, satisfy \ac{NIWD}, and do not allow broadcasting. These properties arise from the atomicity of the identity transformation, which significantly limits the set of operations that can be performed in these theories on a physical system without causing disturbance.

More specifically, \acp{MOPT} do not allow any transformations other than permutations that leave a physical system undisturbed. In the case of \acp{MSOPT}, although they may admit transformations that do not disturb the system locally (e.g., \eqref{eqt:opt:transf:broadcast:classical} in the case of \ac{MSBCT}), all non-reversible transformations irreversibly destroy correlations with the environment.

Furthermore, the atomicity of the identity transformation also implies that these theories do not admit reversible transformations other than permutations. This property allows, in some cases, the separation of an \ac{OPT} from its minimal strongly causal version. For example this separation occurs in \ac{QT} and \ac{BCT}, but not for \ac{CT}. While this result may initially seem unremarkable, it provides valuable insights into the information-processing capabilities of \ac{CT} and \ac{QT}.\\
In the classical scenario, any conceivable computation can be achieved through a sequence of measurements followed by conditioned preparations. However, this is not the case in \ac{QT}, where there are evolutions that cannot be reproduced simply by measuring and re-preparing a system. The fact that \ac{CT} can be fully recovered from its minimal version (\ac{MCT}) by simply adding all conditional instruments also highlights the crucial role of entangled states in proving the atomicity of the identity in \acp{MSOPT}.

The difference between a \ac{MOPT} and the associated \ac{MSOPT} is, instead, always guaranteed by the presence of conditional instruments. For example, instruments of the form \eqref{eqt:opt:transf:broadcast:classical} are always allowed in \acp{MSOPT} and never in \acp{MOPT}.

Getting back to the study of the properties of \ac{MSBCT}, the results presented in Ref.~\cite{darianoClassicalTheoriesEntanglement2020} guarantee that in this theory no mixed state has a purification, and that superposition of states are not admitted. Furthermore, due to the form of the reversible transformations of the theory, \ac{MSBCT} also violates essential uniqueness of purification~\cite{darianoQuantumTheoryFirst2016}. 
	
We also highlight that, even though the notion of classicality adopted by us is related to the simpliciality of the state space, this theory is also Kochen-Specker~\cite{bellProblemHiddenVariables1966,kochenProblemHiddenVariables1975,merminSimpleUnifiedForm1990,merminHiddenVariablesTwo1993} and generalised noncontextual~\cite{spekkensContextualityPreparationsTransformations2005,schmidStructureTheoremGeneralizednoncontextual2024,schmidCharacterizationNoncontextualityFramework2021,soltaniLocalNoncontextualOntological}.

Interestingly, these last few results also hold even when considering just \ac{MBCT}, i.e., the minimal version of \ac{BCT} without conditioned instruments.

\ac{MSBCT} and \ac{MBCT} also allow us to answer some open questions raised in the conclusions of our previous work~\cite{erbaMeasurementIncompatibilityStrictly2024}. First, they provide another example of theories exhibiting irreversibility despite the full compatibility of the observation-instruments. This proves that \ac{MCT} is not the only \ac{OPT} possessing these properties simultaneously. Second, \ac{MSBCT} shows that adding conditional instruments to a classical theory is not sufficient to guarantee the possibility of generalised broadcasting, highlighting the necessity of the local tomography assumption. This is also the reason why our result does not contradict the one presented in Refs.~\cite{barnumCloningBroadcastingGeneric2006,barnumGeneralizedNoBroadcastingTheorem2007}.

The construction of \ac{MSBCT}, along with the minimal versions of \ac{CT} and \ac{QT}, also further characterises the relationships between different physical properties, providing a useful contribution to the axiomatisation program of \ac{QT}. Specifically, we demonstrate that the three properties of simpliciality, strong causality, and local discriminability are pairwise independent (\autoref{sec:axioms-for-ct}).
  
As final remarks, we list some topics that in our opinion deserve further investigation.

First, we note that the class of \acp{MOPT} enables the construction of \acp{OPT} by focusing solely on the state and measurement spaces. This approach is similar to that used in prepare-and-measure scenarios commonly found in the \ac{GPT} literature. Therefore, given that \acp{MOPT} are minimal working examples of \acp{OPT} that can be extracted from \acp{GPT} defined in the prepare-and-measure scenario, they represent a possible bridge between these two areas of research.

Second, the fact that \acp{MOPT} and \acp{MSOPT} satisfy the \ac{NIWD} property suggests the potential for studying whether they could support information-theoretically secure cryptographic protocols for key distribution, similar to \ac{QT}~\cite{bennettQuantumCryptographyUsing1992,bennettExperimentalQuantumCryptography1992,barnettInformationtheoreticLimitsQuantum1993,ekertEavesdroppingQuantumcryptographicalSystems1994,kentUnconditionallySecureBit1999,kentUnconditionallySecureBit1999,loUnconditionalSecurityQuantum1999,mayersUnconditionalSecurityQuantum2001,nielsenQuantumComputationQuantum2010,bennettQuantumCryptographyPublic2014}. The strong restrictions on the set of transformations would allow for greater control over all the variables involved in a communication process, thus allowing the problem to be studied in a simplified setting.  This would provide interesting insights into the properties that generic information processing theories are required to satisfy to admit intrinsically secure cryptographic protocols. In particular, it would be interesting to determine whether an \ac{OPT} like \ac{MSBCT} could admit one, given that it is classical.  We hypothesize that this may be feasible by leveraging the fact that an eavesdropper’s intervention could be detected, as any attempt to obtain information about the system would irreversibly alter it due to the \ac{NIWD} property.

In conclusion, we observe that \acp{MSOPT} open the possibility of studying the paradigm of \emph{measurement-based computation}~\cite{gottesmanDemonstratingViabilityUniversal1999,raussendorfOneWayQuantumComputer2001,leungTwoqubitProjectiveMeasurements2002,nielsenQuantumComputationMeasurement2003,raussendorfMeasurementbasedQuantumComputation2003,leungQuantumComputationMeasurements2004,jozsaIntroductionMeasurementBased2005} in an operational setting, given that this computational paradigm consists in carrying out a series of conditional measurements on a suitable initial entangled state. Such a study would open the possibility to further characterize the computational properties of generic operational theories~\cite{barrettInformationProcessingGeneralized2007,barrettComputationalLandscapeGeneral2019}. For example, one could study under which conditions an operational theory is computationally equivalent to its minimal strongly causal version, as in the case of \ac{CT}, or which additional assumptions are needed.

We finally observe that our findings could have implications concerning Ref.~\cite{jokinenNobroadcastingCharacterizesOperational2024}, where the authors extensively study the relation between the no-broadcasting theorem and contextuality for probabilistic theories. In this scenario, the classical noncontextual theories with no-brodcasting here developed could be used as test toy-models for the implications proved in Ref.~\cite{jokinenNobroadcastingCharacterizesOperational2024} and for their possible generalization to a broader class of theories.

\section*{Acknowledgments}
D.R.~is currently visiting Perimeter Institute for Theoretical Physics and would like to thank the quantum foundations group there for their hospitality during his visit. This research was supported in part by Perimeter Institute for Theoretical Physics. Research at Perimeter Institute is supported by the Government of Canada through the Department of Innovation, Science, and Economic Development, and by the Province of Ontario through the Ministry of Colleges and Universities. M.E.~acknowledges financial support from the National Science Centre, Poland (Opus project, Categorical Foundations of the Non-Classicality of Nature, Project No.~2021/41/B/ST2/03149). A.T.~acknowledges the financial support of Elvia and Federico Faggin Foundation (Silicon Valley Community Foundation Project ID No.~2020-214365). P.P.~acknowledges financial support from European Union - Next Generation EU through the PNNR MUR Project No.~PE0000023-NQSTI.

\bibliography{MSBCT-njphys-v2}

\begin{thebibliography}{10}

\bibitem{chiribellaProbabilisticTheoriesPurification2010}
G.~Chiribella, G.~M. D'Ariano, and P.~Perinotti.
\newblock ``Probabilistic theories with purification''.
\newblock \href{https://dx.doi.org/10.1103/PhysRevA.81.062348}{Physical Review
  A {\bf 81}, 062348}~(2010).

\bibitem{darianoQuantumTheoryFirst2016}
G.~M. D'Ariano, G.~Chiribella, and P.~Perinotti.
\newblock ``Quantum {{Theory}} from {{First Principles}}: {{An Informational
  Approach}}''.
\newblock \href{https://dx.doi.org/10.1017/9781107338340}{Cambridge University
  Press}. ~(2016).
\newblock 1 edition.

\bibitem{chiribellaInformationalDerivationQuantum2011}
G.~Chiribella, G.~M. D'Ariano, and P.~Perinotti.
\newblock ``Informational derivation of quantum theory''.
\newblock \href{https://dx.doi.org/10.1103/PhysRevA.84.012311}{Physical Review
  A {\bf 84}, 012311}~(2011).

\bibitem{hardyDisentanglingNonlocalityTeleportation1999}
L.~Hardy.
\newblock ``Disentangling {{Nonlocality}} and {{Teleportation}}''~(1999).
\newblock
  \href{http://arxiv.org/abs/quant-ph/9906123}{arXiv:quant-ph/9906123}.

\bibitem{barnumCloningBroadcastingGeneric2006}
H.~Barnum, J.~Barrett, M.~Leifer, and A.~Wilce.
\newblock ``Cloning and {{Broadcasting}} in {{Generic Probabilistic
  Theories}}''~(2006).
\newblock
  \href{http://arxiv.org/abs/quant-ph/0611295}{arXiv:quant-ph/0611295}.

\bibitem{barrettInformationProcessingGeneralized2007}
J.~Barrett.
\newblock ``Information processing in generalized probabilistic theories''.
\newblock \href{https://dx.doi.org/10.1103/PhysRevA.75.032304}{Physical Review
  A {\bf 75}, 032304}~(2007).

\bibitem{spekkensEvidenceEpistemicView2007}
R.~W. Spekkens.
\newblock ``Evidence for the epistemic view of quantum states: {{A}} toy
  theory''.
\newblock \href{https://dx.doi.org/10.1103/PhysRevA.75.032110}{Physical Review
  A {\bf 75}, 032110}~(2007).

\bibitem{janottaGeneralizedProbabilisticTheories2013}
P.~Janotta and R.~Lal.
\newblock ``Generalized probabilistic theories without the no-restriction
  hypothesis''.
\newblock \href{https://dx.doi.org/10.1103/PhysRevA.87.052131}{Physical Review
  A {\bf 87}, 052131}~(2013).

\bibitem{darianoFeynmanProblemFermionic2014}
G.~M. D'Ariano, F.~Manessi, P.~Perinotti, and A.~Tosini.
\newblock ``The {{Feynman}} problem and fermionic entanglement: {{Fermionic}}
  theory versus qubit theory''.
\newblock \href{https://dx.doi.org/10.1142/S0217751X14300257}{International
  Journal of Modern Physics A {\bf 29}, 1430025}~(2014).

\bibitem{darianoFermionicComputationNonlocal2014}
G.~M. D'Ariano, F.~Manessi, P.~Perinotti, and A.~Tosini.
\newblock ``Fermionic computation is non-local tomographic and violates
  monogamy of entanglement''.
\newblock \href{https://dx.doi.org/10.1209/0295-5075/107/20009}{Europhysics
  Letters {\bf 107}, 20009}~(2014).

\bibitem{darianoClassicalityLocalDiscriminability2020}
G.~M. D'Ariano, M.~Erba, and P.~Perinotti.
\newblock ``Classicality without local discriminability: {{Decoupling}}
  entanglement and complementarity''.
\newblock \href{https://dx.doi.org/10.1103/PhysRevA.102.052216}{Physical Review
  A {\bf 102}, 052216}~(2020).

\bibitem{schmidShadowsSubsystemsGeneralized2024}
D.~Schmid, J.~H. Selby, V.~P. Rossi, R.~D. Baldij{\~a}o, and A.~B. Sainz.
\newblock ``Shadows and subsystems of generalized probabilistic theories: When
  tomographic incompleteness is not a loophole for contextuality
  proofs''~(2024).
\newblock  \href{http://arxiv.org/abs/2409.13024}{arXiv:2409.13024}.

\bibitem{erbaMeasurementIncompatibilityStrictly2024}
M.~Erba, P.~Perinotti, D.~Rolino, and A.~Tosini.
\newblock ``Measurement incompatibility is strictly stronger than
  disturbance''.
\newblock \href{https://dx.doi.org/10.1103/PhysRevA.109.022239}{Physical Review
  A {\bf 109}, 022239}~(2024).

\bibitem{perinottiCellularAutomataOperational2020}
P.~Perinotti.
\newblock ``Cellular automata in operational probabilistic theories''.
\newblock \href{https://dx.doi.org/10.22331/q-2020-07-09-294}{Quantum {\bf 4},
  294}~(2020).

\bibitem{perinottiCausalInfluenceOperational2021}
P.~Perinotti.
\newblock ``Causal influence in operational probabilistic theories''.
\newblock \href{https://dx.doi.org/10.22331/q-2021-08-03-515}{Quantum {\bf 5},
  515}~(2021).

\bibitem{buschNoInformationDisturbance2009}
P.~Busch.
\newblock ````{{No Information Without Disturbance}}'': {{Quantum Limitations}}
  of {{Measurement}}''.
\newblock In Quantum {{Reality}}, {{Relativistic Causality}}, and {{Closing}}
  the {{Epistemic Circle}}: {{Essays}} in {{Honour}} of {{Abner Shimony}}.
\newblock \href{https://dx.doi.org/10.1007/978-1-4020-9107-0_13}{Pages
  229--256}.
\newblock The {{Western Ontario Series}} in {{Philosophy}} of {{Science}}.
  Springer Netherlands, Dordrecht~(2009).

\bibitem{heinosaariNofreeinformationPrincipleGeneral2019}
T.~Heinosaari, L.~Lepp{\"a}j{\"a}rvi, and M.~Pl{\'a}vala.
\newblock ``No-free-information principle in general probabilistic theories''.
\newblock \href{https://dx.doi.org/10.22331/q-2019-07-08-157}{Quantum {\bf 3},
  157}~(2019).

\bibitem{darianoInformationDisturbanceOperational2020}
G.~M. D'Ariano, P.~Perinotti, and A.~Tosini.
\newblock ``Information and disturbance in operational probabilistic
  theories''.
\newblock \href{https://dx.doi.org/10.22331/q-2020-11-16-363}{Quantum {\bf 4},
  363}~(2020).

\bibitem{woottersSingleQuantumCannot1982}
W.~K. Wootters and W.~H. Zurek.
\newblock ``A single quantum cannot be cloned''.
\newblock \href{https://dx.doi.org/10.1038/299802a0}{Nature {\bf 299},
  802--803}~(1982).

\bibitem{dieksCommunicationEPRDevices1982}
D.~Dieks.
\newblock ``Communication by {{EPR}} devices''.
\newblock \href{https://dx.doi.org/10.1016/0375-9601(82)90084-6}{Physics
  Letters A {\bf 92}, 271--272}~(1982).

\bibitem{yuenAmplificationQuantumStates1986}
H.~P. Yuen.
\newblock ``Amplification of quantum states and noiseless photon amplifiers''.
\newblock \href{https://dx.doi.org/10.1016/0375-9601(86)90660-2}{Physics
  Letters A {\bf 113}, 405--407}~(1986).

\bibitem{barnumNoncommutingMixedStates1996}
H.~Barnum, C.~M. Caves, C.~A. Fuchs, R.~Jozsa, and B.~Schumacher.
\newblock ``Noncommuting {{Mixed States Cannot Be Broadcast}}''.
\newblock \href{https://dx.doi.org/10.1103/PhysRevLett.76.2818}{Physical Review
  Letters {\bf 76}, 2818--2821}~(1996).

\bibitem{walkerClassicalBroadcastingPossible2007}
T.~A. Walker and S.~L. Braunstein.
\newblock ``Classical {{Broadcasting}} is {{Possible}} with {{Arbitrarily High
  Fidelity}} and {{Resolution}}''.
\newblock \href{https://dx.doi.org/10.1103/PhysRevLett.98.080501}{Physical
  Review Letters {\bf 98}, 080501}~(2007).

\bibitem{barnumGeneralizedNoBroadcastingTheorem2007}
H.~Barnum, J.~Barrett, M.~Leifer, and A.~Wilce.
\newblock ``Generalized {{No-Broadcasting Theorem}}''.
\newblock \href{https://dx.doi.org/10.1103/PhysRevLett.99.240501}{Physical
  Review Letters {\bf 99}, 240501}~(2007).

\bibitem{pianiNoLocalBroadcastingTheoremMultipartite2008}
M.~Piani, P.~Horodecki, and R.~Horodecki.
\newblock ``No-{{Local-Broadcasting Theorem}} for {{Multipartite Quantum
  Correlations}}''.
\newblock \href{https://dx.doi.org/10.1103/PhysRevLett.100.090502}{Physical
  Review Letters {\bf 100}, 090502}~(2008).

\bibitem{luoQuantumNoBroadcasting2010}
S.~Luo.
\newblock ``On {{Quantum No-Broadcasting}}''.
\newblock \href{https://dx.doi.org/10.1007/s11005-010-0389-1}{Letters in
  Mathematical Physics {\bf 92}, 143--153}~(2010).

\bibitem{chiribellaQuantumPrinciples2016}
G.~Chiribella, G.~M. D'Ariano, and P.~Perinotti.
\newblock ``Quantum from {{Principles}}''.
\newblock In G.~Chiribella and R.~W. Spekkens, editors, Quantum {{Theory}}:
  {{Informational Foundations}} and {{Foils}}.
\newblock \href{https://dx.doi.org/10.1007/978-94-017-7303-4_6}{Volume 181,
  pages 171--221}.
\newblock Springer Netherlands, Dordrecht~(2016).

\bibitem{hardyQuantumTheoryFive2001}
L.~Hardy.
\newblock ``Quantum {{Theory From Five Reasonable Axioms}}''~(2001).
\newblock
  \href{http://arxiv.org/abs/quant-ph/0101012}{arXiv:quant-ph/0101012}.

\bibitem{fuchsQuantumMechanicsQuantum2002}
C.~A. Fuchs.
\newblock ``Quantum {{Mechanics}} as {{Quantum Information}} (and only a little
  more)''~(2002).
\newblock
  \href{http://arxiv.org/abs/quant-ph/0205039}{arXiv:quant-ph/0205039}.

\bibitem{brassardInformationKey2005}
G.~Brassard.
\newblock ``Is information the key?''.
\newblock \href{https://dx.doi.org/10.1038/nphys134}{Nature Physics {\bf 1},
  2--4}~(2005).

\bibitem{darianoProbabilisticTheoriesWhat2010}
G.~M. D'Ariano.
\newblock ``Probabilistic theories: {{What}} is special about {{Quantum
  Mechanics}}?''.
\newblock In A.~Bokulich and G.~Jaeger, editors, Philosophy of {{Quantum
  Information}} and {{Entanglement}}.
\newblock \href{https://dx.doi.org/10.1017/CBO9780511676550.007}{Pages
  85--126}.
\newblock Cambridge University Press~(2010).
\newblock 1 edition.

\bibitem{darianoTestingAxiomsQuantum2010}
G.~M. D'Ariano and A.~Tosini.
\newblock ``Testing axioms for quantum theory on probabilistic toy-theories''.
\newblock \href{https://dx.doi.org/10.1007/s11128-010-0172-3}{Quantum
  Information Processing {\bf 9}, 95--141}~(2010).

\bibitem{masanesDerivationQuantumTheory2011}
L.~Masanes and M.~P. M{\"u}ller.
\newblock ``A derivation of quantum theory from physical requirements''.
\newblock \href{https://dx.doi.org/10.1088/1367-2630/13/6/063001}{New Journal
  of Physics {\bf 13}, 063001}~(2011).

\bibitem{dakicQuantumTheoryEntanglement2011}
B.~Daki{\'c} and {\v C}.~Brukner.
\newblock ``Quantum {{Theory}} and {{Beyond}}: {{Is Entanglement Special}}?''.
\newblock In H.~Halvorson, editor, Deep {{Beauty}}: {{Understanding}} the
  {{Quantum World}} through {{Mathematical Innovation}}.
\newblock \href{https://dx.doi.org/10.1017/CBO9780511976971.011}{Pages
  365--392}.
\newblock Cambridge University Press, Cambridge~(2011).

\bibitem{barnumTeleportationGeneralProbabilistic2008}
H.~Barnum, J.~Barrett, M.~Leifer, and A.~Wilce.
\newblock ``Teleportation in {{General Probabilistic Theories}}''~(2008).
\newblock  \href{http://arxiv.org/abs/0805.3553}{arXiv:0805.3553}.

\bibitem{barnumInformationProcessingConvex2011}
H.~Barnum and A.~Wilce.
\newblock ``Information {{Processing}} in {{Convex Operational Theories}}''.
\newblock \href{https://dx.doi.org/10.1016/j.entcs.2011.01.002}{Electronic
  Notes in Theoretical Computer Science {\bf 270}, 3--15}~(2011).

\bibitem{plavalaGeneralProbabilisticTheories2021}
M.~Pl{\'a}vala.
\newblock ``General probabilistic theories: {{An}} introduction''.
\newblock \href{https://dx.doi.org/10.1016/j.physrep.2023.09.001}{Physics
  Reports {\bf 1033}, 1--64}~(2023).

\bibitem{coeckeKindergartenQuantumMechanics2006}
B.~Coecke.
\newblock ``Kindergarten {{Quantum Mechanics}}: {{Lecture Notes}}''.
\newblock \href{https://dx.doi.org/10.1063/1.2158713}{AIP Conference
  Proceedings {\bf 810}, 81--98}~(2006).

\bibitem{coeckePicturingQuantumProcesses2017}
B.~Coecke and A.~Kissinger.
\newblock ``Picturing {{Quantum Processes}}: {{A First Course}} in {{Quantum
  Theory}} and {{Diagrammatic Reasoning}}''.
\newblock \href{https://dx.doi.org/10.1017/9781316219317}{Cambridge University
  Press}. Cambridge~(2017).

\bibitem{ludwigFoundationsQuantumMechanics1985}
G.~Ludwig.
\newblock ``Foundations of {{Quantum Mechanics}}''.
\newblock \href{https://dx.doi.org/10.1007/978-3-662-28726-2}{Springer Berlin
  Heidelberg}. Berlin, Heidelberg~(1985).

\bibitem{wilceTestSpacesOrthoalgebras2000}
A.~Wilce.
\newblock ``Test {{Spaces}} and {{Orthoalgebras}}''.
\newblock In B.~Coecke, D.~Moore, and A.~Wilce, editors, Current {{Research}}
  in {{Operational Quantum Logic}}: {{Algebras}}, {{Categories}},
  {{Languages}}.
\newblock \href{https://dx.doi.org/10.1007/978-94-017-1201-9_4}{Pages 81--114}.
\newblock Fundamental {{Theories}} of {{Physics}}. Springer Netherlands,
  Dordrecht~(2000).

\bibitem{darianoClassicalTheoriesEntanglement2020}
G.~M. D'Ariano, M.~Erba, and P.~Perinotti.
\newblock ``Classical theories with entanglement''.
\newblock \href{https://dx.doi.org/10.1103/PhysRevA.101.042118}{Physical Review
  A {\bf 101}, 042118}~(2020).

\bibitem{bravyiFermionicQuantumComputation2002}
S.~B. Bravyi and A.~Y. Kitaev.
\newblock ``Fermionic {{Quantum Computation}}''.
\newblock \href{https://dx.doi.org/10.1006/aphy.2002.6254}{Annals of Physics
  {\bf 298}, 210--226}~(2002).

\bibitem{lugliFermionicStateDiscrimination2020}
M.~Lugli, P.~Perinotti, and A.~Tosini.
\newblock ``Fermionic {{State Discrimination}} by {{Local Operations}} and
  {{Classical Communication}}''.
\newblock \href{https://dx.doi.org/10.1103/PhysRevLett.125.110403}{Physical
  Review Letters {\bf 125}, 110403}~(2020).

\bibitem{perinottiShannonTheoryQuantum2023}
P.~Perinotti, A.~Tosini, and L.~Vaglini.
\newblock ``Shannon {{Theory}} for {{Quantum Systems}} and {{Beyond}}:
  {{Information Compression}} for {{Fermions}}''.
\newblock In A.~Plotnitsky and E.~B{\"a}umer, editors, The {{Quantum-Like
  Revolution}}.
\newblock \href{https://dx.doi.org/10.1007/978-3-031-12986-5_6}{Pages
  135--156}.
\newblock Springer International Publishing, Cham~(2023).

\bibitem{maclaneCategoriesWorkingMathematician1978}
S.~Mac~Lane.
\newblock ``Categories for the {{Working Mathematician}}''.
\newblock \href{https://dx.doi.org/10.1007/978-1-4757-4721-8}{Volume~5 of
  Graduate {{Texts}} in {{Mathematics}}}.
\newblock Springer New York. New York, NY~(1978).

\bibitem{awodeyCategoryTheory2006}
S.~Awodey.
\newblock ``Category {{Theory}}''.
\newblock
  \href{https://dx.doi.org/10.1093/acprof:oso/9780198568612.001.0001}{Oxford
  University Press}. ~(2006).

\bibitem{heunenCategoriesQuantumTheory2019}
C.~Heunen and J.~Vicary.
\newblock ``Categories for {{Quantum Theory}}: {{An Introduction}}''.
\newblock \href{https://dx.doi.org/10.1093/oso/9780198739623.001.0001}{Oxford
  University PressOxford}. ~(2019).
\newblock 1 edition.

\bibitem{selbyContextualityIncompatibility2023}
J.~H. Selby, D.~Schmid, E.~Wolfe, A.~B. Sainz, R.~Kunjwal, and R.~W. Spekkens.
\newblock ``Contextuality without {{Incompatibility}}''.
\newblock \href{https://dx.doi.org/10.1103/PhysRevLett.130.230201}{Physical
  Review Letters {\bf 130}, 230201}~(2023).

\bibitem{selbyAccessibleFragmentsGeneralized2023}
J.~H. Selby, D.~Schmid, E.~Wolfe, A.~B. Sainz, R.~Kunjwal, and R.~W. Spekkens.
\newblock ``Accessible fragments of generalized probabilistic theories, cone
  equivalence, and applications to witnessing nonclassicality''.
\newblock \href{https://dx.doi.org/10.1103/PhysRevA.107.062203}{Physical Review
  A {\bf 107}, 062203}~(2023).

\bibitem{bluhmSimulationQuantumMultimeters2024}
A.~Bluhm, L.~Lepp{\"a}j{\"a}rvi, and I.~Nechita.
\newblock ``On the simulation of quantum multimeters''~(2024).
\newblock  \href{http://arxiv.org/abs/2402.18333}{arXiv:2402.18333}.

\bibitem{leppajarviPostprocessingQuantumInstruments2021}
L.~Lepp{\"a}j{\"a}rvi and M.~Sedl{\'a}k.
\newblock ``Postprocessing of quantum instruments''.
\newblock \href{https://dx.doi.org/10.1103/PhysRevA.103.022615}{Physical Review
  A {\bf 103}, 022615}~(2021).

\bibitem{daffertshoferClassicalNoCloningTheorem2002}
A.~Daffertshofer, A.~R. Plastino, and A.~Plastino.
\newblock ``Classical {{No-Cloning Theorem}}''.
\newblock \href{https://dx.doi.org/10.1103/PhysRevLett.88.210601}{Physical
  Review Letters {\bf 88}, 210601}~(2002).

\bibitem{darianoIncompatibilityObservablesChannels2022}
G.~M. D'Ariano, P.~Perinotti, and A.~Tosini.
\newblock ``Incompatibility of observables, channels and instruments in
  information theories''.
\newblock \href{https://dx.doi.org/10.1088/1751-8121/ac88a7}{Journal of Physics
  A: Mathematical and Theoretical {\bf 55}, 394006}~(2022).

\bibitem{heinosaariInvitationQuantumIncompatibility2016}
T.~Heinosaari, T.~Miyadera, and M.~Ziman.
\newblock ``An invitation to quantum incompatibility''.
\newblock \href{https://dx.doi.org/10.1088/1751-8113/49/12/123001}{Journal of
  Physics A: Mathematical and Theoretical {\bf 49}, 123001}~(2016).

\bibitem{bennettQuantumCryptographyUsing1992}
C.~H. Bennett.
\newblock ``Quantum cryptography using any two nonorthogonal states''.
\newblock \href{https://dx.doi.org/10.1103/PhysRevLett.68.3121}{Physical Review
  Letters {\bf 68}, 3121--3124}~(1992).

\bibitem{bennettExperimentalQuantumCryptography1992}
C.~H. Bennett, F.~Bessette, G.~Brassard, L.~Salvail, and J.~Smolin.
\newblock ``Experimental quantum cryptography''.
\newblock \href{https://dx.doi.org/10.1007/BF00191318}{Journal of Cryptology
  {\bf 5}, 3--28}~(1992).

\bibitem{barnettInformationtheoreticLimitsQuantum1993}
S.~M. Barnett and S.~J.~D. Phoenix.
\newblock ``Information-theoretic limits to quantum cryptography''.
\newblock \href{https://dx.doi.org/10.1103/PhysRevA.48.R5}{Physical Review A
  {\bf 48}, R5--R8}~(1993).

\bibitem{ekertEavesdroppingQuantumcryptographicalSystems1994}
A.~K. Ekert, B.~Huttner, G.~M. Palma, and A.~Peres.
\newblock ``Eavesdropping on quantum-cryptographical systems''.
\newblock \href{https://dx.doi.org/10.1103/PhysRevA.50.1047}{Physical Review A
  {\bf 50}, 1047--1056}~(1994).

\bibitem{kentUnconditionallySecureBit1999}
A.~Kent.
\newblock ``Unconditionally {{Secure Bit Commitment}}''.
\newblock \href{https://dx.doi.org/10.1103/PhysRevLett.83.1447}{Physical Review
  Letters {\bf 83}, 1447--1450}~(1999).

\bibitem{loUnconditionalSecurityQuantum1999}
H.~Lo and H.~F. Chau.
\newblock ``Unconditional {{Security}} of {{Quantum Key Distribution}} over
  {{Arbitrarily Long Distances}}''.
\newblock \href{https://dx.doi.org/10.1126/science.283.5410.2050}{Science {\bf
  283}, 2050--2056}~(1999).

\bibitem{mayersUnconditionalSecurityQuantum2001}
D.~Mayers.
\newblock ``Unconditional security in quantum cryptography''.
\newblock \href{https://dx.doi.org/10.1145/382780.382781}{Journal of the ACM
  {\bf 48}, 351--406}~(2001).

\bibitem{nielsenQuantumComputationQuantum2010}
M.~A. Nielsen and I.~L. Chuang.
\newblock ``Quantum {{Computation}} and {{Quantum Information}}: 10th
  {{Anniversary Edition}}''.
\newblock \href{https://dx.doi.org/10.1017/CBO9780511976667}{Cambridge
  University Press}. Cambridge~(2010).
\newblock 10th anniversary edition edition.

\bibitem{bennettQuantumCryptographyPublic2014}
C.~H. Bennett and G.~Brassard.
\newblock ``Quantum cryptography: {{Public}} key distribution and coin
  tossing''.
\newblock \href{https://dx.doi.org/10.1016/j.tcs.2014.05.025}{Theoretical
  Computer Science {\bf 560}, 7--11}~(2014).

\bibitem{barrettComputationalLandscapeGeneral2019}
J.~Barrett, N.~De~Beaudrap, M.~J. Hoban, and C.~M. Lee.
\newblock ``The computational landscape of general physical theories''.
\newblock \href{https://dx.doi.org/10.1038/s41534-019-0156-9}{npj Quantum
  Information {\bf 5}, 41}~(2019).

\bibitem{joyalBraidedTensorCategories1993}
A.~Joyal and R.~Street.
\newblock ``Braided {{Tensor Categories}}''.
\newblock \href{https://dx.doi.org/10.1006/aima.1993.1055}{Advances in
  Mathematics {\bf 102}, 20--78}~(1993).

\bibitem{10.1063/1.2219356}
G.~M. D'Ariano.
\newblock ``How to {{Derive}} the {{Hilbert}}-{{Space Formulation}} of
  {{Quantum Mechanics From Purely Operational Axioms}}''.
\newblock \href{https://dx.doi.org/10.1063/1.2219356}{AIP Conference
  Proceedings {\bf 844}, 101--128}~(2006).

\bibitem{hardyLimitedHolismRealVectorSpace2012}
L.~Hardy and W.~K. Wootters.
\newblock ``Limited {{Holism}} and {{Real-Vector-Space Quantum Theory}}''.
\newblock \href{https://dx.doi.org/10.1007/s10701-011-9616-6}{Foundations of
  Physics {\bf 42}, 454--473}~(2012).

\bibitem{bellProblemHiddenVariables1966}
J.~S. Bell.
\newblock ``On the {{Problem}} of {{Hidden Variables}} in {{Quantum
  Mechanics}}''.
\newblock \href{https://dx.doi.org/10.1103/RevModPhys.38.447}{Reviews of Modern
  Physics {\bf 38}, 447--452}~(1966).

\bibitem{kochenProblemHiddenVariables1975}
S.~Kochen and E.~P. Specker.
\newblock ``The problem of hidden variables in quantum mechanics''.
\newblock In C.~A. Hooker, editor, The Logico-Algebraic Approach to Quantum
  Mechanics: {{Volume I}}: {{Historical}} Evolution.
\newblock \href{https://dx.doi.org/10.1007/978-94-010-1795-4_17}{Pages
  293--328}.
\newblock Springer Netherlands, Dordrecht~(1975).
\newblock Originally published in Journal of Mathematics and Mechanics 17,
  59--87 (1967).
\newblock  (Available at: \url{http://www.jstor.org/stable/24902153}).

\bibitem{merminSimpleUnifiedForm1990}
N.~D. Mermin.
\newblock ``Simple unified form for the major no-hidden-variables theorems''.
\newblock \href{https://dx.doi.org/10.1103/PhysRevLett.65.3373}{Physical Review
  Letters {\bf 65}, 3373--3376}~(1990).

\bibitem{merminHiddenVariablesTwo1993}
N.~D. Mermin.
\newblock ``Hidden variables and the two theorems of {{John Bell}}''.
\newblock \href{https://dx.doi.org/10.1103/RevModPhys.65.803}{Reviews of Modern
  Physics {\bf 65}, 803--815}~(1993).

\bibitem{spekkensContextualityPreparationsTransformations2005}
R.~W. Spekkens.
\newblock ``Contextuality for preparations, transformations, and unsharp
  measurements''.
\newblock \href{https://dx.doi.org/10.1103/PhysRevA.71.052108}{Physical Review
  A {\bf 71}, 052108}~(2005).

\bibitem{schmidStructureTheoremGeneralizednoncontextual2024}
D.~Schmid, J.~H. Selby, M.~F. Pusey, and R.~W. Spekkens.
\newblock ``A structure theorem for generalized-noncontextual ontological
  models''.
\newblock \href{https://dx.doi.org/10.22331/q-2024-03-14-1283}{Quantum {\bf 8},
  1283}~(2024).

\bibitem{schmidCharacterizationNoncontextualityFramework2021}
D.~Schmid, J.~H. Selby, E.~Wolfe, R.~Kunjwal, and R.~W. Spekkens.
\newblock ``Characterization of {{Noncontextuality}} in the {{Framework}} of
  {{Generalized Probabilistic Theories}}''.
\newblock \href{https://dx.doi.org/10.1103/PRXQuantum.2.010331}{PRX Quantum
  {\bf 2}, 010331}~(2021).

\bibitem{soltaniLocalNoncontextualOntological}
S.~Soltani, M.~Erba, D.~Schmid, and J.~H. Selby.
\newblock ``A local and noncontextual ontological model for bilocal classical
  theory''.
\newblock Forthcoming.

\bibitem{gottesmanDemonstratingViabilityUniversal1999}
D.~Gottesman and I.~L. Chuang.
\newblock ``Demonstrating the viability of universal quantum computation using
  teleportation and single-qubit operations''.
\newblock \href{https://dx.doi.org/10.1038/46503}{Nature {\bf 402},
  390--393}~(1999).

\bibitem{raussendorfOneWayQuantumComputer2001}
R.~Raussendorf and H.~J. Briegel.
\newblock ``A {{One-Way Quantum Computer}}''.
\newblock \href{https://dx.doi.org/10.1103/PhysRevLett.86.5188}{Physical Review
  Letters {\bf 86}, 5188--5191}~(2001).

\bibitem{leungTwoqubitProjectiveMeasurements2002}
D.~W. Leung.
\newblock ``Two-qubit {{Projective Measurements}} are {{Universal}} for
  {{Quantum Computation}}''~(2002).
\newblock
  \href{http://arxiv.org/abs/quant-ph/0111122}{arXiv:quant-ph/0111122}.

\bibitem{nielsenQuantumComputationMeasurement2003}
M.~A. Nielsen.
\newblock ``Quantum computation by measurement and quantum memory''.
\newblock \href{https://dx.doi.org/10.1016/S0375-9601(02)01803-0}{Physics
  Letters A {\bf 308}, 96--100}~(2003).

\bibitem{raussendorfMeasurementbasedQuantumComputation2003}
R.~Raussendorf, D.~E. Browne, and H.~J. Briegel.
\newblock ``Measurement-based quantum computation on cluster states''.
\newblock \href{https://dx.doi.org/10.1103/PhysRevA.68.022312}{Physical Review
  A {\bf 68}, 022312}~(2003).

\bibitem{leungQuantumComputationMeasurements2004}
D.~W. Leung.
\newblock ``Quantum computation by measurements''.
\newblock \href{https://dx.doi.org/10.1142/S0219749904000055}{International
  Journal of Quantum Information {\bf 02}, 33--43}~(2004).

\bibitem{jozsaIntroductionMeasurementBased2005}
R.~Jozsa.
\newblock ``An introduction to measurement based quantum computation''~(2005).
\newblock
  \href{http://arxiv.org/abs/quant-ph/0508124}{arXiv:quant-ph/0508124}.

\bibitem{jokinenNobroadcastingCharacterizesOperational2024}
P.~Jokinen, M.~Weilenmann, M.~Pl{\'a}vala, J.~Pellonp{\"a}{\"a}, J.~Kiukas, and
  R.~Uola.
\newblock ``No-broadcasting characterizes operational contextuality''~(2024).
\newblock  \href{http://arxiv.org/abs/2406.07305}{arXiv:2406.07305}.

\end{thebibliography}
\bibliographystyle{myQuantum}

\widetext
\appendix

\section{Relation between atomicity of the identity transformation and irreversibility [Proof of \autoref{thm:opt:irrev:idAtomicity}]}
\label{opt:atomicityAndIrreversibility}

\begin{reptheorem}{thm:opt:irrev:idAtomicity}
	Let \OPTMath{} be a causal \ac{OPT} with a system $\system{A} \in \Sys{\OPTMath}$ such that $\sysDimension{A} \geq 2$, and the identity transformation $\identityTest{A}$ to it associated is atomic. Then, there exists an instrument $\eventTest{T}{x}{X} \in \Instr{A}{B}$, for some system $\system{B} \in \Sys{\OPTMath}$ that is intrinsically irreversible. Hence, the theory has irreversibility~\cite{erbaMeasurementIncompatibilityStrictly2024}.
\end{reptheorem}

\begin{proof}
Suppose by contradiction that this is not the case, i.e.~there exists a system \system{A} with $\sysDimension{A} \neq 1$, where the identity is atomic and there are no intrinsically irreversible instruments having system \system{A} as an input. By \autoref{opt:intrIrrev:identity} this implies that whatever instrument $\eventTest{T}{x}{X} \in \Instr{A}{B}$ one considers it is always possible to find a suitable collection of conditional instruments such that:

\begin{align*}
	\label{eqt:doesNotExclude}
	\myQcircuit{
		&\s{A}\qw&\gate{\event{T}{x}}&\s {B}\qw&\qw&
	} =& \; \sum_{\outcomeIncludedConditioned{z}{S}{x}}
	\myQcircuit{
		&\s{A}\qw&\multigate{1}{\event{C}{z}}&\qw&\s {B}\qw&\qw&
		\\
		&\pureghost{}&\pureghost{\event{C}{z}}&\s{E}\qw&\measureD{\observationUniqueDeterministic}
	},\\[10pt]
	\myQcircuit{
		&\qw&\qw&\s{A}\qw&\qw&\qw&\qw&
	} =& \; \sum_{\outcomeIncluded{z}{Z}}
	\myQcircuit{
		&\s{A}\qw&\multigate{1}{\event{C}{z}}&\s {B}\qw&\multigate{1}{\conditionedEvent{P}{}{z}}&\s{A}\qw&\qw&
		\\
		&\pureghost{}&\pureghost{\event{C}{z}}&\s{E}\qw&\ghost{\conditionedEvent{P}{}{z}}&\pureghost{}&
	}.
\end{align*}
Consequently, since the identity is atomic, it holds that:
\begin{equation*}
	\myQcircuit{
		&\s{A}\qw&\multigate{1}{\event{C}{z}}&\s {B}\qw&\multigate{1}{\conditionedEvent{P}{}{z}}&\s{A}\qw&\qw&
		\\
		&\pureghost{}&\pureghost{\event{C}{z}}&\s{E}\qw&\ghost{\conditionedEvent{P}{}{z}}&\pureghost{}&
	} = \quad\!\! \probabilityEvent{p}{z} \;
	\myQcircuit{
		&\qw&\s{A}\qw&\qw&\qw&
	},
\end{equation*}
where \probabilityEventTest{p}{z}{Z} is a suitable probability distribution. Applying the deterministic effect on both sides of the above equation obtains:
\begin{equation*}
	\myQcircuit{
		&\s{A}\qw&\multigate{1}{\event{C}{z}}&\s {B}\qw&\multigate{1}{\conditionedEvent{P}{}{z}}&\s{A}\qw&\measureD{\observationUniqueDeterministic}&
		\\
		&\pureghost{}&\pureghost{\event{C}{z}}&\s{E}\qw&\ghost{\conditionedEvent{P}{}{z}}&\pureghost{}&
	} = \quad\!\! \probabilityEvent{p}{z} \;
	\myQcircuit{
		&\qw&\s{A}\qw&\measureD{\observationUniqueDeterministic}&
	},
\end{equation*}
which, given that $\conditionedTest{P}{}{z}$ is deterministic $\forall \outcomeIncluded{z}{Z}$, implies:
\begin{equation*}
	\myQcircuit{
		&\s{A}\qw&\multigate{1}{\event{C}{z}}&\s {B}\qw&\measureD{\observationUniqueDeterministic}&
		\\
		&\pureghost{}&\pureghost{\event{C}{z}}&\s{E}\qw&\measureD{\observationUniqueDeterministic}&
	} = \quad\!\! \probabilityEvent{p}{z} \;
	\myQcircuit{
		&\qw&\s{A}\qw&\measureD{\observationUniqueDeterministic}&
	},
\end{equation*}
and, consequently, by summing over $\outcomeIncludedConditioned{z}{S}{x}$ for all \outcomeIncluded{x}{X}:
\begin{align}\label{eqt:atomicityVSirreversibility}
	\myQcircuit{
		&\s{A}\qw&\gate{\event{T}{x}}&\s {B}\qw&\measureD{\observationUniqueDeterministic}&
	} = \quad\!\! \probabilityEvent{p}{x} \;
	\myQcircuit{
		&\qw&\s{A}\qw&\measureD{\observationUniqueDeterministic}&
	}.
\end{align}
A particular set of instruments that satisfies the latter relation consists of those where a fixed deterministic state is prepared following a measurement:
\begin{align*}
	\myQcircuit{ 
		&\s{A}\qw&\gate{\event{T}{x}}&\s{B}\qw&\qw&
	} = \quad\!\!
	\myQcircuit{
		&\s{A}\qw&\measureD{\observationEvent{a}{x}}&\prepareC{\preparationEventNoDown{\rho}}&\s{B}\qw&\qw&
	} \quad \forall \outcomeIncluded{x}{X},
\end{align*}
where $\observationEventTest{a}{x}{X} \in \Obs{A}$ is a generic observation-instrument of the theory and $\eventNoDown{\rho} \in \StN{A}$ is a generic deterministic state. In this particular case, \eqref{eqt:atomicityVSirreversibility} implies that 
\begin{align*}
	\myQcircuit{
		&\s{A}\qw&\measureD{\observationEvent{a}{x}}&
	} = \quad\!\! \probabilityEvent{p}{x} \;
	\myQcircuit{
		&\qw&\s{A}\qw&\measureD{\observationUniqueDeterministic}&
	} \quad \forall \outcomeIncluded{x}{X}.
\end{align*}
Given that this holds for any observation-instrument of system \system{A}, this implies that for this system the only allowed observation-instruments are a randomisation of the deterministic event $\left\{ \probabilityEvent{p}{x} \observationUniqueDeterministic \right\}_{\outcomeIncluded{x}{X}}$. Since a system of this kind has $\dim\EffR{A} = 1$, as can be proven by direct calculation, it must also be $\dim\EffR{A} = \dim \StR{A} = \sysDimension{A} = 1$, contradicting the hypotheses.
\end{proof}

\section{Characterization of permutations on bipartite systems [Proof of \autoref{thm:OPT:symmetric:permutations:generalForm}]}
\label{app:opt:symmetric:permutations:generalForm}

\begin{reptheorem}{thm:OPT:symmetric:permutations:generalForm}[General form of permutations on bipartite systems]
In every symmetric \ac{OPT} with unique decomposition, for any permutation acting on a bipartite system $\system{AB}$ there exist suitable systems $\system{A}'$, $\system{B}'$, $\system{A}''$, $\system{B}''$, and reversible transformations $\Braid_{1}$, $\Braid_{2}$, $\Braid_{3}$, $\Braid_{4}$ such that
	\begin{equation}
		\myQcircuitSupMat{
			&\s{A}\qw&\multigate{1}{\Braid}&\s{C}\qw&\qw&
			\\
			&\s{B}\qw&\ghost{\Braid}&\s{D}\qw&\qw&
		} =
		\myQcircuitSupMat{
			&\s{A}\qw&\multigate{1}{\Braid_{3}}&\qw&\s{\prim{A}}\qw&\qw&\multigate{1}{\Braid_{4}}&\s{C}\qw&\qw&
			\\
			&\pureghost{}&\pureghost{\Braid_{3}}&\s{\secondE{A}}\qw&\braidingSym&\s{\prim{B}}\qw&\ghost{\Braid_{4}}\qw&
			\\
			&\s{B}\qw&\multigate{1}{\Braid_{1}}&\s{\prim{B}}\qw&\braidingGhost&\s{\secondE{A}}\qw&\multigate{1}{\Braid_{2}}&\s{D}\qw&\qw&
			\\
			&\pureghost{}&\pureghost{\Braid_{1}}&\qw&\s{\secondE{B}}\qw&\qw&\ghost{\Braid_{2}}&
		}, 
		\tag{\ref{eqt:OPT:permutation:symm:formula}}
	\end{equation}
	where \system{A}, \system{B} are generic systems of the theory and \system{C}, \system{D} are systems such that \system{CD} has the same decomposition in elementary systems as \system{AB}. In general, any of \system{A}, \system{B}, \system{C}, \system{D} can be the trivial system, and the same holds also for \system{\prim{A}}, \system{\secondE{A}}, \system{\prim{B}}, \system{\secondE{B}}.
\end{reptheorem}

\begin{proof}
	Let us start by considering the ordered decomposition of \system{AB} in its elementary subsystems:
	\begin{equation*}
		\{\system{A}_{1},\dots,\system{A}_{n},\system{B}_{1},\dots,\system{B}_{m}\}.
	\end{equation*} 
	The action of $\Braid$ is to permute its input systems (\autoref{lem:opt:permutations:characterisation}):
	\begin{equation*}
		\{\system{A}_{1},\dots,\system{A}_{n},\system{B}_{1},\dots,\system{B}_{m}\}  
	\end{equation*}
	\begin{equation*}
		\downarrow \Braid
	\end{equation*}
	\begin{equation*}
		\{\sigma\left( \system{A}_{1} \right), \dots, \sigma\left( \system{A}_{n} \right),\sigma\left( \system{B}_{1} \right),\dots,\sigma\left( \system{B}_{m} \right)\} =  \{\system{C}_{1},\dots,\system{C}_{l},\system{D}_{1},\dots,\system{D}_{k}\}
	\end{equation*}
	where, thanks the hypothesis of uniqueness of the decomposition, it holds that $\sigma\left( \system{A}_{1} \right) = C_{1}$ and so on. If we now define $N = \left\{ 1, \dots, n \right\}, M = \left\{ 1, \dots, m \right\}, L = \left\{ 1, \dots, l \right\}, K = \left\{ 1, \dots, k \right\}$, the most general transformation that can happen due to the action of $\Braid$ is that
	\begin{equation*}
		\begin{aligned}
			&\left\{ \system{A}_{i}\right\} _{i \in \prim{N}} \transfArrow{\Braid} \left\{ \system{C}_{i}\right\} _{i \in \prim{L}},\\
			&\left\{ \system{A}_{j}\right\} _{j \in \secondE{N}} \transfArrow{\Braid} \left\{ \system{D}_{j}\right\}_{j \in \prim{K}},
		\end{aligned}
	\end{equation*}
	where $N = \prim{N} \bigcup \secondE{N}$, $|\prim{N}| = |\prim{L}|$, $|\secondE{N}| = |\prim{K}|$, and analogously for \system{B},
	\begin{equation*}
		\begin{aligned}
			&\left\{ \system{B}_{i}\right\} _{i \in \prim{M}} \transfArrow{\Braid} \left\{ \system{C}_{i}\right\} _{i \in \secondE{L}},\\
			&\left\{ \system{B}_{j}\right\} _{j \in \secondE{M}} \transfArrow{\Braid} \left\{ \system{D}_{j}\right\}_{j \in \secondE{K}},
		\end{aligned}
	\end{equation*}
	where $M = \prim{M} \bigcup \secondE{M}$, $|\prim{M}| = |\secondE{L}|$, $|\secondE{M}| = |\secondE{K}|$, and $L = \prim{L} \bigcup \secondE{L}$, $K = \prim{K} \bigcup \secondE{K}$. With $|S|$ we denote the cardinality of the set $S$.
	
	Now we want to show that this permutation can always be achieved thorough a transformation with the same form as that of \eqref{eqt:OPT:permutation:symm:formula}.	We begin by observing that in the case of system \system{A} one can always find a permutation that reorganizes the systems in such a way that the ones that are mapped into states of \system{C} are on the top and the ones that are mapped into \system{D} are on the bottom,
	\begin{equation*}
		\myQcircuitSupMat{
			&\s{A}\qw&\gate{\Braid_{3}}&\qw&\qw&\sEnsembleDouble{A}{i}{N'}{j}{N''}\qw&\qw&\qw&\splitter&\qw&\qw&\sEnsemble{A}{i}{N'}\qw&\qw&\qw&\qw&
			\\
			&\pureghost{}&\pureghost{\Braid_{3}}&\pureghost{}&\pureghost{}&\pureghost{}&\pureghost{}&\pureghost{}&\splitterGhost&\qw&\qw&\sEnsemble{A}{j}{N''}\qw&\qw&\qw&\qw&
		},
	\end{equation*}
	where the ordering within $N'$ and $N''$ is not important. The same holds for \system{B},
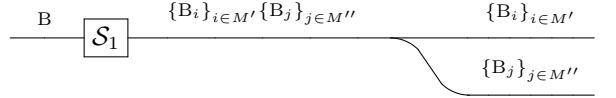
	\begin{equation*}
		\myQcircuitSupMat{
			&\s{B}\qw&\gate{\Braid_{1}}&\qw&\qw&\sEnsembleDouble{B}{i}{M'}{j}{M''}\qw&\qw&\qw&\splitter&\qw&\qw&\sEnsemble{B}{i}{M'}\qw&\qw&\qw&\qw&
			\\
			&\pureghost{}&\pureghost{\Braid_{1}}&\pureghost{}&\pureghost{}&\pureghost{}&\pureghost{}&\pureghost{}&\splitterGhost&\qw&\qw&\sEnsemble{B}{j}{M''}\qw&\qw&\qw&\qw&
		}.
	\end{equation*}
	Now we have to exchange the subsystems of \system{A} that are mapped into \system{D}  and those of \system{B} that are mapped into \system{C}. This can be achieved by swapping $\left\{ \system{A}_{j}\right\}_{j \in \secondE{N}}$ with $\left\{ \system{B}_{i}\right\} _{i \in \prim{M}}$:
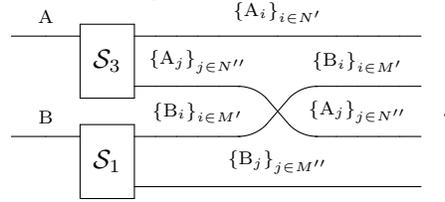
	\begin{equation*}
		\myQcircuitSupMat{
			&\s{A}\qw&\multigate{1}{\Braid_{3}}&\qw&\qw&\qw&\qw&\sEnsemble{A}{i}{N'}\qw&\qw&\qw&\qw&\qw&\qw&
			\\
			&\pureghost{}&\pureghost{\Braid_{3}}&\qw&\qw&\sEnsemble{A}{j}{N''}\qw&\qw&\braidingSym&\qw&\sEnsemble{B}{i}{M'}\qw&\qw&\qw&\qw&
			\\
			&\s{B}\qw&\multigate{1}{\Braid_{1}}&\qw&\qw&\sEnsemble{B}{i}{M'}\qw&\qw&\braidingGhost&\qw&\sEnsemble{A}{j}{N''}\qw&\qw&\qw&\qw&
			\\
			&\pureghost{}&\pureghost{\Braid_{1}}&\qw&\qw&\qw&\qw&\sEnsemble{B}{j}{M''}\qw&\qw&\qw&\qw&\qw&\qw&
		}.
	\end{equation*}
	Finally, $\Braid^{2}$ and $\Braid^{4}$ correctly reorder the subsystems to obtain \system{C} and \system{D}:
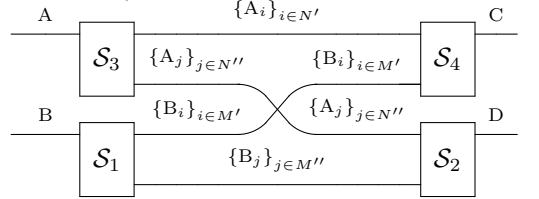
	\begin{equation*}
		\myQcircuitSupMat{
			&\s{A}\qw&\multigate{1}{\Braid_{3}}&\qw&\qw&\qw&\qw&\sEnsemble{A}{i}{N'}\qw&\qw&\qw&\qw&\qw&\multigate{1}{\Braid_{4}}&\s{C}\qw&\qw&
			\\
			&\pureghost{}&\pureghost{\Braid_{3}}&\qw&\qw&\sEnsemble{A}{j}{N''}\qw&\qw&\braidingSym&\qw&\sEnsemble{B}{i}{M'}\qw&\qw&\qw&\ghost{\Braid_{4}}\qw&
			\\
			&\s{B}\qw&\multigate{1}{\Braid_{1}}&\qw&\qw&\sEnsemble{B}{i}{M'}\qw&\qw&\braidingGhost&\qw&\sEnsemble{A}{j}{N''}\qw&\qw&\qw&\multigate{1}{\Braid_{2}}&\s{D}\qw&\qw&
			\\
			&\pureghost{}&\pureghost{\Braid_{1}}&\qw&\qw&\qw&\qw&\sEnsemble{B}{j}{M''}\qw&\qw&\qw&\qw&\qw&\ghost{\Braid_{2}}&
		}.
	\end{equation*}
	
	Therefore we have shown that, for any permutation $\Braid$, it is always possible to find a decomposition as in \eqref{eqt:OPT:permutation:symm:formula} achieving the same permutation of elementary systems. From the fact that permutations can be completely characterised by how they permute their input systems, the equality between the two transformations follows (\autoref{lem:opt:permutations:characterisation}).
\end{proof}

\section{Characterization of Minimal Operational Probabilistic Theories' instruments and transformations [Proofs of \autoref{lem:OPT:minimal:generalInstr} and \autoref{thm:OPT:minimal:symmetric:generalInstr}]}	
\label{app:mopt:instr:characterisation}

\begin{reptheorem}{lem:OPT:minimal:generalInstr}
	In every \ac{MOPT} any instrument $\eventTest{T}{x}{X} \in \Instr{A}{B}$ obtained as parallel and sequential composition of the elements of \eqref{eqt:OPT:minimal:def:test} is of the form:
	\begin{equation}
		\myQcircuit{
			&\prepareC{\preparationEventTest{\rho}{y}{Y}}&\s{C}\qw&\multigate{1}{\Braid}&\s{D}\qw&\measureD{\observationEventTest{a}{z}{Z}}&
			\\
			&\s{A}\qw&\qw&\ghost{\Braid}&\qw&\s{B}\qw&\qw&
		}, 
		\tag{\ref{eqt:OPT:minimal:instr:general}}
	\end{equation}
	where $\Braid \in \RevTransf{AC}{DB}$ is a suitable braid transformation, $\preparationEventTest{\rho}{y}{Y} \in \Prep{C}$, $\observationEventTest{a}{z}{Z} \in \Obs{A}$, the outcome space $\outcomeSpace{X} = \cartesianProduct{\outcomeSpace Y}{\outcomeSpace Z}$, and \system{A}, \system{B}, \system{C}, $\system{D} \in \Sys{\OPTMath}$ may also be equal to the trivial system~\cite{erbaMeasurementIncompatibilityStrictly2024}.
\end{reptheorem}

\begin{proof}
	The proof of \autoref{lem:OPT:minimal:generalInstr} is in a certain way constructive. Let $\eventTest{T}{x}{X} \in \Instr{A}{B}$ be an instrument that satisfies the hypotheses of the theorem. If this instrument does not contain within its decomposition either a preparation-instrument, or an observation-instrument, then it must be that $\eventTest{T}{x}{X} = \Braid$, i.e., \eqref{eqt:OPT:minimal:transf:general} with $\system{C} = \system{D} = \trivialSystem$. Suppose then that $n$ preparation- and observation-instruments are involved in the decomposition, and let us focus on an observation-instrument---the procedure remains the same if one chooses to start with a preparation-instrument. In the case in which an observation-instrument $\left\{ \observationEventNoDown{a}_{\outcome{z}_{1}} \right\}_{\outcome{z}_{1} \in \outcomeSpace{Z}_{1}} \in \Obs{D_{1}}$ is present, it is possible to isolate it and rewrite the instrument in the following way:
	\begin{equation*}
		\myQcircuitSupMat{
			&\s{A}\qw&\gate{\eventTest{T}{x}{X}}&\s{B}\qw&\qw&
		} =
		\myQcircuitSupMat{
			&\s{A}\qw&\multigate{2}{\eventTest{G}{y}{Y}}&\qw&\s{D_{2}}\qw&\qw&\multigate{2}{\eventTest{H}{z}{Z}}&\s{B}\qw&\qw&
			\\
			&\pureghost{}&\pureghost{\eventTest{G}{y}{Y}}&\s{D_{1}}\qw&\measureD{\left\{ \observationEventNoDown{a}_{\outcome{z}_{1}} \right\}_{\outcome{z}_{1} \in \outcomeSpace{Z}_{1}}}&\pureghost{}&\pureghost{\eventTest{H}{z}{Z}}&\pureghost{}&
			\\
			&\pureghost{}&\pureghost{\eventTest{G}{y}{Y}}&\qw&\s{D_{3}}\qw&\qw&\ghost{\eventTest{H}{z}{Z}}&\pureghost{}
		},
	\end{equation*}
	where $\eventTest{G}{y}{Y}$ and $\eventTest{H}{z}{Z}$ are such that $\eventTest{T}{x}{X} = \sequentialComp{\eventTest{H}{z}{Z}}{\sequentialComp{\left( \parallelComp{\identityTest{D_{2}}}{\parallelComp{\left\{ \observationEventNoDown{a}_{\outcome{z}_{1}} \right\}_{\outcome{z}_{1} \in \outcomeSpace{Z}_{1}}}{\identityTest{D_{3}}}} \right)}{\eventTest{G}{y}{Y}}}$ and $\system{D_{1}}, \system{D_{2}}, \system{D_{3}} \in \Sys{\OPTMath}$ are suitable systems. Using the reversibility of the braiding, it is possible to write
	\allowDisplayBreaks{
		\begin{equation*}
			\begin{aligned}
				& \myQcircuitSupMat{
					&\s{A}\qw&\multigate{2}{\eventTest{G}{y}{Y}}&\s{D_{2}}\qw&\braiding&\s{D_{1}}\qw&\braidingInv&\s{D_{2}}\qw&\qw&\qw&\multigate{2}{\eventTest{H}{z}{Z}}&\s{B}\qw&\qw&
					\\
					&\pureghost{}&\pureghost{\eventTest{G}{y}{Y}}&\s{D_{1}}\qw&\braidingGhost&\s{D_{2}}\qw&\braidingGhost&\s{D_{1}}\qw&\measureD{\left\{ \observationEventNoDown{a}_{\outcome{z}_{1}} \right\}_{\outcome{z}_{1} \in \outcomeSpace{Z}_{1}}}&\pureghost{}&\pureghost{\eventTest{H}{z}{Z}}&\pureghost{}&
					\\
					&\pureghost{}&\pureghost{\eventTest{G}{y}{Y}}&\qw&\qw&\qw&\s{D_{3}}\qw&\qw&\qw&\qw&\ghost{\eventTest{H}{z}{Z}}&\pureghost{}
				}\\[10pt]
				&  = \myQcircuitSupMat{
					&\s{A}\qw&\multigate{2}{\eventTest{G}{y}{Y}}&\s{D_{2}}\qw&\braiding&\s{D_{1}}\qw&\measureD{\left\{ \observationEventNoDown{a}_{\outcome{z}_{1}} \right\}_{\outcome{z}_{1} \in \outcomeSpace{Z}_{1}}}&\braidingInvId&\s{D_{1}}\qw&\multigate{2}{\eventTest{H}{z}{Z}}&\s{B}\qw&\qw&
					\\
					&\pureghost{}&\pureghost{\eventTest{G}{y}{Y}}&\s{D_{1}}\qw&\braidingGhost&\qw&\s{D_{2}}\qw&\braidingGhost&\pureghost{}&\pureghost{\eventTest{H}{z}{Z}}&\pureghost{}&
					\\
					&\pureghost{}&\pureghost{\eventTest{G}{y}{Y}}&\qw&\qw&\qw&\s{D_{3}}\qw&\qw&\qw&\ghost{\eventTest{H}{z}{Z}}&\pureghost{}
				}\\[10pt]
				& = \myQcircuitSupMat{
					&\s{A}\qw&\multigate{2}{ \left\{ \eventNoDown{G}'_{y} \right\}_{\outcomeIncluded{y}{Y}} }&\s{D_{1}}\qw&\measureD{\left\{ \observationEventNoDown{a}_{\outcome{z}_{1}} \right\}_{\outcome{z}_{1} \in \outcomeSpace{Z}_{1}}}&
					\\
					&\pureghost{}&\pureghost{\left\{ \eventNoDown{G}'_{y} \right\}_{\outcomeIncluded{y}{Y}}}&\s{D_{2}}\qw&\qw&\multigate{1}{\eventTest{H}{z}{Z}}&\s{B}\qw&\qw&
					\\
					&\pureghost{}&\pureghost{\left\{ \eventNoDown{G}'_{y} \right\}_{\outcomeIncluded{y}{Y}}}&\s{D_{3}}\qw&\qw&\ghost{\eventTest{H}{z}{Z}}&
				}\\[10pt]
				& = \myQcircuitSupMat{
					&\s{A}\qw&\multigate{1}{\left\{ \eventNoDown{T}^{1}_{x'} \right\}_{\outcomeIncluded{x'}{X'}}}&\s{D_{1}}\qw&\measureD{\left\{ \observationEventNoDown{a}_{\outcome{z}_{1}} \right\}_{\outcome{z}_{1} \in \outcomeSpace{Z}_{1}}}&
					\\
					&\pureghost{}&\pureghost{ \left\{ \eventNoDown{T}^{1}_{x'} \right\}_{\outcomeIncluded{x'}{X'}} }&\s{B}\qw&\qw&\qw&\qw&
				}.
			\end{aligned}
		\end{equation*}
	}	
	Now it is sufficient to iterate the procedure on $\left\{ \eventNoDown{T}^{1}_{x'_{1}} \right\}_{\outcomeIncluded{x'_{1}}{X'_{1}}}$ until, after $n$ steps, one obtains a transformation $\left\{ \eventNoDown{T}^{n}_{x'_{n}} \right\}_{\outcome{x}'_{n} \in \outcomeSpace{X}'_{n}} = \Braid$, which is the searched result.
	\begin{equation*}
		\myQcircuitSupMat{
			&\prepareC{\preparationEventTest{\rho}{y_{1}}{Y_{1}}}&\s{C_1}\qw&\multigate{3}{\left\{ \eventNoDown{T}^{n}_{x'_{n}} \right\}_{\outcome{x}'_{n} \in \outcomeSpace{X}'_{n}}}&\s{D_1}\qw&\measureD{\left\{ \observationEventNoDown{a}_{\outcome{z}_{1}} \right\}_{\outcome{z}_{1} \in \outcomeSpace{Z}_{1}}}&
			\\
			&\vdots&\vdots&\pureghost{\left\{ \eventNoDown{T}^{n}_{x'_{n}} \right\}_{\outcome{x}'_{n} \in \outcomeSpace{X}'_{n}}}&\vdots&\vdots&
			\\
			&\prepareC{\left\{\preparationEventNoDown{\rho}_{\outcome{y}_{l}}\right\}_{\outcome{y}_{l} \in \outcomeSpace{Y}_{l}} }&\sIndexDown{C}{l}\qw&\ghost{\left\{ \eventNoDown{T}^{n}_{x'_{n}} \right\}_{\outcome{x}'_{n} \in \outcomeSpace{X}'_{n}}}&\sIndexDown{D}{m}\qw&\measureD{\left\{ \observationEventNoDown{a}_{\outcome{x}_{m}} \right\}_{\outcome{x}_{m} \in \outcomeSpace{X}_{m}}}&
			\\
			&\s{A}\qw&\qw&\ghost{\left\{ \eventNoDown{T}^{n}_{x'_{n}} \right\}_{\outcome{x}'_{n} \in \outcomeSpace{X}'_{n}}}&\qw&\s{B}\qw&\qw&
		} =
		\myQcircuit{
			&\prepareC{\preparationEventTest{\rho}{y}{Y}}&\s{C}\qw&\multigate{1}{\Braid}&\s{D}\qw&\measureD{\observationEventTest{a}{z}{Z}}&
			\\
			&\s{A}\qw&\qw&\ghost{\Braid}&\qw&\s{B}\qw&\qw&
		}.
	\end{equation*} 
\end{proof}

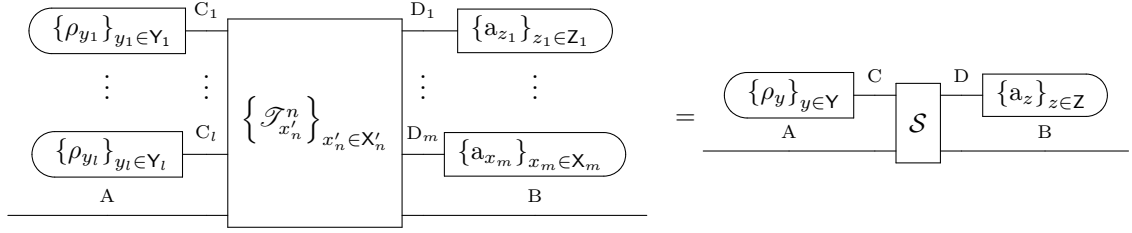
\begin{reptheorem}{thm:OPT:minimal:symmetric:generalInstr}
	In every symmetric \ac{MOPT} any instrument $\eventTest{T}{x}{X} \in \Instr{A}{B}$ obtained as parallel and sequential composition of the elements of \eqref{eqt:OPT:minimal:def:test} is of the form:
	\begin{equation}
		\myQcircuit{
			&\pureghost{}&\multiprepareC{1}{\preparationEventTest{\rho}{y}{Y}}&\qw&\s{C}\qw&\qw&\multimeasureD{1}{\observationEventTest{a}{z}{Z}}&
			\\			&\pureghost{}&\pureghost{\preparationEventTest{\rho}{y}{Y}}&\s{\prim{B}}\qw&\braidingSym&\s{\prim{A}}\qw&\ghost{\observationEventTest{a}{z}{Z}}&
			\\
			&\s{A}\qw&\multigate{1}{\Braid_{1}}&\s{\prim{A}}\qw&\braidingGhost&\s{\prim{B}}\qw&\multigate{1}{\Braid_{2}}&\s{B}\qw&\qw&
			\\
			&\pureghost{}&\pureghost{\Braid_{1}}&\qw&\s{E}\qw&\qw&\ghost{\Braid_{2}}&
		},
		\tag{\ref{eqt:OPT:minimal:symmetric:instr:generic}}
	\end{equation}
	where $\Braid_{1}, \Braid_{2} \in \RevTransfA{\OPTMath}$ are suitable permutations, $\preparationEventTest{\rho}{y}{Ys} \in \Prep{C \prim{B}}$, $\observationEventTest{a}{z}{Z} \in \Obs{C \prim{A}}$, the outcome space $\outcomeSpace{X} = \cartesianProduct{Y}{Z}$, and \system{A}, \system{B}, $\system{\prim{A}}$, $\system{\prim{B}}$, $\system{C}$, $\system{E} \in \Sys{\OPTMath}$ may also be equal to the trivial system~\cite{erbaMeasurementIncompatibilityStrictly2024}.
\end{reptheorem}

\begin{proof}	
	One has just to apply \autoref{thm:OPT:symmetric:permutations:generalForm} to \eqref{eqt:OPT:minimal:instr:general} obtaining
	\begin{equation*}
		\myQcircuit{
			&\s{A}\qw&\gate{\eventTest{T}{x}{X}}&\s{B}\qw&\qw&
		} =
		\myQcircuit{
			&\prepareC{\left\{ \preparationEventNoDown{\rho}'_{\outcome{y}} \right\}_{\outcomeIncluded{y}{Y}}}&\s{\prim{C}}\qw&\multigate{1}{\Braid_{3}}&\qw&\s{C}\qw&\qw&\multigate{1}{\Braid_{4}}&\s{\prim{D}}\qw&\measureD{\left\{ \observationEventNoDown{a}'_{\outcome{z}} \right\}_{\outcomeIncluded{z}{Z}}}&
			\\
			&\pureghost{}&\pureghost{}&\pureghost{\Braid_{3}}&\s{\prim{B}}\qw&\braidingSym&\s{\prim{A}}\qw&\ghost{\Braid_{4}}&
			\\
			&\pureghost{}&\s{A}\qw&\multigate{1}{\Braid_{1}}&\s{\prim{A}}\qw&\braidingGhost&\s{\prim{B}}\qw&\multigate{1}{\Braid_{2}}&\s{B}\qw&\qw&
			\\
			&\pureghost{}&\pureghost{}&\pureghost{\Braid_{1}}&\qw&\s{E}\qw&\qw&\ghost{\Braid_{2}}&
		}.
	\end{equation*}

	Now absorbing $\Braid_{3}, \Braid_{4}$ into $\left\{ \preparationEventNoDown{\rho}'_{\outcome{y}} \right\}_{\outcomeIncluded{y}{Y}}$ and $\left\{ \observationEventNoDown{a}'_{\outcome{z}} \right\}_{\outcomeIncluded{z}{Z}}$ respectively, the proof is concluded:
	\begin{equation*}
		\myQcircuit{
			&\s{A}\qw&\gate{\eventTest{T}{x}{X}}&\s{B}\qw&\qw&
		} =
		\myQcircuit{
			&\pureghost{}&\multiprepareC{1}{\preparationEventTest{\rho}{y}{Y}}&\qw&\s{C}\qw&\qw&\multimeasureD{1}{\observationEventTest{a}{z}{Z}}&
			\\
			&\pureghost{}&\pureghost{\preparationEventTest{\rho}{y}{Y}}&\s{\prim{B}}\qw&\braidingSym&\s{\prim{A}}\qw&\ghost{\observationEventTest{a}{z}{Z}}&
			\\
			&\s{A}\qw&\multigate{1}{\Braid_{1}}&\s{\prim{A}}\qw&\braidingGhost&\s{\prim{B}}\qw&\multigate{1}{\Braid_{2}}&\s{B}\qw&\qw&
			\\
			&\pureghost{}&\pureghost{\Braid_{1}}&\qw&\s{E}\qw&\qw&\ghost{\Braid_{2}}&
		}.
	\end{equation*}
\end{proof}

\subsection{Characterization of the limits [Proofs of \autoref{thm:OPT:minimal:transf:stabilization} and \autoref{thm:OPT:minimal:transf:goodDeterministic}]}

\begin{reptheorem}{thm:OPT:minimal:transf:stabilization}
	In a symmetric \ac{MOPT} whenever one considers a Cauchy sequence of generic transformations obtained as parallel and sequential composition of the elements in \eqref{eqt:OPT:minimal:def:event},
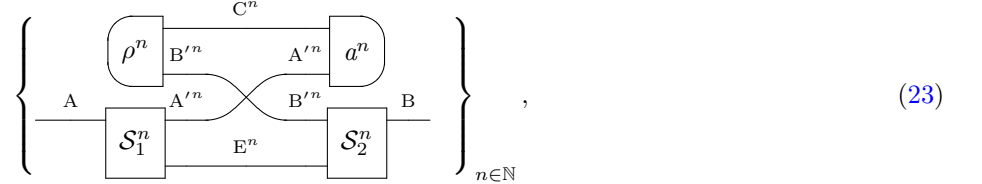
	\begin{equation}
		\left\{
		\myQcircuit{
			&\pureghost{}&\multiprepareC{1}{\preparationEventNoDownSequence{\rho}{n}}&\qw&\sSequence{C}{n}\qw&\qw&\multimeasureD{1}{\observationEventNoDownSequence{a}{n}}&
			\\
			&\pureghost{}&\pureghost{\preparationEventNoDownSequence{\rho}{n}}&\sSequencePrime{B}{n}\qw&\braidingSym&\sSequencePrime{A}{n}\qw&\ghost{\observationEventNoDownSequence{a}{n}}&
			\\								&\s{A}\qw&\multigate{1}{\Braid_{1}^{n}}&\sSequencePrime{A}{n}\qw&\braidingGhost&\sSequencePrime{B}{n}\qw&\multigate{1}{\Braid_{2}^{n}}&\s{B}\qw&\qw&
			\\
			&\pureghost{}&\pureghost{\Braid_{1}^{n}}&\qw&\sSequence{E}{n}\qw&\qw&\ghost{\Braid_{2}^{n}}&
		}
		\right\}_{n \in \mathbb{N}},
		\tag{\ref{eqt:sequenceJelly}}
	\end{equation}	
	there always exists a subsequence where the systems $\systemSequence{E}{n}$, $\systemSequence{A'}{n}$, $\systemSequence{B'}{n}$ and the permutations $\Braid_{1}^{n}$, $\Braid_{2}^{n}$ are fixed:
	\begin{equation*}
		\left\{
		\myQcircuit{
			&\pureghost{}&\multiprepareC{1}{\preparationEventNoDownSequence{\rho}{n}}&\qw&\sSequence{C}{n}\qw&\qw&\multimeasureD{1}{\observationEventNoDownSequence{a}{n}}&
			\\
			&\pureghost{}&\pureghost{\preparationEventNoDownSequence{\rho}{n}}&\s{B'}\qw&\braidingSym&\s{A'}\qw&\ghost{\observationEventNoDownSequence{a}{n}}&
			\\
			&\s{A}\qw&\multigate{1}{\Braid_{1}}&\s{A'}\qw&\braidingGhost&\s{B'}\qw&\multigate{1}{\Braid_{2}}&\s{B}\qw&\qw&
			\\
			&\pureghost{}&\pureghost{\Braid_{1}}&\qw&\s{E}\qw&\qw&\ghost{\Braid_{2}}&
		}
		\right\}_{n \in \mathbb{N}}.
	\end{equation*}	
\end{reptheorem}

\begin{proof}
	Let us start by considering a Cauchy sequence such as the one in \eqref{eqt:sequenceJelly}. The proof can be divided in two steps
	\begin{enumerate}[I]
		\item Let us consider the two decompositions in elementary systems of \system{A} and \system{B}, which we know to be unique (\autoref{def:opt:system:uniqueDec}), and composed at most of a finite number of elementary systems. We remind now that permutations are completely characterised by how they permute the input wires, and since there is only a finite number of ways of permuting a finite number of elements, there must exist at least a pair of permutations $\Braid^{1}$ and $\Braid^{2}$ that appear infinitely many times together in the decomposition of the elements of the sequence~\eqref{eqt:sequenceJelly}. We can now concentrate on the subsequence with this pair of permutations
		\begin{equation}
			\left\{
			\myQcircuit{
				&\pureghost{}&\multiprepareC{1}{\preparationEventNoDown{\rho}_{n}}&\qw&\sSequence{C}{n}\qw&\qw&\multimeasureD{1}{\observationEventNoDown{a}_{n}}&
				\\
				&\pureghost{}&\pureghost{\preparationEventNoDown{\rho}_{n}}&\sSequencePrime{B}{n}\qw&\braidingSym&\sSequencePrime{A}{n}\qw&\ghost{\observationEventNoDown{a}_{n}}&
				\\									&\s{A}\qw&\multigate{1}{\Braid_{1}}&\sSequencePrime{A}{n}\qw&\braidingGhost&\sSequencePrime{B}{n}\qw&\multigate{1}{\Braid_{2}}&\s{B}\qw&\qw&
				\\
				&\pureghost{}&\pureghost{\Braid_{1}}&\qw&\sSequence{E}{n}\qw&\qw&\ghost{\Braid_{2}}&
			}
			\right\}_{n \in \mathbb{N}},
			\tag{\ref{eqt:sequenceJelly}}
		\end{equation}	
		being \eqref{eqt:sequenceJelly} Cauchy also its subsequences will be Cauchy and they will have the same limit.
		
		\item We now focus our attention on the systems $\system{E}_{n}$. Due to the fact that in the previous point we have fixed $\Braid^{1}$ and the fact that the decomposition in elementary systems is unique, the systems contained within the composite system $\system{A}_{n}'\system{E_{n}}$ will not change. Therefore, the only thing that can change at the variation of $n$ is how they are grouped. 
		For example if $\system{A}_{n}' = \system{S}_{1}$ and $\system{E}_{n} = \left(\system{S}_{2} \system{S}_{3} \system{S}_{4}\right)$, for a different value $\prim{n} \neq n$, it must be  $\system{A}_{n'}' = \system{S}_{1}\system{S}_{2}$ and $\system{E}_{\prim{n}} = \system{S}_{3} \system{S}_{4}$, or $\system{A}_{n'}' = \left(\system{S_{1}}\system{S_{2}}\system{S_{3}}\right)$ and $\system{E_{\prim{n}}} = \system{S}_{4}$, or any other possible regrouping (also the original one) in which the order of the $\system{S}_{i}$ does not change.
		Given that  $\system{A}_{n}'\system{E}_{n}$ can be composed only of a finite number of systems, and analogously for $\system{B}_{n}'\system{E}_{n}$, it is always possible to find at least a system \system{E} that appears infinitely many times in the considered sequence. By fixing \system{E}, then also the systems $\system{A}'$ and $\system{B}'$ are automatically fixed. This is the searched result.
	\end{enumerate}
\end{proof}

\begin{reptheorem}{thm:OPT:minimal:transf:goodDeterministic}
	In a causal symmetric \ac{MOPT} the limits of Cauchy sequences of deterministic transformations are still of the form~\cite{erbaMeasurementIncompatibilityStrictly2024}
	\begin{equation}
		\minimalDeterministicCausalDestroyReprep{A}{B}{\prim{A}}{\prim{B}}{E}{\Braid_{1}}{\Braid_{2}}{\rho}.
		\tag{\ref{eqt:OPT:minimal:transf:lem:causalDeterm}}
	\end{equation}
\end{reptheorem}
\begin{proof}
	Let us start by considering a Cauchy sequence of deterministic transformations from \system{A} to \system{B}, which by \eqref{eqt:OPT:minimal:transf:lem:causalDeterm} we know to be of the form
	\begin{equation}
		\label{proof:OPT:minimal:transf:thm:goodDeterministic:1}
		\left\{
		\minimalDeterministicCausalDestroyReprepSequencePrime{A}{B}{A}{B}{E}{\Braid_{1}^{n}}{\Braid_{2}^{n}}{\rho_{n}}{n}
		\right\}_{n \in \mathbb{N}}.
	\end{equation}
	\autoref{thm:OPT:minimal:transf:stabilization} guarantees the existence of a subsequence where the only elements that can vary thorough the sequence are the states $\preparationEventNoDownSequenceAll{\rho}{n} \subset \StN{\prim{B}}$:
	\begin{equation*}
		\left\{ \minimalDeterministicCausalDestroyReprep{A}{B}{\prim{A}}{\prim{B}}{E}{\Braid_{1}}{\Braid_{2}}{\rho_{n}} \right\}_{n \in \mathbb{N}}.
	\end{equation*}	
	The elements of this subsequence satisfy the following inequality relations $\forall n,m \in \mathbb{N}$:
	\allowDisplayBreaks{
		\begin{align*}
			& \left\lVert \quad \minimalDeterministicCausalDestroyReprep{A}{B}{\prim{A}}{\prim{B}}{E}{\Braid_{1}}{\Braid_{2}}{\preparationEventNoDownSequence{\rho}{n}}  \right. \\[10pt] & \hspace{1cm} - \; \left. \minimalDeterministicCausalDestroyReprep{A}{B}{\prim{A}}{\prim{B}}{E}{\Braid_{1}}{\Braid_{2}}{\preparationEventNoDownSequence{\rho}{m}}
			\right\rVert_{op}  \\[10pt] & = 
			\normOp{ \quad
				\myQcircuit{
					&\s{\prim{A}}\qw&\measureD{\observationUniqueDeterministic}&\pureghost{}&\prepareC{\preparationEventNoDownSequence{\rho}{n}}&\s{\prim{B}}\qw&\qw&
					\\
					&\qw&\qw&\s{E}\qw&\qw&\qw&\qw&
				} - \quad
				\myQcircuit{
					&\s{\prim{A}}\qw&\measureD{\observationUniqueDeterministic}&\pureghost{}&\prepareC{\preparationEventNoDownSequence{\rho}{m}}&\s{\prim{B}}\qw&\qw&
					\\
					&\qw&\qw&\s{E}\qw&\qw&\qw&\qw&
				}
			}  \\[10pt] & =
			\normOp{ \quad
				\myQcircuit{
					&\s{\prim{A}}\qw&\measureD{\observationUniqueDeterministic}&\prepareC{\preparationEventNoDownSequence{\rho}{n}}&\s{\prim{B}}\qw&\qw&
				} - \quad
				\myQcircuit{
					&\s{\prim{A}}\qw&\measureD{\observationUniqueDeterministic}&\prepareC{\preparationEventNoDownSequence{\rho}{m}}&\s{\prim{B}}\qw&\qw&
				}
			}  \\[10pt] & =
			\normOp{
				\myQcircuit{
					&\prepareC{\preparationEventNoDownSequence{\rho}{n}}&\s{\prim{B}}\qw&\qw&
				} - \quad
				\myQcircuit{
					&\prepareC{\preparationEventNoDownSequence{\rho}{m}}&\s{\prim{B}}\qw&\qw&
				}	
			},
		\end{align*}
	}
\noindent where in the penultimate step \autoref{lem:opt:norm:operational:monotonicity} was used, while the equality in the last step can be proved using \autoref{lem:opt:norm:operational:invariaceAncilas}. This implies that the sequence of deterministic states of this particular subsequence of \eqref{proof:OPT:minimal:transf:thm:goodDeterministic:1} is Cauchy. Therefore, we can conclude that the subsequence considered in this point, and consequently \eqref{proof:OPT:minimal:transf:thm:goodDeterministic:1}, converges to 
	\begin{equation*}
		\minimalDeterministicCausalDestroyReprep{A}{B}{\prim{A}}{\prim{B}}{E}{\Braid_{1}}{\Braid_{2}}{\rho},
	\end{equation*}
	where $\eventNoDown{\rho} = \lim_{n \to \infty} \preparationEventNoDownSequence{\rho}{n}$. 
	With this we conclude our proof, since we found the desired result.
\end{proof}

\section{Iteration procedure for transformations in the Proof of \autoref{thm:MSOPT:symmetric:idAtomicity}}
\label{app:proof:MSOPT:atomicity}
Here we provide the detailed analysis of the structure of single transformations composing the instruments that appear in the sequence \eqref{eqt:proof:MSOPT:atomicity:3} and prove that all of them are proportional to the identity if their coarse-graining is the identity itself~\eqref{eqt:proof:MSOPT:atomicity:4}, which means that the identity is an atomic map. Given that we are now interested in studying transformations, in the following discussion we will assume a fixed outcome of the conditional instrument, that is $\left( \outcome{\boldsymbol{x}'}, \outcome{x_{k-1}}, \outcome{x_{k}} \right) \in \cartesianProduct{\outcomeSpace{\boldsymbol{X}}'}{\cartesianProduct{\outcomeSpace{X}_{k-1}}{\outcomeSpace{X}_{k}}}$. We highlighted in the list of outcomes those of the last two transformations in the sequential composition since those are be the ones we will focus on at first. These transformations have generally the form
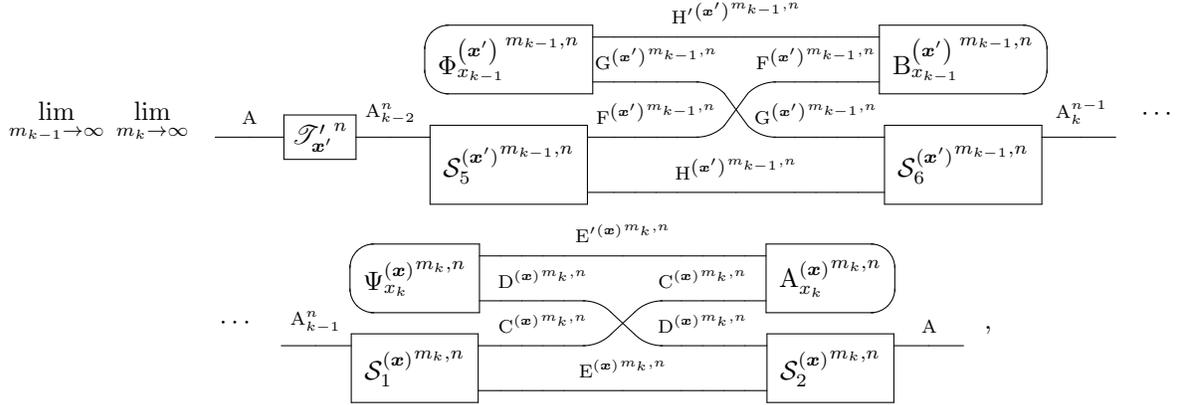
\begin{align*}
	\lim_{m_{k-1} \to \infty} \; \lim_{m_{k} \to \infty} \quad
	&\myQcircuit{
		&\pureghost{}&\pureghost{}&\pureghost{}&\multiprepareC{1}{{\preparationEventNoDown{\Phi}_{\outcome{x_{k-1}}}^{\left(\outcome{\boldsymbol{x}'}\right)}}^{m_{k-1},n}}&\qw&\qw&\qw&\qw&\sSequenceConditioned{H'}{\boldsymbol{x}'}{m_{k-1},n}\qw&\qw&\qw&\qw&\qw&\multimeasureD{1}{{\observationEventNoDown{B}_{\outcome{x_{k-1}}}^{\left(\outcome{\boldsymbol{x}'}\right)}}^{m_{k-1},n}}&
		\\
		&\pureghost{}&\pureghost{}&\pureghost{}&\pureghost{{\preparationEventNoDown{\Phi}_{\outcome{x_{k-1}}}^{\left(\outcome{\boldsymbol{x}'}\right)}}^{m_{k-1},n}}&\qw&\qw&\sSequenceConditioned{G}{\boldsymbol{x}'}{m_{k-1},n}\qw&\qw&\braidingSym&\qw&\sSequenceConditioned{F}{\boldsymbol{x}'}{m_{k-1},n}\qw&\qw&\qw&\ghost{{\observationEventNoDown{B}_{\outcome{x_{k-1}}}^{\left(\outcome{\boldsymbol{x}'}\right)}}^{m_{k-1},n}}&
		\\
		&\s{A}\qw&\gate{{\event{T}{\boldsymbol{x}'}'}^{n}}&\sSequenceIndex{A}{k-2}{n}\qw&\multigate{1}{{\permutation_{5}^{(\outcome{\boldsymbol{x}'})}}^{m_{k-1},n}}&\qw&\qw&\sSequenceConditioned{F}{\boldsymbol{x}'}{m_{k-1},n}\qw&\qw&\braidingGhost&\qw&\sSequenceConditioned{G}{\boldsymbol{x}'}{m_{k-1},n}\qw&\qw&\qw&\multigate{1}{{\permutation_{6}^{(\outcome{\boldsymbol{x}'})}}^{m_{k-1},n}}&\sSequenceIndex{A}{k}{n-1}\qw&\qw&
		\\
		&\pureghost{}&\pureghost{}&\pureghost{}&\pureghost{{\permutation_{5}^{(\outcome{\boldsymbol{x}'})}}^{m_{k-1},n}}&\qw&\qw&\qw&\qw&\sSequenceConditioned{H}{\boldsymbol{x}'}{m_{k-1},n}\qw&\qw&\qw&\qw&\qw&\ghost{{\permutation_{6}^{(\outcome{\boldsymbol{x}'})}}^{m_{k-1},n}}&\pureghost{}&
	} \cdots \\[10pt] & \cdots \quad
	\myQcircuit{
		&\pureghost{}&\multiprepareC{1}{{\preparationEventNoDown{\Psi}_{\outcome{x_{k}}}^{\left(\outcome{\boldsymbol{x}}\right)}}^{m_{k},n}}&\qw&\qw&\qw&\qw&\sSequenceConditioned{E'}{\boldsymbol{x}}{m_{k},n}\qw&\qw&\qw&\qw&\qw&\multimeasureD{1}{{\observationEventNoDown{A}_{\outcome{x_{k}}}^{\left(\outcome{\boldsymbol{x}}\right)}}^{m_{k},n}}&
		\\
		&\pureghost{}&\pureghost{{\preparationEventNoDown{\Psi}_{\outcome{x_{k}}}^{\left(\outcome{\boldsymbol{x}}\right)}}^{m_{k},n}}&\qw&\qw&\sSequenceConditioned{D}{\boldsymbol{x}}{m_{k},n}\qw&\qw&\braidingSym&\qw&\sSequenceConditioned{C}{\boldsymbol{x}}{m_{k},n}\qw&\qw&\qw&\ghost{{\observationEventNoDown{A}_{\outcome{x_{k}}}^{\left(\outcome{\boldsymbol{x}}\right)}}^{m_{k},n}}&
		\\
		&\sSequenceIndex{A}{k-1}{n}\qw&\multigate{1}{{\permutation_{1}^{(\outcome{\boldsymbol{x}})}}^{m_{k},n}}&\qw&\qw&\sSequenceConditioned{C}{\boldsymbol{x}}{m_{k},n}\qw&\qw&\braidingGhost&\qw&\sSequenceConditioned{D}{\boldsymbol{x}}{m_{k},n}\qw&\qw&\qw&\multigate{1}{{\permutation_{2}^{(\outcome{\boldsymbol{x}})}}^{m_{k},n}}&\s{A}\qw&\qw&
		\\
		&\pureghost{}&\pureghost{{\permutation_{1}^{(\outcome{\boldsymbol{x}})}}^{m_{k},n}}&\qw&\qw&\qw&\qw&\sSequenceConditioned{E}{\boldsymbol{x}}{m_{k},n}\qw&\qw&\qw&\qw&\qw&\ghost{{\permutation_{2}^{(\outcome{\boldsymbol{x}})}}^{m_{k},n}}&\pureghost{}&
	},
\end{align*}
where the limits with respect to the variable $m_{k-1}$ and $m_{k}$ are due to the fact that the transformations that compose the instruments of \eqref{eqt:proof:MSOPT:atomicity:3} are those of \acp{MOPT}. Therefore, they may also be limits of sequences of transformations of the form \eqref{eqt:OPT:minimal:symmetric:transf:generic}.        
Given~\autoref{thm:OPT:minimal:transf:stabilization}, one can already remove the dependence of many elements from the indexes $m_{k-1}$ and $m_{k}$:
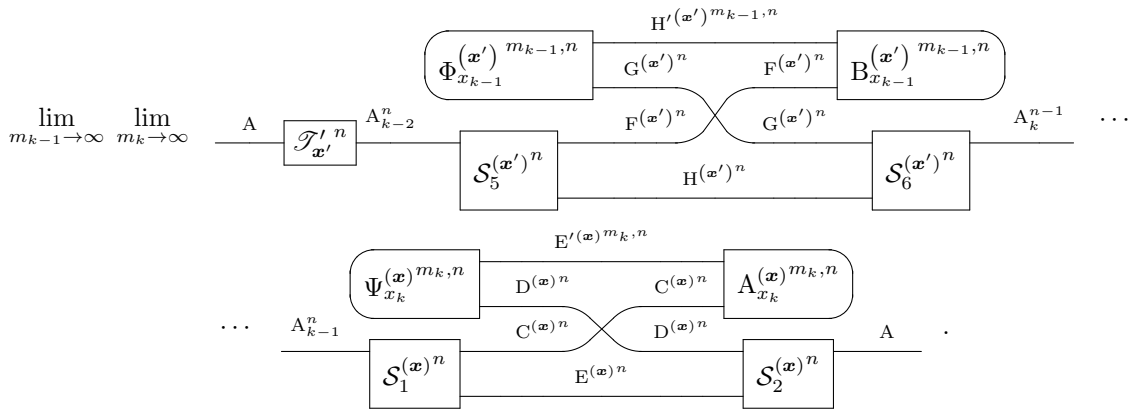
\begin{align*}
	\lim_{m_{k-1} \to \infty} \; \lim_{m_{k} \to \infty} \quad
	&\myQcircuit{
		&\pureghost{}&\pureghost{}&\pureghost{}&\multiprepareC{1}{{\preparationEventNoDown{\Phi}_{\outcome{x_{k-1}}}^{\left(\outcome{\boldsymbol{x}'}\right)}}^{m_{k-1},n}}&\qw&\qw&\qw&\sSequenceConditioned{H'}{\boldsymbol{x}'}{m_{k-1},n}\qw&\qw&\qw&\qw&\multimeasureD{1}{{\observationEventNoDown{B}_{\outcome{x_{k-1}}}^{\left(\outcome{\boldsymbol{x}'}\right)}}^{m_{k-1},n}}&
		\\
		&\pureghost{}&\pureghost{}&\pureghost{}&\pureghost{{\preparationEventNoDown{\Phi}_{\outcome{x_{k-1}}}^{\left(\outcome{\boldsymbol{x}'}\right)}}^{m_{k-1},n}}&\qw&\qw&\sSequenceConditioned{G}{\boldsymbol{x}'}{n}\qw&\braidingSym&\qw&\sSequenceConditioned{F}{\boldsymbol{x}'}{n}\qw&\qw&\ghost{{\observationEventNoDown{B}_{\outcome{x_{k-1}}}^{\left(\outcome{\boldsymbol{x}'}\right)}}^{m_{k-1},n}}&
		\\
		&\s{A}\qw&\gate{{\event{T}{\boldsymbol{x}'}'}^{n}}&\sSequenceIndex{A}{k-2}{n}\qw&\multigate{1}{{\permutation_{5}^{(\outcome{\boldsymbol{x}'})}}^{n}}&\qw&\qw&\sSequenceConditioned{F}{\boldsymbol{x}'}{n}\qw&\braidingGhost&\qw&\sSequenceConditioned{G}{\boldsymbol{x}'}{n}\qw&\qw&\multigate{1}{{\permutation_{6}^{(\outcome{\boldsymbol{x}'})}}^{n}}&\sSequenceIndex{A}{k}{n-1}\qw&\qw&
		\\
		&\pureghost{}&\pureghost{}&\pureghost{}&\pureghost{{\permutation_{5}^{(\outcome{\boldsymbol{x}'})}}^{n}}&\qw&\qw&\qw&\sSequenceConditioned{H}{\boldsymbol{x}'}{n}\qw&\qw&\qw&\qw&\ghost{{\permutation_{6}^{(\outcome{\boldsymbol{x}'})}}^{n}}&\pureghost{}&
	} \cdots \\[10pt] & \cdots \quad
	\myQcircuit{
		&\pureghost{}&\multiprepareC{1}{{\preparationEventNoDown{\Psi}_{\outcome{x_{k}}}^{\left(\outcome{\boldsymbol{x}}\right)}}^{m_{k},n}}&\qw&\qw&\qw&\sSequenceConditioned{E'}{\boldsymbol{x}}{m_{k},n}\qw&\qw&\qw&\qw&\multimeasureD{1}{{\observationEventNoDown{A}_{\outcome{x_{k}}}^{\left(\outcome{\boldsymbol{x}}\right)}}^{m_{k},n}}&
		\\
		&\pureghost{}&\pureghost{{\preparationEventNoDown{\Psi}_{\outcome{x_{k}}}^{\left(\outcome{\boldsymbol{x}}\right)}}^{m_{k},n}}&\qw&\qw&\sSequenceConditioned{D}{\boldsymbol{x}}{n}\qw&\braidingSym&\qw&\sSequenceConditioned{C}{\boldsymbol{x}}{n}\qw&\qw&\ghost{{\observationEventNoDown{A}_{\outcome{x_{k}}}^{\left(\outcome{\boldsymbol{x}}\right)}}^{m_{k},n}}&
		\\
		&\sSequenceIndex{A}{k-1}{n}\qw&\multigate{1}{{\permutation_{1}^{(\outcome{\boldsymbol{x}})}}^{n}}&\qw&\qw&\sSequenceConditioned{C}{\boldsymbol{x}}{n}\qw&\braidingGhost&\qw&\sSequenceConditioned{D}{\boldsymbol{x}}{n}\qw&\qw&\multigate{1}{{\permutation_{2}^{(\outcome{\boldsymbol{x}})}}^{n}}&\s{A}\qw&\qw&
		\\
		&\pureghost{}&\pureghost{{\permutation_{1}^{(\outcome{\boldsymbol{x}})}}^{n}}&\qw&\qw&\qw&\sSequenceConditioned{E}{\boldsymbol{x}}{n}\qw&\qw&\qw&\qw&\ghost{{\permutation_{2}^{(\outcome{\boldsymbol{x}})}}^{n}}&\pureghost{}&
	}.
\end{align*}
Following the reasoning made for \eqref{eqt:proof:MSOPT:atomicity:6}, one can the proceed to remove the dependence from the index $n$ on the last transformation in our composition
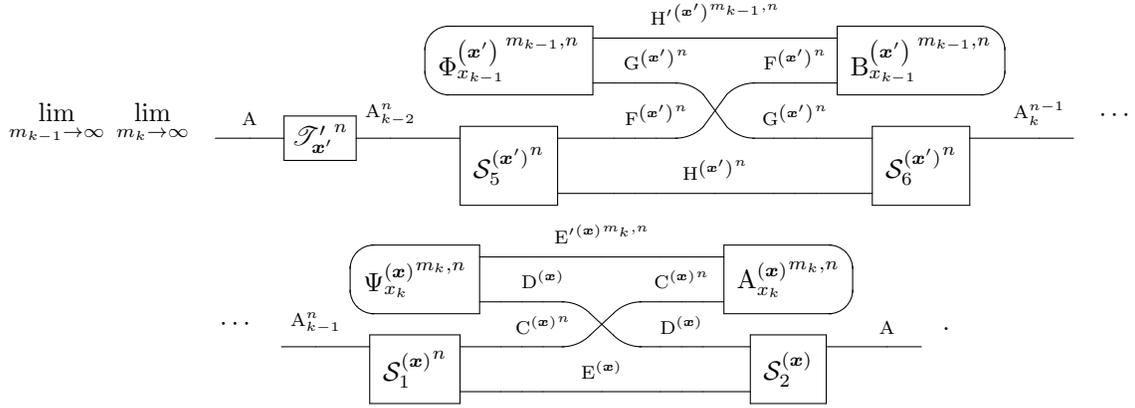
\begin{align*}
	\lim_{m_{k-1} \to \infty} \; \lim_{m_{k} \to \infty} \quad
	&\myQcircuit{
		&\pureghost{}&\pureghost{}&\pureghost{}&\multiprepareC{1}{{\preparationEventNoDown{\Phi}_{\outcome{x_{k-1}}}^{\left(\outcome{\boldsymbol{x}'}\right)}}^{m_{k-1},n}}&\qw&\qw&\qw&\sSequenceConditioned{H'}{\boldsymbol{x}'}{m_{k-1},n}\qw&\qw&\qw&\qw&\multimeasureD{1}{{\observationEventNoDown{B}_{\outcome{x_{k-1}}}^{\left(\outcome{\boldsymbol{x}'}\right)}}^{m_{k-1},n}}&
		\\
		&\pureghost{}&\pureghost{}&\pureghost{}&\pureghost{{\preparationEventNoDown{\Phi}_{\outcome{x_{k-1}}}^{\left(\outcome{\boldsymbol{x}'}\right)}}^{m_{k-1},n}}&\qw&\qw&\sSequenceConditioned{G}{\boldsymbol{x}'}{n}\qw&\braidingSym&\qw&\sSequenceConditioned{F}{\boldsymbol{x}'}{n}\qw&\qw&\ghost{{\observationEventNoDown{B}_{\outcome{x_{k-1}}}^{\left(\outcome{\boldsymbol{x}'}\right)}}^{m_{k-1},n}}&
		\\
		&\s{A}\qw&\gate{{\event{T}{\boldsymbol{x}'}'}^{n}}&\sSequenceIndex{A}{k-2}{n}\qw&\multigate{1}{{\permutation_{5}^{(\outcome{\boldsymbol{x}'})}}^{n}}&\qw&\qw&\sSequenceConditioned{F}{\boldsymbol{x}'}{n}\qw&\braidingGhost&\qw&\sSequenceConditioned{G}{\boldsymbol{x}'}{n}\qw&\qw&\multigate{1}{{\permutation_{6}^{(\outcome{\boldsymbol{x}'})}}^{n}}&\sSequenceIndex{A}{k}{n-1}\qw&\qw&
		\\
		&\pureghost{}&\pureghost{}&\pureghost{}&\pureghost{{\permutation_{5}^{(\outcome{\boldsymbol{x}'})}}^{n}}&\qw&\qw&\qw&\sSequenceConditioned{H}{\boldsymbol{x}'}{n}\qw&\qw&\qw&\qw&\ghost{{\permutation_{6}^{(\outcome{\boldsymbol{x}'})}}^{n}}&\pureghost{}&
	} \cdots \\[10pt] & \cdots \quad
	\myQcircuit{
		&\pureghost{}&\multiprepareC{1}{{\preparationEventNoDown{\Psi}_{\outcome{x_{k}}}^{\left(\outcome{\boldsymbol{x}}\right)}}^{m_{k},n}}&\qw&\qw&\qw&\sSequenceConditioned{E'}{\boldsymbol{x}}{m_{k},n}\qw&\qw&\qw&\qw&\multimeasureD{1}{{\observationEventNoDown{A}_{\outcome{x_{k}}}^{\left(\outcome{\boldsymbol{x}}\right)}}^{m_{k},n}}&
		\\
		&\pureghost{}&\pureghost{{\preparationEventNoDown{\Psi}_{\outcome{x_{k}}}^{\left(\outcome{\boldsymbol{x}}\right)}}^{m_{k},n}}&\qw&\qw&\sConditioned{D}{\boldsymbol{x}}\qw&\braidingSym&\qw&\sSequenceConditioned{C}{\boldsymbol{x}}{n}\qw&\qw&\ghost{{\observationEventNoDown{A}_{\outcome{x_{k}}}^{\left(\outcome{\boldsymbol{x}}\right)}}^{m_{k},n}}&
		\\
		&\sSequenceIndex{A}{k-1}{n}\qw&\multigate{1}{{\permutation_{1}^{(\outcome{\boldsymbol{x}})}}^{n}}&\qw&\qw&\sSequenceConditioned{C}{\boldsymbol{x}}{n}\qw&\braidingGhost&\qw&\sConditioned{D}{\boldsymbol{x}}\qw&\qw&\multigate{1}{{\permutation_{2}^{(\outcome{\boldsymbol{x}})}}}&\s{A}\qw&\qw&
		\\
		&\pureghost{}&\pureghost{{\permutation_{1}^{(\outcome{\boldsymbol{x}})}}^{n}}&\qw&\qw&\qw&\sConditioned{E}{\boldsymbol{x}}\qw&\qw&\qw&\qw&\ghost{{\permutation_{2}^{(\outcome{\boldsymbol{x}})}}}&\pureghost{}&
	}.
\end{align*}
But, we have prove that $\systemConditioned{D}{\boldsymbol{x}} = \trivialSystem$. Therefore, the transformation becomes
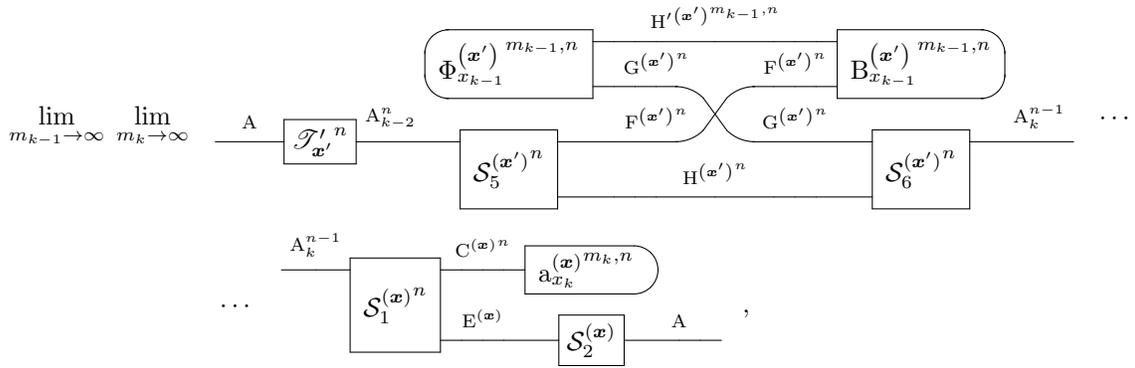
\begin{align*}
	\lim_{m_{k-1} \to \infty} \; \lim_{m_{k} \to \infty} \quad
	&\myQcircuit{
		&\pureghost{}&\pureghost{}&\pureghost{}&\multiprepareC{1}{{\preparationEventNoDown{\Phi}_{\outcome{x_{k-1}}}^{\left(\outcome{\boldsymbol{x}'}\right)}}^{m_{k-1},n}}&\qw&\qw&\qw&\sSequenceConditioned{H'}{\boldsymbol{x}'}{m_{k-1},n}\qw&\qw&\qw&\qw&\multimeasureD{1}{{\observationEventNoDown{B}_{\outcome{x_{k-1}}}^{\left(\outcome{\boldsymbol{x}'}\right)}}^{m_{k-1},n}}&
		\\
		&\pureghost{}&\pureghost{}&\pureghost{}&\pureghost{{\preparationEventNoDown{\Phi}_{\outcome{x_{k-1}}}^{\left(\outcome{\boldsymbol{x}'}\right)}}^{m_{k-1},n}}&\qw&\qw&\sSequenceConditioned{G}{\boldsymbol{x}'}{n}\qw&\braidingSym&\qw&\sSequenceConditioned{F}{\boldsymbol{x}'}{n}\qw&\qw&\ghost{{\observationEventNoDown{B}_{\outcome{x_{k-1}}}^{\left(\outcome{\boldsymbol{x}'}\right)}}^{m_{k-1},n}}&
		\\
		&\s{A}\qw&\gate{{\event{T}{\boldsymbol{x}'}'}^{n}}&\sSequenceIndex{A}{k-2}{n}\qw&\multigate{1}{{\permutation_{5}^{(\outcome{\boldsymbol{x}'})}}^{n}}&\qw&\qw&\sSequenceConditioned{F}{\boldsymbol{x}'}{n}\qw&\braidingGhost&\qw&\sSequenceConditioned{G}{\boldsymbol{x}'}{n}\qw&\qw&\multigate{1}{{\permutation_{6}^{(\outcome{\boldsymbol{x}'})}}^{n}}&\sSequenceIndex{A}{k}{n-1}\qw&\qw&
		\\
		&\pureghost{}&\pureghost{}&\pureghost{}&\pureghost{{\permutation_{5}^{(\outcome{\boldsymbol{x}'})}}^{n}}&\qw&\qw&\qw&\sSequenceConditioned{H}{\boldsymbol{x}'}{n}\qw&\qw&\qw&\qw&\ghost{{\permutation_{6}^{(\outcome{\boldsymbol{x}'})}}^{n}}&\pureghost{}&
	}\cdots \\[10pt] & \cdots \quad
	\myQcircuit{
		&\sSequenceIndex{A}{k}{n-1}\qw&\multigate{1}{{\permutation_{1}^{(\outcome{\boldsymbol{x}})}}^{n}}&\qw&\sSequenceConditioned{C}{\boldsymbol{x}}{n}\qw&\qw&\measureD{{\observationEventNoDown{a}_{\outcome{x_{k}}}^{\left(\outcome{\boldsymbol{x}}\right)}}^{m_{k},n}}&
		\\
		&\pureghost{}&\pureghost{{\permutation_{1}^{(\outcome{\boldsymbol{x}})}}^{n}}&\qw&\sConditioned{E}{\boldsymbol{x}}\qw&\qw&\gate{{\permutation_{2}^{(\outcome{\boldsymbol{x}})}}}&\s{A}\qw&\qw&\qw&
	},
\end{align*}
where
\begin{equation*}
	\myQcircuit{
		&\prepareC{{\preparationEventNoDown{\Psi}_{\outcome{x_{k}}}^{\left(\outcome{\boldsymbol{x}}\right)}}^{m_{k},n}}&\qw&\qw&\sSequenceConditioned{E'}{\boldsymbol{x}}{m_{k},n}\qw&\qw&\multimeasureD{1}{{\observationEventNoDown{A}_{\outcome{x_{k}}}^{\left(\outcome{\boldsymbol{x}}\right)}}^{m_{k},n}}&
		\\
		&\qw&\qw&\qw&\sSequenceConditioned{C}{\boldsymbol{x}}{n}\qw&\qw&\ghost{{\observationEventNoDown{A}_{\outcome{x_{k}}}^{\left(\outcome{\boldsymbol{x}}\right)}}^{m_{k},n}}&
	} = \quad\!\! \myQcircuit{
		&\qw&\sSequenceConditioned{C}{\boldsymbol{x}}{n}\qw&\qw&\measureD{{\observationEventNoDown{a}_{\outcome{x_{k}}}^{\left(\outcome{\boldsymbol{x}}\right)}}^{m_{k},n}}&
	}.
\end{equation*}
The latter equation, following the same passages that lead to \eqref{eqt:proof:MSOPT:atomicity:10}, becomes
\begin{equation*}
	\lim_{m_{k-1} \to \infty} \; \lim_{m_{k} \to \infty} \;
	\myQcircuit{
		&\pureghost{}&\pureghost{}&\pureghost{}&\multiprepareC{1}{{\preparationEventNoDown{\Phi}_{\outcome{x_{k-1}}}^{\left(\outcome{\boldsymbol{x}'}\right)}}^{m_{k-1},n}}&\qw&\qw&\qw&\sSequenceConditioned{H'}{\boldsymbol{x}'}{m_{k-1},n}\qw&\qw&\qw&\qw&\multimeasureD{1}{{\observationEventNoDown{B}_{\outcome{x_{k-1}}}^{\left(\outcome{\boldsymbol{x}'}\right)}}^{m_{k-1},n}}&
		\\
		&\pureghost{}&\pureghost{}&\pureghost{}&\pureghost{{\preparationEventNoDown{\Phi}_{\outcome{x_{k-1}}}^{\left(\outcome{\boldsymbol{x}'}\right)}}^{m_{k-1},n}}&\qw&\qw&\sSequenceConditioned{G}{\boldsymbol{x}'}{n}\qw&\braidingSym&\qw&\sSequenceConditioned{F}{\boldsymbol{x}'}{n}\qw&\qw&\ghost{{\observationEventNoDown{B}_{\outcome{x_{k-1}}}^{\left(\outcome{\boldsymbol{x}'}\right)}}^{m_{k-1},n}}&
		\\
		&\s{A}\qw&\gate{{\event{T}{\boldsymbol{x}'}'}^{n}}&\sSequenceIndex{A}{k-2}{n}\qw&\multigate{1}{{\permutation_{5}^{(\outcome{\boldsymbol{x}'})}}^{n}}&\qw&\qw&\sSequenceConditioned{F}{\boldsymbol{x}'}{n}\qw&\braidingGhost&\qw&\sSequenceConditioned{G}{\boldsymbol{x}'}{n}\qw&\qw&\multigate{1}{{\permutation_{6}^{(\outcome{\boldsymbol{x}'})}}^{n}}&\sSequenceIndex{A}{k}{n-1}\qw&\multigate{1}{\permutation_{4}^{n}}&\sSequence{C}{n}\qw&\measureD{{{\observationEventNoDown{a}_{\outcome{x_{k}}}'}^{\left(\outcome{\boldsymbol{x}}\right)}}^{m_{k},n}}&
		\\
		&\pureghost{}&\pureghost{}&\pureghost{}&\pureghost{{\permutation_{5}^{(\outcome{\boldsymbol{x}'})}}^{n}}&\qw&\qw&\qw&\sSequenceConditioned{H}{\boldsymbol{x}'}{n}\qw&\qw&\qw&\qw&\ghost{{\permutation_{6}^{(\outcome{\boldsymbol{x}'})}}^{n}}&\pureghost{}&\pureghost{\permutation_{4}^{n}}&\s{A}\qw&\qw&\qw&\qw&
	},
\end{equation*}
or equivalently
\begin{equation}
	\label{eqt:proof:MSOPT:atomicity:15}
	\lim_{m_{k-1} \to \infty} \; \lim_{m_{k} \to \infty} \quad \myQcircuit{
		&\pureghost{}&\pureghost{}&\pureghost{}&\multiprepareC{1}{{\preparationEventNoDown{\Phi}_{\outcome{x_{k-1}}}^{\left(\outcome{\boldsymbol{x}'}\right)}}^{m_{k-1},n}}&\qw&\qw&\qw&\sSequenceConditioned{H'}{\boldsymbol{x}'}{m_{k-1},n}\qw&\qw&\qw&\qw&\multimeasureD{1}{{\observationEventNoDown{B}_{\outcome{x_{k-1}}}^{\left(\outcome{\boldsymbol{x}'}\right)}}^{m_{k-1},n}}&
		\\				&\pureghost{}&\pureghost{}&\pureghost{}&\pureghost{{\preparationEventNoDown{\Phi}_{\outcome{x_{k-1}}}^{\left(\outcome{\boldsymbol{x}'}\right)}}^{m_{k-1},n}}&\qw&\qw&\sSequenceConditioned{G}{\boldsymbol{x}'}{n}\qw&\braidingSym&\qw&\sSequenceConditioned{F}{\boldsymbol{x}'}{n}\qw&\qw&\ghost{{\observationEventNoDown{B}_{\outcome{x_{k-1}}}^{\left(\outcome{\boldsymbol{x}'}\right)}}^{m_{k-1},n}}&
		\\				&\s{A}\qw&\gate{{\event{T}{\boldsymbol{x}'}'}^{n}}&\sSequenceIndex{A}{k-2}{n}\qw&\multigate{1}{{\permutation_{5}^{(\outcome{\boldsymbol{x}'})}}^{n}}&\qw&\qw&\sSequenceConditioned{F}{\boldsymbol{x}'}{n}\qw&\braidingGhost&\qw&\sSequenceConditioned{G}{\boldsymbol{x}'}{n}\qw&\qw&\multigate{1}{{\permutation_{7}^{(\outcome{\boldsymbol{x}'})}}^{n}}&\sSequence{C}{n}\qw&\measureD{{{\observationEventNoDown{a}_{\outcome{x_{k}}}'}^{\left(\outcome{\boldsymbol{x}}\right)}}^{m_{k},n}}&
		\\
		&\pureghost{}&\pureghost{}&\pureghost{}&\pureghost{{\permutation_{5}^{(\outcome{\boldsymbol{x}'})}}^{n}}&\qw&\qw&\qw&\sSequenceConditioned{H}{\boldsymbol{x}'}{n}\qw&\qw&\qw&\qw&\ghost{{\permutation_{7}^{(\outcome{\boldsymbol{x}'})}}^{n}}&\s{A}\qw&\qw&\qw&\qw&
	},
\end{equation}
where the local permutation on the system $\systemSequence{C}{n}$, as done in \eqref{eqt:proof:MSOPT:atomicity:10}, was absorbed within the effect ${\observationEventNoDown{a}_{\outcome{x_{k}}}^{\left(\outcome{\boldsymbol{x}}\right)}}^{m_{k},n}$. It is now explicit the reason why before we decided to maintain the dependence form the outcomes for the deterministic effect $\observationUniqueDeterministic\left(\outcome{\boldsymbol{x}}\right)$, acting on $\systemSequence{C}{n}$. The reason being that while the deterministic effect does not depend on it, the observation-instrument $\left\{ {{\observationEventNoDown{a}_{\outcome{x_{k}}}'}^{\left(\outcome{\boldsymbol{x}}\right)}}^{m_{k},n} \right\}_{\outcomeIncluded{x_{k}}{X_{K}}} \in \ObsOPT{\systemSequence{C}{n}}$ does.
We have now arrived at the slightly delicate point exposed before that leads to \eqref{eqt:proof:MSOPT:atomicity:14}. Here, things are a bit more tricky, since we cannot split the observation-instrument as we did previously for the deterministic effect. However, with a little of patience, also this situation can be dealt with. Let us start again by splitting the systems ${\systemConditioned{G}{\boldsymbol{x}'}}^{n} = \systemSequence{J}{n}{\systemConditioned{G'}{\boldsymbol{x}'}}^{n}$ and $\systemSequence{C}{n} = \systemSequence{J}{n} \systemSequence{C'}{n}$,\footnote{We recall that the system $\systemSequence{J}{n}$ that appears is defined up to a local permutation that can be absorbed within the observation-instrument.} and study the the permutation
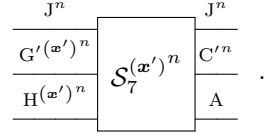
\begin{equation*}
	\myQcircuit{
		&\qw&\sSequence{J}{n}\qw&\qw&\multigate{2}{{\permutation_{7}^{(\outcome{\boldsymbol{x}'})}}^{n}}&\sSequence{J}{n}\qw&\qw&
		\\
		&\qw&\sSequenceConditioned{G'}{\boldsymbol{x}'}{n}\qw&\qw&\ghost{{\permutation_{7}^{(\outcome{\boldsymbol{x}'})}}^{n}}&\sSequence{C'}{n}\qw&\qw&
		\\
		&\qw&\sSequenceConditioned{H}{\boldsymbol{x}'}{n}\qw&\qw&\ghost{{\permutation_{7}^{(\outcome{\boldsymbol{x}'})}}^{n}}&\s{A}\qw&\qw&
	}.
\end{equation*}
Given that permutations are completely characterized by how they permute their input and output systems (\autoref{lem:opt:permutations:characterisation}), if one finds two different permutations that permute in the same way their input and output wires, then they are the same transformation. Therefore, an equivalent transformation to ${\permutation_{7}^{(\outcome{\boldsymbol{x}'})}}^{n}$ is given by 
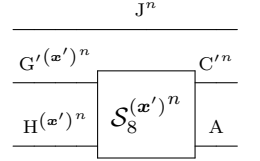
\begin{equation*}
	\myQcircuitBox{
		&\qw&\qw&\qw&\sSequence{J}{n}\qw&\qw&\qw&
		\\
		&\qw&\sSequenceConditioned{G'}{\boldsymbol{x}'}{n}\qw&\qw&\multigate{1}{{\permutation_{8}^{(\outcome{\boldsymbol{x}'})}}^{n}}&\sSequence{C'}{n}\qw&\qw&
		\\
		&\qw&\sSequenceConditioned{H}{\boldsymbol{x}'}{n}\qw&\qw&\ghost{{\permutation_{8}^{(\outcome{\boldsymbol{x}'})}}^{n}}&\s{A}\qw&\qw&
	},
\end{equation*}
which substituted in \eqref{eqt:proof:MSOPT:atomicity:15} becomes

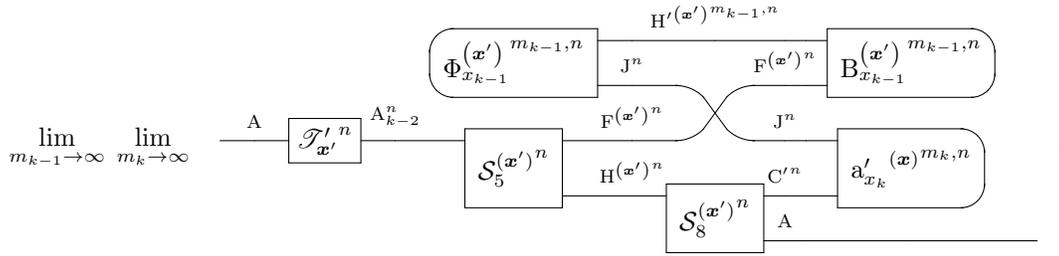
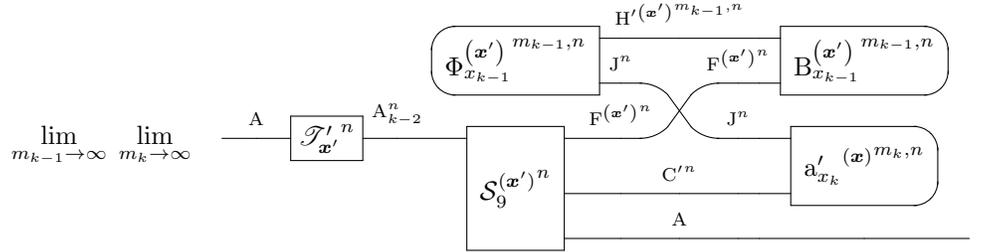
\begin{equation*}
	\lim_{m_{k-1} \to \infty} \; \lim_{m_{k} \to \infty} \quad
	\myQcircuit{
		&\pureghost{}&\pureghost{}&\pureghost{}&\multiprepareC{1}{{\preparationEventNoDown{\Phi}_{\outcome{x_{k-1}}}^{\left(\outcome{\boldsymbol{x}'}\right)}}^{m_{k-1},n}}&\qw&\sSequenceConditioned{H'}{\boldsymbol{x}'}{m_{k-1},n}\qw&\qw&\qw&\multimeasureD{1}{{\observationEventNoDown{B}_{\outcome{x_{k-1}}}^{\left(\outcome{\boldsymbol{x}'}\right)}}^{m_{k-1},n}}&
		\\
		&\pureghost{}&\pureghost{}&\pureghost{}&\pureghost{{\preparationEventNoDown{\Phi}_{\outcome{x_{k-1}}}^{\left(\outcome{\boldsymbol{x}'}\right)}}^{m_{k-1},n}}&\sSequence{J}{n}\qw&\braidingSym&\sSequenceConditioned{F}{\boldsymbol{x}'}{n}\qw&\qw&\ghost{{\observationEventNoDown{B}_{\outcome{x_{k-1}}}^{\left(\outcome{\boldsymbol{x}'}\right)}}^{m_{k-1},n}}&
		\\
		&\s{A}\qw&\gate{{\event{T}{\boldsymbol{x}'}'}^{n}}&\sSequenceIndex{A}{k-2}{n}\qw&\multigate{1}{{\permutation_{5}^{(\outcome{\boldsymbol{x}'})}}^{n}}&\sSequenceConditioned{F}{\boldsymbol{x}'}{n}\qw&\braidingGhost&\sSequence{J}{n}\qw&\qw&\multimeasureD{1}{{{\observationEventNoDown{a}_{\outcome{x_{k}}}'}^{\left(\outcome{\boldsymbol{x}}\right)}}^{m_{k},n}}&
		\\
		&\pureghost{}&\pureghost{}&\pureghost{}&\pureghost{{\permutation_{5}^{(\outcome{\boldsymbol{x}'})}}^{n}}&\sSequenceConditioned{H}{\boldsymbol{x}'}{n}\qw&\multigate{1}{{\permutation_{8}^{(\outcome{\boldsymbol{x}'})}}^{n}}&\sSequence{C'}{n}\qw&\qw&\ghost{{{\observationEventNoDown{a}_{\outcome{x_{k}}}'}^{\left(\outcome{\boldsymbol{x}}\right)}}^{m_{k},n}}&
		\\
		&\pureghost{}&\pureghost{}&\pureghost{}&\pureghost{}&\pureghost{}&\pureghost{{\permutation_{8}^{(\outcome{\boldsymbol{x}'})}}^{n}}&\s{A}\qw&\qw&\qw&\qw&\qw&
	},
\end{equation*}
or equivalently,
\begin{equation*}
	\lim_{m_{k-1} \to \infty} \; \lim_{m_{k} \to \infty} \quad
	\myQcircuit{
		&\pureghost{}&\pureghost{}&\pureghost{}&\multiprepareC{1}{{\preparationEventNoDown{\Phi}_{\outcome{x_{k-1}}}^{\left(\outcome{\boldsymbol{x}'}\right)}}^{m_{k-1},n}}&\qw&\sSequenceConditioned{H'}{\boldsymbol{x}'}{m_{k-1},n}\qw&\qw&\qw&\multimeasureD{1}{{\observationEventNoDown{B}_{\outcome{x_{k-1}}}^{\left(\outcome{\boldsymbol{x}'}\right)}}^{m_{k-1},n}}&
		\\
		&\pureghost{}&\pureghost{}&\pureghost{}&\pureghost{{\preparationEventNoDown{\Phi}_{\outcome{x_{k-1}}}^{\left(\outcome{\boldsymbol{x}'}\right)}}^{m_{k-1},n}}&\sSequence{J}{n}\qw&\braidingSym&\sSequenceConditioned{F}{\boldsymbol{x}'}{n}\qw&\qw&\ghost{{\observationEventNoDown{B}_{\outcome{x_{k-1}}}^{\left(\outcome{\boldsymbol{x}'}\right)}}^{m_{k-1},n}}&
		\\
		&\s{A}\qw&\gate{{\event{T}{\boldsymbol{x}'}'}^{n}}&\sSequenceIndex{A}{k-2}{n}\qw&\multigate{2}{{\permutation_{9}^{(\outcome{\boldsymbol{x}'})}}^{n}}&\sSequenceConditioned{F}{\boldsymbol{x}'}{n}\qw&\braidingGhost&\sSequence{J}{n}\qw&\qw&\multimeasureD{1}{{{\observationEventNoDown{a}_{\outcome{x_{k}}}'}^{\left(\outcome{\boldsymbol{x}}\right)}}^{m_{k},n}}&
		\\
		&\pureghost{}&\pureghost{}&\pureghost{}&\pureghost{{\permutation_{9}^{(\outcome{\boldsymbol{x}'})}}^{n}}&\qw&\sSequence{C'}{n}\qw&\qw&\qw&\ghost{{{\observationEventNoDown{a}_{\outcome{x_{k}}}'}^{\left(\outcome{\boldsymbol{x}}\right)}}^{m_{k},n}}&
		\\
		&\pureghost{}&\pureghost{}&\pureghost{}&\pureghost{{\permutation_{9}^{(\outcome{\boldsymbol{x}'})}}^{n}}&\qw&\s{A}\qw&\qw&\qw&\qw&\qw&
	},
\end{equation*}
where the fact that it is possible to find a subsequence in which ${\systemConditioned{G'}{\boldsymbol{x}'}}^{n} = \trivialSystem$ was also exploited.

Following the same procedure as before, one can then also remove the dependence from the set of outcomes $\outcome{\boldsymbol{x}'}$ both from the permutation and the system ${\systemConditioned{F}{\boldsymbol{x}'}}^{n}$
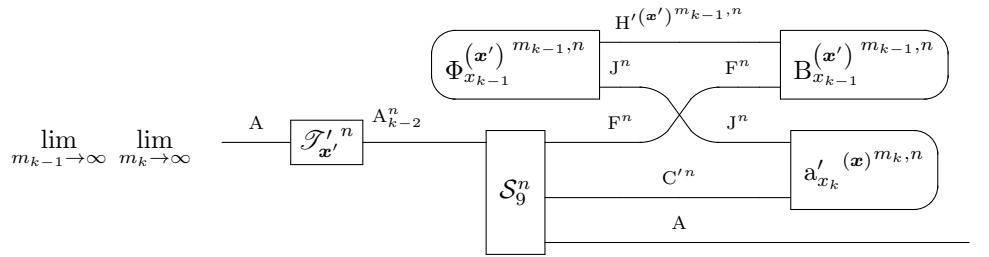
\begin{equation*}
	\lim_{m_{k-1} \to \infty} \; \lim_{m_{k} \to \infty} \quad
	\myQcircuit{
		&\pureghost{}&\pureghost{}&\pureghost{}&\multiprepareC{1}{{\preparationEventNoDown{\Phi}_{\outcome{x_{k-1}}}^{\left(\outcome{\boldsymbol{x}'}\right)}}^{m_{k-1},n}}&\qw&\sSequenceConditioned{H'}{\boldsymbol{x}'}{m_{k-1},n}\qw&\qw&\qw&\multimeasureD{1}{{\observationEventNoDown{B}_{\outcome{x_{k-1}}}^{\left(\outcome{\boldsymbol{x}'}\right)}}^{m_{k-1},n}}&
		\\
		&\pureghost{}&\pureghost{}&\pureghost{}&\pureghost{{\preparationEventNoDown{\Phi}_{\outcome{x_{k-1}}}^{\left(\outcome{\boldsymbol{x}'}\right)}}^{m_{k-1},n}}&\sSequence{J}{n}\qw&\braidingSym&\sSequence{F}{n}\qw&\qw&\ghost{{\observationEventNoDown{B}_{\outcome{x_{k-1}}}^{\left(\outcome{\boldsymbol{x}'}\right)}}^{m_{k-1},n}}&
		\\
		&\s{A}\qw&\gate{{\event{T}{\boldsymbol{x}'}'}^{n}}&\sSequenceIndex{A}{k-2}{n}\qw&\multigate{2}{\permutation_{9}^{n}}&\sSequence{F}{n}\qw&\braidingGhost&\sSequence{J}{n}\qw&\qw&\multimeasureD{1}{{{\observationEventNoDown{a}_{\outcome{x_{k}}}'}^{\left(\outcome{\boldsymbol{x}}\right)}}^{m_{k},n}}&
		\\
		&\pureghost{}&\pureghost{}&\pureghost{}&\pureghost{\permutation_{9}^{n}}&\qw&\sSequence{C'}{n}\qw&\qw&\qw&\ghost{{{\observationEventNoDown{a}_{\outcome{x_{k}}}'}^{\left(\outcome{\boldsymbol{x}}\right)}}^{m_{k},n}}&
		\\
		&\pureghost{}&\pureghost{}&\pureghost{}&\pureghost{\permutation_{9}^{n}}&\qw&\s{A}\qw&\qw&\qw&\qw&\qw&
	}.
\end{equation*}

Exploiting then the fact that

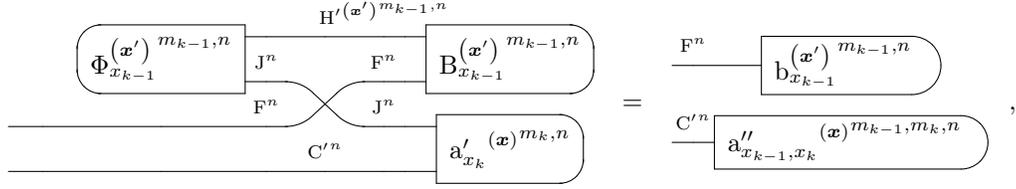
\begin{equation*}
	\myQcircuit{
		&\pureghost{}&\multiprepareC{1}{{\preparationEventNoDown{\Phi}_{\outcome{x_{k-1}}}^{\left(\outcome{\boldsymbol{x}'}\right)}}^{m_{k-1},n}}&\qw&\qw&\sSequenceConditioned{H'}{\boldsymbol{x}'}{m_{k-1},n}\qw&\qw&\multimeasureD{1}{{\observationEventNoDown{B}_{\outcome{x_{k-1}}}^{\left(\outcome{\boldsymbol{x}'}\right)}}^{m_{k-1},n}}&
		\\
		&\pureghost{}&\pureghost{{\preparationEventNoDown{\Phi}_{\outcome{x_{k-1}}}^{\left(\outcome{\boldsymbol{x}'}\right)}}^{m_{k-1},n}}&\sSequence{J}{n}\qw&\braidingSym&\sSequence{F}{n}\qw&\qw&\ghost{{\observationEventNoDown{B}_{\outcome{x_{k-1}}}^{\left(\outcome{\boldsymbol{x}'}\right)}}^{m_{k-1},n}}&
		\\
		&\qw&\qw&\sSequence{F}{n}\qw&\braidingGhost&\sSequence{J}{n}\qw&\qw&\multimeasureD{1}{{{\observationEventNoDown{a}_{\outcome{x_{k}}}'}^{\left(\outcome{\boldsymbol{x}}\right)}}^{m_{k},n}}&
		\\
		&\qw&\qw&\qw&\sSequence{C'}{n}\qw&\qw&\qw&\ghost{{{\observationEventNoDown{a}_{\outcome{x_{k}}}'}^{\left(\outcome{\boldsymbol{x}}\right)}}^{m_{k},n}}&
	} = \quad\!\!
	\myQcircuit{
		&\sSequence{F}{n}\qw&\measureD{{\observationEventNoDown{b}_{\outcome{x_{k-1}}}^{\left(\outcome{\boldsymbol{x}'}\right)}}^{m_{k-1},n}}&
		\\
		&\sSequence{C'}{n}\qw&\measureD{{{\observationEventNoDown{a}_{\outcome{x_{k-1},x_{k}}}''}^{\left(\outcome{\boldsymbol{x}}\right)}}^{m_{k-1},m_{k},n}}&
	},
\end{equation*}
for some suitable observation-instrument, one arrives to
\begin{equation}
	\label{eqt:proof:MSOPT:atomicity:16}
	\lim_{m_{k-1} \to \infty} \; \lim_{m_{k} \to \infty} \;
	\myQcircuit{
		&\s{A}\qw&\gate{{\event{T}{\boldsymbol{x}'}'}^{n}}&\sSequenceIndex{A}{k-2}{n}\qw&\multigate{2}{\permutation_{9}^{n}}&\sSequence{F}{n}\qw&\measureD{{\observationEventNoDown{b}_{\outcome{x_{k-1}}}^{\left(\outcome{\boldsymbol{x}'}\right)}}^{m_{k-1},n}}&
		\\
		&\pureghost{}&\pureghost{}&\pureghost{}&\pureghost{\permutation_{9}^{n}}&\sSequence{C'}{n}\qw&\measureD{{{\observationEventNoDown{a}_{\outcome{x_{k-1},x_{k}}}''}^{\left(\outcome{\boldsymbol{x}}\right)}}^{m_{k-1},m_{k},n}}&
		\\
		&\pureghost{}&\pureghost{}&\pureghost{}&\pureghost{\permutation_{9}^{n}}&\s{A}\qw&\qw&\qw&\qw&\qw&
	},
\end{equation}
which is still of the form \eqref{eqt:proof:MSOPT:atomicity:15}, as can be seen by expanding ${\event{T}{\boldsymbol{x}'}'}^{n}$ and highlighting the $k-2$-th conditional step. So the procedure we followed from \eqref{eqt:proof:MSOPT:atomicity:15} can then be iterated. In particular, by stopping before the last step, what one finds is of the form
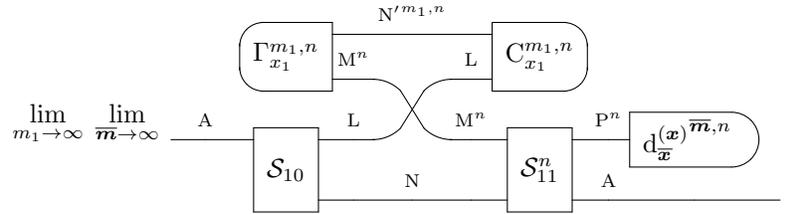
\begin{equation*}
	\lim_{m_{1} \to \infty} \; \lim_{\overline{\boldsymbol m} \to \infty} \;
	\myQcircuit{
		&\pureghost{}&\multiprepareC{1}{\preparationEvent{\Gamma}{x_{1}}^{m_{1},n}}&\qw&\sSequencePrime{N}{m_{1},n}\qw&\qw&\multimeasureD{1}{\observationEvent{C}{x_{1}}^{m_{1},n}}&
		\\
		&\pureghost{}&\pureghost{\preparationEvent{\Gamma}{x_{1}}^{m_{1},n}}&\sSequence{M}{n}\qw&\braidingSym&\s{L}\qw&\ghost{\observationEvent{C}{x_{1}}^{m_{1},n}}&
		\\
		&\s{A}\qw&\multigate{1}{\permutation_{10}}&\s{L}\qw&\braidingGhost&\sSequence{M}{n}\qw&\multigate{1}{\permutation_{11}^{n}}&\sSequence{P}{n}\qw&\measureD{{\observationEvent{d}{\overline{\boldsymbol x}}^{\left( \outcome{\boldsymbol{x}} \right)}}^{\overline{\boldsymbol m},n}}&
		\\
		&\pureghost{}&\pureghost{\permutation_{10}}&\qw&\s{N}\qw&\qw&\ghost{\permutation_{11}^{n}}&\s{A}\qw&\qw&\qw&
	},
\end{equation*}
where $\overline{\boldsymbol m} = m_{2}, \ldots, m_{k}$, $\outcome{\overline{\boldsymbol x}} = \left( \outcome{x_{2}, \ldots, x_{k}} \right)$, $\permutation_{11}^{n}$ is a suitable permutation, $\systemSequence{P}{n}$ is a suitable system, $\left\{ {\observationEvent{d}{\overline{\boldsymbol x}}^{\left( \outcome{\boldsymbol{x}} \right)}}^{\overline{\boldsymbol m},n} \right\}_{\overline{\boldsymbol x} \in \overline{\outcomeSpace{X}} = \cartesianProduct{\outcomeSpace{X}_{2}}{\cartesianProduct{\cdots}{\outcomeSpace{X}_{k}}}}$ is a suitable observation-instrument, the dependence form the index $m_{1}$ was already removed from wherever possible and we suppose to consider the subsequence where the leftmost permutation with its systems is fixed. We highlight that the fact that it is possible to remove the dependence form the index $n$ of the sequence from the system \system{L} is extremely important to conclude our proof. Since the effect ${\observationEvent{d}{\overline{\boldsymbol x}}^{\left( \outcome{\boldsymbol{x}} \right)}}^{\boldsymbol{m},n}$ must be completely ``absorbed'' by the state $\preparationEvent{\Gamma}{x_{1}}^{m_{1},n}$ to avoid  a transformation which cannot decompose the identity, it must be $\systemSequence{M}{n} = \systemSequence{M'}{n}\systemSequence{P}{n}$ with $\systemSequence{M'}{n} = \trivialSystem$, due to the observations made following \eqref{eqt:proof:MSOPT:atomicity:7}. Therefore, exploiting again the fact that permutations are completely characterized by how they permute the systems wires and grouping whatever is possible, the transformation becomes
\begin{equation}
	\label{eqt:proof:MSOPT:atomicity:17}
	\lim_{m_{1} \to \infty} \; \lim_{\overline{\boldsymbol m} \to \infty} \;
	\myQcircuit{
		&\pureghost{}&\multiprepareC{1}{\preparationEvent{\Gamma}{x_{1}}^{m_{1},n}}&\qw&\sSequencePrime{N}{m_{1},n}\qw&\qw&\multimeasureD{1}{\observationEvent{C}{x_{1}}^{m_{1},n}}&
		\\
		&\pureghost{}&\pureghost{\preparationEvent{\Gamma}{x_{1}}^{m_{1},n}}&\sSequence{P}{n}\qw&\braidingSym&\s{L}\qw&\ghost{\observationEvent{C}{x_{1}}^{m_{1},n}}&
		\\
		&\s{A}\qw&\multigate{1}{\permutation_{12}}&\s{L}\qw&\braidingGhost&\sSequence{P}{n}\qw&\measureD{{\observationEvent{d}{\overline{\boldsymbol x}}^{\left( \outcome{\boldsymbol{x}} \right)}}^{\overline{\boldsymbol m},n}}&
		\\
		&\pureghost{}&\pureghost{\permutation_{12}}&\qw&\s{A}\qw&\qw&\qw&\qw&\qw&
	}.
\end{equation}
We highlight that also $\permutation_{12}$ still does not depend on $n$, because it was obtained by composing $\permutation_{10}$ with a permutation $\permutation \in \RevTransf{N}{A}$ and neither \system{N} nor \system{A} depend form the index $n$ of the sequence.
Proceeding now to resolve the calculations related to the state and the effects the following transformation is obtained
\begin{equation*}
	\lim_{m_{1} \to \infty} \; \lim_{\overline{\boldsymbol m} \to \infty} \;
	\myQcircuit{
		&\s{A}\qw&\multigate{1}{\permutation_{12}}&\s{L}\qw&\measureD{{\observationEvent{c}{\outcome{x}_{1},\overline{\boldsymbol x}}^{\left( \outcome{\boldsymbol{x}} \right)}}^{m_{1},\overline{\boldsymbol m},n}}&
		\\
		&\pureghost{}&\pureghost{\permutation_{12}}&\s{A}\qw&\qw&\qw&
	},
\end{equation*}
which is ill defined unless $\system{L} = \trivialSystem$.

Finally, substituting this latter result into \eqref{eqt:proof:MSOPT:atomicity:17} one obtains that any transformation that composes the instrument \eqref{eqt:proof:MSOPT:atomicity:3} must be of the form 
\begin{equation}
	\label{eqt:proof:MSOPT:atomicity:18}
	\lim_{m_{1} \to \infty} \; \lim_{\overline{\boldsymbol m} \to \infty} \; 
	\myQcircuit{
		&\pureghost{}&\multiprepareC{1}{\preparationEvent{\Gamma}{x_{1}}^{m_{1},n}}&\qw&\sSequencePrime{N}{m_{1},n}\qw&\qw&\measureD{\observationEvent{C}{x_{1}}^{m_{1},n}}&
		\\
		&\pureghost{}&\pureghost{\preparationEvent{\Gamma}{x_{1}}^{m_{1},n}}&\qw&\sSequence{P}{n}\qw&\qw&\measureD{{\observationEvent{d}{\overline{\boldsymbol x}}^{\left( \outcome{\boldsymbol{x}} \right)}}^{\overline{\boldsymbol m},n}}&
		\\
		&\qw&\qw&\qw&\s{A}\qw&\qw&\qw&\qw&\qw&
	}
\end{equation}
if one wants the full coarse-graining of the instrument to be equal to the identity, i.e., for \eqref{eqt:proof:MSOPT:atomicity:4} to hold.

\end{document}